\pdfoutput=1
\newif\ifdraft \draftfalse % Makes comments go away
\newif\ifauthornames \authornamesfalse
\newif \ifnips \nipstrue % turn this to true to convert to NIPS format

\nipstrue % Set to true to make into NIPS format
\authornamestrue % Set to false to hide author names in supplementary file
\draftfalse % Set to true to include author comments
\ifnips
\documentclass{article}
\PassOptionsToPackage{numbers, sort&compress}{natbib}
\usepackage[final]{nips_2017}
\usepackage[utf8]{inputenc} % allow utf-8 input
\usepackage[T1]{fontenc}    % use 8-bit T1 fonts
\usepackage{url}            % simple URL typesetting
\usepackage{amsfonts}       % blackboard math symbols
\usepackage{nicefrac}       % compact symbols for 1/2, etc.
\usepackage{microtype}      % microtypography
\else
\documentclass[11pt]{article}
\usepackage{natbib}
\fi

\usepackage[ruled,noend]{algorithm2e} % For algorithms

\SetAlFnt{\small}
\SetAlCapFnt{\small}
\SetAlCapNameFnt{\small}
\SetAlCapHSkip{0pt}
\IncMargin{-\parindent}
\usepackage{algpseudocode}

\ifnips
\else
\usepackage{fullpage}
\fi
\usepackage{multirow}
\usepackage{amsmath, amssymb, amsthm}
\usepackage{mathtools}
\usepackage{xspace}
\usepackage{verbatim}
\usepackage{color}
\usepackage{bbm}

\usepackage[dvipsnames]{xcolor}
\definecolor[named]{MyGray}{cmyk}{0,0.1,0.1,0.1}
\definecolor{DarkGreen}{rgb}{0.1,0.5,0.1}
\definecolor{DarkRed}{rgb}{0.5,0.1,0.1}
\definecolor{DarkBlue}{rgb}{0.1,0.1,0.5}
\usepackage[]{hyperref}
\hypersetup{
    unicode=false,          % non-Latin characters in Acrobat's bookmarks
    pdftoolbar=true,        % show Acrobat toolbar?
    pdfmenubar=true,        % show Acrobat menu?
    pdffitwindow=false,      % page fit to window when opened
    pdftitle={},    % title
    pdfauthor={}
    pdfsubject={},   % subject of the document
    pdfnewwindow=true,      % links in new window
    pdfkeywords={keywords}, % list of keywords
    colorlinks=true,       % false: boxed links; true: colored links
    linkcolor=DarkRed,          % color of internal links
    citecolor=DarkGreen,        % color of links to bibliography
    filecolor=DarkRed,      % color of file links
    urlcolor=DarkBlue,          % color of external links
}
\usepackage{cleveref}

\usepackage{enumitem}

\allowdisplaybreaks
\usepackage{caption}
\usepackage{subcaption}

\usepackage[]{color-edits}
\addauthor{Ryan}{blue}
\addauthor{MD}{Cerulean}
\addauthor{JWV}{red}
\addauthor{SL}{brown}

\newcommand{\rynote}[1]{\ifdraft \Ryancomment{#1}\else\ignorespaces\fi}
\newcommand{\jenn}[1]{\ifdraft\JWVcomment{#1}\else\ignorespaces\fi}

\newcommand{\mdcomment}[1]{\ifdraft\MDcomment{#1}\else\ignorespaces\fi}
\newcommand{\mdedit}[1]{\ifdraft\MDedit{#1}\else{#1}\ignorespaces\fi}

\newcommand{\ignore}[1]{}

\newcommand\R{\mathbb{R}}

\newcommand{\cA}{\mathcal{A}}
\newcommand{\cB}{\mathcal{B}}

\newcommand{\cD}{\mathcal{D}}

\newcommand{\cG}{\mathcal{G}}

\newcommand{\cI}{\mathcal{I}}

\newcommand{\cL}{\mathcal{L}}
\newcommand{\cM}{\mathcal{M}}

\renewcommand{\vec}[1]{\boldsymbol{#1}}
\newcommand{\nbuyers}{N}
\newcommand{\nsec}{K}
\newcommand{\outcome}{\omega}
\newcommand{\Outcome}{\Omega}
\newcommand{\npayoff}{\phi}
\newcommand{\payoff}{\vec{\phi}}
\newcommand{\Payoff}{\Phi}
\newcommand{\bundle}{r}
\newcommand{\bbundle}{\vec{\bundle}}
\newcommand{\barbundle}{{\bar{\bundle}}}
\newcommand{\barbbundle}{\vec{\bar{\bundle}}}
\newcommand{\barcash}{\bar{c}}
\newcommand{\barbcash}{\bar{\vec{c}}}
\newcommand{\state}{s}
\newcommand{\bstate}{\vec{\state}}
\newcommand{\liq}{b}
\newcommand{\lmsr}{T}

\newcommand{\bprob}{\vec{p}}
\newcommand{\price}{\mu}
\newcommand{\bprice}{\vec{\price}}
\newcommand{\belief}[1]{\tilde\price_{#1}}
\newcommand{\bbelief}[1]{\tilde\bprice_{#1}}
\newcommand{\bprobbelief}[1]{\tilde\bprob_{#1}}
\newcommand{\TRUE}{\textup{true}}
\newcommand{\bpricetruth}{\bprice^\TRUE}
\newcommand{\bprobtruth}{\bprob^\TRUE}

\newcommand{\aggbprice}{\bar\bprice}
\newcommand{\aggprice}[1]{\bar\price_{#1}}
\newcommand{\iterbprice}[1]{\bprice^{#1}}

\newcommand{\iterbbundle}[1]{\bbundle^{#1}}

\newcommand{\eqbprice}{\bprice^\star}
\newcommand{\eqprice}[1]{\price^\star_{#1}}
\newcommand{\eqbundle}{\bundle^\star}
\newcommand{\eqbbundle}{\vec{\bundle}^\star}
\newcommand{\eqcash}{c^\star}
\newcommand{\eqbcash}{\vec{c}^\star}
\newcommand{\btheta}{\vec{\theta}}
\newcommand{\bttheta}{\vec{\tilde\theta}}
\newcommand{\ttheta}{\tilde\theta}
\newcommand{\bdelta}{\vec{\delta}}
\newcommand{\obj}{F}
\newcommand{\Fee}{\bar{F}}
\newcommand{\belem}{\vec{e}}

\newcommand{\optbbundle}{\bbundle^\star}

\newcommand{\dualobj}{\obj^*}
\newcommand{\Span}{\mathrm{span}}
\newcommand{\Neff}{N_{\mathrm{eff}}}

\newcommand{\realtheta}{\btheta^\TRUE}

\newcommand{\ASD}{\texttt{\upshape{ASD}}\xspace}
\newcommand{\SSD}{\texttt{\upshape{SSD}}\xspace}
\newcommand{\BCD}{\ASD}
\newcommand{\SCD}{\SSD}
\newcommand{\LMSR}{\texttt{\upshape{LMSR}}\xspace}
\newcommand{\IND}{\texttt{\upshape{IND}}\xspace}
\DeclareMathOperator{\dom}{dom}
\DeclareMathOperator{\ri}{ri}
\DeclareMathOperator{\cl}{cl}

\DeclareMathOperator{\interior}{int}
\DeclareMathOperator{\aff}{aff}
\DeclareMathOperator{\conv}{conv}
\DeclareMathOperator{\range}{range}
\DeclareMathOperator{\Null}{null}
\DeclareMathOperator{\diag}{diag}
\newcommand{\wo}{\backslash}
\newcommand{\within}{\simeq}
\newcommand{\Rinf}{\R\cup\set{\infty}}
\newcommand{\trans}{^\intercal}
\newcommand{\phtrans}{^{\vphantom{\intercal}}}

\newcommand{\sigmalow}{\sigma_{\textup{low}}}%{\underline{\sigma}}
\newcommand{\sigmahigh}{\sigma_{\textup{high}}}%{\overline{\sigma}}
\renewcommand{\kappa}{\gamma}
\newcommand{\kappalow}{\kappa_{\textup{low}}}
\newcommand{\kappahigh}{\kappa_{\textup{high}}}
\DeclareMathOperator*{\myargmin}{\arg\!\min}

\DeclareMathOperator*{\argmin}{\arg\!\min}
\DeclareMathOperator*{\argmax}{\arg\!\max}
\newcommand{\bg}{\vec{g}}
\newcommand{\bh}{\vec{h}}
\newcommand{\bs}{\vec{s}}

\newcommand{\bq}{\vec{q}}

\newcommand{\one}{\mathbf{1}}
\newcommand{\zero}{\mathbf{0}}
\newcommand{\vdelta}{\vec{\delta}}
\newcommand{\bu}{\vec{u}}
\newcommand{\bv}{\vec{v}}
\newcommand{\ba}{\vec{a}}

\newcommand{\by}{\vec{y}}
\newcommand{\bz}{\vec{z}}

\newcommand{\eps}{\varepsilon}
\newcommand{\veps}{\vec{\eps}}
\newcommand{\bpriceavg}{\hat{\bprice}}
\newcommand{\cind}{\mathbb{I}}

\newcommand{\Sin}{S_{\textup{in}}}
\newcommand{\Sout}{S_{\textup{out}}}
\newcommand{\lambdain}{\lambda_{\textup{in}}}
\newcommand{\lambdaout}{\lambda_{\textup{out}}}

\newcommand\defn{\mathrm{defn}}

\newcommand{\ex}[1]{\mathbb{E}\left[#1\right]}
\newcommand{\hatex}[1]{\mathbb{\hat{E}}\left[#1\right]}
\newcommand{\var}[1]{\mathbb{V}\left(#1\right)}
\newcommand{\hatexNoBrack}{\mathbb{\hat{E}}}
\newcommand{\exNoBrack}{\mathbb{E}}
\newcommand{\varNoBrack}{\mathbb{V}}
\newcommand{\Ex}[2]{\mathbb{E}_{#1}\left[#2\right]}
\newcommand{\ExNoBrack}[1]{\mathbb{E}_{#1}}
\newcommand{\prob}[1]{\mathbb{P}\left\{#1\right\}}%{\mathrm{Pr}\left[#1\right]}
\newcommand{\card}[1]{\lvert#1\rvert}
\newcommand{\set}[1]{\{#1\}}

\newcommand{\bigSet}[1]{\bigl\{#1\bigr\}}
\newcommand{\BigSet}[1]{\Bigl\{#1\Bigr\}}
\newcommand{\braces}[1]{\{#1\}}
\newcommand{\Braces}[1]{\left\{#1\right\}}

\newcommand{\bracks}[1]{[#1]}
\newcommand{\bigBracks}[1]{\bigl[#1\bigr]}
\newcommand{\BigBracks}[1]{\Bigl[#1\Bigr]}

\newcommand{\Bracks}[1]{\left[#1\right]}
\newcommand{\parens}[1]{(#1)}
\newcommand{\Parens}[1]{\left(#1\right)}
\newcommand{\bigParens}[1]{\bigl(#1\bigr)}
\newcommand{\BigParens}[1]{\Bigl(#1\Bigr)}

\newcommand{\given}{\mathbin{\vert}}
\newcommand{\bigGiven}{\mathbin{\bigm\vert}}
\newcommand{\BigGiven}{\mathbin{\Bigm\vert}}

\newcommand{\norm}[1]{\lVert#1\rVert}
\newcommand{\Norm}[1]{\left\lVert#1\right\rVert}
\newcommand{\bigNorm}[1]{\bigl\lVert#1\bigr\rVert}
\newcommand{\BigNorm}[1]{\Bigl\lVert#1\Bigr\rVert}

\newcommand{\KL}[2]{\textup{KL}\left(#1 \Vert #2\right)}

\newcommand{\bigKL}[2]{\textup{KL}\bigParens{#1 \bigm\Vert #2}}

\newcommand{\etaKL}{\eta_\textup{KL}}

\newtheorem{theorem}{Theorem}[section]
\newtheorem{lemma}[theorem]{Lemma}

\newtheorem{conjecture}[theorem]{Conjecture}
\newtheorem{proposition}[theorem]{Proposition}
\theoremstyle{definition}
\newtheorem{definition}[theorem]{Definition}
\newtheorem{example}[theorem]{Example}

\theoremstyle{remark}

\newcommand{\Thm}[1]{Theorem~\ref{thm:#1}}
\newcommand{\Def}[1]{Definition~\ref{def:#1}}
\newcommand{\Lem}[1]{Lemma~\ref{lem:#1}}

\newcommand{\Prop}[1]{Proposition~\ref{prop:#1}}
\newcommand{\Sec}[1]{Section~\ref{sec:#1}}
\newcommand{\App}[1]{Appendix~\ref{app:#1}}
\newcommand{\Eq}[1]{Eq.~\eqref{eq:#1}}
\newcommand{\Eqs}[2]{Eqs.~\eqref{eq:#1} and~\eqref{eq:#2}}
\newcommand{\Fig}[1]{Fig.~\ref{fig:#1}}

\newcommand{\convplus}{convex$^+$\xspace}
\newcommand{\convexityplus}{convexity$^+$\xspace}
\newcommand{\Convplus}{Convex$^+$\xspace}

\newcommand{\squishlist}{
   \begin{list}{$\bullet$}
    { \setlength{\itemsep}{0pt}      \setlength{\parsep}{4pt}
      \setlength{\topsep}{4pt}       \setlength{\partopsep}{0pt}
     \setlength{\leftmargin}{2em} \setlength{\labelwidth}{1.5em}
      \setlength{\labelsep}{0.5em} } }
\newcommand{\squishend}{  \end{list}  }

\newcounter{qcounter}
\newenvironment{qenum}
 {\begin{list}{\arabic{qcounter}.}
 {\usecounter{qcounter} \setlength{\topsep}{0in} \setlength{\partopsep}{0in}
  \setlength{\parsep}{0in} \setlength{\itemsep}{\parskip}
  \setlength{\leftmargin}{0.21in} \setlength{\rightmargin}{0in}
  \setlength{\listparindent}{0.0in} \setlength{\labelwidth}{0.07in}
  \setlength{\labelsep}{0.1in} \setlength{\itemindent}{-0.04in}}}
 {\end{list}}

\title{A Decomposition of Forecast Error in\\ Prediction Markets}
\ifauthornames
\author{Miroslav Dud\'ik \\
Microsoft Research, New York, NY \\
\texttt{mdudik@microsoft.com}
\And
S\'{e}bastien Lahaie \\
 Google, New York, NY \\
\texttt{slahaie@google.com}
\And
Ryan Rogers \\
University of Pennsylvania, Philadelphia, PA \\
\texttt{rrogers386@gmail.com}
\And
Jennifer Wortman Vaughan \\
Microsoft Research, New York, NY \\
\texttt{jenn@microsoft.com}}
\else
\author[1]{Anonymous Author(s)}
\affil[1]{Unknown Institution}
\fi

\begin{document}

\maketitle

%\mdcomment{Make plots color-blind friendly. Also consider making them black-and-white-printer friendly (but I don't think it's too important). Generally, I would suggest revisiting the plot legibility.}

\begin{abstract}
We analyze sources of error in prediction market forecasts in order to
bound the difference between a security's price and the ground truth
it estimates. We consider cost-function-based prediction markets in which an
automated market maker adjusts security prices according to the
history of trade. We decompose the forecasting error into three
components: \emph{sampling error}, arising because traders only
possess noisy estimates of ground truth;
% \emph{risk-aversion effect}, arising because traders reveal beliefs
% only through self-interested trade;
\emph{market-maker bias}, resulting from the use of a particular
market maker (i.e., cost function) to facilitate trade; and
\emph{convergence error}, arising because, at any point in time,
market prices may still be in flux. Our goal is to make explicit the
tradeoffs between these error components, influenced by design decisions
such as the functional form of the cost function and the amount of
liquidity in the market. We consider a specific model in which
traders have exponential utility and exponential-family beliefs
representing noisy estimates of ground truth. In this setting,
sampling error vanishes as the number of traders grows, but there is a
tradeoff between the other two components. We provide both upper and
lower bounds on market-maker bias and convergence error, and
demonstrate via numerical simulations that these bounds are tight. Our
results yield new insights into the question of how to set the
market's liquidity parameter and into the forecasting benefits of enforcing
coherent prices across securities.
% ends ignore
% ends ignore
%%% Local Variables:
%%% mode: latex
%%% TeX-master: "main"
%%% End:

\end{abstract}

\ifnips
\else
\clearpage
\tableofcontents
\clearpage
\fi

\section{Introduction}
\label{sec:intro}

A prediction market is a marketplace in which participants can trade securities with payoffs that depend on the outcomes of future events~\cite{wolfers2004prediction}. Consider the simple setting in which we are interested in predicting the outcome of a political election: whether the incumbent or challenger will win. A prediction market might issue a security that pays out \$1 per share if the incumbent wins, and \$0 otherwise.
The market price $p$ of this security should always lie between 0 and 1, and can be construed as an event probability. If a trader believes that the likelihood of the incumbent winning is greater than $p$, she will buy shares with the expectation of making a profit.
%, pushing the security price closer to her belief. If she believes the likelihood is less, she will (short) sell.
Market prices increase when there is more interest in buying and decrease when there is more interest in selling. By this process, the market aggregates traders' information into a consensus forecast, represented by the market price.
With sufficient activity, prediction markets are competitive with alternative forecasting methods such as polls~\cite{berg2008results}, but while there is a mature literature on sources of error and bias in polls, the impact of prediction market structure on forecast accuracy is still an active area of research~\cite{rothschild2009forecasting}.

We consider prediction markets in which all trades occur through a centralized entity known as a \emph{market maker}. Under this market structure, security prices are dictated by a fixed \emph{cost function} and the current number of outstanding shares~\cite{CP07}. The basic conditions that a cost function should satisfy
%in order to bound the market maker's loss and avoid arbitrage
to correctly elicit beliefs,
while bounding the market maker's loss,
%and avoiding arbitrage
are now well-understood, chief among them being convexity~\cite{ACV13}. Nonetheless, the class of allowable cost functions remains broad, and the literature so far provides little formal guidance on the specific form of cost function to use in order to achieve good forecast accuracy, including how to set the \emph{liquidity parameter} which controls price responsiveness to trade. In practice, the impact of the liquidity parameter is difficult to quantify a priori, so implementations typically resort to calibrations based on market simulations~\citep{dudik2013combinatorial,slamka2013prediction}.
%Simulations also suggest that cost functions which automatically maintain coherence among prices of logically related securities
%have predictive advantages over independently priced securities~\cite{dudik2013combinatorial}, but there has been little prior theoretical work aimed at understanding why.
Prior work also suggests that maintaining coherence among prices of logically related securities has informational advantages~\citep{dudik2013combinatorial}, but there has been little work aimed at understanding why.
%
%\JWVedit{Additionally, while simulations suggest that \emph{combinatorial markets}, in which the prices of logically related securities are automatically kept coherent, have predictive advantages over independently priced securities~\cite{dudik2013combinatorial}, there has been little prior theoretical work aimed at understanding why.}
%\slnote{we don't end up studying combinatorial markets in this paper
%  after all, so do we really want to bring this up?}
%\mdcomment{I have toned down the combinatorial aspect, so it can be viewed as related to \LMSR vs \IND, but I'm also fine dropping.}

This paper provides a framework to quantify the impact of the choice of cost function on forecast accuracy. We introduce a decomposition of forecast error, in analogy with the bias-variance decomposition familiar from statistics or the approximation-estimation-optimization decomposition for large-scale machine learning~\citep{bousquet2008tradeoffs}. Our decomposition consists of three components. First, there is the \emph{sampling error} resulting from the fact that the market consists of a finite population of traders, each holding a noisy estimate of ground truth.
% Second, this error may be compounded by a \emph{risk-aversion effect}, due to the fact that varying risk attitudes among the traders can unduly bias market prices towards certain private signals.
Second, there is a \emph{market-maker bias} which stems from the use of a cost function to provide liquidity and induce trade.
%Intuitively, a market with high liquidity (low price responsiveness) may not adequately incorporate the traders' information into the market prices; the actual form of the cost function also impacts this bias in non-obvious ways.
Third, there is \emph{convergence error} due to the fact that the market prices may not have fully converged to their equilibrium point.
%If liquidity is too low, prices fluctuate around the equilibrium, whereas if liquidity is too high, prices are slow to incorporate the traders' information. More generally, the impact of the overall form of the cost function on convergence remains to be understood.

The central contribution of this paper is a theoretical characterization of the market-maker bias and convergence error, the two components of this decomposition that depend on market structure as defined by the form of
the cost function and level of liquidity.
%Taken together, our results quantify the tradeoffs inherent in various choices of liquidity levels and cost functions.
%
We consider a tractable model of agent behavior, originally studied by \citet{AKLS14}, in which traders have exponential utility functions and beliefs drawn from an exponential family. Under this model it is possible to characterize the market's equilibrium prices in terms of the traders' belief and risk aversion parameters, and thereby quantify the discrepancy between current market prices and ground truth.
To analyze market convergence, we consider the trader dynamics introduced by \citet{FR15},
under which trading can be viewed as randomized block-coordinate descent on
a suitable potential function.

Our analysis is \emph{local} in that the bounds depend on the market equilibrium prices.
%and therefore on trader beliefs.
This allows us to exactly identify the main asymptotic terms of error. We demonstrate via numerical experiments that these asymptotic bounds are accurate early on and therefore can be used to compare market designs.

We make the following specific contributions:
\begin{qenum}
\item We precisely define the three components of the forecasting error.
%\emph{sampling error}, \emph{market-maker bias} and \emph{convergence error}.

\item We show that the market-maker bias equals $cb\pm O(b^2)$ as $b\to 0$, where $b$ is the liquidity parameter, and $c$ is an explicit constant that depends on the cost function and trader beliefs.

\item We show that the convergence error decreases with the number of trades $t$ as $\gamma^t$ with $\gamma=1-\Theta(b)$.
% where $b$ is the liquidity parameter.
We provide explicit upper and lower bounds on $\gamma$ that depend on the cost function and trader beliefs.  In the process, we
%extend the results of Nesterov \citep{Nest12} to
prove a new local convergence bound for block-coordinate descent.
    %of the form $1-cb-O(b^2)\le \gamma\le 1-c'b+O(b^2)$ where $c$ and $c'$

\item We use our explicit formulas for bias and convergence error to compare two common cost functions: independent markets (\IND), under which security prices vary independently, and the logarithmic market scoring rule (\LMSR)~\cite{H03}, which enforces logical relationships between security prices.
    We show that at the same value of the market-maker bias, \IND requires at least half-as-many and at most twice-as-many trades
    as \LMSR to achieve the same convergence error.
\end{qenum}

We consider a specific utility model (exponential utility), but our bias and convergence analysis immediately carry over if we assume that each trader is optimizing a risk measure (rather than an exponential utility function) similar to the setup of~\citet{FR15}. Exponential utility was chosen because it was previously well studied and allowed us to focus on the analysis of the cost function and liquidity.
The role of the liquidity parameter in trading off the bias and convergence error has been informally recognized in the literature~\citep{chen2010new,H03,othman2013practical},
but our precise definition of market-maker bias and explicit formulas for the bias and convergence error are novel.
\citet{AKLS14} provide results that can be used to derive the bias for \LMSR, but not for generic cost functions,
so they do not enable comparison of biases of different costs. \citet{FR15} observe that
the convergence error can be locally bounded as $\gamma^t$, but they only provide an upper bound and do not show how
$\gamma$ is related to the liquidity or cost function. Our analysis establishes both upper and lower bounds
on convergence and relates $\gamma$ explicitly to the liquidity and cost function.
This is necessary for a meaningful
comparison of cost function families.
Thus our framework provides the first meaningful way to compare the error tradeoffs inherent in different choices of cost functions and liquidity levels.

\section{Preliminaries}
\label{sec:prelims}

%We first review the cost-function-based market making framework~\citep{CP07,ACV13} and describe the basic components of our trader model.

%Here and throughout the paper,
%Vectors are denoted by small boldface italics (e.g., $\bstate$, $\bprice$), while matrices are denoted by large italics (e.g., $A$, $P$).
%For a pair of symmetric square matrices $A$ and $B$ we write $A\preceq B$ iff $\bu\trans A\bu\le\bu\trans B\bu$ for all $\bu$.
We use the notation $[N]$ to denote the set $\{1, \dotsc, N\}$.
Given a convex function $f:\R^d\to\Rinf$, its \emph{effective domain}, denoted $\dom f$, is the set of points where $f$ is finite. Whenever $\dom f$ is non-empty, the \emph{conjugate} $f^*:\R^d\to\Rinf$ is defined by
$f^*(\bv)\coloneqq\sup_{\bu\in\R^d} [\bv\trans\bu-f(\bu)]$.
We write $\norm{\cdot}$ for the Euclidean norm.
%Relevant mathematical background is introduced along the way as needed, and
A centralized mathematical reference is provided in Appendix~\ref{app:math}.\footnote{A longer version of this paper containing the appendix is available on arXiv and the authors' websites.}

\paragraph{Cost-function-based market makers}

We study cost-function-based prediction markets~\citep{ACV13}.
%in the general setting formalized by \citet{ACV13}.
Let $\Outcome$ be a finite set of mutually exclusive and exhaustive states of the world.
A market administrator, known as \emph{market maker},
% in this context,
wishes to elicit information about the likelihood of various states $\omega\in\Omega$, and to that end offers to buy and sell any number of shares of $\nsec$ \emph{securities}.
%at some current market price.
%that depends on the history of the trade.
Securities are associated with coordinates of a payoff function $\payoff: \Outcome \to \R^\nsec$, where each share of the $k${\small th} security is worth $\phi_k(\outcome)$ in the event that the true state of the world is
$\outcome \in \Outcome$.
%We use $\Payoff \in \R^{|\Outcome| \times K}$ to denote the matrix in which the $\outcome$th row is $\payoff(\outcome)$.
%
Traders arrive in the market sequentially and trade with the market maker.
%
%, who is always willing to buy or sell the $\nsec$ securities at some current market price that depends on the history of trade.
The market price is fully determined by a convex potential function $C$ called the \emph{cost function}.  In particular, if the market maker has previously sold $s_k \in \R$ shares of each security $k$ and a trader would like to purchase a bundle consisting of $\delta_k \in \R$ shares of each, the trader is charged $C(\bstate + \pmb{\delta}) - C(\bstate)$.
The \emph{instantaneous price} of security $k$ is then $\partial C(\bstate) / \partial s_k$.   Note that negative values of $\delta_k$ are allowed and correspond to the trader (short) selling security $k$.

Let $\cM\coloneqq\conv\set{\payoff(\omega):\:\omega\in\Omega}$ be the convex hull of the set of payoff vectors. It is exactly the set of expectations $\ex{\payoff(\omega)}$ across all possible probability distributions
over $\Omega$, which we call \emph{beliefs}. We refer to elements of $\cM$ as \emph{coherent prices}. \citet{ACV13} characterize the conditions that a cost function must satisfy in order to guarantee important properties such as bounded loss for the market maker and no possibility of arbitrage. To start, we assume only that $C:\R^K\to\R$ is convex and differentiable and that $\cM\subseteq\dom C^*$, which corresponds to the bounded loss property.
%In particular, they show that the price function $\nabla C$ must be convex and map to the convex hull of the image of the payoff function $\payoff(\Outcome)$.
%However, parts of our analysis make use of stronger assumptions on $C$.

\ignore{
Throughout this paper, we repeatedly return to several running examples of cost functions.
The most commonly used cost function is Hanson's logarithmic market scoring rule (LMSR) \citep{H03} which is defined for \emph{complete markets} with $K = |\Omega|$.
}

\begin{example}[Logarithmic Market Scoring Rule: LMSR~\citep{H03}]
Consider a \emph{complete market} with a single security for each outcome worth \$1 if that outcome occurs and \$0 otherwise, i.e., $\Omega=[K]$ and $\npayoff_k(\omega)=\one\braces{k=\omega}$ for all $k$.
The LMSR cost function and instantaneous security prices are given by
\begin{equation}
{\textstyle
   C(\bstate) = \log\left(\sum_{k=1}^K e^{\state_k}\right)
}
\label{eq:LMSR}
\qquad
\text{and}
\qquad
\frac{\partial C(\bstate)}{\partial \state_k} = \frac{e^{\state_k}}{\sum_{\ell=1}^K e^{\state_{\ell}}}, \ \forall k \in [K].
\end{equation}
Its conjugate is the entropy function, $C^*(\bprice)=\sum_k \mu_k\log\mu_k+\cind\set{\bprice\in\Delta_K}$, where $\Delta_K$ is the simplex
in $\R^K$
and $\cind\set{\cdot}$ is the
convex indicator, equal to zero if its argument is true and infinity if false. Thus, in this case $\cM=\Delta_K=\dom C^*$.
\end{example}
Notice that the LMSR security prices are coherent because they always sum to one.
%(i.e., they correspond to a valid probability distribution over the $K$ outcomes).
This prevents arbitrage opportunities for traders. Our second running example does not have this property.
%This cost function is also for a complete market with the same $K = |\Omega|$ securities.
%and payoff matrix $\Payoff$ as the LMSR.
%The costs it produces are equivalent to the costs that would be produced using separate binary LMSR cost functions to independently price each security.  As a result, the price of security $k$ does not change as other securities are bought and sold.

\begin{example}[Sum of Independent LMSRs: IND]
Let $\Omega=[K]$ and $\npayoff_k(\omega)=\one\braces{k=\omega}$ for all $k$.
The cost function and instantaneous security prices for the \emph{sum of independent LMSRs} are given by
\begin{equation}
{\textstyle
C(\bstate) = \sum_{k=1}^K \log \left( 1+ e^{\state_k} \right)
}
\qquad
\text{and}
\qquad
\frac{\partial C(\bstate)}{\partial \state_k}
= \frac{e^{\state_k}}{1+e^{\state_k}} , \ \forall k \in [K],
\label{eq:IND}
\end{equation}
%Its conjugate is the sum of binary entropies,
$C^*(\bprice)=\sum_k \bracks{\mu_k\log\mu_k+(1-\mu_k)\log(1-\mu_k)}+\cind\set{\bprice\in[0,1]^K}$, $\cM=\Delta_K$, and $\dom C^*=[0,1]^K$.
\end{example}

\begin{comment}
The LMSR generalizes to incomplete markets with arbitrary security payoffs. This generalized LMSR ensures prices are coherent, but not always tractable to compute~\cite{CL+08}.
\begin{example}[Generalized LMSR]
Let $\payoff$ be an arbitrary payoff function. The cost function and instantaneous prices for the \emph{generalized LMSR} are given by
\begin{equation}
C(\bstate)=\log \left(\sum_{\omega \in \Omega} e^{\payoff(\omega)\cdot\bstate}\right)
\qquad
\text{and}
\qquad
\frac{\partial C(\bstate)}{\partial \state_k} = \sum_{\omega \in \Omega} \npayoff_k(\omega)  e^{\payoff(\omega)\cdot\bstate - C(\bstate)} , \ \forall k \in [K].
\label{eq:genlmsr}
\end{equation}
This cost function is also called \emph{log partition function} and will be discussed again in \Sec{exp}.
We do not describe $C^*$ here, but only point out that, similar to LMSR, $\dom C^*=\cM$.
\end{example}
\end{comment}

\ignore{
We note that the associated price function $\nabla\lmsr$ will be \emph{coherent}, meaning that the prices will be in the marginal polytope $\cM = \left\{ \mu \in \R^\nsec : \Payoff \mu \in \Delta_{|\Outcome|}\right\}$.\rynote{This might not be the place to talk about coherent prices.}
}

When choosing a cost function, one important consideration is \emph{liquidity}, that is, how quickly prices change in response to trades. Any cost function $C$ can be viewed as a member of a parametric family of cost functions of the form
$
C_\liq(\bstate)\coloneqq\liq C(\bstate/\liq)
$
across all $\liq > 0$.  With larger values of $\liq$, larger trades are required to move market prices by some fixed amount, and the worst-case loss of the market maker is larger; with smaller values, small purchases can result in big changes to the market price.
%and market maker loss is smaller as well.

\paragraph{Basic model}

%In order to discuss the aggregation of information and the error of market predictions, it is necessary to assume a model of trader behavior and to specify how traders' information relates to some notion of ground truth.  To this end,

In our analysis of error we assume that there exists an unknown true probability distribution $\bprobtruth \in \Delta_{|\Outcome|}$ over the outcome set $\Omega$.  The true expected payoffs of the $K$ market securities are then given by the vector
$
\bpricetruth \coloneqq \Ex{\outcome \sim \bprobtruth}{\payoff(\outcome)}
%= \Payoff\trans \bprobtruth.
$.

We assume that there are $\nbuyers$ traders and that each trader $i \in [N]$ has a private belief $\bprobbelief{i}$ over outcomes.
%The expected security payoffs under trader $i$'s belief are given by the vector
%$
%\bbelief{i} \coloneqq \Ex{\outcome \sim \bprobbelief{i}}{\payoff(\outcome)}
%%=\Payoff\trans\bprobbelief{i}.
%$.
%For now we make no assumptions about the way in which $\bprobbelief{i}$ relates to the ground truth $\bprobtruth$, though we will do so later when discussing sampling error.
%
We additionally assume that each trader $i$ has a \emph{utility function} $u_i: \R \to \R$ for wealth and would like to maximize expected utility subject to her beliefs.  For now we assume that $u_i$ is differentiable and concave, meaning that each trader is risk averse, though later we focus on exponential utility.  The expected utility of trader $i$ owning a security bundle $\bbundle_i \in \R^K$ and cash $c_i$ is
$
  U_i(\bbundle_i, c_i)\coloneqq\Ex{\outcome \sim \bprobbelief{i}}{u_i\bigParens{c_i + \payoff(\outcome)\cdot \bbundle_i}}
.
$
We assume that each trader begins with zero cash. This is without loss of generality because we could incorporate any initial cash holdings into $u_i$.
%%% Local Variables:
%%% mode: latex
%%% TeX-master: "main"
%%% End:

\section{A Decomposition of Error}
\label{sec:decomposition}

% ends ignore of decomposition overview

In this section, we decompose the market's forecast error into three major components. The first is \emph{sampling error}, which arises because traders have only noisy observations of the ground truth. The second is \emph{market-maker bias}, which arises because the shape of the cost function impacts the traders' willingness to invest. Finally, \emph{convergence error} arises due to the fact that at any particular point in time the market prices may not have fully converged. To formalize our decomposition, we introduce two new notions of equilibrium.

Our first notion of equilibrium, called a \emph{market-clearing equilibrium}, does not assume the existence of a market maker, but rather assumes that traders trade only among themselves, and so no additional securities or cash are available beyond the traders' initial allocations. This equilibrium is described by security prices $\aggbprice\in\R^K$ and allocations $(\barbbundle_i, \barcash_i)$ of security bundles and cash to each trader $i$ such that, given her allocation, no trader wants to buy or sell any bundle of securities at those prices. Trader bundles and cash are summarized as $\barbbundle=(\barbbundle_i)_{i\in[N]}$ and $\barbcash=(\barcash_i)_{i\in[N]}$.

%First, a \emph{market-clearing equilibrium} is described by security prices $\aggbprice\in\R^K$ and allocations $(\barbbundle_i, \barcash_i)$ of security bundles and cash to each trader $i$ such that, given her allocation, no trader wants to %buy or sell any bundle of securities at the those prices. Trader bundles and cash are summarized as $\barbbundle=(\barbbundle_i)_{i\in[N]}$ and $\barbcash=(\barcash_i)_{i\in[N]}$.
%\jennedit{Note that this definition does not assume the existence of a market maker, but rather assumes that no additional securities or cash are available beyond the traders' initial allocations---traders trade only among themselves.}

\ignore{
\mdcomment{One of the reviewers suggested the following: \textit{It may be easier to understand the definition of market-clearing and market-maker equilibria if the orders of their definition were reversed. The market-clearing equilibrium definition comes immediately after the definition of the market, but does not require that the traders buy/sell securities through the market, and this took me some time to understand.}
I'm not sure whether we need or should rearrange things. Perhaps we can just emphasize things better (or do nothing).}
}

\begin{definition}[Market-clearing equilibrium]\label{defn:market_clearing}
  A triple $(\barbbundle, \barbcash, \aggbprice)$ is a market-clearing equilibrium if $\sum_{i=1}^\nbuyers \barbbundle_i = \zero$, $\sum_{i=1}^N \barcash_i = 0$, and for all $i \in [\nbuyers]$,
$
\zero \in \argmax_{\vdelta \in \R^K} U_i(\barbbundle_i + \vdelta,\,\barcash_i - \vdelta \cdot \aggbprice) .
$
We call $\aggbprice$ \emph{market-clearing prices} if there exist $\barbbundle$ and $\barbcash$ such that $(\barbbundle, \barbcash, \aggbprice)$ is a market-clearing equilibrium.
Similarly, we call $\barbbundle$ a \emph{market-clearing allocation} if there exists a corresponding equilibrium.
\end{definition}
%
%Let's examine this definition more closely.

The requirements on $\sum_{i=1}^\nbuyers \barbbundle_i$ and $\sum_{i=1}^\nbuyers \barcash_i$ guarantee that no additional securities or cash have been created. In other words, there exists some set of trades among traders that would lead to the market-clearing allocation, although the definition says nothing about how the equilibrium is reached.
%
\begin{comment}
Since utility is concave, the final condition is equivalent to requiring that for every security $k \in [\nsec]$ and trader $i \in [\nbuyers]$, we have
\begin{equation}
\frac{\partial U_i(\barbbundle_i, \barcash_i)}{\partial \barbundle_{i,k}} - \aggprice{k} \frac{\partial U_i(\barbbundle_i, \barcash_i)}{\partial \barcash_{i}} = 0.
\label{eqn:mceqcondition}
\end{equation}
We will make use of this alternative formulation later. \jenn{Do we? Could we cut this?}
\end{comment}
%

% JWV: mentioning below to save space
%In general the market-clearing prices $\aggbprice$ are not unique, but they will be unique for the specific utility functions that we study.

\ignore{
\jenn{Is there more intuition we could add about what these prices mean? Relationship to equilibrium notions studied in econ/finance?}
}

\ignore{
\mdcomment{In general, the equilibrium is not unique. The equilibria can be obtained as solutions to problems of the form $\min_{\bbundle,\vec{c}}\sum_i \alpha_i U_i(\bbundle_i,c_i)$ where $\alpha_i>0$. Different
choices of $\alpha_i$ may give rise to different equilibria. I don't think that this is worth pointing out. Instead I suggest that we just say that the equilibrium is in general not unique,
but that we will study it for utilities which will yield uniqueness. We could introduce our specific form of $U_i$ earlier in which case we can now make a forward pointer saying that the equilibrium prices are unique.}
}

Since we rely on a market maker to orchestrate trade, our markets generally do not reach the market-clearing equilibrium. Instead, we introduce the notion of \emph{market-maker equilibrium}. This equilibrium is again described by a set of security prices $\eqbprice$ and trader allocations $(\eqbbundle_i, \eqcash_i)$, summarized as $(\eqbbundle,\eqbcash)$, such that no trader wants to trade at these prices given her allocation. The difference is that we now require $\eqbbundle$ and $\eqbcash$ to be reachable via some sequence of trade with the market maker instead of via trade among only the traders, and $\eqbprice$ must be
the market prices after such a sequence of trade.
%the prices offered by the market maker after such a sequence of trade.
%
\begin{definition}[Market-maker equilibrium]
  A triple $(\eqbbundle\!, \eqbcash\!, \eqbprice)$ is a market-maker equilibrium for cost function $C_\liq$ if, for the market state $\bstate^\star=\sum_{i=1}^\nbuyers \eqbbundle_i$,
we have
$\sum_{i=1}^\nbuyers \eqcash_i\! = C_\liq(\zero)-C_\liq(\bstate^\star) $,
$\eqbprice\! = \nabla C_\liq(\bstate^\star)$,
and for all $i\in [N]$,
$\zero \in \argmax_{\vdelta \in \R^K} U_i\bigParens{
                       \eqbbundle_i\! + \vdelta,\,
                       \eqcash_i\! - C_\liq(\bstate^\star\! + \vdelta) + C_\liq(\bstate^\star)
                       }
$.
\ignore{
\[
\sum_{i=1}^\nbuyers \eqcash_i\! + C_\liq(\bstate^\star) - C_\liq(\zero) {=} 0
,
\;\;
\eqbprice\! {=} \nabla C_\liq(\bstate^\star)
,
\;\;
\zero \in \argmax_{\vdelta \in \R^K} U_i\bigParens{
                       \eqbbundle_i\! + \vdelta,\,
                       \eqcash_i\! - C_\liq(\bstate^\star\! + \vdelta) + C_\liq(\bstate^\star)
                       }
\ \ \forall i \in [N].
\]
} % ends ignore
We call $\eqbprice$ \emph{market-maker equilibrium prices} if there exist $\eqbbundle$ and $\eqbcash$ such that $(\eqbbundle\!, \eqbcash\!, \eqbprice)$ is a market-maker equilibrium.
Similarly, we call $\eqbbundle$ a \emph{market-maker equilibrium allocation} if there exists a corresponding equilibrium.
We sometimes write
%equilibrium prices as
$\eqbprice(\liq;C)$ to show the dependence of $\eqbprice$ on $C$ and $\liq$.
\end{definition}
%
\begin{comment}
Since utility is concave and the cost function is convex, the final condition is again equivalent to requiring that for every security $k \in [\nsec]$ and trader $i \in [\nbuyers]$, we have \jenn{Do we use this?}
\begin{equation}
\frac{\partial U_i(\eqbbundle_i, \eqcash_i)}{\partial \eqbundle_{i,k}} - \eqprice{k} \frac{\partial U_i(\eqbbundle_i, \eqcash_i)}{\partial \eqcash_{i}} = 0.
\label{eqn:mmeqcondition}
\end{equation}
\end{comment}
%

The market-clearing prices $\aggbprice$ and the market-maker equilibrium prices $\eqbprice(\liq;C)$ are not unique in general, but are unique for the specific utility functions that we study in this paper.

% old def
\ignore{
\begin{definition}[Market-Maker Equilibrium Prices]
We then say that $\eqbprice(\liq;C)$ is a market equilibrium price for market maker $C_\liq$ as long as there exists bundles $\bbundle = (\bbundle_i)_{i=1}^\nbuyers$ and cash $(c_i)_{i=1}^\nbuyers$
\begin{enumerate}
\item We can write $\eqbprice(\liq;C) = \nabla C_\liq\left(\sum_{i=1}^\nbuyers\bbundle_i\right)$ and $\sum_{i=1}^\nbuyers c_i = C_\liq(\zero) - C_\liq\left(\sum_{i=1}^\nbuyers \bbundle_i\right)$
\item For every security $k \in [\nsec]$ and buyer $i \in [\nbuyers]$, we have
$$
\frac{\partial U_i(\bbundle_i, c_i)}{\partial \bundle_{i,k}} - \eqprice{k}(\liq;C) \frac{\partial U_i(\bbundle_i, c_i)}{\partial c_{i}} = 0
$$
\end{enumerate}
\end{definition}
} % ends ignore of old def

Using these notions of equilibrium, we can formally define our error components. Sampling error is the difference between the true security values and the market-clearing equilibrium prices.  The bias is the difference between the market-clearing equilibrium prices and the market-maker equilibrium prices.  Finally, the convergence error is the difference between the market-maker equilibrium prices and the market prices $\iterbprice{t}(b;C)$ at a particular round $t$.  Putting this together, we have that
\begin{equation}
\bpricetruth - \iterbprice{t}(b;C) = \underbrace{\bpricetruth - \aggbprice}_{\text{Sampling Error}} + \underbrace{\aggbprice - \eqbprice(\liq;C)}_{\text{Bias}} + \underbrace{\eqbprice(\liq;C) - \iterbprice{t}(b;C)}_{\text{Convergence Error}} .
\label{eq:total_error}
\end{equation}

%Since our primary goal is to understand the impact of the cost function $C$ and liquidity parameter $b$ on error, we devote most of our attention to analyzing the market-maker bias and convergence error, though we also briefly explore the sampling error and {risk-aversion bias}.

\ignore{
Note that the sampling error and XXX do not depend on the choice of the cost function family $C$, liquidity parameter $\liq$, or particular dynamics of trade.  We will touch on these only briefly before providing a more detailed analysis of the bias and convergence.
}
%%% Local Variables:
%%% mode: latex
%%% TeX-master: "main"
%%% End:

\section{The Exponential Trader Model}
\label{sec:exp}
\label{SEC:EXP}

%We showed that the error in market predictions can be decomposed into several terms.

%We now turn to the task of bounding the sources of forecast error in market predictions and examining how they are affected by the choice of cost function $C$ and liquidity parameter $b$.

For the remainder of the paper, we work with the exponential trader model introduced by \citet{AKLS14} in which traders have exponential utility functions and exponential-family beliefs.  Under this model, both the market-clearing prices and market-maker equilibrium prices are unique
%(though the equilibrium allocations might not be)
and can be expressed cleanly in terms of potential functions~\cite{FR15}, yielding a tractable analysis. The results of this section are immediate consequences of prior work~\cite{AKLS14,FR15}, but our equilibrium concepts bring them into a common framework.

%\subsection{Trader Model}

We consider a specific \emph{exponential family}~\cite{barndorff1982exponential} of probability distributions over $\Omega$ defined as
$
p(\outcome;\btheta) = e^{\payoff(\outcome)\cdot \btheta -
  \lmsr(\btheta)}$,
where $\btheta \in R^K$ is the \emph{natural parameter} of the distribution,
and
$\lmsr$ is the \emph{log partition function},
%\begin{equation}
$\lmsr(\btheta) \coloneqq \log \left(\sum_{\outcome \in \Outcome}
  e^{\payoff(\outcome)\cdot \btheta}\right)$.
%\label{eqn:logpartition}
%\end{equation}
%For the log partition function,
The gradient
$\nabla T(\btheta)$ coincides with the expectation of $\payoff$ under $p(\cdot;\btheta)$, and
% or at least say that “\ntabla T(\btheta)\in\cM”. We also need to say that
$\dom T^*=\conv\set{\payoff(\omega):\:\omega\in\Omega}=\cM$.

%All relevant information about an outcome $\omega$ is captured through $\payoff(\omega)$,
%%which is referred to as the sufficient statistic.
%which in our setting is the security payoff function as defined in
%Section~\ref{sec:prelims}.

% Exponential family distributions are both analytically convenient
% and widely used.
% Exponential families distributions include the binomial, Poisson, Gaussian, Dirichlet, and exponential distributions.

\ignore{
For example, we might have \begin{equation}
    \ttheta_{i,k}\sim\textup{Normal}(\theta_k^\TRUE,\sigma^2)
\enspace,
\label{eq:belief_gen}
\end{equation}
where $\sigma^2$ quantifies the error in the agent beliefs.
}

\begin{comment}
We then model the true probability and the sampled natural parameters $\{\bttheta_i\}_{i=1}^\nbuyers$ of the traders over the securities in the following way for $\sigma>0$ and $\realtheta \in \R^\nsec$
$$
\{\Payoff \bttheta_i\}_{i=1}^\nbuyers \stackrel{i.i.d.}{\sim} N(\Payoff\realtheta,\sigma^2I_{|\Outcome|})
$$
\end{comment}

Following \citet{AKLS14}, we assume that each trader $i$ has exponential-family beliefs with natural parameter $\bttheta_i$. From the perspective of trader $i$, the expected payoffs of the $K$ market securities can then be expressed as the vector $\bbelief{i}$ with $\belief{i,k}\coloneqq\sum_{\outcome \in \Outcome} \npayoff_{k}(\outcome) p(\outcome;\bttheta_i)$.

As in \citet{AKLS14}, we also assume that traders are risk averse with exponential utility for wealth, so the utility of trader $i$ for wealth $W$ is
$
u_i(W) = -(1/a_i)e^{-a_i W},
$
where $a_i$ is the the trader's risk aversion coefficient.  We assume that the traders' risk aversion coefficients are fixed.
%where $a_i$ is the trader's Arrow-Pratt coefficient of absolute risk aversion.  Larger values of $a_i$ imply higher levels of risk aversion. We assume that the traders' risk aversion coefficients are fixed.

Using the definitions of the expected utility $U_i$, the exponential family distribution $p(\cdot; \bttheta_i)$, the log partition function $T$, and the exponential utility $u_i$, it is straightforward to show~\cite{AKLS14} that %the expected utility of trader $i$ can be written
\begin{equation}
  U_i(\bbundle_i,c_i)
=-\frac{1}{a_i} e^{-T(\bttheta_i)-a_i c_i}
        \textstyle
        \sum_{\omega \in \Omega} e^{\payoff(\omega) \cdot (\bttheta_i-a_i\bbundle_i)}
        \displaystyle
=-\frac{1}{a_i} e^{T(\bttheta_i-a_i\bbundle_i)-T(\bttheta_i)-a_i c_i}.
\label{eq:expectedutility}
\end{equation}
%Similar expressions were used in the analysis of \citet{AKLS14}.

%\subsection{Potential Functions and Uniqueness of Equilibrium Prices}

Under this trader model, we can use the techniques of \citet{FR15} to construct potential functions which yield alternative characterizations of the equilibria as solutions of minimization problems.
%These characterizations enable us to prove uniqueness of the equilibrium prices and play a crucial role in our analysis of bias and convergence error.
%
Consider first a market-clearing equilibrium.
Define
$
  F_i(\bstate)\coloneqq\frac{1}{a_i}T(\bttheta_i+a_i\bstate)
$
for each trader $i$.
From \Eq{expectedutility} we can observe that $-F_i(-\bbundle_i) + c_i$ is a monotone transformation of trader $i$'s utility.
%
% cut for space
\ignore{
\footnote{We could have alternatively defined $F_i(\bstate)\coloneqq (1/a_i)T(\bttheta_i-a_i\bstate)$ so that we could work with $F_i(\bbundle_i)$ instead of $F_i(-\bbundle_i)$. The form we chose turns out to be more convenient for the expression in \Eq{eqbdual:short}. \jenn{Fill in.}\rynote{Added.  This is due to the connection to infimal convolution -- the conjugate of a sum of functions is the sum of the conjugates as long as the arguments add to zero.}}
}
Since each trader's utility is locally maximized at a market-clearing equilibrium, the sum of traders' utilities is also locally maximized, as is $\sum_{i=1}^N (-F_i(-\bbundle_i) + c_i)$.  Since the equilibrium conditions require that $\sum_{i=1}^N c_i = 0$, the security allocation associated with any market-clearing equilibrium must be a local minimum of
%the function
%
%\begin{align}
% \Fee(\bbundle) &\coloneqq
$\sum_{i=1}^N F_i(-\bbundle_i)$.
%\enspace.
%\end{align}
This idea is formalized in the following theorem. The proof
follows from an analysis of the KKT conditions of the equilibrium. (See the appendix
%in the supplemental material
for all omitted proofs.)

% ends ignore of old version

\begin{theorem}
\label{thm:agg:short}
Under the exponential trader model, a market-clearing equilibrium always exists and market-clearing prices are unique.
Market-clearing allocations and prices are exactly the solutions of the following optimization problems:
\begin{equation}\label{eq:aggbdual:short}
  \bar{\bbundle}\in\argmin_{\bbundle:\:\sum_{i=1}^N \bbundle_i=\zero} \BigBracks{{\textstyle\sum_{i=1}^N F_i(-\bbundle_i)}}
\enspace,
\qquad
  \aggbprice = \argmin_{\bprice\in\R^K} \BigBracks{{\textstyle \sum_{i=1}^N F_i^*(\bprice)}}
\enspace.
\end{equation}
\end{theorem}

Using a similar argument, we can show that the allocation associated with any market-maker equilibrium is a local minimum of the function
%\begin{align}
  $F(\bbundle) \coloneqq
  \sum_{i=1}^N F_i(-\bbundle_i) + C_b\bigParens{\sum_{i=1}^N\bbundle_i}.$
%\enspace.
%\end{align}
%

%This is formalized in the following theorem.

\begin{theorem}
\label{thm:eq:short}
Under the exponential trader model, a market-maker equilibrium always exists and equilibrium prices are unique.
Market-maker equilibrium allocations and prices are exactly the solutions of the following optimization problems:
\begin{equation}\label{eq:eqbdual:short}
  \optbbundle\in\argmin_{\bbundle} F(\bbundle)
\enspace,
\qquad
  \eqbprice = \argmin_{\bprice\in\R^K} \BigBracks{{\textstyle \sum_{i=1}^N F_i^*(\bprice) + bC^*(\bprice)}}
\enspace.
\end{equation}
\end{theorem}

\paragraph{Sampling error}

We finish this section with an analysis of the first component of error identified
in \Sec{decomposition}: the sampling error.
% the error due to the noise in trader beliefs and due to their risk aversion.
We begin by deriving a more explicit form of market-clearing prices:
%by studying first-order optimality conditions in \Thm{agg:short}:
%
\begin{theorem}
\label{thm:eqprice-char:NIPS}
  Under the exponential trader model, the unique market-clearing
  equilibrium prices can be written as
  $\aggbprice = \Ex{\bar{\btheta}}{\payoff(\omega)}$,
  where
  $\bar{\btheta}\coloneqq\bigParens{\sum_{i=1}^\nbuyers \bttheta_i/a_i}/\bigParens{\sum_{i=1}^\nbuyers 1/a_i}$
  is the risk-aversion-weighted average belief
  and $\mathbb{E}_{\bar{\btheta}}$ is the expectation under $p(\cdot;\bar{\btheta})$.
\end{theorem}
The sampling error arises because the beliefs $\bttheta_i$ are only
noisy signals of the ground truth. From \Thm{eqprice-char:NIPS} we see that this
error may be compounded by the weighting
according to risk aversions, which can skew the prices.
To obtain a concrete bound on the error term $\norm{\bpricetruth - \aggbprice}$,
we need to make some assumptions about risk aversion coefficients, the
true distribution of the outcome, and how this distribution is related to trader beliefs.
%Many more involved models are possible and are beyond the scope of this paper.
For instance, suppose risk aversion coefficients are bounded both from below and above,
the true outcome is drawn from an exponential-family distribution with natural parameter $\btheta^\TRUE$,
%i.e., $\bpricetruth=\Ex{\btheta^\TRUE}{\payoff(\omega)}
%The corresponding vector of expected security payoffs is denoted $\bpricetruth$.
and the beliefs $\bttheta_i$ are independent samples with mean $\btheta^\TRUE$ and a bounded covariance matrix.
Under these assumptions, one can show using standard concentration bounds that with high probability, $\norm{\bpricetruth - \aggbprice}=O(\sqrt{1/N})$ as $N\to\infty$. In other words,
market-clearing prices approach the ground truth as the number of traders increases.
In \App{sampling} we make the dependence on risk aversion and belief noise more explicit.
The analysis of other information
structures (e.g., biased or correlated beliefs) is beyond the scope of this paper; instead, %in the remainder of the paper
we focus
on the two error components that depend on the market design.
%: \emph{market-maker bias} and \emph{convergence error}.

%%% Local Variables:
%%% mode: latex
%%% TeX-master: "main"
%%% End:

%\input{sampling}
\ifnips
\section{Market-maker Bias}
\label{sec:bias}
\label{SEC:BIAS}

We now analyze the market-maker bias---the difference between the marker-maker equilibrium prices $\eqbprice$ and market-clearing prices $\aggbprice$.
% Our main tools are the characterizations of $\eqbprice$ and $\aggbprice$ provided by Theorems~\ref{thm:agg:short} and~\ref{thm:eq:short}.
We first state a global bound that depends on the liquidity $b$ and cost function $C$, but not on trader beliefs,
%We then move to a tighter local analysis, which allows us to precisely relate the bias of LMSR and independent markets.
%
\begin{comment}
We provide both global and local bounds to characterize the dependence on the liquidity $b$ and cost function $C$.  Our global bound shows that the bias is not more than $cb$ for a constant $c$ that depends on $C$ but not on the trader beliefs.  Our tighter local analysis shows that for an appropriate vector $\bu$, which depends on both $\aggbprice$ and $C$, we have $\eqbprice=\aggbprice+b\bu + \veps_b$ where $\norm{\veps_b}=O(b^2)$. An a corollary of this result, we then obtain that independent markets exhibit a larger, but no more than twice as large, bias as LMSR for the same value of $b$.
\end{comment}
%
and show that $\eqbprice\to\aggbprice$ with the rate $O(b)$ as $b\to 0$.
The proof builds on Theorems~\ref{thm:agg:short}
and~\ref{thm:eq:short} and uses the facts that $C^*$ is bounded on
$\cM$ (by our assumptions on $C$), and conjugates $F_i^*$ are strongly convex on $\cM$ (from properties of the log partition function).
\begin{theorem}[Global Bias Bound]
\label{thm:bias:global}
%Let $\sigma$ be the strong convexity constant of $G$ on $\cM$. Then
Under the exponential trader model, for any $C$, there exists a constant $c$ such that
$\norm{\eqbprice(b;C)-\aggbprice}\le cb$ for all $b\ge 0$.
%\le b \cdot\frac{2\norm{P\nabla Q(\aggbprice)}}{\sigma} = O(\liq)
\end{theorem}
%
\begin{comment}
The proof is in \App{bias:global} and we only sketch the key ideas here. First, note that the dual objectives in \Eqs{aggbdual:short}{eqbdual:short} only differ in the term $bC^*(\bprice)$.
Since $\dom F_i^*=\cM$, both optimizations effectively take place over $\cM$ and by our assumptions $C^*$ is upper bounded on $\cM$. From properties of the log partition function, the functions $F_i^*$ are strongly convex on $\cM$ and so is their
sum $\sum_i F_i^*(\bprice)$. This suffices to conclude that the difference between the optimizers $\aggbprice$ and $\eqbprice(b;C)$
is $O(b)$.
%as $b\to0$.
\end{comment}

This result makes use of strong convexity constants
that are valid over the entire set $\cM$, which can be overly conservative when $\eqbprice$ is close to $\aggbprice$.
%for some securities as the curvature of $F^*_i$ varies over $\cM$ non-uniformly across the securities.
Furthermore, it gives us only
an upper bound, which cannot be used to compare different cost function families. In the rest of
this section we pursue a tighter local analysis, based on the properties of $F^*_i$ and $C^*$ at $\aggbprice$.
Our local analysis requires assumptions
that go beyond convexity and differentiability of the cost function.
We call the class
of functions that satisfy these assumptions \emph{\convplus functions}. (See \App{convplus}
for their complete treatment and a more general definition than provided here.)
%\footnote{The gradient space and \convplus functions are defined here for the special case of convex functions that are finite everywhere. \Convplus functions and their properties are fully treated in .}
These functions are related to functions of \emph{Legendre type} (see Sec.~26 of \citet{Rockafellar70}).
 Informally, they are smooth functions
that are strictly convex along directions in a certain space (the \emph{gradient space})
% defined below)
and linear in orthogonal directions.
%
%The strictly convex and linear decomposition is natural for cost functions.
For cost functions,
strict convexity means that prices change in response to arbitrarily small trades, while the linear directions correspond to
%``sure bundles'', namely
bundles with constant payoffs, whose prices are therefore fixed.
%
%We next formally define \convplus functions and provide a few examples.

%\mdcomment{Do we want to mention here that strict convexity is a useful desideratum (with a cite) and that LMSR and IND are \convplus? Re: useful desideratum, I think it pops up in Jenn-Yiling-Jake paper. In the absence of strict convexity, buying/selling might not change the price.}

\begin{definition}
\label{def:Gf:main}
Let $f:\R^d\to\R$ be differentiable and convex.
Its \emph{gradient space}
is the linear space parallel to the affine hull of its gradients,
denoted as
$
  \cG(f)\coloneqq\Span\set{\nabla f(\bu)-\nabla f(\bu')\!:\!\bu,\bu'\!\in\R^d}
$.
\end{definition}

\begin{definition}
\label{def:convplus:main}
We say that a convex function $f:\R^d\to\R$ is \emph{\convplus} if
it has continuous third derivatives and $\range(\nabla^2 f(\bu))=\cG(f)$ for all $\bu\in\R^d$.
\end{definition}

It can be checked that if $P$ is a projection on $\cG(f)$ then there exists
some $\ba$ %a unique point $\ba\in\cG(f)^\perp$, i.e., in the orthogonal complement of $\cG(f)$,
such that $f(\bu)=f(P\bu)+\ba\trans\bu$, so
%the function
$f$ is up to a linear term fully described by its values
on $\cG(f)$.
The condition on the range of the Hessian
ensures that $f$ is strictly convex over $\cG(f)$, so its gradient map is invertible
over $\cG(f)$.
This means that the Hessian %of $f$
can be expressed as a function
of the gradient, i.e., there exists a matrix-valued function $H_f$ such that $\nabla^2 f(\bu)=H_f(\nabla f(u))$
(see \Prop{H:exists}).
The cost functions $C$ for both the LMSR and the sum of independent LMSRs (IND) are \convplus.
%as is the log partition function.

\begin{example}[LMSR as a \convplus function]
For LMSR, the gradient space of $C$ is parallel to the simplex: $\cG(C)=\set{\bu:\one\trans\bu=0}$.
% where $\one$ is the all-ones-vector.
The gradients of $C$ are points in the relative interior of the simplex. Given such a point $\bprice=\nabla C(\bstate)$, the corresponding
Hessian is $\nabla^2 C(\bstate)=H_C(\bprice)=(\diag_{k\in[K]} \mu_k)-\bprice\bprice\trans$, where $\diag_{k\in[K]}\mu_k$ denotes the diagonal matrix
with values $\mu_k$ on the diagonal.
%The matrix $H_C(\bprice)$ is the covariance matrix of the multinomial distribution $\bprice$.
The null space of $H_C(\bprice)$ is
%$\cG(C)^\perp=\set{c\one:c\in\R}$.
$\set{c\one:c\in\R}$, so
$C$ is linear in the all-ones direction (buying one share of each security always has cost one), but strictly convex in directions from $\cG(C)$.
%Finally, note that $C(\bstate)=C(P\bstate)+\ba\trans\bu$ where $P=I_K-\one\one\trans/K$ and $\ba = \one/K$.
\end{example}

\begin{example}[IND as a \convplus function]
For IND, the gradient space is $\R^K$ and the gradients are the points in $(0,1)^K$. In this case,
$H_C(\bprice)=\diag_k [\mu_k(1-\mu_k)]$.
%corresponding to the covariance matrix of $K$ independent Bernoulli variables with expectations $\mu_k$.
This matrix has full rank.
\end{example}

%With \convplus functions defined, we now move to the local analysis of bias.

Our next theorem shows that
for an appropriate vector $\bu$, which depends on $\aggbprice$ and $C$, we have
$
  \eqbprice(\liq;C)=\aggbprice+b\bu + \veps_b,
$
where $\norm{\veps_b}=O(b^2)$.
Here, the $O(\cdot)$ is taken as $b\to 0$, so the error term $\veps_b$ goes to zero faster than the
term $b\bu$,
which we call the \emph{asymptotic bias}.
%which we call the \emph{asymptotic bias vector}, and its magnitude the \emph{asymptotic bias}.
Our analysis is \emph{local} in the sense that the constants hiding within $O(\cdot)$ may depend
on $\aggbprice$.
%In contrast with \Thm{bias:global}, the local
This analysis fully uncovers the
main asymptotic term
%, which is on the order $O(b)$ as stated in \Thm{bias:global}.
and therefore allows comparison of cost families.
In our experiments, we show that the asymptotic bias is an
accurate estimate of the bias even for moderately large values of $b$.
%This tighter analysis requires that the cost function be \convplus.

\begin{comment}
\mdedit{Assumptions/properties required in the next two theorems:
\begin{itemize}
\item $C$ is strictly convex.
\item $G$ and $Q$ exist. They can be constructed via $G(\bprice)\coloneqq\sum_{i=1}^\nbuyers\obj_i^*(P\bprice)$ and $Q(\bprice)\coloneqq C^*(P\bprice)$. We only need the differentiability of $Q$ over the relative interior of $\cM$, which
should follow from strict convexity of $C$. Similarly, differentiability of $G$ should follow from strict convexity of the log partition function.
\item Minimizers $\eqbprice$ and $\aggbprice$ are attained in the relative interior of $\cM$. This should follow from the properties of the conjugates of the log partition function.
\item Since the minima are attained in the relative interior of $\cM$ it is fine to replace the minimization over $\cM$ by the minimization over $\cA$.
\item $G$ is strongly convex on $\cM$. This should again follow from the properties of the log partition function.
\end{itemize}}
\end{comment}

\begin{theorem}[Local Bias Bound]
\label{thm:bias:local}
Assume that the cost function $C$ is \convplus. Then
%the difference between the market-maker equilibrium price $\eqbprice(\liq;C)$ and the market-clearing equilibrium price $\aggbprice$ is
\[
\textstyle
%\BigNorm{\eqbprice(\liq;C)\;-\;\BigParens{\aggbprice- b\cdot\frac{\bar{a}}{N} H_T(\aggbprice)\partial C^*(\aggbprice)}}\le O(b^2)
  \eqbprice(\liq;C)=\aggbprice-b(\bar{a}/N) H_T(\aggbprice)\partial C^*(\aggbprice)
  +\veps_b
\enspace,
\quad
  \text{where $\norm{\veps_b}=O(b^2)$.}
\]
In the expression above, $\bar{a}=N/(\sum_{i=1}^N 1/a_i)$ is the harmonic mean of risk-aversion coefficients
and $H_T(\aggbprice)\partial C^*(\aggbprice)$ is guaranteed to consist of a single point even when $\partial C^*(\aggbprice)$ is a set.
%The remainder term $\cR$ can be bounded by the following
%$
%|| \cR || = O(\liq^2).
%$
\end{theorem}
The theorem is proved by a careful application of Taylor's Theorem and crucially uses properties of conjugates of \convplus functions, which we derive in \App{convplus}.
It gives us a formula to calculate the asymptotic bias for any cost function for a particular value of $\aggbprice$, or evaluate the worst-case bias against some set of possible market-clearing prices. It also constitutes an important step in comparing cost function
families. To compare the convergence error of two costs $C$ and $C'$ in the next section, we require that their liquidities $b$ and $b'$
be set so that they have (approximately) the same bias, i.e.,
$\Norm{\eqbprice(\liq';C') - \aggbprice}\approx\Norm{\eqbprice(\liq;C) - \aggbprice}$. \Thm{bias:local} tells us that this can be achieved by the linear rule $\liq'=\liq/\eta$ where
$\eta=\norm{H_T(\aggbprice)\partial {C'}^*(\aggbprice)}\,/\,\norm{H_T(\aggbprice)\partial C^*(\aggbprice)}$. For $C=\LMSR$ and $C'=\IND$, we prove that the corresponding $\eta\in[1,2]$. Equivalently, this means that for the same value of $\liq$ the asymptotic bias of \IND is at least as large as that of \LMSR, but no more than twice as large:
\begin{theorem}
\label{thm:bias:two}
For any $\aggbprice$ there exists $\eta\in [1,2]$ such that for all $b$,
$\norm{\eqbprice(\liq/\eta;\IND) - \aggbprice}
=
\norm{\eqbprice(\liq;\LMSR) - \aggbprice}
\pm O(b^2)$.
For this same $\eta$, also
$\norm{\eqbprice(\liq;\IND) - \aggbprice}
=
\eta\norm{\eqbprice(\liq;\LMSR) - \aggbprice}
\pm O(b^2)$.
\end{theorem}

\Thm{bias:local} also captures an intuitive relationship which can guide the market maker in adjusting the market liquidity $b$ as the number of traders $N$ and their risk aversion coefficients $a_i$ vary. In particular, holding $\aggbprice$
and the cost function fixed, we can maintain the same amount of bias by setting $b\propto N/\bar{a}$. Note that $1/a_i$ plays the role of the budget of trader $i$ in the sense that at fixed prices, the trader will spend an amount of cash proportional to $1/a_i$. Thus $N/\bar{a}=\sum_i (1/a_i)$ corresponds to the total amount of available cash among the traders in the market. Similarly, the market maker's worst-case loss, amounting to the market maker's cash, is proportional to $b$, so setting $b\propto\sum_i (1/a_i)$ is natural.

%%% Local Variables:
%%% mode: latex
%%% TeX-master: "main"
%%% End:

\else
\input{bias}
\fi
\ifnips
\section{Convergence Error}
\label{sec:convergence}
\label{SEC:CONVERGENCE}

We now study the convergence error, namely the difference between the prices $\bprice^t$ at round $t$ and the market-maker equilibrium prices $\eqbprice$.
To do so,
% of market prices,
we must posit a model of how the traders interact with the market.
% as a function of their current holdings of securities $\bbundle_i$ and cash $c_i$,
%and the current market state $\bstate$.
Following \citet{FR15}, we assume that in each round, a trader
$i\in[N]$, chosen uniformly at random, buys a
bundle $\bdelta\in\R^\nsec$ that optimizes her utility given the
current market state $\bstate$ and her existing security and cash
allocations, $\bbundle_i$ and $c_i$.
The resulting updates of the allocation vector $\bbundle=(\bbundle_i)_{i=1}^N$ correspond to randomized block-coordinate descent on the potential function $F(\bbundle)$ with blocks $\bbundle_i$
%corresponding to individual traders
(see \App{trader:dynamics} and~\citet{FR15}).
We refer to this model as the \emph{all-security (trader) dynamics} (\ASD).\footnote{%
In \App{conv}, we also analyze the \emph{single-security (trader) dynamics} (\SSD), in which
a randomly chosen trader randomly picks a single security to trade,
corresponding to randomized coordinate descent on $F$.}
We apply and extend the analysis of block-coordinate descent to
%analyze convergence error in
this setting.  We focus on \convplus functions and conduct local
convergence analysis around the minimizer of $F$. Our
experiments
%in \Sec{experiments}
demonstrate that the local analysis
accurately estimates the convergence rate.

Let $\bbundle^\star$ denote an arbitrary minimizer of $F$ and let $F^\star$ be the minimum value of $F$.
%As with $\eqbprice$, we will write $\bbundle^\star(b;C)$ and $F^\star(b;C)$ whenever we want to
%emphasize the dependence on the liquidity and cost.
Also, let $\bbundle^t$ denote the allocation vector and $\bprice^t$ the market price vector after the $t${\small th} trade.
Instead of directly analyzing the convergence error $\norm{\bprice^t-\eqbprice}$,
we bound the suboptimality $F(\bbundle^t)-F^\star$ since $\norm{\bprice^t-\eqbprice}^2=\Theta\parens{F(\bbundle^t)-F^\star}$
for \convplus costs $C$ under \ASD (see \App{conv:ASD}).

\Convplus functions are locally strongly convex
%(along a relevant subspace)
and have a Lipschitz-continuous gradient, so the standard analysis of block-coordinate descent~\citep{Nest12,FR15} implies linear convergence,
i.e., $\ex{F(\bbundle^t)}-F^\star\le O(\gamma^t)$ for some $\gamma<1$,
where the expectation is under the randomness of the algorithm.
%These rates were also derived for the market convergence error under \ASD\citep{FR15}.
We refine the standard analysis by (1) proving not only upper, but also lower bounds on the convergence rate, and (2) proving an explicit dependence of $\gamma$
on the cost function $C$ and the liquidity $b$. These two refinements
are crucial for comparison of cost families, as we demonstrate with the
comparison of \LMSR and \IND.
We begin by formally defining bounds on local convergence of any randomized
iterative algorithm that minimizes a function $F(\bbundle)$ via a sequence of iterates $\bbundle^t$.

%\jenn{This definition is kind of weird without the set-up that we had
%  in the previous version. We need to be talking about an interative
%  algorithm for minimizing $F$ with iterates $\bbundle^t$.}

\begin{definition}
We say that $\kappahigh$ is an \emph{upper bound on the local convergence rate} of an algorithm
if, with probability 1 under the randomness of the algorithm,
the algorithm
%either finds a minimum in a finite number of iterations, or
reaches an iteration $t_0$ such that for some $c>0$ and all $t\ge t_0$,
$
        \ex{F(\bbundle^t)\bigGiven\bbundle^{t_0}}-F^\star\le c\kappahigh^{t-t_0}.
$
%\end{definition}
%
%\begin{definition}
We say that $\kappalow$ is a \emph{lower bound on the local convergence rate} if $\kappahigh\ge\kappalow$ holds
for all upper bounds $\kappahigh$.
\end{definition}

To state explicit bounds, we use the notation $D\coloneqq\diag_{i\in[N]} a_i$ and $P\coloneqq I_N-\one\one\trans/N$, where $I_N$ is the $N\times N$ identity matrix
and $\one$ is the all-ones vector.
%The matrix $P$ is the projection matrix on the set of centered vectors, i.e., vectors $\bu$ in $\R^N$ such that $\one\trans\bu=0$.
We write $M^+$ for the pseudoinverse of a matrix $M$ and $\lambda_{\min}(M)$ and $\lambda_{\max}(M)$ for its smallest and largest positive eigenvalues.

\begin{theorem}[Local Convergence Bound]
\label{thm:conv:short}
Assume that $C$ is \convplus.
Let $H_T\coloneqq H_T(\aggbprice)$ and $H_C\coloneqq H_C(\aggbprice)$.
%, and $D_C$ be the diagonal matrix with the diagonal of $H_C$.
For the all-securities dynamics, the local convergence rate is bounded between
%$  \kappahigh^{\ASD}  =
%1-2b/\nbuyers\cdot\lambda_{\min}(PDP)\cdot\lambda_{\min}\bigParens{\;H_T^{1/2}
%  H_C^+ H_T^{1/2}\;}+O(b^2)\enspace$
%and
%$\kappalow^{\ASD} = 1-2b/\nbuyers\cdot\lambda_{\max}(PDP)\cdot\lambda_{\max}\bigParens{\;H_T^{1/2} H_C^+ H_T^{1/2}\;}-O(b^2)$.
%
%\ignore{
\begin{align*}
  \kappahigh^{\ASD} & = 1-\tfrac{2b}{\nbuyers}\cdot\lambda_{\min}(PDP)\cdot\lambda_{\min}\bigParens{\;H_T^{1/2} H_C^+ H_T^{1/2}\;}+O(b^2)\enspace,
\\
  \kappalow^{\ASD} & = 1-\tfrac{2b}{\nbuyers}\cdot\lambda_{\max}(PDP)\cdot\lambda_{\max}\bigParens{\;H_T^{1/2} H_C^+ H_T^{1/2}\;}-O(b^2)
\enspace.
\end{align*}
%} % end ignore
\end{theorem}

In our proof, we first establish both lower and upper bounds on convergence of a generic block-coordinate descent that extend the results of \citet{Nest12}.
%and may be of independent interest.
We then analyze the behavior of the algorithm
%over two consecutive iterations
%at a time
for the specific structure of our objective
to obtain explicit lower and upper bounds. Our bounds prove linear convergence with the rate $\gamma=1-\Theta(b)$. Since
the convergence gets worse as $b\to 0$, there is a trade-off with the bias, which decreases as $b\to 0$.

Theorems~\ref {thm:bias:local} and~\ref{thm:conv:short} enable systematic quantitative comparisons of cost families. For simplicity, assume that $N\ge 2$ and all risk aversions
are $a$, so $\lambda_{\min}(PDP)=\lambda_{\max}(PDP)=a$. To compare
convergence rates of two costs $C$ and $C'$, we need to control for bias. As discussed after \Thm{bias:local}, their
biases are (asymptotically) equal if their liquidities are linearly related as $b'=b/\eta$ for a suitable $\eta$. \Thm{conv:short} then states that $C'_{b'}$ requires
(asymptotically) at most a factor of $\rho$ as many trades as $C_b$ to achieve the same convergence error, where
$  \rho\coloneqq
%b\lambda_{\max}\parens{H_T^{1/2} H_C^+ H_T^{1/2}}/
%b'\lambda_{\min}\parens{H_T^{1/2} H_{C'}^+ H_T^{1/2}}
%=
\eta\cdot
\lambda_{\max}\parens{H_T^{1/2} H_C^+ H_T^{1/2}} /
\lambda_{\min}\parens{H_T^{1/2} H_{C'}^+ H_T^{1/2}}.$
\ignore{
\[
  \rho\coloneqq
\frac{b\lambda_{\max}\parens{H_T^{1/2} H_C^+ H_T^{1/2}}}
     {b'\lambda_{\min}\parens{H_T^{1/2} H_{C'}^+ H_T^{1/2}}}
= \eta\cdot
\frac{\lambda_{\max}\parens{H_T^{1/2} H_C^+ H_T^{1/2}}}
     {\lambda_{\min}\parens{H_T^{1/2} H_{C'}^+ H_T^{1/2}}}
%= \frac{\Norm{H_T\partial {C'}^*(\aggbprice)}}{\Norm{H_T\partial C^*(\aggbprice)}}
%\cdot
%\frac{\lambda_{\max}\parens{H_T^{1/2} H_C^+ H_T^{1/2}}}
%     {\lambda_{\min}\parens{H_T^{1/2} H_{C'}^+ H_T^{1/2}}}
\enspace.
\]
} % ends ignore
Similarly, $C_b$ requires at most a factor of $\rho'$ as many trades as $C'_{b'}$, with $\rho'$ defined symmetrically to~$\rho$. For $C=\LMSR$ and $C'=\IND$,
we can show that $\rho\le 2$ and $\rho'\le 2$, yielding the following result:

\begin{theorem}
\label{thm:conv:two}
Assume that $N\ge 2$ and all risk aversions are equal to $a$.
Consider running \LMSR with liquidity $b$ and \IND with liquidity $b'=b/\eta$ such that their asymptotic biases are equal. Denote the
iterates of the two runs of the market as $\bprice^t_\LMSR$ and $\bprice^t_\IND$ and the respective market-maker equilibria
as $\eqbprice_\LMSR$ and $\eqbprice_\IND$.
Then, with probability 1, there exist $t_0$ and $t_1\ge t_0$ such that for all $t\ge t_1$ and sufficiently small $b$
\[
   \ExNoBrack{t_0}\bigBracks{\bigNorm{\bprice^{2t(1+\eps)}_\LMSR-\eqbprice_\LMSR}^2}
   \le
   \ExNoBrack{t_0}\bigBracks{\bigNorm{\bprice^t_\IND-\eqbprice_\IND}^2}
   \le
   \ExNoBrack{t_0}\bigBracks{\bigNorm{\bprice^{(t/2)(1-\eps)}_\LMSR-\eqbprice_\LMSR}^2}
\enspace,
\]
where $\eps=O(b)$ and $\exNoBrack_{t_0}[\cdot]=\exNoBrack[\cdot\given\bbundle^{t_0}]$ conditions on the $t_0${\small th} iterate
of a given run.
\end{theorem}

This result means that \LMSR and \IND are roughly equivalent (up to a factor of two) in terms of the number of trades required to achieve a given accuracy. This is somewhat surprising as this implies that maintaining price coherence
does not offer strong informational advantages (at least when traders are individually coherent, as assumed here). However, while there is little difference between the two costs in terms of accuracy, there is a difference in terms of the worst-case loss.
For $K$ securities, the worst-case loss of \LMSR with the liquidity $b$ is $b\log K$, and the worst-case loss of \IND with the liquidity $b'$ is $b'K\log2$. If liquidities are chosen as in \Thm{conv:two}, so that $b'$ is up to a factor-of-two smaller than $b$, then the worst-case loss of \IND is at least $(bK/2)\log 2$, which is always worse than the \LMSR's loss of $b\log K$, and the ratio of the two losses increases as $K$ grows.

When all risk aversion coefficients
are equal to some constant $a$, then the dependence of \Thm{conv:short} on the number of traders $N$ and their risk aversion is similar to the dependence in \Thm{bias:local}. \mdedit{For instance, to guarantee that $\gamma$ stays below a certain level for varying $N$ and $a$ requires $b=\Omega(N/a)$.}

\iffalse
\begin{example}[Convergence of \LMSR under \ASD]
We next demonstrate the tightness of our bounds for \ASD.
%
Consider the setting when $N\ge 2$ and the risk aversion of all traders equals $a$. Then $PDP=aPI_NP=aP$, and since $P$ is a non-zero projection matrix,
we obtain $\lambda_{\min}(PDP)=a\lambda_{\min}(P)=a$ and similarly $\lambda_{\max}(PDP)=a$. Furthermore,
if the cost is \LMSR then $H_C=H_T$, and thus $H_T^{1/2} H_C^+ H_T^{1/2}=I_K$. Therefore, \Thm{conv:short} yields
the objective decreases at the rate $\gamma^t$ and the convergence error
at the rate $\gamma^{t/2}$ with $\gamma=1-2ab/N+O(b^2)$, and the linear term in $b$ cannot be improved.
\end{example}
%when we have at least two traders $a_{\min}\le\lambda_{\min}(PDP)\le\lambda_{\max}(PDP)\le a_{\max}$
\fi

%%% Local Variables:
%%% mode: latex
%%% TeX-master: "main"
%%% End:

\else
\input{convergence}
\fi
\ifnips
\section{Numerical Experiments}
\label{sec:experiments}

\ignore{
We have analyzed the dependence of the market-maker bias and convergence error on the functional form of the cost function $C$ and the liquidity parameter $b$, providing upper and lower bounds.
}

\begin{comment}
\mdcomment{Include or emphasize the following points:
\begin{itemize}
\item 20 sequences were sufficient to get accurate (i.e., low-variance) estimates of expected performance. In particular, the error bars in the convergence plots are $\pm\text{[enter number]}$ on the $\log_{10}$ scale. [After Ryan's results, I don't think that this strategy is quite sound, because the largest error bars are $\pm 0.1$ on the $\log_{10}$ scale and that appears too large for some of the lines. Consider the three options that I suggested in the e-mail thread?]
\item For simplicity, we set all traders to have a risk aversion equal to 1. Note that setting all the risk aversions to some other value $c$ would yield the same results under the rescaling of liquidities by $c$. When the risk aversion is not equivalent,\dots
\end{itemize}}
\end{comment}

We evaluate the tightness of our theoretical bounds via numerical simulation.
We consider a complete market over $K = 5$ securities and simulate $N = 10$ traders
%with exponential-family beliefs and exponential utilities
with risk aversion coefficients equal to~1. These values of $N$ and $K$ are large enough to demonstrate the tightness of our results, but small enough that simulations are tractable.  While our theory comprehensively covers heterogeneous risk aversions and the dependence on the number of traders and securities, we have chosen to keep these values fixed, so that we can more cleanly explore the impact of liquidity and number of trades.
We consider the two most commonly studied cost functions: \LMSR and \IND.
We fix the ground-truth
%likelihood of each outcome (or equivalently, ground truth expected security payoffs) $\bprice^\TRUE$ and the corresponding
natural parameter $\btheta^\TRUE$ and independently sample the belief $\bttheta_i$ of each trader from $\textup{Normal}(\btheta^\TRUE,\sigma^2I_\nsec)$, with $\sigma=5$.  We consider a \emph{single-peaked} ground truth distribution with $\theta^\TRUE_1 = \log(1- \nu(K-1))$ and $\theta^\TRUE_k = \log\nu$ for $k \neq 1$, with $\nu = 0.02$. Trading is simulated according to the all-security dynamics (\ASD) as described at the start of Section~\ref{sec:convergence}.  In Appendix~\ref{app:experimentsNIPS}, we show qualitatively similar results using a uniform ground truth distribution and single-security dynamics (\SSD).

%\paragraph{Bias/Convergence Tradeoff}
We first examine the tradeoff that arises between market-maker bias and convergence error as the liquidity parameter is adjusted.
%Since our main interest is in the effect of the cost function $C$ and liquidity parameter $b$ on error, we ignore the sources of error that do not depend on the choice of cost function.
\Fig{experiments} (left) shows the combined bias and convergence error,
$\norm{\iterbprice{t}-\aggbprice}$, as a function of liquidity and the number of trades $t$ (indicated by the color of the line)
%when averaged over 20 random sequences of trade
for the two cost functions, averaged over twenty random trading sequences.
%(Other choices of norm lead to similar results.)
The minimum point on each curve tells us the optimal value of the liquidity parameter $b$ for the particular cost function and particular number of trades. When the market is run for a short time, larger values of $b$ lead to lower error.  On the other hand, smaller values of $b$ are preferable as the number of trades grows, with the combined error approaching 0 for small $b$.
%
%\jenn{Do we want to say something about why we chose to look only at these particular (very small) values of $b$?}
%MD: I don't think we need to

%\paragraph{Market-Maker Bias}
%We next focus in on the market-maker bias to empirically evaluate our bounds from \Sec{bias}.
%From \Thm{bias:local}, we know that
In \Fig{experiments} (center) we plot the bias $\norm{\eqbprice(\liq;C) - \aggbprice}$ as a function of $b$ for both \LMSR and \IND. We compare this with the theoretical approximation
$
\norm{\eqbprice(\liq;C)-\aggbprice} \approx b(\bar{a}/N)\norm{H_T(\aggbprice)\partial C^*(\aggbprice)}
$
from \Thm{bias:local}.  Although \Thm{bias:local} only gives an asymptotic guarantee as $\liq \to 0$, the approximation is fairly accurate even for moderate values of $b$. In agreement with \Thm{bias:two},
%the slope of \IND is higher than the slope of \LMSR, i.e.,
the bias of \IND is higher than that of \LMSR at any fixed value of $b$, but by no more than a factor of two.

\begin{figure}[t!]
\begin{subfigure}{.32\textwidth}
  \centering
  \includegraphics[width=\linewidth]{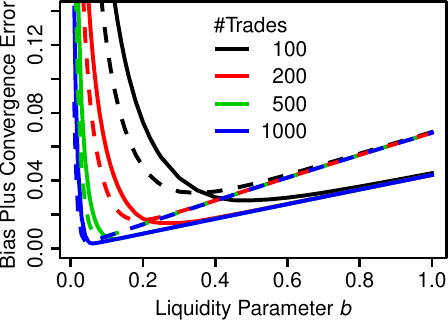}
   % \caption{The tradeoff between marker-maker bias and convergence error. \label{fig:Uplots}}
\end{subfigure}%
\hfill
\begin{subfigure}{.32\textwidth}
  \centering
    \includegraphics[width=\linewidth]{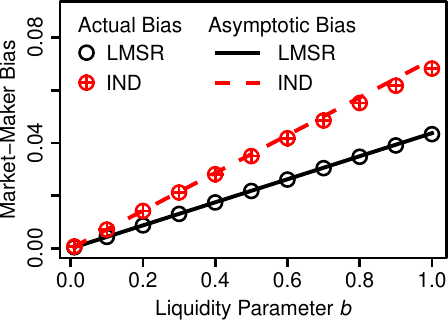}
    %\caption{Market-maker bias as a function of $b$. \label{fig:bias}}
\end{subfigure}%
\hfill
\begin{subfigure}{.32\textwidth}
  \centering
  \includegraphics[width=\linewidth]{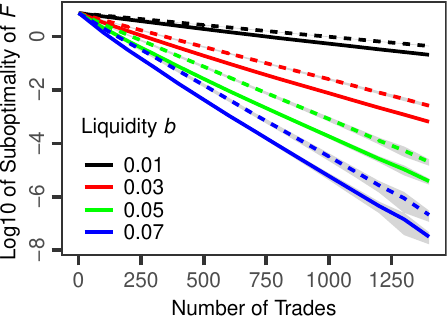}
   % \caption{Convergence in the objective value. \label{fig:semi_logB}}
\end{subfigure}%
\caption{(Left) The tradeoff between market-maker bias and convergence. Solid lines are for \LMSR, dashed for \IND, the color indicates the number of trades.
         (Center) Market-maker bias as a function of $b$.
         (Right) Convergence in the objective. Shading indicates 95\% confidence based on 20 trading sequences.}
\label{fig:experiments}%
\vspace{-8pt}
\end{figure}

%\paragraph{Convergence Error}
%Finally, we evaluate our convergence error bounds from \Sec{convergence}.
In \Fig{experiments} (right) we plot the log of
%the suboptimality,
$\hatexNoBrack[F(\bbundle^t)]-F^\star$ as a function of the number of trades $t$
for our two cost functions and several liquidity levels. Even for small $t$ the curves are close to linear, showing that the local linear convergence rate kicks in essentially from the start of trade in our simulations. In other words, there exist some $\hat{c}$ and $\hat{\gamma}$ such that, empirically, we have
$
\hatexNoBrack[F(\bbundle^t)]-F^\star \approx \hat{c} \hat{\gamma}^t
$,
or equivalently,
$
\log(\hatexNoBrack[F(\bbundle^t)]-F^\star) \approx \log\hat{c} +t\log\hat{\gamma}
$.
Plugging the belief values into \Thm{conv:short}, the slope of the curve for \LMSR should be $\log_{10}\hat{\gamma}\approx -0.087b$ for sufficiently small $b$, and the slope for \IND should be between $-0.088b$ and $-0.164b$. In \App{experimentsNIPS}, we verify that this is the case.

%In summary, these numerical simulations illustrate that our explicit local bounds on bias and convergence error are tight and that the local regime comes into effect very early on.
%Our theoretical framework therefore provides a meaningful way to quantitatively evaluate error tradeoffs inherent in different choices of cost functions and liquidity levels.

%%% Local Variables:
%%% mode: latex
%%% TeX-master: "main"
%%% End:

\else
\input{experiments}
\fi
\section{Conclusion}

\ignore{
We have proposed a decomposition of forecasting error in prediction
markets into three components: \emph{sampling error},
\emph{market-maker bias}, and \emph{convergence error}.  We have
derived explicit local bounds on bias and convergence error and
illustrated in numerical experiments that these bounds are tight and
that the local regime comes into effect very early on.  
Our theoretical framework is to our knowledge the first to provide a
meaningful way to quantitatively evaluate error tradeoffs inherent in
different choices of cost functions and liquidity levels.
}

Our theoretical framework provides a meaningful way to quantitatively
evaluate the error tradeoffs inherent in different choices of cost
functions and liquidity levels. We find, for example, that to maintain
a fixed amount of bias, one should set the liquidity parameter $b$
proportional to a measure of the amount of cash that traders are
willing to spend. We also find that, although the LMSR maintains
coherent prices while IND does not, the two are equivalent up to a
factor of two in terms of the number of trades required to reach any
fixed accuracy, though LMSR has lower worst-case loss.

We have assumed that traders' beliefs are
individually coherent. Experimental evidence suggests that LMSR might
have additional informational advantages over IND when traders'
beliefs are incoherent or each trader is informed about only a subset
of events~\cite{OTL14}. We touch on this in
Appendix~\ref{app:incoherent}, but leave a full exploration of the
impact of different assumptions on trader beliefs to future work.

\newpage
\bibliographystyle{plainnat}
\bibliography{./refs}

%\end{document}

%
%\ifnips
%\else
%\iffalse
\newpage
\appendix
\section{Mathematical Background}
\label{app:math}

\subsection{Vectors, Matrices, Intervals}
\label{app:matrices}

Vectors are denoted by small boldface italics $\bu,\bprice,\dotsc$, matrices by large
italics $A,B,P,\dotsc$.
The norm $\norm{\cdot}$ denotes the standard Euclidean norm for vectors,
and the operator norm for matrices, i.e., $\norm{A}=\sup_{\bu:\:\norm{\bu}=1}\norm{A\bu}$. The \emph{range
of a matrix}, denoted $\range(\cdot)$, is the span of its columns. For a symmetric matrix $A$, its
pseudoinverse, denoted $A^+$, is the unique symmetric matrix with $\range(A^+)=\range(A)$ such that
$A^+A=AA^+=P$ where $P$ is the projection on $\range(A)$.

For symmetric matrices $A$ and $B$, we write $A\preceq B$ to denote that $B-A$ is positive-semidefinite. We
use the notation $A\within B\pm C$ to denote $B-C\preceq A\preceq B+C$.
For scalars, we similarly write $a\within b\pm c$ to mean $a\in[b-c,b+c]$.

We use $\lambda_{\min}(A)$
and $\lambda_{\max}(A)$ to denote the smallest and the largest positive eigenvalue of a symmetric
positive-semidefinite matrix. We also write $\lambda_{\min}(A,B)$ and $\lambda_{\max}(A,B)$ for the generalized
minimum and maximum eigenvalues defined as follows whenever $B\ne\zero$:
\[
   \lambda_{\min}(A,B)\coloneqq\min_{\bu\in\range(B)\wo\set{\zero}} \frac{\bu\trans A\bu}{\bu\trans B\bu}
\enspace,
\qquad
   \lambda_{\max}(A,B)\coloneqq\max_{\bu\in\range(B)\wo\set{\zero}} \frac{\bu\trans A\bu}{\bu\trans B\bu}
\enspace.
\]
The following formulas follow by substituting $\bv=B^{1/2}\bu$:
\begin{gather*}
   \lambda_{\min}(A,B)=\lambda_{\min}\bigParens{\,(B^{1/2})^+A(B^{1/2})^+\,}
\quad
   \text{if $\range(A)\subseteq\range(B)$,}
\\
   \lambda_{\max}(A,B)=\lambda_{\max}\bigParens{\,(B^{1/2})^+A(B^{1/2})^+\,}
\quad
   \text{if $\range(A)\supseteq\range(B)$.}
\end{gather*}
By Eq.~(180) of Petersen and Pedersen \citep{PetersenPe12}, we also have for any matrix $A$ (not necessarily square):
\begin{equation}
\label{eq:lambda:commute}
  \lambda_{\min}(AA\trans)=\lambda_{\min}(A\trans A)
\enspace,
\quad
  \lambda_{\max}(AA\trans)=\lambda_{\max}(A\trans A)
\enspace.
\end{equation}

The following two results relate $A$ and its pseudoinverse $A^+$:
\begin{proposition}
\label{prop:inv:ineq}
Let $A$ and $B$ be symmetric matrices such that $\range(A)=\range(B)$. Then $A\preceq B$ if and only if $B^+\preceq A^+$.
\end{proposition}
\begin{proof}
Assume that $\range(A)=\range(B)=\cL$. Then we have the following equivalences
\begin{align*}
A\preceq B
&\text{\ \ iff\ \ }
 \bu\trans A\bu\le\bu\trans B\bu
 \text{ for all $\bu\in\cL\wo\set{\zero}$}
\\
&\text{\ \ iff\ \ }
 1\le\min_{\bu\in\cL\wo\set{\zero}} \frac{\bu\trans B\bu}{\bu\trans A\bu}
\\
&\text{\ \ iff\ \ }
 1\le\lambda_{\min}(B,A)
   =\lambda_{\min}\bigParens{(A^{1/2})^+B(A^{1/2})^+}
\\
&\hphantom{\text{\ \ iff\ \ }
 1\le\lambda_{\min}(B,A)}
   =\lambda_{\min}\bigParens{B^{1/2} A^+ B^{1/2}}
   =\lambda_{\min}(A^+,B^+)
\tag*{(by Eq.~\ref{eq:lambda:commute})}
\\
&\text{\ \ iff\ \ }
 1\le\min_{\bu\in\cL\wo\set{\zero}} \frac{\bu\trans A^+\bu}{\bu\trans B^+\bu}
\\
\tag*{\qedhere}
&\text{\ \ iff\ \ }
 B^+\preceq A^+
\end{align*}
\end{proof}

\begin{proposition}[Blockwise Inversion]
\label{prop:woodbury}
Let $A,B,C\in\R^{d\times d}$, where $A$ and $C$ are symmetric, and let $S=C-B\trans A^+ B$. Assume the following conditions hold:
\begin{itemize}[noitemsep]
\item $\range(A)\cap\range(C)=\set{\zero}$.
\item $\range(B)\subseteq\range(A)$.
\item $\range(B\trans)\subseteq\range(C)$.
\item $\range(S)=\range(C)$.
\end{itemize}
Then
\[
  (A+B+B\trans+C)^+
  =A^+
   +
   (I-A^+ B)S^+(I-A^+ B)\trans
\enspace.
\]
\end{proposition}
\begin{proof}
This follows by the formula in Section 9.1.3 of Petersen and Pedersen \citep{PetersenPe12} for the block matrix
\[
\begin{pmatrix}
    U\trans AU & U\trans BV
\\  V\trans BU & V\trans CV
  \end{pmatrix}
\]
where $U$ and $V$ are matrices with an orthonormal basis of $\range(A)$ and $\range(C)$,
respectively.
\end{proof}

The final result of this section relates the optimization of a quadratic form with a matrix $A$
to the quadratic form with the matrix $A^+$:

\begin{proposition}
\label{prop:pseudo}
Let $A\in\R^{d\times d}$ be a symmetric positive-semidefinite matrix and let $\bu\in\range(A)$. Then
\[
  \min_{\bdelta\in\R^d} \Parens{\bdelta\trans\bu+\frac12\bdelta\trans A\bdelta}
  =-\frac12\bu\trans A^+\bu
\enspace.
\]
\end{proposition}
\begin{proof}
The result follows from convex conjugacy of quadratic convex functions, see page 108 of Rockafellar \citep{Rockafellar70}.
\begin{comment}
Let $P$ be the projection matrix on the range of $A$. Since $\range(A)=\range(A^{1/2})$, we can write $P=A^{1/2}(A^+)^{1/2}$.
Our derivation is based on the fact that $\bu\in\range(A)$ and so $\bu=A^{1/2}(A^+)^{1/2}\bu$:
%
\begin{align}
\notag
  \min_{\bdelta\in\R^d} \Parens{\bdelta\trans\bu+\frac12\bdelta\trans A\bdelta}
&
=
  \min_{\bdelta\in\R^d} \Parens{\bdelta\trans A^{1/2}(A^+)^{1/2}\bu+\frac12\bdelta\trans A^{1/2} A^{1/2}\bdelta}
\\
\label{eq:pseudo:1}
&
=
  \min_{\bdelta'\in\range(A)} \Parens{(\bdelta')\trans(A^+)^{1/2}\bu+\frac12\norm{\bdelta'}^2}
\\
\label{eq:pseudo:2}
&
=
  \BigParens{-(A^+)^{1/2}\bu}\trans(A^+)^{1/2}\bu+\frac12\Norm{(A^+)^{1/2}\bu}^2
\\
\notag
&
=
  -\frac12\bu\trans A^+\bu
\enspace.
\end{align}
%
\Eq{pseudo:2} follows by setting the gradient in \Eq{pseudo:1} equal to zero
and finding out that the objective is minimized at $\bdelta'=-(A^+)^{1/2}\bu$, which lies in $\range(A)$.
\end{comment}
\end{proof}

\subsection{Convex Analysis}
\label{app:convex}

If $f:\R^d\to\Rinf$ then the \emph{epigraph} of $f$ is the set of points in $\R^d\times\R$ that lie on or above the
graph of $f$. The function $f$ is called \emph{convex} if its epigraph is convex. It is called \emph{closed} if its epigraph is closed,
and \emph{proper} if it is not identically equal to $\infty$.
The \emph{effective domain} of $f$ is the set of
points where it is finite, denoted $\dom f$.
%\coloneqq\set{\bu\in\R^d:\:f(\bu)<\infty}$.
The convex hull of a set $S$, written
$\conv S$, is the smallest convex set containing $S$. For instance,
the simplex in $\R^d$ is the convex hull of the vectors of standard
basis.

An \emph{affine subspace} of
$\R^d$ is any set that can be written as $\cA=\set{\ba+\bu:\:\bu\in
  \cL}$ for some fixed vector $\ba$ where $\cL$ is a linear subspace of $\R^d$. We refer to $\cL$
as the linear space \emph{parallel to $\cA$}. The \emph{affine hull}
of a non-empty set $S$, denoted $\aff S$, is the smallest affine set that contains $S$. The \emph{relative interior} of a set $S$,
denoted $\ri S$, is the interior of $S$ under the topology of its affine hull. The \emph{relative boundary} of a set $S$ consists of points in the
closure of $S$
that are not in its relative interior, $(\cl S)\wo(\ri S)$. For instance, the affine hull of the simplex consists
of vectors $\bu$ such that $\bu\trans\one=1$,
where $\one$ is the all-ones vector. The parallel linear space is $\set{\bu:\:\bu\trans\one=0}$. The simplex
has an empty interior, but its relative interior consists of points $\set{\bu\in(0,1)^d:\:\bu\trans\one=1}$. The relative boundary of the simplex consists
of those points in the simplex which have at least one coordinate equal to zero.

For a convex $f:\R^d\to\Rinf$, the \emph{subdifferential} of $f$ at $\bu$
is defined as $\partial f(\bu)\coloneqq\set{\bv\in\R^d:\:f(\bu')\ge f(\bu)+\bv\trans(\bu'-\bu),\,\forall\bu'\in\R^d}$. Any convex $f:\R^d\to\Rinf$
is \emph{subdifferentiable}, i.e., its subdifferential is non-empty, at all points in $\ri\dom f$. If the subdifferential is a singleton,
it coincides with the gradient.
Given a proper convex function $f:\R^d\to\Rinf$, we define its \emph{convex conjugate} $f^*:\R^d\to\Rinf$ by
$f^*(\bprice)\coloneqq \sup_{\bu\in\R^d}[\bu\trans\bprice - f(\bu)]$. For a closed proper convex function $f$, its conjugate is also closed, proper and convex,
and the following statements are equivalent:
\[
\bprice\in\partial f(\bu)
\quad
\text{iff}
\quad
\bu\in\partial f^*(\bu)
\quad
\text{iff}
\quad
f(\bu)+f^*(\bprice)=\bu\trans\bprice
\enspace.
\]

We will use the following variant of a duality result
known as Fenchel's duality:
\begin{theorem}[Fenchel's duality]
\label{thm:fenchel:duality}
Let $f:\R^d\to\R$ and $g:\R^K\to\Rinf$ be closed proper convex functions and $A\in\R^{K\times d}$. Assume that
there exists some $\bu\in\R^d$ such that $A\bu\in\ri(\dom g)$ and some $\bprice\in\ri(\dom g^*)$ such that $A\trans\bprice\in\ri(\dom f^*)$.
%Further assume
%that $g$ is polyhedral
%and there exists $\mub\in\relint(\dom f^*)$ such that
%$\A^\top\mub\in\dom g^*$.
Then
\[
   \inf_{\bu\in\R^d}
   \Bracks{f(-\bu) + g(A\bu)}
   =
   \sup_{\bprice\in\R^K}
   \Bracks{-f^*(A\trans\bprice) - g^*(\bprice)}
\enspace
\]
and both the supremum and the infimum are attained. Vectors $\hat{\bu}$ and $\hat{\bprice}$ are their respective
solutions if and only if $A\trans\hat{\bprice}\in\partial f(-\hat{\bu})$ and $\hat{\bprice}\in\partial g(A\hat{\bu})$.
\end{theorem}
\begin{proof}
The results follows from Corollary 31.2.1 and Theorem 31.3 of Rockafellar \citep{Rockafellar70}.
\end{proof}

\subsection{\Convplus Functions}
\label{app:convplus}

Throughout the paper we work with functions that satisfy additional assumptions beyond convexity. We refer to them as
\emph{\convplus functions}. They have a close relationship to functions of \emph{Legendre type} (see Sec.~26 of \citet{Rockafellar70}).
Before we define \convplus functions,
we introduce the \emph{gradient space}, which plays a role in their structure.

\begin{definition}
\label{def:grad:space}
Let $f:\R^d\to\Rinf$ be differentiable on the interior of its domain $D\coloneqq\interior\dom f$.
Its \emph{gradient space}, denoted $\cG(f)$, is the linear space parallel to the affine hull
of the set of its gradients,
\[
  \cG(f)\coloneqq\Span\set{\nabla f(\bu)-\nabla f(\bu'):\:\bu,\bu'\in D}
\enspace.
\]
\end{definition}

\begin{comment}
\begin{definition}
\label{def:convplus}
We say that a function $f:\R^d\to\R$ is \emph{\convplus} if it satisfies the following conditions:
\begin{enumerate}[noitemsep]
\item $f$ is convex;
\item $f$ has continuous third derivatives;
\item $\range(\nabla^2 f(\bu))=\cG(f)$ for all $\bu\in\R^d$.
\end{enumerate}
\end{definition}
\end{comment}

\begin{definition}
\label{def:convplus}
Let $f:\R^d\to\Rinf$ and $D\coloneqq\interior\dom f$.
We say that $f$ is \emph{\convplus} if it satisfies the following conditions:
\begin{enumerate}[noitemsep]
\item $f$ is closed and convex.
\item $D$ is non-empty.
\item $f$ has continuous third derivatives on $D$.
\item $\range(\nabla^2 f(\bu))=\cG(f)$ for all $\bu\in D$.
\item $\lim_{t\to\infty}\norm{\nabla f(\bu_t)}=\infty$ whenever $\bu_1,\bu_2,\dotsc$ is
a sequence in $D$ converging to a boundary point of $D$.
\end{enumerate}
\end{definition}

\begin{proposition}
\label{prop:convplus}
Let $f:\R^d\to\Rinf$ be \convplus, $D\coloneqq\interior\dom f$. Let $P$ be the projection on $\cG(f)$ and $\ba$ the unique point in $\cA\cap\cG(f)^\perp$, where $\cA$ is the affine hull
of the set of gradients of $f$.
%\ba=(I-P)\nabla f(\bu_0)$ for some $\bu_0\in D$.
Then the following statements hold:
\begin{enumerate}[noitemsep]
\item $f(\bu)=f(P\bu)+\ba\trans\bu$.
%for any $\bu\in\R^d$.
\item $\nabla f(\bu)=\nabla f(P\bu)$.
%for any $\bu\in D$.
\item $\nabla^2 f(\bu)=\nabla^2 f(P\bu)$.
%for any $\bu\in D$.
\item $f$ is strictly convex on $D\cap\cG(f)$.
\item $\nabla f$ is one-to-one on $D\cap\cG(f)$.
\end{enumerate}
\end{proposition}
\begin{proof}
We prove the first statement by the Mean Value Theorem. First, since $f$ is differentiable on an open convex $D$, it
must be actually continuously differentiable on $D$ (by Corollary 25.5.1 of \citet{Rockafellar70}),
so the Mean Value Theorem can be applied on $D$. Let $\bu\in D$ and $\bv\in\cG(f)^\perp$. Then for any $\bu'=\bu+t\bv\in D$,
we have, for some $\bar{\bu}$ on the line segment connecting $\bu$ and $\bu'$,
\begin{equation}
\label{eq:cplus1:D}
  f(\bu')=f(\bu)+[\nabla f(\bar{\bu})]\trans t\bv=f(\bu)+t\ba\trans\bv
\end{equation}
where the second equality follows because $\bv\perp\cG(f)$. We argue that the entire line $\set{\bu+t\bv:t\in\R}$ must be contained in $D$. For contradiction, assume
it intersects the boundary of $D$ at $\bu^\star=\bu+t^\star\bv$, and say $t^\star>0$. Consider an increasing sequence $0=t_1,t_2,\dotsc$ converging to $t^\star$.
\Eq{cplus1:D} holds for $\bu'$ replaced by $\bu_i=\bu+t_i\bv$ as well as points $\tilde{\bu}_i=\tilde{\bu}+t_i\bv$,
where $\tilde{\bu}$ is in
a small enough neighborhood of $\bu$ along directions in $\cG(f)$. This means that
$\nabla f(\bu_i)=P\nabla f(\bu)+\ba=\nabla f(\bu)$. However, this is not possible for \convplus functions, because the norms of their gradients go to $\infty$ towards
the boundary. Similar argument holds for $t^\star<0$. This means that the entire line $\set{\bu+t\bv:t\in\R}$ must be in $D$. This holds for arbitrary $\bu\in D$,
so $D$ can be written as $D=D_0+\cG(f)^\perp$ where $D_0\subseteq\cG(f)$. \Eq{cplus1:D} now implies that statement (1) holds over $D$. Since
$f$ is closed, the statement also holds over $\dom f$, which then necessarily has form $\dom f=S_0+\cG(f)^\perp$ where $S_0\subseteq\cG(f)$. Therefore, statement (1) also
holds for $\bu\not\in\dom f$.

The remaining statements are more straightforward.
To prove the second statement, note that for any $\bu'\in D$, its gradient can be decomposed as $\nabla f(\bu')=P\nabla f(\bu')+\ba$. Thus,
$\nabla f(\bu)= P\nabla f(P\bu)+\ba=\nabla f(P\bu)$.
For the third statement, we have $\nabla^2 f(\bu)=P[\nabla^2 f(P\bu)]P=\nabla^2 f(P\bu)$, because
$\range(\nabla^2 f(P\bu))=\cG(f)$.
The fourth statement is equivalent to the fifth statement and they follow because $\range(\nabla^2 f(P\bu))=\cG(f)$.
\end{proof}

This proposition immediately implies that the Hessian of $f$ can be expressed as a function of the gradient of $f$. We will denote such
a function $H_f$:

\begin{proposition}
\label{prop:H:exists}
Let $f:\R^d\to\Rinf$ be \convplus and $D\coloneqq\interior\dom f$. Let $D'\coloneqq\set{\nabla f(\bu):\:\bu\in D}$ be the set of its gradients. Then there exists
a map $H_f:D'\to\R^{d\times d}$ such that $\nabla^2 f(\bu)=H_f(\nabla f(\bu))$ for all $\bu\in D$.
\end{proposition}
\begin{proof}
Let $D_0=D\cap\cG(f)$. \Prop{convplus} implies that $\nabla f$ is a bijection from $D_0$ to $D'$. Denoting its inverse from $D'$ to $D_0$ as $\bh$, we can then
define the map $H_f$ via $H_f(\bprice)=\nabla^2 f(\bh(\bprice))$. Now, for any $\bu\in D$, we have
\begin{equation}
\tag*{\qedhere}
   \nabla^2 f(\bu)=\nabla^2 f(P\bu)=\nabla^2 f(\bh(\nabla f(P\bu)))=H_f(\nabla f(P\bu))=H_f(\nabla f(\bu))
\enspace.
\end{equation}
\end{proof}

\begin{proposition}
\label{prop:convplus:conj}
Let $f:\R^d\to\Rinf$ be a \convplus function and let $\cA=\aff\dom f^*$. Then there exists a \convplus function $g:\R^d\to\Rinf$ such that
the following hold:
%for all $\bprice\in\ri\dom f^*$:
\begin{enumerate}[noitemsep]
%\item $g$ is \convplus.
\item $g$ agrees with $f^*$ on $\cA$.
\item $\cG(g)$ is parallel to $\cA$, or equivalently $\cG(g)=\cG(f)$.
\item For $\bprice\in\ri\dom f^*$: $\partial f^*(\bprice)=\nabla g(\bprice)+\cG(f)^\perp$.
\item For $\bprice\in\ri\dom f^*$: $\nabla^2 g(\bprice)=H_f^+(\bprice)$.
\end{enumerate}
\end{proposition}
\begin{proof}
We begin by representing $f$ via a function of Legendre type and then rely on the properties of such functions to obtain $g$.
In particular, we note that by \Prop{convplus}, the function $f$ is defined by its values on $\cG(f)$, except for a linear term,
and so we construct a function $f_0$ that is a transformation of $f$ on $\cG(f)$ and is of Legendre type.

To start, let $D=\interior\dom f$, and $D_0=D\cap\cG(f)$.
Assume that $\cG(f)$ is $d_0$ dimensional $d_0\le d$, and let $A\in\R^{d\times d_0}$ be a matrix whose columns form
an orthonormal basis of $\cG(f)$. Then $A\trans A=I_{d_0}$, i.e., $A\trans$ is the left inverse of $A$. The matrix
$A$ is an injective linear map from $\R^{d_0}\to\R^d$, but it is also a bijection from $\R^{d_0}$ to $\cG(f)$.
We also have $AA\trans=P$,
where $P$ is the projection on $\cG(f)$.
Define the function $f_0:\R^{d_0}\to\Rinf$ as
\[
  f_0(\bv)\coloneqq f(A\bv)
\enspace.
\]
Let $S=\interior\dom f_0$, so $S=A\trans D_0$.
We have
\[
  \nabla f_0(\bv)=A\trans\nabla f(A\bv),
\;\;
  \nabla^2 f_0(\bv)=A\trans\nabla^2 f(A\bv)A,
\;\;
  \nabla^3 f_0(\bv)[\cdot,\cdot,\cdot]=\nabla^3 f(A\bv)[A(\cdot),A(\cdot),A(\cdot)],
\]
so in particular $f_0$ has continuous third derivatives over $S$.
We next argue that $f_0$ is of \emph{Legendre type} in the sense of Rockafellar \citep{Rockafellar70}, page 258. For that we need to check that it satisfies the following conditions:
\begin{itemize}
\item[(a)]\emph{$S$ is non-empty.}\\
This follows, because $D_0$ is non-empty and $D_0\subseteq\cG(f)$. Now, $A\trans$, as a linear map, is a bijection from $\cG(f)$ to $\R^{d_0}$,
so the set $S=A\trans D_0$ is also non-empty.

\item[(b)]\emph{$f_0$ is differentiable throughout $S$.}\\
Similarly to the previous property, this holds, because $f$ is differentiable throughout $D_0$.

\item[(c)]\emph{$\lim_{t\to\infty}\norm{\nabla f_0(\bv_t)}=\infty$ whenever $\bv_1,\bv_2,\dotsc$ is
a sequence in $S$ converging to a boundary point of $S$.}\\
If $\bu_1,\bu_2,\dotsc$ is any sequence in $D_0$ converging to the relative boundary of $D_0$, then
this point is on the boundary of $D$ and therefore $\norm{\nabla f(\bu_t)}\to\infty$, because $f$ is \convplus. Now, suppose
we are given a sequence $\bv_1,\bv_2,\dotsc$ in $S$ converging to a boundary point of $S$. Then $\bu_t=A\bv_t$ is exactly a sequence in $D_0$ converging to the relative boundary of $D_0$,
so $\norm{\nabla f(A\bv_t)}\to\infty$. Since $\nabla f(A\bv_t)\in\cG(f)$ and the row space of $A\trans$ coincides with $\cG(f)$, we also have
$\norm{\nabla f_0(\bv_t)}=\norm{A\trans\nabla f(A\bv_t)}\to\infty$.

\item[(d)]\emph{$f_0$ is strictly convex on $S$.}\\
Since $f_0$ has a continuous Hessian on $S$, it suffices to show that its Hessian is full rank, i.e.,
%throughout $S$
its range is $\R^{d_0}$. For any $\bv\in S$, we have $\range(\nabla^2 f_0(\bv))=\range(A\trans\nabla^2 f(A\bv)A)$, which must be $\R^{d_0}$, because $\range(\nabla^2 f(A\bv))=\cG(f)$
and $A$ is a bijection from $\R^{d_0}$ to $\cG(f)$.
\end{itemize}

We now express $f$ in terms of $f_0$.
Note that by conjugacy, the affine hull of gradients of $f$ coincides with $\cA=\aff\dom f^*$.
By \Prop{convplus}, the function $f$ is defined by its values on $\cG(f)$, except for the linear term, described by the unique $\ba\in\cA\cap\cG^\perp$. We defined $f_0$ to
exactly represent $f$ on $\cG(f)$, so for any $\bu'\in\cG(f)$, we have $f_0(A\trans\bu')=f(AA\trans\bu')=f(\bu')$, because $AA\trans=P$. Thus, for any $\bu\in\R^d$ we have,
by \Prop{convplus},
\begin{equation}
\label{eq:f:as:f0}
  f(\bu)=f(P\bu)+\ba\trans\bu=f_0(A\trans P\bu)+\ba\trans\bu=f_0(A\trans\bu)+\ba\trans\bu
\enspace,
\end{equation}
because the row space of $A\trans$ coincides with $\cG(f)$.
This relationship between $f$ and $f_0$ implies the following relationship for their conjugates (by Theorems~12.3 and~16.3 of \citet{Rockafellar70}):
\begin{align}
  f^*(\bprice)
  &=\inf_{\by\in\R^d:\:A\by=\bprice-\ba} f^*_0(\by)
\enspace,
\end{align}
where the infimum of an empty set is $\infty$. The linear map $A$ is injective and $\range(A)=\cG(f)$, so the linear system
$A\bv=\bprice-\ba$ has a single solution $\bv=A\trans(\bprice-\ba)$ when $(\bprice-\ba)\in\cG(f)$, and
no solutions when $(\bprice-\ba)\not\in\cG(f)$. Therefore,
\[
  f^*(\bprice)=
\begin{cases}
  f^*_0\bigParens{A\trans(\bprice-\ba)}
&\text{if $\bprice\in\cA$,}
\\
  \infty
&\text{if $\bprice\not\in\cA$.}
\end{cases}
\]

Since $f_0$ is of Legendre type, so is its conjugate $f_0^*$ (see Theorem~26.5 of \citet{Rockafellar70}), which means that $f_0^*$ with $S^*\coloneqq\interior\dom f_0^*$
satisfies the properties (a)--(d).
The function $g$ is constructed as follows:
\[
  g(\bprice)\coloneqq f^*_0\bigParens{A\trans(\bprice-\ba)}=f^*_0(A\trans\bprice)
\enspace,
\]
where the equality follows because $\ba\in\cG(f)^\perp$.
This differs from the expression for $f^*$ in that it does not equal to $\infty$ outside $\cA$. Before we argue that $g$ is \convplus, we analyze
the gradient, Hessian and third derivatives of $f^*_0$. From the properties of the conjugates, we know that $\nabla f^*_0$ is the inverse of $\nabla f_0$, so for all $\by\in S^*$,
\[
  \nabla f^*_0(\by)=[\nabla f_0]^{-1}(\by)
\enspace.
\]
Since $\nabla f_0$ is continuously differentiable and its derivative (i.e., the Hessian of $f_0$) is an invertible matrix, the inverse of $\nabla f_0$ is also
continuously differentiable and its derivative (i.e., the Hessian of $f_0^*$) is
\begin{equation}
\label{eq:f0:inv}
  \nabla^2 f^*_0(\by)
  = \bigBracks{\nabla^2 f_0\bigParens{[\nabla f_0]^{-1}(\by)}}^{-1}
  = \bigBracks{\nabla^2 f_0\bigParens{\nabla f_0^*(\by)}}^{-1}
\enspace.
\end{equation}
In particular note that $\nabla^2 f^*_0(\by)$ is also an invertible matrix. We next
show that $f_0^*$ has a continuous third derivative,
by using the chain rule to argue that $\nabla^2 f^*_0$ is
continuously differentiable. This follows, because $\nabla^2 f^*_0$ is the composition of (i) the matrix inversion $[\cdot]^{-1}$,
(ii) the Hessian map $\nabla^2 f_0$, and (iii) the conjugate gradient $\nabla f_0^*$, and all of them are continuously differentiable
at points where the respective derivatives are taken; specifically, the matrix inversion is taken at an invertible
matrix $\nabla^2 f_0\bigParens{\nabla f_0^*(\by)}$, the Hessian at $\nabla f_0^*(\by)\in S$, and the conjugate gradient at
$\by\in S^*$.

Now we can show that $g$ is \convplus. From the definition of $g$,
%\[
%  g(\bprice)=f^*_0(A\trans\bprice),
%\]
we have
\begin{equation}
\label{eq:g:1:2}
  \nabla g(\bprice)=A\nabla f^*_0(A\trans\bprice),
\quad
  \nabla^2 g(\bprice)=A\nabla f^*_0(A\trans\bprice)A\trans,
\end{equation}
and
\begin{equation}
\label{eq:g:3}
  \nabla^3 g(\bprice)[\cdot,\cdot,\cdot]=\nabla^3 f^*_0(A\trans\bprice)[A\trans(\cdot),A\trans(\cdot),A\trans(\cdot)].
\end{equation}
Let $D^*\coloneqq\interior\dom g$. From the definition of $g$, we have $D^*=AS^*+\cG(f)^\perp$. Note that $\cG(f_0^*)=\R^{n_0}$, because $\dom f_0$ has a non-empty
interior. Therefore, by \Eq{g:1:2}, $\cG(g)=\range(A)=\cG(f)$. We next verify that $g$ satisfies properties (1)--(5) of \convexityplus, relying on
the properties (a)--(d) satisfied by $f_0^*$ and $S^*$:
\begin{enumerate}[noitemsep]
\item \emph{$g$ is closed and convex.}\\
This follows, because $f_0^*$ is closed and convex.
\item \emph{$D^*$ is non-empty.}\\
This follows, since $f_0^*$ satisfies (a), so $S^*$ is non-empty, and so is $D^*$.
\item \emph{$g$ has continuous third derivatives on $D^*$.}\\
This follows by \Eq{g:3}, because $f_0^*$ has continuous third derivatives on $S^*$.
\item \emph{$\range(\nabla^2 g(\bprice))=\cG(g)$ for all $\bprice\in D^*$.}\\
Since the Hessians of $f_0^*$ are full rank, \Eq{g:1:2} implies
that $\range(\nabla^2 g(\bprice))=\range(A)=\cG(g)$.
\item \emph{$\lim_{t\to\infty}\norm{\nabla g(\bprice_t)}=\infty$ whenever $\bprice_1,\bprice_2,\dotsc$ is
a sequence in $D^*$ converging to a boundary point of $D^*$.}\\
If $\bprice_t$ converges to a point on the border of $D^*=AS^*+\cG(f)^\perp$, then the sequence of points
$\by_t=A\trans\bprice_t$ converges to the border of $S^*$. Since $f_0^*$ satisfies property~(c), this means that
$\norm{\nabla f_0^*(\by_t)}\to\infty$. And since $A$ is injective, also
$\norm{\nabla g(\bprice_t)}=\norm{A\nabla f^*_0(\by_t)}\to\infty$.
\end{enumerate}

We now prove that $g$ has the properties stated in the theorem:
\begin{enumerate}[noitemsep]
%\item $g$ is \convplus.
\item \emph{$g$ agrees with $f^*$ on $\cA$.}\\
Immediate from the definition of $g$.
\item \emph{$\cG(g)$ is parallel to $\cA$, or equivalently $\cG(g)=\cG(f)$.}\\
As we already argued, \Eq{g:1:2} and the fact that $\cG(f_0^*)=\R^{n_0}$ imply that $\cG(g)=\range(A)=\cG(f)$.
\item \emph{For $\bprice\in\ri\dom f^*$: $\partial f^*(\bprice)=\nabla g(\bprice)+\cG(f)^\perp$.}\\
Note that $\bprice\in\cA$, so $(\bprice-\ba)\in\cG(f)$. The statement follows by the following chain of equivalences
\begin{align*}
\bu\in\partial f^*(\bprice)
&\text{\ \ iff\ \ }
   \nabla f(\bu)=\bprice
\\
&\text{\ \ iff\ \ }
   A\nabla f_0(A\trans\bu)+\ba=\bprice
\tag*{\text{(by Eq.~\ref{eq:f:as:f0})}}
\\
&\text{\ \ iff\ \ }
   \nabla f_0(A\trans\bu)=A\trans(\bprice-\ba)=A\trans\bprice
\tag*{\text{(because $(\bprice-\ba)\in\cG(f)=\range(A)$)}}
\\
&\text{\ \ iff\ \ }
   \nabla f_0^*(A\trans\bprice)=A\trans\bu
\\
&\text{\ \ iff\ \ }
   A\nabla f_0^*(A\trans\bprice)=AA\trans\bu=P\bu
\tag*{\text{(because $A$ is injective)}}
\\
&\text{\ \ iff\ \ }
   \nabla g(\bprice)=P\bu
\tag*{\text{(by Eq.~\ref{eq:g:1:2})}}
\\
&\text{\ \ iff\ \ }
   \bu\in\nabla g(\bprice)+\cG(f)^\perp
\tag*{\text{(because $\nabla g(\bprice)\in\range(A)=\cG(f)$)}}
\end{align*}
\item \emph{For $\bprice\in\ri\dom f^*$: $\nabla^2 g(\bprice)=H_f^+(\bprice)$.}\\
From \Eqs{g:1:2}{f0:inv}, we have
\begin{align*}
   \nabla^2 g(\bprice)
&= A[\nabla^2 f_0^*(A\trans\bprice)]A\trans
\\
&= A\bigBracks{\nabla^2 f_0\bigParens{\nabla f_0^*(A\trans\bprice)}}^{-1}A\trans
\enspace.
\end{align*}
By \Prop{H:exists}, $H_f(\bprice)=\nabla^2 f(\bu)$ for any $\bu$ such that $\nabla f(\bu)=\bprice$,
which is equivalent to $\bu\in\partial f^*(\bprice)$. Above, we have shown that $\nabla g(\bprice)\in\partial f^*(\bprice)$,
so $H_f(\bprice)=\nabla^2 f(\nabla g(\bprice))$. We continue the derivation of $H_f(\bprice)$ using
\Eqs{f:as:f0}{g:1:2}:
\begin{align*}
    H_f(\bprice)
&= \nabla^2 f(\nabla g(\bprice))
\\
&= A\bigBracks{\nabla^2 f_0\bigParens{A\trans\nabla g(\bprice)}}A\trans
\\
&= A\bigBracks{\nabla^2 f_0\bigParens{A\trans A\nabla f_0^*(A\trans\bprice)}}A\trans
\\
&= A\bigBracks{\nabla^2 f_0\bigParens{\nabla f_0^*(A\trans\bprice)}}A\trans
\enspace,
\end{align*}
where the last equation follows, because $A\trans A=I_{d_0}$. Since $\cG(g)=\cG(f)$,
the ranges of $\nabla^2 g(\bprice)$ and $H_f(\bprice)$ coincide with $\cG(f)$. From
the above derivations of $\nabla^2 g(\bprice)$ and $H_f(\bprice)$, we also have
\[
  [\nabla^2 g(\bprice)] H_f(\bprice) = H_f(\bprice)[\nabla^2 g(\bprice)] = AA\trans = P
\enspace,
\]
so indeed $\nabla^2 g(\bprice)=H_f^+(\bprice)$.\hfill\qedhere
\end{enumerate}
\end{proof}

Thanks to the continuity of third derivatives of $f$ and the continuity of second derivatives of $g$ from \Prop{convplus:conj}, we can
easily prove a local Lipschitz property for $H_f$:
\begin{proposition}
\label{prop:H:Lipschitz}
Let $f:\R^d\to\Rinf$ be \convplus and $D'$ be the set of its gradients (necessarily open within $\aff D'$).
Let $\bprice\in D'$ and let $B$ be a closed ball (in $\aff D'$) centered at $\bprice$ and fully contained in $D'$.
%In other words, for some $\delta>0$,
%\[
%  D'\supseteq B\coloneqq\bigSet{\bprice'\in\aff D':\:\norm{\bprice'-\bprice}\le\delta}
%\enspace.
%\]
Then there exists a constant $c$ such that for all $\bprice'\in B$
\[
  H_f(\bprice')\within\bigParens{1\pm c\norm{\bprice'-\bprice}}H_f(\bprice)
\enspace.
\]
\end{proposition}
\begin{proof}
Let $g$ be the function from \Prop{convplus:conj}. Similarly as we argued in the proof of \Prop{convplus:conj},
we can write
\[
   H_f(\bprice)=\nabla^2 f(\nabla g(\bprice))
\enspace,
\]
because $\nabla g(\bprice)\in\partial f^*(\bprice)$, and thus
$\nabla f(\nabla g(\bprice))=\bprice$. Since, $g$ has continuous second derivatives and $f$ has
continuous third derivatives, the Mean Value Theorem implies that $\nabla g$ is Lipschitz continuous on any compact
subset of $D'$, and $\nabla^2 f$ is Lipschitz continuous on any compact subset of $\interior\dom f$. Since $B$
is compact, and so is its image under $\nabla g$ by continuity of $\nabla g$, we obtain
that both $\nabla g$ and $\nabla^2 f$ are Lipschitz on required sets, and so $H_f$ is also Lipschitz within $B$, with
some constant $L$, i.e.,
\[
  \norm{H_f(\bprice')-H_f(\bprice)}\le L\norm{\bprice'-\bprice}
\enspace.
\]
Since $\range(H_f(\bprice'))=\range(H_f(\bprice))=\cG(f)$, this implies that
\[
  H_f(\bprice')\within H_f(\bprice)\pm L\norm{\bprice'-\bprice}P
\enspace,
\]
where $P$ is the projection on $\cG(f)$. Since $\range(H_f(\bprice))=\cG(f)$, we have $P\preceq\sigma^{-1}H_f(\bprice)$
where $\sigma=\lambda_{\min}(H_f(\bprice))$ is the smallest positive eigenvalue of $H_f(\bprice)$. Thus, we have
\begin{equation}
\tag*{\qedhere}
  H_f(\bprice')\within \bigParens{1\pm L\sigma^{-1}\norm{\bprice'-\bprice}} H_f(\bprice)
\enspace.
\end{equation}
\end{proof}

\begin{comment}
\subsection{Assumptions on the Cost Function}

Let $\cM\coloneqq\conv\set{\payoff(\omega):\:\omega\in\Omega}$ be the convex hull of the set of possible
payoff vectors, i.e., the set of all coherent payoff expectations. We assume that:
%
\begin{enumerate}[noitemsep]
%\item $g$ is \convplus.
\item $C:\R^K\to\R$ is \convplus;
\item $\cM\subseteq\dom C^*$;
\item $\nabla C$ is Lipschitz, i.e., there exists a constant $L_C$ such that $\norm{\nabla C(\bu')-\nabla C(\bu)}\le L_C\norm{\bu'-\bu}$ for all $\bu',\bu\in\R^K$.
\end{enumerate}

\mdcomment{Like Jenn said, we probably only need a suitable condition on the dual $C^*$. Specifically, it should suffice to require
that $C$ is differentiable and there exists a \convplus function $Q$ which agrees with $C^*$ on $\cM$ and also satisfies $\cG(Q)=\cG(C)$. We should then
obtain \Prop{H:exists} and \Prop{H:Lipschitz} with $H_C$ only defined over $\ri\cM$.}
\end{comment}

\subsection{Lipschitz Gradients and Strong Convexity}

In addition to (or instead of) \convexityplus, some of our results require Lipschitz gradients or, dually, strong convexity.
To be precise, we say that a differentiable function $f:\R^d\to\R$ has a Lipschitz gradient
if there exists a constant $L$ such that $\norm{\nabla f(\bu)-\nabla f(\bv)}\le L\norm{\bu-\bv}$ for all $\bu,\bv\in\R^d$.
%and $L$ is called the Lipschitz constant.
If $f$ is twice differentiable, it suffices to check that $\nabla^2 f(\bu)\preceq L I_d$ for all $\bu$,
where $I_d\in\R^{d\times d}$ is the identity.
%Both \LMSR and \IND have Lipschitz gradients since $H_C(\bprice)\preceq I_K$ in both cases.
%Similarly the log partition function has Lipschitz gradient, because
%$\nabla^2 T(\btheta)$ is the covariance matrix of a bounded
%random variable, namely $\payoff(\omega)$ when $\omega\sim p(\omega;\btheta)$.

We say that $f$ is strongly convex with the strong convexity constant $\sigma$
if
\[
  f(\bv)\ge f(\bu)+\bg\trans(\bv-\bu)+\frac12\sigma\norm{\bv-\bu}^2
\enspace,
\]
for all $\bv,\bu\in\R^d$ and $\bg\in\partial f(\bu)$. A standard convex analysis
result states that if $f:\R^d\to\R$ has a gradient with Lipschitz constant $L$ then $f^*$ is strongly convex with the strong convexity
constant $\sigma=1/L$ (see Prop.~12.60 of \citet{RockafellarWe09}). 
\section{Proofs and Additional Results for \Sec{exp}}
\label{app:exp}

\subsection{Proof of \Thm{agg:short}}

We prove a more explicit version of the theorem:
\begin{theorem}
\label{thm:agg:app}
Under the exponential trader model, $(\bar{\bbundle},\bar{\vec{c}}, \aggbprice)$ is a market-clearing equilibrium if and only if
\[
  \bar{\bbundle}\in\argmin_{\bbundle:\:\sum_{i=1}^N \bbundle_i=\zero} \sum_{i=1}^N F_i(-\bbundle_i)
 ,
\qquad
  \sum_{i=1}^N \bar{c}_i=0
 ,
\qquad
  \text{and}
\qquad
  \aggbprice=\nabla T(\bttheta_i-a_i\bar{\bbundle}_i) \ \ \forall i \in [N]
 .
\]
A market-clearing equilibrium always exists. Furthermore, for any market-clearing equilibrium, the equilibrium prices are unique solutions of the following dual problem:
\begin{equation}
%\label{eq:aggbdual}
\notag
  \aggbprice = \argmin_{\bprice\in\R^K} \BigBracks{\sum_i F_i^*(\bprice)}
\enspace.
\end{equation}
\end{theorem}

\begin{proof}
We first express the market-clearing equilibrium definition using the trader potential functions $F_i$ instead of trader utilities.
Since $[-F_i(-\bbundle_i)+c_i]$ is a monotone one-to-one transformation of the utility $U_i(\bbundle_i,c_i)$, we get the following equivalences
\begin{align*}
\zero \in \argmax_{\vdelta \in \R^K} U_i(\barbbundle_i + \vdelta,\,\barcash_i - \vdelta \cdot \aggbprice)
&\text{\ \ iff\ \ }
    \zero \in \argmin_{\vdelta \in \R^K} \bigBracks{F_i(-\barbbundle_i -\vdelta) -\barcash_i + \vdelta \cdot \aggbprice}
\\
&\text{\ \ iff\ \ }
    \nabla F_i(-\barbbundle_i)=\aggbprice
\enspace,
\end{align*}
where the last step follows by setting the gradient of the objective to zero at $\vdelta=\zero$. Thus,
we have that $(\barbbundle, \barbcash, \aggbprice)$ is a market-clearing equilibrium iff
\begin{equation}
\label{eq:agg:F}
  \sum_{i=1}^\nbuyers \barbbundle_i = \zero,
\quad
  \sum_{i=1}^N \barcash_i = 0,
\quad
   \nabla F_i(-\barbbundle_i)=\aggbprice
   \text{ for all $i \in [\nbuyers]$}.
\end{equation}

We now analyze the minimization of the potential $\sum_i F_i(-\bbundle_i)$ subject to the market clearing
constraint $\sum_i\bbundle_i=\zero$.
We express this constraint using the convex indicator function $\cind\set{\cdot}$,
which equals zero if its argument is true and $\infty$ when its false. We also introduce the matrix $A\in\R^{K\times NK}$
with the block structure $A\coloneqq(I_K\;I_K\;\dotsb\;I_K)$ where $I_K$ is the $K\times K$ identity matrix. Thus,
$A$ implements the summation over the blocks $\bbundle_i$, since $A\bbundle=\sum_i\bbundle_i$. With this notation,
the potential minimization problem can be written as
\begin{equation}
\label{eq:agg:primal}
   \min_{\bbundle\in\R^{NK}}
   \BigBracks{
      \underbrace{\sum_{i=1}^N F_i(-\bbundle_i)}_{f(-\bbundle)}
      +
      \underbrace{\vphantom{\sum_{i=1}^N}
                  \cind\set{A\bbundle=0}}_{g(A\bbundle)}
   }
\enspace,
\end{equation}
where we introduced the functions $f(\bbundle)=\sum_i F_i(\bbundle_i)$ and $g(\bstate)=\cind\set{\bstate=\zero}$ for $\bstate\in\R^K$. Now,
if certain conditions are satisfied, we can apply Fenchel's duality (\Thm{fenchel:duality}) and obtain that the value of the primal~\eqref{eq:agg:primal} equals
the value of the following dual problem
\[
   \max_{\bprice\in\R^K}
   \BigBracks{
     -f^*(A\trans\bprice)
     -g^*(\bprice)
   }
\enspace,
\]
which is equivalent to
\begin{equation}
\label{eq:agg:dual}
   \max_{\bprice\in\R^K}
   \BigBracks{
     -\sum_{i=1}^N F_i^*(\bprice)
   }
\enspace,
\end{equation}
because $g^*(\bprice)=0$ and $f^*(\by)=\sum_i F_i^*(\by_i)$, for $\by\in\R^{NK}$, so $f^*(A\trans\bprice)=\sum_i F_i^*(\bprice)$.
It remains to verify that the preconditions of \Thm{fenchel:duality} are satisfied.
First, we need to check that there exists $\bbundle$ such that $A\bbundle\in\ri(\dom g)$. Since $\ri(\dom g)=\set{\zero}$, the vector $\bbundle=\zero$ satisfies this.
We also need to check that there exist $\bprice$ such that $A\trans\bprice\in\ri(\dom f^*)$, which for our choices of $A$ and $f$ is equivalent to
$\bprice\in\ri(\dom F_i^*)$ for all $i\in[N]$. Since $\dom F_i^*=\dom T^*=\cM$, any $\bprice\in\cM$ satisfies this. Thus, conclusions of \Thm{fenchel:duality} hold.

The conclusions state that both the primal and the dual are attained, and $\hat{\bbundle}$ and $\hat{\bprice}$ are their solutions if and only if
$A\trans\hat{\bprice}=\nabla f(-\hat{\bbundle})$ and $\hat{\bprice}\in\partial g(A\hat{\bbundle})$. From the definitions of $A$ and $f$, the first condition is
equivalent to
\[
   \hat{\bprice}=\nabla F(-\hat{\bbundle}_i)=\nabla T(\bttheta_i-a_i\hat{\bbundle}_i)
\enspace.
\]
The second condition is by conjugacy equivalent to $A\hat{\bbundle}=\nabla g(\hat{\bprice})=\zero$, i.e.,
\[
  \sum_{i=1}^N \hat{\bbundle}_i = \zero
\enspace.
\]
This establishes that $\hat{\bbundle}$ and $\hat{\bprice}$ are solutions to the primal~\eqref{eq:agg:primal} and dual~\eqref{eq:agg:dual}, if and only
if they satisfy the conditions in~\eqref{eq:agg:F}, i.e., if and only if they form a market-clearing equilibrium. This proves the theorem except
for the uniqueness of the equilibrium prices $\aggbprice$. The uniqueness follows from the fact that $\aggbprice$ minimizes $\sum_i F_i^*(\bprice)$,
and the functions $F_i^*$ are strongly convex on their domain~$\cM$, which in turn follows because the functions $F_i$ have Lipschitz gradients (a property
they inherit from the log partition function $T$).
\end{proof}

\subsection{Proof of \Thm{eq:short}}

We prove a more explicit version of the theorem:
\begin{theorem}
\label{thm:eq:app}
Under the exponential trader model, $(\optbbundle,\vec{c}^\star, \eqbprice)$ is a market-maker equilibrium for cost function $C_b$ if and only if,
for the market state $\bstate^\star=\sum_{i=1}^\nbuyers \eqbbundle_i$,
\[
  \optbbundle\in\argmin_{\bbundle} F(\bbundle)
 ,
\quad
  \sum_{i=1}^N c^\star_i=C_b(\zero)-C_b(\bstate^\star)
%\BigParens{\sum_{i=1}^N\bbundle_i}
 ,
\quad
  \text{and}
\quad
  \eqbprice=\nabla C_b(\bstate^\star)=\nabla T(\bttheta_i-a_i\optbbundle_i) \ \forall i \in [N]
 .
\]
A market-maker equilibrium always exists. Furthermore, for any market-maker equilibrium, the equilibrium prices are unique solutions of the following dual problem:
\begin{equation}
%\label{eq:eqbdual}
\notag
  \eqbprice = \argmin_{\bprice\in\R^K} \BigBracks{\sum_i F_i^*(\bprice) + bC^*(\bprice)}
\enspace.
\end{equation}
\end{theorem}

\begin{proof}
We proceed similarly to the proof of \Thm{agg:app}. We first express the market-maker equilibrium definition using trader potentials
instead of trader utilities:
\begin{align*}
&\zero \in \argmax_{\vdelta \in \R^K} U_i\BigParens{
                       \eqbbundle_i + \vdelta,\,
                       \eqcash_i - C_\liq(\bstate^\star + \vdelta) + C_\liq(\bstate^\star)
                       }
\\
&\qquad\text{iff\ \ }
    \zero \in \argmin_{\vdelta \in \R^K} \bigBracks{
                       F_i(-\eqbbundle_i -\vdelta) -\eqcash_i + C_\liq(\bstate^\star + \vdelta) - C_\liq(\bstate^\star)
    }
\\
&\qquad\text{iff\ \ }
    \nabla F_i(-\eqbbundle_i)=\nabla C_\liq(\bstate^\star)
\enspace.
\end{align*}
Thus, we have that $(\eqbbundle, \eqbcash, \eqbprice)$ is a market-maker equilibrium iff,
for the market state $\bstate^\star=\sum_{i=1}^\nbuyers \eqbbundle_i$,
\begin{equation}
\label{eq:eq:F}
   \sum_{i=1}^\nbuyers \eqcash_i + C_\liq(\bstate^\star) - C_\liq(\zero) = 0,
\quad
   \eqbprice = \nabla C_\liq(\bstate^\star),
\quad
   \nabla F_i(-\barbbundle_i)=\nabla C_\liq(\bstate^\star)
   \text{ for all $i \in [\nbuyers]$}.
\end{equation}

We next use Fenchel's duality to
analyze minimization of the potential $F(\bbundle)=\sum_i F_i(-\bbundle_i)+C_b(\sum_i\bbundle_i)$.
We again define $A\coloneqq(I_K\;I_K\;\dotsb\;I_K)$ and
$f(\bbundle)\coloneqq\sum_i F_i(\bbundle_i)$, but in this case set $g(\bstate)=C_b(\bstate)$. Therefore,
by \Thm{fenchel:duality}, we obtain the following correspondence between the primal and the dual:
\begin{align}
\label{eq:eq:primal}
   \min_{\bbundle\in\R^{NK}}
   \BigBracks{
      \sum_{i=1}^N F_i(-\bbundle_i)
      +
      C_b(A\bbundle)
   }
&=
   \min_{\bbundle\in\R^{NK}}
   \BigBracks{
      f(-\bbundle)
      +
      g(A\bbundle)
   }
\\
\notag
&=
   \max_{\bprice\in\R^K}
   \BigBracks{
     -f^*(A\trans\bprice)
     -g^*(\bprice)
   }
\\
\label{eq:eq:dual}
&=
   \max_{\bprice\in\R^K}
   \BigBracks{
     -\sum_{i=1}^N F_i^*(\bprice)
     -bC^*(\bprice)
   }
\enspace,
\end{align}
where we used the fact that $g(\bbundle)=b C(\bbundle/b)$ and therefore $g^*(\bprice)=b C^*(\bbundle)$
(this is immediate from the definition of conjugate).
It remains to check the preconditions of \Thm{fenchel:duality}.
First, we need to check that there exists $\bbundle$ such that $A\bbundle\in\ri(\dom g)$. Since $\ri(\dom C_b)=\R^K$, this is vacuous and any vector $\bbundle$ satisfies this.
We also need to check that there exist $\bprice$ such that $A\trans\bprice\in\ri(\dom f^*)$, which is equivalent to
$\bprice\in\ri(\dom F_i^*)$ for all $i\in[N]$. Since $\dom F_i^*=\dom T^*=\cM$, any $\bprice\in\cM$ satisfies this. Thus, conclusions of \Thm{fenchel:duality} hold.

The conclusions state that both the primal and the dual are attained, and $\hat{\bbundle}$ and $\hat{\bprice}$ are their solutions if and only if
$A\trans\hat{\bprice}=\nabla f(-\hat{\bbundle})$ and $\hat{\bprice}=\nabla g(A\hat{\bbundle})$. As in the proof of \Thm{agg:app}, for our $A$ and $f$, the first condition is
equivalent to
\[
   \hat{\bprice}=\nabla F(-\hat{\bbundle}_i)=\nabla T(\bttheta_i-a_i\hat{\bbundle}_i)
\enspace.
\]
The second condition, $g(\bstate)=C_b(\bstate)$, is
\[
  \hat{\bprice}=\nabla C_b\bigParens{\sum_{i=1}^N \hat{\bbundle}_i}
\enspace.
\]
This establishes that $\hat{\bbundle}$ and $\hat{\bprice}$ are solutions to the primal~\eqref{eq:eq:primal} and dual~\eqref{eq:eq:dual}, if and only
if they satisfy the conditions in~\eqref{eq:eq:F}, i.e., if and only if they form a market-maker equilibrium. It remains to show
that $\eqbprice$ is unique. As before this follows by strong convexity of $F_i^*$ and the fact that $\eqbprice$ minimizes $\sum_i F_i^*(\bprice)+bC^*(\bprice)$.
\end{proof}

\subsection{Proof of \Thm{eqprice-char:NIPS}}

By \Thm{agg:short}, $\aggbprice = \argmin_{\bprice\in\R^K} [\sum_i F_i^*(\bprice)]$ and from the first-order optimality
\[
  \zero\in\partial\BigBracks{\sum_i F_i^*(\aggbprice)}
\enspace.
\]
Since $F_i(\bstate)=\frac{1}{a_i}T(\bttheta_i+a_i\bstate)$, the properties of the conjugates (Theorems~12.3 and~16.1 of~\citet{Rockafellar70})
yield
\[
  F_i^*(\bprice)=\frac{1}{a_i}\bigParens{ T^*(\bprice)-\bttheta_i\cdot\bprice }
\enspace.
\]
Thus,
\begin{align*}
 \sum_i F_i^*(\aggbprice)
&=
   \BigBracks{\sum_i 1/a_i}T^*(\bprice)
-
   \BigBracks{\sum_i \bttheta_i/a_i}\cdot\bprice
\\
&\implies
 \partial\BigBracks{\sum_i F_i^*(\aggbprice)}
=
   \BigBracks{\sum_i 1/a_i}\partial T^*(\bprice)
-
   \BigBracks{\sum_i \bttheta_i/a_i}
\enspace.
\end{align*}
Therefore, $\zero\in\partial[\sum_i F_i^*(\aggbprice)]$ iff
\[
   \frac{\sum_i \bttheta_i/a_i}
        {\sum_i 1/a_i}
   \in
   \partial T^*(\aggbprice)
\enspace,
\]
which is equivalent to
\[
 \aggbprice =
 \nabla \lmsr\Parens{
 \frac{\sum_i \bttheta_i/a_i}
      {\sum_i 1/a_i}
 }
 =
 \nabla \lmsr(\bar{\btheta})
 =
 \Ex{\bar{\btheta}}{\payoff(\omega)}
\enspace,
\]
where $\bar{\btheta}\coloneqq\bigParens{\sum_i \bttheta_i/a_i}/\bigParens{\sum_i 1/a_i}$
and the last equality follows from the properties of the log partition function.

\subsection{Sampling Error}
\label{app:sampling}

The market's forecasting ability is fundamentally limited by the information present
among the population of traders, and the traders' risk attitudes in
communicating their information via trades. In this appendix, we
quantify these sources of error by analyzing the discrepancy
$\norm{\bpricetruth - \aggbprice}$ between the true expected
security values and the market-clearing equilibrium prices.

The characterization of $\aggbprice$ in \Thm{eqprice-char:NIPS} reflects
two possible sources of error. First, the beliefs
$\bttheta_i$ are typically noisy signals of the ground truth. Second,
beliefs are weighted according to risk aversions $a_i$, which can skew
the prices. To formalize the latter
concept, we write
%
%$$\Neff = \Parens{\sum_i a_i^{-1}}^2 \bigg/ \Parens{\sum_i a_i^{-2}}$$
$$\Neff = \frac{\Parens{\sum_i a_i^{-1}}^2}{\Parens{\sum_i a_i^{-2}}}$$
to denote the \emph{effective sample size} of the weighted average.
When risk aversion coefficients are equal across agents, we have $\Neff = N$, and
when one agent has much smaller risk aversion than the others,
$\Neff \rightarrow 1$. As the next result shows, the magnitude of the
sampling error depends on the effective sample size as it relates to
the number of securities and the variance in trader beliefs.
\begin{theorem}\label{thm:sampling-error}
  Under the exponential trader model, assume that the beliefs
  $\bttheta_i$ are drawn independently for each trader $i \in [N]$
  with mean $\exNoBrack[\bttheta_i] = \btheta^\TRUE$ and covariance
  $\varNoBrack(\bttheta_i) \preceq \sigma^2 I_K$ for some $\sigma^2 \geq 0$.
  For any $\delta \in (0,1)$, the market-clearing prices $\aggbprice$
  satisfy, with probability at least $1 - \delta$,
$
\norm{\aggbprice - \bpricetruth} \leq
O\left(\, \sigma \sqrt{K/(\Neff\,\delta)} \,\right).
$
Furthermore, assuming that each $a_i$ lies in a bounded range
$[a_{\min}, a_{\max}]$ where $a_{\min}, a_{\max} > 0$,
we have that $\Neff \rightarrow \infty$ as $N \rightarrow \infty$.
\label{thm:Neff}
\end{theorem}
\begin{proof}
  We write $w_i = a_i^{-1} / (\sum_j a_j^{-1})$ for the weights in the
  average. Note that $\Neff = (\sum_i w_i^2)^{-1}$. By the fact that beliefs
  are independent, we have:
  $$
  \ex{\sum_{i=1}^N w_i \bttheta_i}
  %= \sum_{i=1}^N w_i \ex{\bttheta_i}
  = \btheta^\TRUE
  \enspace \mbox{and} \enspace
  \var{\sum_{i=1}^N w_i \bttheta_i} \preceq \Neff^{-1} \sigma^2 I_K.
  $$
% $$
%   \ex{\sum_{i=1}^N w_i \bttheta_i}
% = \sum_{i=1}^N \ex{w_i} \ex{\bttheta_i}
% = \btheta^\TRUE\, \ex{\sum_{i=1}^N w_i}
% = \btheta^\TRUE.
% $$
%   By the law of total variance, and the fact that beliefs are independent across traders, we have
%   \begin{eqnarray*}
%     \var{\sum_{i=1}^N w_i \bttheta_i} & = &
%     \ex{\var{\left. \sum_{i=1}^N w_i \bttheta_i \:\right|\: \bbw}}
%     + \var{\ex{\left. \sum_{i=1}^N w_i \bttheta_i \:\right|\: \bbw}} \\
%     & = &
%     \ex{\sum_{i=1}^N w_i^2 \var{\bttheta_i}}
%           + \var{\sum_{i=1}^N w_i \ex{\bttheta_i}} \\
%    & = &
%     \ex{\Neff^{-1}} \sigma^2 I_K .
%   \end{eqnarray*}
  By applying the multidimensional version of Chebyshev's inequality,
  we therefore have
$$
\mathrm{Pr}\Parens{\,
\left\|\, \sum_{i=1}^N w_i \bttheta_i - \btheta^\TRUE \,\right\| > t \,}
\leq \frac{K\sigma^2}{\Neff\, t^2}.
$$
The result then follows from the fact that
$\aggbprice=\nabla\lmsr(\sum_i w_i \bttheta_i)$ by \Thm{eqprice-char:NIPS}, the fact
that $\bpricetruth = \nabla \lmsr(\btheta^\TRUE)$, and the Lipschitz
continuity of $\nabla \lmsr$.

For the final claim, we have
\begin{equation}
\tag*{\qedhere}
\Neff = \frac{\Parens{\sum_i a_i^{-1}}^2}{\Parens{\sum_i a_i^{-2}}}
\geq \frac{\Parens{\sum_i a_{\max}^{-1}}^2}{\Parens{\sum_i
    a_{\min}^{-2}}}
= N (a_{\min} / a_{\max})^2.
\end{equation}
\end{proof}
%
%To interpret this result, note that the characterization of
%market-clearing prices in~\Eq{aggbdual:short} evokes standard
%statistical estimators such as the maximum likelihood estimate (after
%rescaling the objective by $1/N$).
%
Theorem~\ref{thm:sampling-error}
implies that as the number of traders grows large, the market prices
$\aggbprice$ converge to the ground truth $\bpricetruth$ in
probability. It is important to note that this relies on
\Thm{eqprice-char:NIPS}, which is an artifact of exponential
utility---for general utilities, market-clearing prices may not be
consistent in the statistical sense, and this extra discrepancy would
need to be quantified in the error decomposition.

For a finite number of traders, the bound in
Theorem~\ref{thm:sampling-error} increases with the belief variance
and the number of securities, as one would expect. It decreases with
the effective sample size: the information incorporated into
market-clearing prices improves when risk aversions are more uniform, and when the number of agents increases.

\ignore{
We stress that
the distributional assumptions of the theorem are only needed to bound
the sampling error, and are not required in the remainder of the
paper, which deals with sources of error once traders' beliefs have
been drawn and fixed.
}
%%% Local Variables:
%%% mode: latex
%%% TeX-master: "main"
%%% End:

\section{Proofs and Additional Results for \Sec{bias}}
\label{app:bias}

\subsection{Proof of \Thm{bias:global}}
\label{app:bias:global}

Since gradients $\nabla F_i$ are Lipschitz, the functions $F_i^*$ are \emph{strongly convex} over their domain, which is $\cM$. Therefore, their
sum $G(\bprice)\coloneqq\sum_i F_i^*(\bprice)$ is also strongly convex on $\cM$ with some strong convexity constant $\sigma$. Since $\aggbprice\in\ri\cM$,
$G$ is subdifferentiable at $\aggbprice$, and since $\aggbprice$ minimizes $G$, any element $\bg\in\partial G(\aggbprice)$ satisfies $\bg\trans(\bprice-\aggbprice)=0$
for all $\bprice\in\cM$. This together with strong convexity yields the lower bound
$
   \sum_i F_i^*(\bprice)\ge \sum_i F_i^*(\aggbprice)+\frac12\sigma\norm{\bprice-\aggbprice}^2
$.
At the same time, $C^*(\bprice)$ is bounded below by a linear function of the form $C^*(\aggbprice)+\bu\trans(\bprice-\aggbprice)$, because $C^*$ is subdifferentiable at $\aggbprice$ since $\aggbprice\in\ri\cM\subseteq\ri\dom C^*$.

Now from the optimality of $\eqbprice$ and the lower bounds on $\sum_i F_i^*(\bprice)$ and $C^*(\bprice)$, we have
\begin{align*}
  \sum_i F_i^*(\aggbprice) + \liq C^*(\aggbprice)
& \ge
  \sum_i F_i^*(\eqbprice) + \liq C^*(\eqbprice)
\\
& \ge
  \bigParens{\sum_i F_i^*(\aggbprice)+\frac12\sigma\norm{\eqbprice-\aggbprice}^2}
  + \liq\bigParens{C^*(\aggbprice)+\bu\trans(\eqbprice-\aggbprice)}
\\
& \implies
  \norm{\eqbprice-\aggbprice}^2
  \leq -\frac{2\liq}{\sigma} \bu\trans(\eqbprice-\aggbprice)
\\
& \implies
  \norm{\eqbprice-\aggbprice} \leq \frac{2\liq}{\sigma} \norm{\bu}
\enspace.
\end{align*}

\subsection{A Remark on Partial and Incoherent Beliefs}
\label{app:incoherent}

The proof of \Thm{bias:global} crucially relies on the fact that $\dom(\sum_i F_i^*)=\cM\subseteq\dom C^*$, i.e., that the cost function $C$ does not force any additional constraints on $\bprice$ beyond those already represented by the trader potentials $F_i$. This is natural in our setting, because trader utilities restrict the equilibrium prices to lie in the smallest set including all coherent price vectors, $\cM$. This property of the trader utilities means that the traders would be always willing to trade if the prices were outside the set $\cM$. If the trader utilities did not have this property, for instance, if each trader was interested in only a few securities, or their beliefs were incoherent, then this result might not hold. In such a setting, we might end up with $\aggbprice\not\in\dom C^*$. At best, we could then show that
\[
   \lim_{b \to 0} \eqbprice(\liq;C) = \myargmin_{\bprice\in\dom C^*} \Bracks{\sum_{i=1}^\nbuyers \dualobj_i(\bprice) }.
\]
This, of course, agrees with \Thm{bias:global} for our specific setting when the restriction to $\dom C^*$ creates no
additional constraints, because $\dom(\sum_i F_i^*)\subseteq\dom C^*$.

\subsection{Proof of \Thm{bias:local}}
\label{app:bias:local}

The proof will proceed by analyzing the Taylor expansion of the dual objective characterizing~$\eqbprice$.
%
%\begin{equation}
%\label{eq:bias:eq:app}
%       \eqbprice=\myargmin_{\bprice} \Bracks{\sum_{i=1}^N F_i^*(\bprice) + bC^*(\bprice)}
%\enspace.
%\end{equation}
%
However, the functions $F_i^*$ and $C^*$ might not be differentiable in the standard sense, because their domains
might have empty interiors (and only non-empty relative interiors). Fortunately, $F_i$ and $C$ are \convplus,
so by \Prop{convplus:conj} there exist \convplus functions $G_i$ and $R$ that coincide with $F^*_i$ and $C^*$ on
$\aff\dom F^*_i$ and $\aff\dom C^*$. These functions are three times continuously differentiable, which
is what we need to obtain the third order Taylor expansion.

Note that $\dom F^*_i=\cM$ for all $i\in[N]$. Let $\cA$ denote the affine hull of $\cM=\dom F^*_i$.
Since $\cM\subseteq\dom C^*$, the dual in~\Eq{eqbdual:short} is equivalent to
\begin{equation}
\label{eq:bias:eq:app}
       \eqbprice=\myargmin_{\bprice\in\cA} \Bracks{\sum_{i=1}^N G_i(\bprice) + bR(\bprice)}
\enspace.
\end{equation}
Note that $\aggbprice\in\ri\cM$, because by the definition $\aggbprice=\nabla T(\bar{\bstate})$ for some $\bar{\bstate}$ and the
gradients of $T$ are in $\ri\cM$. Thus,
functions $G_i$ and $R$ are differentiable and have Hessians at $\aggbprice$ (by \convexityplus). We apply
the Taylor expansion at $\aggbprice$ to analyze the value of the objective at $\eqbprice$.
Let $G(\bprice)\coloneqq\sum_i G_i(\bprice)$. By Mean Value Theorem, we have
\begin{align*}
     \nabla G(\eqbprice)&=\nabla G(\aggbprice)+\nabla^2 G(\bprice_G)(\aggbprice-\eqbprice)
\\
     \nabla R(\eqbprice)&=\nabla R(\aggbprice)+\nabla^2 R(\bprice_R)(\aggbprice-\eqbprice)
\end{align*}
for some $\bprice_G$ and $\bprice_R$ on the line segment connecting $\eqbprice$ with $\aggbprice$.
By \Thm{bias:global}, $\norm{\eqbprice-\aggbprice}=O(b)$ as $b\to 0$, and thus also $\norm{\bprice_G-\aggbprice}=O(b)$ and $\norm{\bprice_R-\aggbprice}=O(b)$.
By the continuity of third derivatives of $G$ and $R$ in the neighborhood of $\aggbprice$, the Hessians of $G$ and $R$ are Lipschitz in some neighborhood of $\aggbprice$,
which means that
\begin{align*}
  \nabla^2 G(\bprice_G)&=\nabla^2 G(\aggbprice) + \Delta_G
\\
  \nabla^2 R(\bprice_R)&=\nabla^2 R(\aggbprice) + \Delta_R
\end{align*}
where $\Delta_G$ and $\Delta_R$ are matrices with
$\norm{\Delta_G}=O(\norm{\bprice_G-\aggbprice})=O(b)$ and $\norm{\Delta_R}=O(\norm{\bprice_R-\aggbprice})=O(b)$.

We next calculate Hessians $\nabla^2 G(\aggbprice)$ and $\nabla^2 R(\aggbprice)$. First, note that
\[
  \nabla^2 G(\bprice)=\sum_i \nabla^2 G_i(\bprice)
\enspace,
\]
and by \Prop{convplus:conj}, we have
\[
  \nabla^2 G_i(\bprice)=H_{F_i}^+(\bprice)=(1/a_i) H_T^+(\bprice)
\enspace,
\]
where the last equality follows, because $H_{F_i}(\bprice)=a_i H_T(\bprice)$ from the definition of $F_i$. Thus,
\[
  \nabla^2 G(\bprice)=\sum_i (1/a_i) H_T^+(\bprice)=(N/\bar{a}) H_T^+(\bprice)
\enspace.
\]
We also have
\[
  \nabla^2 R(\bprice)
  = H_C^+(\bprice)
\enspace.
\]

By \Prop{convplus:conj}, the affine space $\cA$ is parallel to $\cG(T)$.
%Let $P$ be the projection on $\cG(T)$.
From the optimality of $\eqbprice$ in \eqref{eq:bias:eq:app}, we have $\bigParens{ \nabla G(\eqbprice) + \liq \nabla R(\eqbprice) } \perp \cG(T)$.
Thus, writing $P$ for the projection on $\cG(T)$, we obtain
\begin{align}
\notag
\zero
&=P \BigParens{ \nabla G(\eqbprice) + \liq \nabla R(\eqbprice) }
\\
\label{eq:bias:app1}
&= P \nabla G(\aggbprice) + P \nabla^2 G(\aggbprice)(\eqbprice - \aggbprice) + \underbrace{P\Delta_G(\eqbprice - \aggbprice)}_{\veps_G}
   + b P \nabla R(\aggbprice) + \underbrace{b P \bigBracks{\nabla^2 R(\aggbprice)+\Delta_R}(\eqbprice - \aggbprice)}_{\veps_R}
\\
\label{eq:bias:app2}
&=  \nabla^2 G(\aggbprice)(\eqbprice - \aggbprice)
   + b P \nabla R(\aggbprice) + \veps_G+\veps_R
\enspace.
\end{align}
In \Eq{bias:app1}, the terms $\veps_G$ and $\veps_R$ have norms $O(b^2)$, because $\norm{\eqbprice-\aggbprice}$, $\norm{\Delta_G}$ and $\norm{\Delta_R}$ are all at most $O(b)$.
In \Eq{bias:app2}, we use that $P\nabla G(\aggbprice)=\zero$ by optimality of $\aggbprice$. We also use that $P\nabla^2 G(\aggbprice)=\nabla^2 G(\aggbprice)$, because $\range(\nabla^2 G(\aggbprice))=\cG(T)$.

Since $(\eqbprice - \aggbprice)\in\cG(T)$, multiplying \Eq{bias:app2} by $[\nabla^2 G(\aggbprice)]^+$, we obtain
\begin{align*}
  \eqbprice - \aggbprice
  +
  \underbrace{[\nabla^2 G(\aggbprice)]^+(\veps_G+\veps_R)}_{\veps}
&=
  -b[\nabla^2 G(\aggbprice)]^+ P\nabla R(\aggbprice)
\\[-6pt]
&=
  -b(\bar{a}/N) H_T(\aggbprice) P\nabla R(\aggbprice)
\\
&=
  -b(\bar{a}/N) H_T(\aggbprice) \partial C^*(\aggbprice)
\enspace,
\end{align*}
where the last step follows, because $H_T(\aggbprice)P=H_T(\aggbprice)$, and by \Prop{convplus:conj}, $\partial C^*(\aggbprice)=\nabla R(\aggbprice)+\cG^\perp(C)$, and $\cG^\perp(C)\subseteq\cG^\perp(T)$.
The theorem now follows by noting that $\norm{\veps}=O(b^2)$.

\subsection{Proof of \Thm{bias:two}}
\label{app:bias:two}

We prove a slightly stronger statement that holds not only for the bias measured under the Euclidean norm,
but also when it is measured by the KL divergence. We use a more compact notation $\bprice^\IND(\liq)$ and $\bprice^\LMSR(\liq)$ for
$\eqbprice(\liq;\IND)$ and $\eqbprice(\liq;\LMSR)$.

\begin{theorem}
For any $\aggbprice$ there exist $\eta\in[1,2]$ and $\etaKL\in[1,2]$ such that for all $b$
\begin{align}
\label{eq:two:LS}
  \Norm{\bprice^\IND(\liq/\eta) - \aggbprice}
  &=
  \Norm{\bprice^\LMSR(\liq) - \aggbprice}
  +O(b^2)
\enspace,
\\
\label{eq:two:KL:1}
  \bigKL{\bprice^\IND(\liq/\etaKL)}{\aggbprice}
  &=
  \bigKL{\bprice^\LMSR(\liq)}{\aggbprice}
  +O(b^3)
\enspace,
\\
\label{eq:two:KL:2}
  \bigKL{\aggbprice}{\bprice^\IND(\liq/\etaKL)}
  &=
  \bigKL{\aggbprice}{\bprice^\LMSR(\liq)}
  +O(b^3)
\enspace.
\end{align}
For these same $\eta$ and $\etaKL$, we also have, for all $b$,
\begin{align}
\label{eq:two:LS:equiv}
  \norm{\bprice^\IND(\liq) - \aggbprice}
  &=
  \eta\norm{\bprice^\LMSR(\liq) - \aggbprice}
\pm O(b^2)
\enspace,
\\
\label{eq:two:KL:1:equiv}
  \bigKL{\bprice^\IND(\liq)}{\aggbprice}
  &=
  \etaKL^2\bigKL{\bprice^\LMSR(\liq)}{\aggbprice}
  +O(b^3)
\enspace,
\\
\label{eq:two:KL:2:equiv}
  \bigKL{\aggbprice}{\bprice^\IND(\liq)}
  &=
  \etaKL^2\bigKL{\aggbprice}{\bprice^\LMSR(\liq)}
  +O(b^3)
\enspace.
\end{align}
\end{theorem}
\begin{proof}
Without loss of generality, we assume that
the coordinates of $\aggbprice$ are sorted in the non-increasing order, i.e.,
$\aggprice{1}\ge\aggprice{2}\ge\dotsb\ge\aggprice{K}$.

Our proof is based on \Thm{bias:local}, which states that
\begin{equation}
\label{eq:bias:remind}
  \eqbprice(\liq;C)-\aggbprice=b\Parens{-\frac{\bar{a}}{N} H_T(\aggbprice)\partial C^*(\aggbprice)}
  +\veps_b
\enspace,
\quad
  \text{where $\norm{\veps_b}=O(b^2)$.}
\end{equation}
Let $H\coloneqq H_T(\aggbprice)=\parens{\diag_{k\in[K]}\aggprice{k}}-\aggbprice\aggbprice\trans$, and let $\bstate^\LMSR$ and $\bstate^\IND$ denote arbitrary elements of $\partial C^*(\aggbprice)$
for the costs $\LMSR$ and $\IND$, respectively. Then \Eq{bias:remind} yields
\[
  \bigNorm{\bprice^\IND(\liq)-\aggbprice}=b\frac{\bar{a}}{N}\cdot\norm{H\bstate^\IND} + O(b^2)
\enspace,
\quad
  \bigNorm{\bprice^\LMSR(\liq)-\aggbprice}=b\frac{\bar{a}}{N}\cdot\norm{H\bstate^\LMSR} + O(b^2)
\enspace,
\]
so \Eqs{two:LS}{two:LS:equiv} follow by setting
\begin{equation}
\label{eq:two:def:eta}
  \eta\coloneqq\frac{\norm{H\bstate^\IND}}{\norm{H\bstate^\LMSR}}=\Parens{\frac{(\bstate^\IND)\trans H^2\bstate^\IND}{(\bstate^\LMSR)\trans H^2\bstate^\LMSR}}^{1/2}
\end{equation}
and it remains to prove that $\eta\in[1,2]$. (We do so below.)

For KL divergence results, we begin by using the fact that all entries of $\aggbprice$ are positive so both $f_1(\bprice)\coloneqq\KL{\bprice}{\aggbprice}$ and
$f_2(\bprice)\coloneqq\KL{\aggbprice}{\bprice}$ have bounded and continuous third derivatives in a sufficiently small neighborhood of $\aggbprice$. Therefore,
by Taylor's theorem, we obtain in this neighborhood
\begin{align*}
  \KL{\bprice}{\aggbprice}=f_1(\bprice)
  &=\underbrace{f_1(\aggbprice)
   +\nabla f_1(\aggbprice)\trans(\bprice-\aggbprice)}_{=0}
   +(\bprice-\aggbprice)\trans\nabla^2 f_1(\aggbprice)\trans(\bprice-\aggbprice)
   +O(\norm{\bprice-\aggbprice}^3)
%\\&
%  =
%   (\bprice-\aggbprice)\trans\nabla^2 f_1(\aggbprice)\trans(\bprice-\aggbprice)
%   +O(\norm{\bprice-\aggbprice}^3)
\\
  \KL{\aggbprice}{\bprice}=f_2(\bprice)
  &=\underbrace{f_2(\aggbprice)
   +\nabla f_2(\aggbprice)\trans(\bprice-\aggbprice)}_{=0}
   +(\bprice-\aggbprice)\trans\nabla^2 f_2(\aggbprice)\trans(\bprice-\aggbprice)
   +O(\norm{\bprice-\aggbprice}^3)
\enspace.
%\\&
%  =
%   (\bprice-\aggbprice)\trans\nabla^2 f_2(\aggbprice)\trans(\bprice-\aggbprice)
%   +O(\norm{\bprice-\aggbprice}^3)
\end{align*}
By direct calculation, $\nabla^2 f_1(\aggbprice)=\nabla^2 f_2(\aggbprice)=\diag_{k\in[K]} (\aggprice{k})^{-1}\eqqcolon M$.
Now, by \Thm{bias:global}, we have $\norm{\eqbprice(\liq;C)-\aggbprice}=O(b)$, and so we obtain
\begin{align}
\label{eq:KL:1:Taylor}
  \bigKL{\eqbprice(\liq;C)}{\aggbprice}
  &=\bigParens{\eqbprice(\liq;C)-\aggbprice}\trans
    M
    \bigParens{\eqbprice(\liq;C)-\aggbprice}
    +O(b^3)
\\
\label{eq:KL:2:Taylor}
  \bigKL{\aggbprice}{\eqbprice(\liq;C)}
  &=\bigParens{\eqbprice(\liq;C)-\aggbprice}\trans
    M
    \bigParens{\eqbprice(\liq;C)-\aggbprice}
    +O(b^3)
\enspace.
\end{align}
We next invoke \Eq{bias:remind}, but before we do so, note that since $H=\parens{\diag_{k\in[K]}\aggprice{k}}-\aggbprice\aggbprice\trans$
and $M=\diag_{k\in[K]} (\aggprice{k})^{-1}$, we have $HMH=H$. Now, invoking \Eq{bias:remind} and plugging it into \Eq{KL:1:Taylor},
we obtain
\begin{align*}
  \bigKL{\bprice^\IND(\liq)}{\aggbprice}
&=b^2\frac{\bar{a}^2}{N^2}\cdot(\bstate^\IND)\trans H\bstate^\IND + O(b^3)
\\
  \bigKL{\bprice^\LMSR(\liq)}{\aggbprice}
&=b^2\frac{\bar{a}^2}{N^2}\cdot(\bstate^\LMSR)\trans H\bstate^\LMSR + O(b^3)
\end{align*}
and similarly for $\KL{\aggbprice}{\cdot}$. Therefore Eqs.~\eqref{eq:two:KL:1}, \eqref{eq:two:KL:2}, \eqref{eq:two:KL:1:equiv} and~\eqref{eq:two:KL:2:equiv} follow by setting
\begin{equation}
\label{eq:two:def:eta:KL}
  \etaKL\coloneqq\Parens{
     \frac{
      (\bstate^\IND)\trans H\bstate^\IND}{
      (\bstate^\LMSR)\trans H\bstate^\LMSR}
  }^{1/2}
\end{equation}
and it remains to prove that $\etaKL\in[1,2]$.

In the remainder of the proof, we show that $\eta$ and $\etaKL$ defined in \Eqs{two:def:eta}{two:def:eta:KL} are in $[1,2]$. We proceed
by \Lem{two:bias} (see below), which shows that $1\le (\bv\trans H\bv)/(\bstate\trans H\bstate)\le 4$ and $1\le(\bv H^2\bv)/(\bstate H^2\bstate)\le 4$
for any sorted vectors $\bstate$ and $\bv$, whose differences between consecutive coordinates are within a factor-of-two of each other.
We only need to show that $\bstate=\bstate^\LMSR$ and $\bv=\bstate^\IND$ satisfy this condition.

Recall that $\bstate^\LMSR$ and $\bstate^\IND$ can be chosen as arbitrary elements of $\partial C^*(\aggbprice)$
for the costs $\LMSR$ and $\IND$. From the properties of conjugates, $\bstate\in\partial C^*(\aggbprice)$ iff $\nabla C(\bstate)=\aggbprice$, so
we can obtain $\bstate^\LMSR$ and $\bstate^\IND$ by inverting the gradients of \LMSR and \IND:
\[
\textstyle
   s^\LMSR_k=\log \aggprice{k}
\enspace,
\quad
   s^\IND_k=\log\Parens{\frac{\aggprice{k}}{1-\aggprice{k}}}
\quad
\text{for all $k\in[K]$.}
\]
Note that both $s^\LMSR_k$ and $s^\IND_k$ are monotone transformations of $\aggprice{k}$, and since $\aggbprice$ is sorted,
so must be $\bstate^\LMSR$ and $\bstate^\IND$. We next show that the differences between the consecutive
coordinates of $\bstate^\LMSR$ and $\bstate^\IND$ are within a factor two of each other. For any $k\in[K-1]$, we have
\[
  s^\LMSR_k-s^\LMSR_{k+1}
  =
  \log\Parens{\frac{ \aggprice{k} }{ \aggprice{k+1} }}
\]
and we also have
\begin{align*}
  s^\IND_k-s^\IND_{k+1}
  &=
  \log\Parens{\frac{ \aggprice{k} }{ \aggprice{k+1} }\cdot \frac{ 1-\aggprice{k+1} }{ 1-\aggprice{k} }}
\\
  &=
  \log\Parens{\frac{ \aggprice{k} }{ \aggprice{k+1} }\cdot \frac{ c_k+\aggprice{k} }{ c_k+\aggprice{k+1} }}
  =
  \log\Parens{\frac{ \aggprice{k} }{ \aggprice{k+1} }} + \log\Parens{\frac{ c_k+\aggprice{k} }{ c_k+\aggprice{k+1} }}
\enspace,
\end{align*}
where $c_k\coloneqq 1-\aggprice{k}-\aggprice{k+1}\ge 0$. Since $\aggprice{k}\ge\aggprice{k+1}$, we therefore have
\[
  0\le\log\Parens{\frac{ c_k+\aggprice{k} }{ c_k+\aggprice{k+1} }}
   \le\log\Parens{\frac{ \aggprice{k} }{ \aggprice{k+1} }}
\]
and therefore
\[
  s^\LMSR_k-s^\LMSR_{k+1}
  \le
  s^\IND_k-s^\IND_{k+1}
  \le
  2(s^\LMSR_k-s^\LMSR_{k+1})
\enspace.
\]
Thus, \Lem{two:bias} with $\bstate=\bstate^\LMSR$ and $\bv=\bstate^\IND$ and $H=H_C(\aggbprice)$ proves that indeed
$\eta\in[1,2]$ and $\etaKL\in[1,2]$.
\end{proof}

\begin{lemma}
\label{lem:two:bias}
Let $\bprice\in\R^K$ be a sorted probability vector, i.e.,
$\price_1\ge\price_2\ge\dotsb\ge\price_K$, and let $H=\parens{\diag_{k\in[K]}\price_k}-\bprice\bprice\trans$ be the covariance
matrix of the associated multinomial distribution. Let $\bs$, $\bv$ be sorted vectors in $\R^K$, i.e.,
$s_1\ge\dotsb\ge s_K$ and $v_1\ge\dotsb\ge v_K$, such that $s_k-s_{k+1}\le v_k-v_{k+1}\le 2(s_k-s_{k+1})$. Then the following
two statements hold
\[
  \bs\trans H\bs
  \le \bv\trans H\bv
  \le 4\cdot\bs\trans H\bs
\enspace,
\quad
  \bs\trans H^2\bs
  \le \bv\trans H^2\bv
  \le 4\cdot\bs\trans H^2\bs
\enspace.
\]
\end{lemma}
\begin{proof}
The proof proceeds in several steps.

\paragraph{Step 1: Decomposition of $\bs$ and $\bv$ into an alternate basis.}

We begin by rewriting $\bs$ and $\bv$ in the basis consisting of vectors $\bz_k\in\R^K$, for $k=1,\dotsc,K$, where
each $\bz_k$ has ones on positions $1$ through $k$, and zeros on the remaining positions, i.e., $z_{kj}=\one\braces{j\le k}$.
Let
\[
   a_k\coloneqq s_k-s_{k+1}
\quad
\text{and}
\quad
   b_k\coloneqq v_k-v_{k+1}
\quad
\text{for }
   k=1,\dotsc,K-1.
\]
The vectors $\bs$ and $\bv$ can then be written as
\begin{equation}
\label{eq:sv:basis}
   \bs = s_K\bz_K+\sum_{k=1}^{K-1} a_k\bz_k
\enspace,
\quad
   \bv = v_K\bz_K+\sum_{k=1}^{K-1} b_k\bz_k
\end{equation}
where
\begin{equation}
\label{eq:ak:bk}
   0\le a_k\le b_k\le 2 a_k
\enspace.
\end{equation}
%by our assumptions on $\bs$ and $\bv$.
From the definition of $H$, we
have $H\bz_K=0$ and so
\begin{equation}
\label{eq:drop:zK}
  \bs\trans H\bs
  = (\bs')\trans H\bs'
\enspace,
\quad
  \bv\trans H\bv = (\bv')\trans H\bv'
\enspace,
\quad
  \bs\trans H^2\bs = (\bs')\trans H^2\bs'
\enspace,
\quad
  \bv\trans H^2\bv = (\bv')\trans H^2\bv'
\enspace,
\end{equation}
where $\bs'$ and $\bv'$ exclude the basis element $\bz_K$, i.e.,
\begin{equation}
\label{eq:svprime:basis}
   \bs' = \sum_{k=1}^{K-1} a_k\bz_k
\enspace,
\quad
   \bv' = \sum_{k=1}^{K-1} b_k\bz_k
\enspace.
\end{equation}
Therefore, we have the decomposition
\begin{equation}
\label{eq:sHs}
  \bs\trans H\bs
  = (\bs')\trans H\bs'
  = \Parens{\sum_{k=1}^{K-1} a_k\bz_k}\trans H\Parens{\sum_{\ell=1}^{K-1} a_\ell\bz_\ell}
  = \sum_{k,\ell\in[K-1]} a_k a_\ell (\bz_k\trans H\bz_\ell\phtrans)
\enspace,
\end{equation}
and similar decompositions for $\bv\trans H\bv$, $\bs\trans H^2\bs$, and $\bv\trans H^2\bv$.

In the remainder of the proof we show that for all $k,\ell\in[K-1]$, we have
$\bz_k\trans H\bz_\ell\phtrans\ge 0$ and $\bz_k\trans H^2\bz_\ell\phtrans\ge 0$, which together with the decomposition
in \Eq{sHs} and similar decompositions implied by \Eq{drop:zK} will imply the theorem, thanks to the fact
that $0\le a_k\le b_k\le 2 a_k$.

\paragraph{Step 2: $\bz_k\trans H\bz_\ell\ge 0$.}

Fix $k,\ell\in[K]$ and assume $k\le\ell$. Then
%First analyze $\bz_k\trans H\bz_\ell\phtrans$:
%
\begin{align}
\notag
  \bz_k\trans H \bz_\ell\phtrans
&=
  \sum_{j\in[K]} \mu_j z_{kj}z_{\ell j}
  -
  (\bz_k\trans\bprice)(\bprice\trans\bz_\ell\phtrans)
\\
\notag
&=
  \underbrace{
  \bigParens{\sum_{j\le k}\mu_j}
  }_{\mu_k^{[1]}}
  -
  \underbrace{
  \bigParens{\sum_{j\le k}\mu_j}
  }_{\mu_k^{[1]}}
  \underbrace{
  \bigParens{\sum_{j\le\ell}\mu_j}
  }_{\mu_\ell^{[1]}}
\\
\label{eq:two:0}
&=
  \mu_k^{[1]}\bigParens{1-\mu_\ell^{[1]}}
%\\
%&
\ge 0
\enspace.
\end{align}
Above, we introduced notation for partial sums
\[
  \mu_k^{[d]}\coloneqq\sum_{j=1}^k \mu_j^d
\enspace,
\]
and used the fact that $\mu_\ell^{[1]}\le 1$, because entries of $\bprice$ are non-negative and sum to one. This proves
that $\bz_k\trans H\bz_\ell\phtrans\ge 0$ when $k\le\ell$. The case $k\ge\ell$ follows by symmetry of $H$.

\paragraph{Step 3: $\bz_k\trans H^2\bz_\ell\ge 0$.}

Again, let $k,\ell\in[K]$ and $k\le\ell$.
%We next analyze $\bz_k\trans H^2\bz_\ell\phtrans$,
First note that
\[
  H^2=(\diag_{j\in[K]}\mu_j^2) - \sum_{j\in[K]}\mu_j^2\belem_j\bprice\trans -\sum_{j\in[K]}\mu_j^2\bprice\belem_j\trans + \mu_K^{[2]}\bprice\bprice\trans
\]
where $\belem_j$ is the $j$th vector of the standard basis and we used our partial sum notation to substitute $\mu_K^{[2]}$ for $\norm{\bprice}^2$.
Thus, we can write
\begin{align}
\notag
  \bz_k\trans H^2 \bz_\ell\phtrans
&=
  \mu_k^{[2]}
  -
  \mu_k^{[2]}\mu_\ell^{[1]}
  -
  \mu_k^{[1]}\mu_\ell^{[2]}
  +
  \mu_K^{[2]}\mu_k^{[1]}\mu_\ell^{[1]}
\\
\label{eq:two:1}
&\ge
  \mu_k^{[2]}
  -
  \mu_k^{[2]}\mu_\ell^{[1]}
  -
  \mu_k^{[1]}\mu_\ell^{[2]}
  +
  \mu_\ell^{[2]}\mu_k^{[1]}\mu_\ell^{[1]}
\\
\notag
&=
  \Parens{\mu_k^{[2]}-\mu_\ell^{[2]}\mu_k^{[1]}}
  \Parens{1-\mu_\ell^{[1]}}
\\
\label{eq:two:2}
&\ge
  \Parens{\mu_k^{[2]}\mu_\ell^{[1]}-\mu_\ell^{[2]}\mu_k^{[1]}}
  \Parens{1-\mu_\ell^{[1]}}
\\
\label{eq:two:3}
&=
  \mu_k^{[1]}\mu_\ell^{[1]}
  \Parens{
    \frac{\mu_k^{[2]}}{\mu_k^{[1]}}
  -
    \frac{\mu_\ell^{[2]}}{\mu_\ell^{[1]}}
  }
  \Parens{1-\mu_\ell^{[1]}}
  \ge 0
\enspace.
\end{align}
In \Eq{two:1}, we used the bound $\mu_K^{[2]}\ge\mu_\ell^{[2]}$. In \Eq{two:2}, we used that
$0\le\mu_\ell^{[1]}\le 1$, so $\mu_k^{[2]}(1-\mu_\ell^{[1]})\ge\mu_k^{[2]}\mu_\ell^{[1]}(1-\mu_\ell^{[1]})$. The final inequality
uses the fact that $\mu_\ell^{[1]}\le 1$ and $\mu_k^{[2]}/\mu_k^{[1]}\ge \mu_\ell^{[2]}/\mu_\ell^{[1]}$. The
latter clearly holds if $\ell=k$ or $\mu_\ell^{[1]}=\mu_k^{[1]}$ (in which case $\mu_j=0$ for $k<j\le\ell$ and thus
$\mu_\ell^{[2]}=\mu_k^{[2]})$. We now argue that $\mu_k^{[2]}/\mu_k^{[1]}\ge \mu_\ell^{[2]}/\mu_\ell^{[1]}$ also
holds when $k<\ell$ and $\mu_{k+1}>0$. We
introduce the interval sum notation
\[
  \mu_{k+1:\ell}^{[d]}\coloneqq\sum_{j=k+1}^\ell \mu_j^d=\mu_\ell^{[d]}-\mu_k^{[d]}
\enspace.
\]
We begin by writing
$\mu_\ell^{[2]}/\mu_\ell^{[1]}$ as the following convex combination
\begin{align}
\notag
  \frac{\mu_\ell^{[2]}}
       {\mu_\ell^{[1]}}
%&
=
  \frac{\mu_k^{[2]}+\mu_{k+1:\ell}^{[2]}}
       {\mu_\ell^{[1]}}
%\\
%\notag
&=\frac{\mu_k^{[1]}}
       {\mu_\ell^{[1]}}
       \cdot
  \frac{\mu_k^{[2]}}
       {\mu_k^{[1]}}
  +
  \frac{\mu_{k+1:\ell}^{[1]}}
       {\mu_\ell^{[1]}}
       \cdot
  \frac{\mu_{k+1:\ell}^{[2]}}
       {\mu_{k+1:\ell}^{[1]}}
\\
\label{eq:two:4}
&=\lambda
       \cdot
  \frac{\mu_k^{[2]}}
       {\mu_k^{[1]}}
  +
  (1-\lambda)
       \cdot
  \frac{\mu_{k+1:\ell}^{[2]}}
       {\mu_{k+1:\ell}^{[1]}}
\enspace,
\end{align}
where we write $\lambda\coloneqq\mu_k^{[1]}/\mu_\ell^{[1]}$.
The expressions weighted in \Eq{two:4} by $\lambda$ and $1-\lambda$
can be viewed as weighted averages of $\mu_j$, with the weights also equal to $\mu_j$. Since $\bprice$ is sorted, we have
\[
  \frac{\mu_k^{[2]}}
       {\mu_k^{[1]}}
  =\sum_{j=1}^k \frac{\mu_j}{\mu_k^{[1]}}\cdot\mu_j
  \ge \mu_k
\qquad
\text{and}
\qquad
  \frac{\mu_{k+1:\ell}^{[2]}}
       {\mu_{k+1:\ell}^{[1]}}
  =\sum_{j=k+1}^\ell \frac{\mu_j}{\mu_{k+1:\ell}^{[1]}}\cdot\mu_j
  \le \mu_{k+1}
\enspace,
\]
so
\[
  \frac{\mu_k^{[2]}}
       {\mu_k^{[1]}}
  \ge \mu_k
  \ge \mu_{k+1}
  \ge
  \frac{\mu_{k+1:\ell}^{[2]}}
       {\mu_{k+1:\ell}^{[1]}}
\enspace.
\]
Plugging this back into \Eq{two:4}, we obtain
\[
   \frac{\mu_\ell^{[2]}}
        {\mu_\ell^{[1]}}
   =
   \lambda
       \cdot
  \frac{\mu_k^{[2]}}
       {\mu_k^{[1]}}
  +
  (1-\lambda)
       \cdot
  \frac{\mu_{k+1:\ell}^{[2]}}
       {\mu_{k+1:\ell}^{[1]}}
  \le
%   \lambda\cdot
%   \frac{\mu_k^{[2]}}
%        {\mu_k^{[1]}}
%  +
%  (1-\lambda)\cdot
%  \frac{\mu_k^{[2]}}
%       {\mu_k^{[1]}}
%  =
  \frac{\mu_k^{[2]}}
       {\mu_k^{[1]}}
\enspace,
\]
finishing the proof of \Eq{two:3}, showing that $\bz_k\trans H^2\bz_\ell\phtrans\ge 0$ when $k\le\ell$. The case $k\ge\ell$
again follows by symmetry of $H^2$.

\paragraph{Step 4: Putting it all together.} Let $M$ be either the matrix $H$ or $H^2$. Since in both cases
$\bz_k\trans M\bz_\ell\phtrans\ge 0$, the inequalities $0\le a_k\le b_k\le 2 a_k$
%\eqref{eq:ak:bk}
imply that
\[
   \sum_{k,\ell\in[K-1]} a_k a_\ell (\bz_k\trans M\bz_\ell\phtrans)
\le
   \sum_{k,\ell\in[K-1]} b_k b_\ell (\bz_k\trans M\bz_\ell\phtrans)
\le
   4\cdot\sum_{k,\ell\in[K-1]} a_k a_\ell (\bz_k\trans M\bz_\ell\phtrans)
\enspace.
\]
This is by the decomposition in \Eq{sHs} and analogous decompositions
for $\bv\trans H\bv$, $\bs\trans H^2\bs$, and $\bv\trans H^2\bv$ equivalent to
\[
   \bs\trans M\bs
\le
   \bv\trans M\bv
\le
   4\cdot
   \bs\trans M\bs
\enspace,
\]
proving the lemma.
\end{proof}

\section{Proofs and Additional Results for \Sec{convergence}}
\label{app:conv}

\subsection{Trader Dynamics}
\label{app:trader:dynamics}

To study convergence properties of the market, we need to posit a model of how the traders arrive in the market and which securities they buy or sell, as a function of their current holdings of securities $\bbundle_i$ and cash $c_i$,
and the current market state $\bstate$. We refer to such a model as \emph{trader dynamics}.
We consider two simple trader dynamics:
\begin{itemize}
\item \emph{All-securities dynamics} (\ASD). In each round, a trader $i\in[N]$ is chosen uniformly at random. This trader then buys a bundle $\bdelta\in\R^\nsec$ which optimizes her utility, i.e., if the current state of the market is $\bstate$ and the current cash and security allocations of the trader are $c_i$ and $\bbundle_i$, then the trader picks $\bdelta$ maximizing $U_i(\bbundle_i+\bdelta,\,c_i-C_b(\bstate+\bdelta)+C_b(\bstate))$.
\item \emph{Single-security dynamics} (\SSD). In each round, a trader $i\in[N]$ is chosen uniformly at random and this trader picks a security $k\in[K]$ uniformly at random. The trader then buys a quantity $\delta\in\R$ of
the $k$th security to optimize her utility. Let $\belem_k$ be the $k$th vector of the standard basis. Then the trader picks $\delta$ maximizing $U_i(\bbundle_i+\delta\belem_k,\,c_i-C_b(\bstate+\delta\belem_k)+C_b(\bstate))$.
\end{itemize}
The all-securities dynamics has been recently studied by \citet{FR15}. This model assumes that traders are able to calculate the optimal bundle over all securities, which may not be a realistic assumption when the number of securities is large. The single-security dynamics, in which we only assume that traders can optimize over a single security at a time, is more appropriate for computationally limited traders.

We formalize both dynamics in a unified analysis via \emph{blocks}.
Specifically, we assume that the coordinates $[N]\times[K]$ are partitioned into blocks $\alpha\in\cA$, where $\alpha\subseteq[N]\times[K]$,
such that
$\biguplus_{\alpha\in\cA} \alpha=[N]\times[K]$. Blocks are disjoint subsets of coordinates of the overall allocation vector $\bbundle\in\R^{NK}$. We further assume that each block $\alpha\in\cA$ is fully contained
within coordinates corresponding to some trader $i$. Thus each block $\alpha$ can be written as $\set{i}\times\beta$ for some $i\in[N]$ and $\beta\subseteq[K]$.
%We write $B_\alpha=\set{i_\alpha}\times\beta_\alpha$ where $\beta_\alpha\subseteq[K]$.
For \ASD, we have $N$ blocks $\cA=\bigSet{ \set{i}\times[K]:\:i\in[N]}$. For \SSD,
we have $NK$ blocks $\cA=\bigSet{ \set{i}\times\set{k}:\: i\in[N],k\in[K] }$.

Let $E_\alpha\in\R^{NK\times\card{\alpha}}$ be the \emph{embedding matrix} for the block $\alpha$. It maps $\card{\alpha}$-dimensional vectors
$\bu\in\R^{\card{\alpha}}$ to vectors $\bv\in\R^{NK}$ that are zero everywhere except for the block $\alpha$, where they coincide with $\bu$.
The range of $E_\alpha$ is exactly the set of vectors that are zero outside the block $\alpha$.
Its transpose $E_\alpha\trans$ projects vectors from $\R^{NK}$ to $\R^{\card{\alpha}}$ by removing all coordinates outside the block~$\alpha$.
For any subset $\beta\subseteq[K]$, we similarly
define the embedding matrices $E_\beta\in\R^{K\times\card{\beta}}$.
%
%\mdcomment{Insert example(s)?}

The next theorem shows that optimizing utility, as in \ASD or \SSD, corresponds to the greedy optimization of the potential function $F$ along the coordinates
of the corresponding block:
%\jenn{Credit Mark and Raf appropriately.}
%
\begin{theorem} Let $\alpha=\set{i}\times\beta$, where $\beta\subseteq[K]$, be a block of coordinates controlled by trader $i$. Assume that trader $i$ has security bundle $\bbundle_i$ and $c_i$ units of cash, and the current market state is~$\bstate$. Then \[
   \argmax_{\bdelta\in\range(E_\beta)} U_i(\bbundle_i+\bdelta,\,c_i-C_b(\bstate+\bdelta)+C_b(\bstate))
   =
%   \argmin_{\bdelta\in\R^{\card{\alpha}}} F(\bbundle+E_\alpha\bdelta)
%\enspace.
   \argmin_{\bdelta\in\range(E_\beta)} F(\bbundle_{-i},\,\bbundle_i+\bdelta)
\enspace,
\]
where $\bbundle_{-i}$ denotes the concatenation of $\bbundle_j$ across $j\ne i$.
\end{theorem}

\begin{proof}
The proof is immediate from the definition of the potential $F$.
\end{proof}

\subsection{Relationship between the Suboptimality of Potential and the Convergence Error}
\label{app:subopt}

In this appendix, we relate the suboptimality of the objective to several other quantities used in
analysis of convergence error. We begin by defining these quantities.

Given an allocation vector $\bbundle\in R^{NK}$, the associated market price (the gradient of the cost)
will be denoted $\bprice_0(\bbundle)$ and the gradients of trader potentials will be denoted $\bprice_i(\bbundle)$:
\begin{align*}
%\label{eq:mu:0}
  \bprice_0(\bbundle)
  &\coloneqq
  \nabla C_b\bigParens{\textstyle\sum_{i=1}^N \bbundle_i}
  =
  \nabla C\bigParens{\textstyle\sum_{i=1}^N \bbundle_i/b}
\\
%\label{eq:mu:i}
  \bprice_i(\bbundle)
  &\coloneqq
  \nabla F_i(-\bbundle_i)
  =
  \nabla T(\bttheta_{i}-a_i\bbundle_i)
\qquad
\text{for $i\in[N]$}
\end{align*}
where $T$ is the log partition function.
%(see Eq.~\ref{eqn:logpartition}).
As in the body of the paper, let $\bbundle^\star$ denote an arbitrary minimizer of $F$ and let $F^\star$ denote the minimum value of $F$.
%As with $\eqbprice$, we will write $\bbundle^\star(b;C)$ and $F^\star(b;C)$ whenever we want to
%emphasize the dependence on the liquidity and cost.
From \Thm{eq:app}, at any equilibrium allocation $\bbundle^\star$,
\[
  \bprice_0(\bbundle^\star)=\eqbprice
  \qquad
  \text{and}
  \qquad
  \bprice_i(\bbundle^\star)=\eqbprice\text{ for all $i\in[N]$.}
\]
Let $\bbundle^t$ denote the allocation vector after the $t$th trade, and $\bprice^t\coloneqq\bprice_0(\bbundle^t)$ be the corresponding market price.
We next show that to bound the convergence error $\norm{\bprice_0(\bbundle^t)-\eqbprice}$ it suffices to bound the suboptimality of the current objective value, $F(\bbundle^t)-F^\star$. In fact,
we show that the suboptimality $F(\bbundle^t)-F^\star$ simultaneously also bounds $\norm{\bprice_i(\bbundle^t)-\eqbprice}$, which can be viewed as a measure of
suboptimality of individual traders and will be used in our later analysis. We first prove this result when $C$ has a Lipschitz-continuous gradient and then for the case when $C$ is \convplus. 
\begin{theorem}
\label{thm:errconnect:L}
If $\nabla C$ has the Lipschitz constant $L_C$ then for any $\bbundle\in\R^{NK}$
% with respect to $\norm{\cdot}$, then
\[
  F(\bbundle)-F^\star
  \ge
\frac{b}{2L_C}\cdot
  \norm{\bprice_0(\bbundle)-\eqbprice}^2
  +
\frac{1}{2L_T}\cdot
  \sum_{i=1}^N
  \frac{1}{a_i}\norm{\bprice_i(\bbundle)-\eqbprice}^2
\enspace,
\]
where $L_T$ is the Lipschitz constant of $\nabla T$.
\end{theorem}
\begin{proof}
First note that since $\nabla C$ and $\nabla T$ have Lipschitz constants $L_C$ and $L_T$, their conjugates
are strongly convex with constants $1/L_C$ and $1/L_T$. Further, by the properties of conjugates (see Theorems~12.3 and~16.1 of~\citet{Rockafellar70}),
the definitions of $F_i$ and $C_b$ yield
\[
  F_i^*(\bprice)=\frac{1}{a_i}\bigParens{ T^*(\bprice)-\bttheta_i\cdot\bprice }
\enspace,
\quad
  C_b^*(\bprice)=bC^*(\bprice)
\enspace,
\]
and so $F_i^*$ and $C_b^*$ are strongly convex, respectively, with constants $1/(a_iL_T)$ and $b/L_C$.

We now invoke the duality result of \Thm{eq:app} to prove our theorem.
Specifically, from \Eqs{eq:primal}{eq:dual}, we have
\[
   F(\bbundle^\star)
   =
     -\sum_{i=1}^N F_i^*(\eqbprice)
     -C_b^*(\eqbprice)
\enspace.
\]
Therefore, for any $\bbundle$, we have
\begin{align}
\notag
  F(\bbundle)-F^\star
&=\sum_{i=1}^N F_i(-\bbundle_i)
  +
  C_b\bigParens{\sum_i\bbundle_i}
  +\sum_{i=1}^N F_i^*(\eqbprice)
  +C_b^*(\eqbprice)
\\
\label{eq:bregman:1}
&=\sum_{i=1}^N
  \BigBracks{
     F_i(-\bbundle_i)
     + F_i^*(\eqbprice)
     + \bbundle_i\trans\eqbprice
  }
  +
  \BigBracks{
     C_b\bigParens{\sum_i\bbundle_i}
     + C_b^*(\eqbprice)
     - \bigParens{\sum_i\bbundle_i}\trans\eqbprice
  }.
\end{align}
Using conjugacy and strong convexity, we next show that the terms in the brackets can be lower-bounded by quadratic functions.

Let $\bstate\coloneqq\sum_i\bbundle_i$. Since $\bprice_i(\bbundle)=\nabla F_i(-\bbundle_i)$ and $\bprice_0(\bbundle)=\nabla C_b(\bstate)$,
\begin{align*}
   (-\bbundle_i)\in\partial F_i^*(\bprice_i(\bbundle))
&\text{\ \ \ and\ \ \ }
   F_i(-\bbundle_i)=-\bbundle_i\trans\bprice_i(\bbundle)-F_i^*(\bprice_i(\bbundle))
\\
   \bstate\in\partial C_b^*(\bprice_0(\bbundle))
&\text{\ \ \ and\ \ \ }
   C_b(\bprice_0(\bbundle))
   =\bstate\trans\bprice_0(\bbundle)-C_b^*(\bprice_0(\bbundle))
\enspace.
\end{align*}
Using these identities and invoking the strong convexity of $F_i$ and $C_b$, we thus obtain
\begin{align*}
     F_i(-\bbundle_i)
     + F_i^*(\eqbprice)
     + \bbundle_i\trans\eqbprice
&=   F_i^*(\eqbprice)
     - (-\bbundle_i)\trans
       \bigParens{\eqbprice-\bprice_i(\bbundle)}
     - F_i^*(\bprice_i(\bbundle))
\\
&\ge \frac{1}{2 a_i L_T}\norm{\bprice_i(\bbundle^t)-\eqbprice}^2
\\
     C_b(\bstate)
     + C_b^*(\eqbprice)
     - \bstate\trans\eqbprice
&=
     C_b^*(\eqbprice)
     - \bstate\trans
       \bigParens{\eqbprice-\bprice_0(\bbundle)}
     - C_b^*(\bprice_0(\bbundle))
\\
&\ge \frac{b}{2 L_C}\norm{\bprice_0(\bbundle)-\eqbprice}^2
\enspace.
\end{align*}
The theorem now follows by applying these lower
bounds in \Eq{bregman:1}.
\end{proof}

\begin{theorem}
\label{thm:errconnect:convplus}
If $C$ is \convplus then there exist $\eps>0$ and $c>0$ such that if $F(\bbundle)-F^\star\le\eps$ then
\[
  F(\bbundle)-F^\star
  \ge
  c\Bracks{b
  \norm{\bprice_0(\bbundle)-\eqbprice}^2
  +
  \sum_{i=1}^N
  \frac{1}{a_i}\norm{\bprice_i(\bbundle)-\eqbprice}^2
  }
\enspace.
\]
\end{theorem}
\begin{proof}
The proof begins similarly to proof of \Thm{errconnect:L}, by establishing the identity
\begin{align}
\notag
  F(\bbundle)-F^\star
&=\sum_{i=1}^N \BigBracks{
     F_i^*(\eqbprice)
     - (-\bbundle_i)\trans
       \bigParens{\eqbprice-\bprice_i(\bbundle)}
     - F_i^*(\bprice_i(\bbundle))
  }
\\
\label{eq:dual:connect}
&\qquad{}
  +
  \BigBracks{
  C_b^*(\eqbprice)
     - \bstate\trans
       \bigParens{\eqbprice-\bprice_0(\bbundle)}
     - C_b^*(\bprice_0(\bbundle))
  }
\enspace.
\end{align}
Let $G_i$ be the \convplus functions that match $F^*_i$ on $\dom F^*_i$ and $R$ be the \convplus function
that matches $C^*$ on $\dom C^*$, and that have the additional properties outlined in \Prop{convplus:conj}. Thus,
\begin{align}
\notag
  F(\bbundle)-F^\star
&=\sum_{i=1}^N \BigBracks{
     G_i(\eqbprice)
     - [\nabla G_i(\bprice_i(\bbundle))]\trans
       \bigParens{\eqbprice-\bprice_i(\bbundle)}
     - G_i(\bprice_i(\bbundle))
  }
\\
\label{eq:bregman:convplus}
&\qquad{}
  +
  b
  \BigBracks{
  R(\eqbprice)
     - [\nabla R(\bprice_0(\bbundle))]\trans
       \bigParens{\eqbprice-\bprice_0(\bbundle)}
     - R(\bprice_0(\bbundle))
  }
\enspace.
\end{align}
By \convexityplus, functions $G_i$ are strictly convex on $\cM=\dom F_i^*$ and $R$ is
also strictly convex on $\cM\subseteq\dom C^*$. Thus, $G_i$ and $R$ are strictly
convex in a neighborhood of $\eqbprice$ (in $\aff\cM$), which implies that the
terms on the right-hand side of \Eq{bregman:convplus} are strictly convex in $\bprice_i(\bbundle)$
and $\bprice_0(\bbundle)$. Since they are minimized at $\eqbprice$, we obtain
$\bprice_i(\bbundle)\to\eqbprice$ and $\bprice_0(\bbundle)\to\eqbprice$ as $F(\bbundle)-F^\star\to 0$.
Thus, by picking a sufficiently small $\eps$, we can guarantee that $\bprice_i(\bbundle)$ and
$\bprice_0(\bbundle)$ are arbitrarily close to $\eqbprice$ whenever $F(\bbundle)-F^\star\le\eps$.
From Taylor's theorem, and the fact that $\nabla^2 G_i\equiv H_{F_i}^+\equiv(1/a_i)H_T^+$ and
$\nabla^2 R\equiv H_C^+$,
we obtain
\begin{align*}
     G_i(\eqbprice)
     - [\nabla G_i(\bprice_i(\bbundle))]\trans
       \bigParens{\eqbprice-\bprice_i(\bbundle)}
     - G_i(\bprice_i(\bbundle))
&=
     \frac{1}{2a_i}\bigParens{\eqbprice-\bprice_i(\bbundle)}\trans
     H_T^+(\bar{\bprice}_i)\bigParens{\eqbprice-\bprice_i(\bbundle)}
\\   
  R(\eqbprice)
     - [\nabla R(\bprice_0(\bbundle))]\trans
       \bigParens{\eqbprice-\bprice_0(\bbundle)}
     - R(\bprice_0(\bbundle))
&=
     \frac{1}{2}\bigParens{\eqbprice-\bprice_0(\bbundle)}\trans
     H_C^+(\bar{\bprice}_0)\bigParens{\eqbprice-\bprice_0(\bbundle)}
.
\end{align*}
Now, envoking \Prop{H:Lipschitz} for \convplus
functions $T$ and $C$, we obtain that for a sufficiently small $\eps$ we have that if $F(\bbundle)-F^\star\le\eps$,
then
\[
  H_T^+(\bar{\bprice}_i)
  \within\Parens{1\pm \frac12}H_T^+(\eqbprice)
\enspace,
\quad
  H_C^+(\bar{\bprice}_0)
  \within\Parens{1\pm \frac12}H_C^+(\eqbprice)
\enspace.
\]
Plugging this back into \Eq{bregman:convplus}, we obtain
\begin{align}
\notag
  F(\bbundle)-F^\star
&\within\Parens{1\pm\frac12}\sum_{i=1}^N
  \frac{1}{2a_i}\bigParens{\eqbprice-\bprice_i(\bbundle)}\trans
     H_T^+(\eqbprice)\bigParens{\eqbprice-\bprice_i(\bbundle)}
\\
\label{eq:subopt:convplus}
&\qquad{}
  +
  b\Parens{1\pm\frac12}
  \frac{1}{2}\bigParens{\eqbprice-\bprice_0(\bbundle)}\trans
     H_C^+(\eqbprice)\bigParens{\eqbprice-\bprice_0(\bbundle)}
\enspace.
\end{align}
The theorem now follows by noting that the ranges of matrices
$H_T(\eqbprice)$ and $H_C(\eqbprice)$ include all directions $\bprice-\bprice'$ where $\bprice,\bprice'\in\cM$, because
$\cM=\dom T^*$ and $\cM\subseteq\dom C^*$.
\end{proof}

\subsection{Local Convergence Rate of Block-Coordinate Descent}

In this section, we consider general unconstrained convex minimization, but use the same notation as in the rest of the paper.
%In particular, we use the notation for blocks introduced in \App{trader:dynamics}.
The key difference from the standard analysis of Nesterov \citep{Nest12} is the focus on the local convergence in the neighborhood of the solution, rather than global convergence.
This analysis is not specific to our setting, and may be of independent interest.

%While most of the steps in the proof
%are identical, there are a few subtle differences, so we present the full proof below.

We consider the optimization problem
\[
   \min_{\bbundle\in\R^{NK}} F(\bbundle)
\enspace,
\]
where $F:\R^{NK}\to\R$ is a differentiable convex function bounded below. We are given a set of blocks $\alpha\in\cA$, which partition the coordinates $[N]\times[K]$. The \emph{block-coordinate descent} algorithm sets the initial iterate $\bbundle^0=\zero$, and in each iteration chooses an index $\alpha\in\cA$ uniformly at random and fully optimizes the objective over the coordinates in block $\alpha$: given the current iterate $\bbundle^t$, the new iterate is $\bbundle^{t+1}=\Psi_\alpha(\bbundle^t)$ where
\begin{equation}
\label{eq:update}
  \Psi_\alpha(\bbundle) \coloneqq  \myargmin_{\bbundle' \in \bbundle + \range(E_\alpha)} \obj(\bbundle')
\end{equation}
and $E_\alpha$ is the embedding matrix for the block $\alpha$ as introduced in \App{trader:dynamics}.
%$\bbundle^{t+1} = \myargmin_{\bbundle' \in \bbundle^t + \range(E_\alpha)} \obj(\bbundle')$.

\citet{Nest12} shows that when the optimization objective is strongly convex,
%(under a certain norm),
it is possible to
achieve the linear convergence rate of the form
$
  \ex{F(\bbundle^t)}-F^\star\le c\kappa^{t}
$
for some constants $c>0$ and $\kappa<1$.
%
%The expectation is under the randomness in the choice of the updates of the algorithm.
%The bound on the expectation can be turned into a bound on $F(\bbundle^t)$ using Markov's inequality, so with probability $1-\rho$, we
%would obtain $F(\bbundle^t)-F^\star\le c\kappa^{t}/\rho$ under strong convexity.
%
While the objective in our setting is not globally strongly convex, it is
strongly convex locally, so the optimization eventually stays within a region where the strong convexity constant is bounded away from zero, yielding
a linear convergence rate. Recall from \Sec{convergence} that $\kappahigh$ is an upper bound on the local convergence rate of an algorithm if
the algorithm with probability 1 reaches an iteration $t_0$ such that for some $c>0$ and all $t\ge t_0$,
\[
  \ex{F(\bbundle^t)\bigGiven\bbundle^{t_0}}-F^\star\le c\kappahigh^{t-t_0}
\enspace.
\]
Also, $\kappalow$ is a lower bound on the local convergence rate of an algorithm if $\kappahigh\ge\kappalow$ holds
for all upper bounds $\kappahigh$.

Our local convergence-rate results are based on various local
properties of $F$, by which we mean properties that hold on some
\emph{proper level set} of $F$, in the sense of the following definition:
\begin{definition}[Level Set]
\label{def:level:set}
For a given $\lambda\in\R$, the level set $S(F,\lambda)$ of a function $F$ is defined as the set of points $\bbundle$ where $F$ is at most $\lambda$:
\[
  S(F,\lambda)\coloneqq\set{\bbundle:\:F(\bbundle)\le\lambda}
\enspace.
\]
The level set $S(F,\lambda)$ is called \emph{proper} if $\lambda>F^\star$.
\end{definition}

Our first result requires that the objective
be locally strongly convex and smooth in the sense of the following definition:
\begin{definition}[Strong Convexity and Smoothness for Matrix Seminorms]
Let $A$ and $B$ be symmetric positive-semidefinite matrices.
We say that a differentiable function $\obj$ is strongly convex on a set $S$ with respect to $A$ and smooth on $S$ with
respect to $B$
if
\[
   \frac12\bdelta\trans A\bdelta \le F(\bbundle+\bdelta)-F(\bbundle)-\bdelta\trans\nabla F(\bbundle) \le \frac12\bdelta\trans B\bdelta
\]
whenever $\bbundle\in S$ and $\bbundle+\bdelta\in S$.
\end{definition}

To state the theorem, we introduce additional notation.
For a matrix $M\in\R^{NK\times NK}$, let $M_{\alpha,\alpha}$ denote the block consisting of rows and columns in~$\alpha$,
i.e., $M_{\alpha,\alpha}\coloneqq E_\alpha\trans M E_\alpha$. And let $\cD:\R^{NK\times NK}\to\R^{NK\times NK}$ be the
operation of retaining only the block diagonal of a matrix, i.e., $\cD(M)\coloneqq\diag_{\alpha\in\cA} M_{\alpha,\alpha}$.
Recall that $M^+$ denotes the pseudoinverse of $M$.
Finally, let $\bg^t\coloneqq\nabla F(\bbundle^t)$ denote the gradient of the objective in the $t$th iteration.

\begin{theorem}
\label{thm:asympt:0}
Assume that $F$ attains a minimum and let $S\coloneqq S(F,\lambda)$ be a proper level set which satisfies the following conditions:
\begin{enumerate}[noitemsep]
\item $F$ is strongly convex and smooth on $S$ with respect to some positive-semidefinite matrices $A$ and $B$ such that $\cG(F)\subseteq\range(A)$, $\cG(F)\subseteq\range(B)$.
\item There exist non-negative constants $\sigmalow\le\sigmahigh\le\infty$ and $\ell<\infty$ such that whenever some iterate $\bbundle^{t_0}$
%of the block-coordinate descent
lies in $S$, then all the consecutive iterates with $t\ge t_0+\ell$ satisfy
\begin{align*}
      \sigmalow\ex{(\bg^t)\trans A^+\bg^t\BigGiven\bbundle^{t_0}}
&
      \le
      \ex{(\bg^t)\trans \cD(B)^+\bg^t\BigGiven\bbundle^{t_0}}
\enspace,
\\
      \ex{(\bg^t)\trans \cD(A)^+\bg^t\BigGiven\bbundle^{t_0}}
&
      \le
      \sigmahigh\ex{(\bg^t)\trans B^+\bg^t\BigGiven\bbundle^{t_0}}
\enspace,
\end{align*}
where the expectation is over the random choice of updates by the algorithm.
\end{enumerate}
Then if $F^\star<F(\bbundle^{t_0})<\lambda$ and $t\ge t_0+\ell$, we have
\[
\max\Braces{1-\frac{\sigmahigh}{\card{\cA}}\;,\; 0}
\le
\frac{\ex{\obj(\iterbbundle{t+1})\bigGiven\iterbbundle{t_0}} - \obj^\star}
     {\ex{\obj(\iterbbundle{t})\given\iterbbundle{t_0}} - \obj^\star}
\le
1-\frac{\sigmalow}{\card{\cA}}
\enspace.
\]
%
\begin{comment}
Let $\bbundle^t$ by the $t$th iterate of the block coordinate descent algorithm and let $S\coloneqq S{F(\bbundle^t)}(F)$. Assume that $F$ is strongly convex and smooth
on $S$ with respect to some positive-semidefinite matrices $A$ and $B$.
Then
\[
\ex{\obj(\iterbbundle{t+1})\given\iterbbundle{t}} - \obj^\star \leq \left(1 - \frac{\sigma}{\card{\cA}} \right) \left(\obj(\iterbbundle{t}) - \obj^\star\right)
\enspace,
\]
where the expectation is over the randomness in the choice of the next update, and
\[
   \sigma=\min_{\bu\in\cL\wo\set{\zero}} \frac{\bu\trans D^+\bu}{\bu\trans A^+\bu}
\enspace,
\]
where $D^+=\diag_{\alpha\in\cA}B_{\alpha,\alpha}^+$.
\end{comment}
\end{theorem}

The proof is deferred to a separate subsection below (\App{proof:asympt:0}). We next show
how to apply \Thm{asympt:0} to obtain bounds on local convergence rate.
For the lower bounds, we assume that the optimization problem is \emph{non-degenerate}
%with respect to the algorithm
in the sense that the probability of reaching an iterate $\bbundle^t$ which attains a minimum is zero;
in other words, only approximate solutions are reached in a finite number of steps.
This condition holds for the potential function from the main paper under all-security dynamics (\ASD), whenever
there are at least three agents with distinct beliefs.
%Since the optimization starts at a deterministic point $\bbundle^0=\zero$, and the randomization
%is among a finite set of choices, there is only a finite set of allocation
%vectors $\bbundle^t$ that can be reached at any given iteration $t$. Under non-degeneracy,
%none of these allocations (for any $t$) are the
%actual minimizers of $F$.

\begin{proposition}
\label{prop:t:t0}
Let $\Sin=S(F,\lambdain)$ and
$\Sout=S(F,\lambdaout)$ be level sets with $F^\star<\lambdain<\lambdaout$, and
assume that the conditions of \Thm{asympt:0} hold for the level set $\Sout$ with
matrices $A$ and $B$, and non-negative scalars $\sigmalow\le\sigmahigh\le\infty$ and $\ell<\infty$.
Then for every $\bbundle^{t_0}\in\Sin$ there exists $c>0$ such that for all $t\ge t_0$
\[
  \ex{F(\bbundle^t)\given\bbundle^{t_0}}-F^\star
  \le c\kappahigh^{t-t_0}
\enspace,
\]
where $\kappahigh=1-\sigmalow/\card{\cA}$. Furthermore, if the optimization problem
is non-degenerate then for every $\bbundle^{t_0}\in\Sin$, which has a non-zero probability of occurring,
and all $t_1\ge t_0$, there exists $c>0$ such that for all $t\ge t_1$
\[
  \ex{F(\bbundle^t)\given\bbundle^{t_1}}-F^\star
  \ge c\kappalow^{t-t_1}
\enspace,
\]
where $\kappalow=\max\set{1-\sigmahigh/\card{\cA},\,0}$.
\end{proposition}
\begin{proof}
For the upper bound, if $F(\bbundle^{t_0})=F^\star$, the bound holds for any $c>0$.
Otherwise, set $c\coloneqq\kappahigh^{-\ell}\bigParens{F(\bbundle^{t_0})-F^\star}$.
Since the optimization does not increase the objective, this guarantees that the bound holds
for $t_0\le t\le t_0+\ell$. For $t\ge t_0+\ell$, \Thm{asympt:0} gives
\[
  \ex{F(\bbundle^{t+1})\given\bbundle^{t_0}}-F^\star
  \le
  \kappahigh\cdot\bigParens{\ex{F(\bbundle^t)\given\bbundle^{t_0}}-F^\star}
\enspace,
\]
and the upper bound now follows by induction.

For the lower bound, note that the non-degeneracy guarantees that
$F(\bbundle^t)>F^\star$ for all $t\ge t_0$. If $\kappalow=0$, the bound holds for any $c$.
Next consider the case when $\kappalow>0$. Pick $t_1\ge t_0$ and set
\[
  c\coloneqq \min_{t_1\le t\le t_1+\ell}\;\kappalow^{-\ell}\bigParens{\ex{F(\bbundle^t)\given\bbundle^{t_1}}-F^\star}
\enspace,
\]
which is non-zero, because $\kappalow\in(0,1]$ and $F(\bbundle^t)>F^\star$.
This guarantees that the bound holds up to $t=t_1+\ell$. For larger $t$,
it follows by \Thm{asympt:0} and induction.
\end{proof}

When the objective is \convplus, we can use the Hessian at the minimum of $F$ to obtain both a quadratic lower and
upper bound on $F$, yielding the following bounds on the local convergence rate:
\begin{theorem}
\label{thm:conv}
Let $F$ be a \convplus function attaining a minimum at $\bbundle^\star$. Let $H^\star\coloneqq\nabla^2 F(\bbundle^\star)$ and assume
there exists $\lambda > F^\star$
%proper level set $S$
and non-negative constants $\sigmalow$, $\sigmahigh$, and $\ell$ such that
\begin{equation}
\label{eq:thm:conv}
      \sigmalow
      \le
      \frac{\ex{\;(\bg^t)\trans \cD(H^\star)^+\bg^t\;\bigGiven\; \bbundle^{t-\ell} \;}}
           {\ex{\;(\bg^t)\trans (H^\star)^+\bg^t\;\bigGiven\; \bbundle^{t-\ell} \;}}
      \le
      \sigmahigh
\end{equation}
whenever $F(\bbundle^{t-\ell}) \leq \lambda$,
where the expectation
is over the choice of updates by the algorithm.
%and $\bg^t=\nabla F(\bbundle^t)$.
%
Then, for all $\eps>0$, the local convergence rate is bounded
above by $\kappahigh=1-\sigmalow/\card{\cA}+\eps$. If the optimization problem
is non-degenerate, the local convergence is also bounded below by $\kappalow=1-\sigmahigh/\card{\cA}$.
%\[
%  \kappalow=1-\frac{\sigmahigh}{\card{\cA}}
%\enspace,
%\qquad
%  \kappahigh=
%  1-\frac{\sigmalow}{\card{\cA}}
%\enspace.
%\]
\end{theorem}
%
%If the expectations in the numerator and denominator of \Eq{thm:conv} are both zero, we view the inequalities as vacuously satisfied. This happens
%if and only if $\nabla F(\bbundle^{t-\ell})=0$, i.e., if $\bbundle^{t-\ell}$ is a minimizer of $F$. Thus, this degenerate case does not affect the statement
%of the theorem.
%
%\jenn{I think we should shorten the following discussion, but I'm not sure what is the most useful thing to say. We want to get across that the bounds are usable.}
%\mdcomment{I think I actually made it longer.}
%
We defer the proof of \Thm{conv} to a separate subsection (\App{proof:conv}).
The ratio bounded in \Eq{thm:conv} can be interpreted as a curvature of the quadratic form $\cD(H^\star)^+$ under the norm described
by the quadratic form $(H^\star)^+$. Larger values of the ratio (larger curvature) mean faster convergence. We will refer to the
bounds $\sigmalow$ and $\sigmahigh$
%from \Thm{conv}
as \emph{lower and upper bounds on local strong convexity} of $F$ (under randomized block-coordinate descent updates).
Any non-trivial lower bound, i.e., $\sigmalow>0$, yields local linear convergence rate since it implies $\kappahigh<1$,
since $\eps$ can be chosen arbitrarily small.

%While \Thm{conv} gives the flexibility to consider only the expected curvature in the directions of gradients at actual iterates,
%it is possible to obtain a simple bound by considering all directions in the span of gradients (and set $\ell=0$).
If we know the Hessian $H^\star$, we can obtain a simple lower bound $\sigmalow$ and a linear convergence
by considering all directions in the span of gradients (and setting $\ell=0$).
The span of gradients
coincides with $\cG(F)$, because $\zero$ is a valid gradient (since $F$ attains a minimum). Thus, we can obtain
$\sigmalow$ by the following generalized eigenvalue calculation:
\begin{equation}
\label{eq:sigmalow}
   \sigmalow=\min_{\bu\in\cG(F)\wo\set{\zero}} \frac{\bu\trans \cD(H^\star)^+\bu}{\bu\trans (H^\star)^+\bu}
            =\lambda_{\min}\bigParens{(H^\star)^{1/2}\cD(H^\star)^+(H^\star)^{1/2}}
\enspace,
\end{equation}
where $\lambda_{\min}(\cdot)$ is the smallest positive eigenvalue.
This is a valid setting of $\sigmalow$, because for all $\bu\in\cG(F)$ we then have $\sigmalow \bu\trans (H^\star)^+\bu \leq \bu\trans \cD(H^\star)^+\bu$ and therefore,
when $F(\bbundle^{t-\ell})\le\lambda$, also
\[
\sigmalow  {\ex{\;(\bg^t)\trans (H^\star)^+\bg^t\;\bigGiven\; \bbundle^{t-\ell} \;}} \leq
{\ex{\;(\bg^t)\trans \cD(H^\star)^+\bg^t\;\bigGiven\; \bbundle^{t-\ell} \;}}.
\]
This value of $\sigmalow$ is non-zero, because $\cG(F)=\range (H^\star)^+\subseteq\range\cD(H^\star)^+$.
We pursue this style of analysis for single-securities dynamics (\SSD), where we derive $\sigmalow$ using \Eq{sigmalow}, but do so directly in terms of Hessians of functions $C$ and $T$
rather than the potential $F$. For all-securities dynamics, we consider $\ell>0$
and take advantage of the averaging effect of expectations in \Eq{thm:conv}, which yields a tighter lower bound $\sigmalow$
and a non-trivial upper bound $\sigmahigh$.

\begin{comment}
Using Markov's inequality we then obtain:
%
\begin{theorem}
Given the same conditions as in \Thm{asympt:0}, for any $t'\ge t$, with probability at least $1-\rho$, we have
\[
  \obj(\iterbbundle{t'}) - \obj^\star
  \leq \frac{1}{\rho}
       \Parens{1 - \frac{\sigma}{\card{\cI}} }^{t'-t}
       \Parens{\obj(\iterbbundle{t}) - \obj^\star}
  \leq e^{(t'-t)\sigma/\card{\cI}}
       \Parens{\frac{\obj(\iterbbundle{t}) - \obj^\star}{\rho}}
\enspace.
\]
\end{theorem}
\end{comment}

\begin{comment}
We then use a result from \cite{RT14} to obtain a high probability bound.
\begin{theorem}
Given the same conditions as in \Cref{thm:expected_rate} on the objective $\obj$, then when $n > \sigma$ and with probability at least $1-\rho$, we have
$$
\obj(\iterbbundle{t}) - \obj^\star\leq \left(\frac{\obj(\iterbbundle{0}) - \obj^\star   }{\rho}  \right) \exp\left( -t n/\sigma \right)
$$
\end{theorem}
\rynote{end new analysis.}
\end{comment}

%%%%%%%%%%%%%%%%
%%%%%%%%%%%%%%%%
%%%%%%%%%%%%%%%%
%%%%%%%%%%%%%%%%
%%%%%%%%%%%%%%%%
%%%%%%%%%%%%%%%%
%%%%%%%%%%%%%%%%
%%%%%%%%%%%%%%%%
%%%%%%%%%%%%%%%%
%%%%%%%%%%%%%%%%

\subsection{Proof of \Thm{asympt:0}}
\label{app:proof:asympt:0}

\begin{proof}[Proof of \Thm{asympt:0}]
The vector of partial derivatives of $F$ within block $\alpha$ will be called \emph{partial gradient} and denoted
$\nabla_\alpha F(\bbundle)\coloneqq E_\alpha\trans\nabla F(\bbundle)$.
Since $F$ has a minimizer, its set of gradients contains a zero, and therefore its gradient space $\cG(F)$
coincides with the span of its gradients (see \Def{grad:space}). We also define the set $\cG_\alpha(F)\coloneqq E_\alpha\trans\cG(F)$,
which then coincides with the span of partial gradients of $F$. Note that by \Prop{D:range}, proved below,
any positive-semidefinite matrix $M$ satisfies $\range(M)\subseteq\range(\cD(M))$.
Since $\cG(F)\subseteq\range(A)$, we therefore obtain $\cG_\alpha(F)\subseteq\range(B_{\alpha,\alpha})$ by the following
reasoning:
\[
   \cG_\alpha(F)=E_\alpha\trans\cG(F)\subseteq E_\alpha\trans\range(A)\subseteq E_\alpha\trans\range(\cD(A))=\range(A_{\alpha,\alpha})
\enspace,
\]
and similarly obtain $\cG_\alpha(F)\subseteq\range(B_{\alpha,\alpha})$.

Assume $F^\star<F(\bbundle^{t_0})<\lambda$ and let $t\ge t_0+\ell$. Consider the bundle $\bbundle\coloneqq\bbundle^{t}$ and analyze
the value of the objective at the next iterate $\bbundle^{t+1}=\Psi_\alpha(\bbundle)$. First note that the objective does
not increase during optimization, so $F(\bbundle)<\lambda$ and in particular $\bbundle\in S$. From \Eq{update}, we then have
\begin{align}
\label{eq:conv:0a}
  \obj(\Psi_\alpha(\bbundle))
&
  =
  \min_{\bbundle'\in\Parens{\bbundle+\range(E_\alpha)}\cap S} \; F(\bbundle')
\\
\label{eq:conv:0b}
&
  \le
  \min_{\bbundle'\in\Parens{\bbundle+\range(E_\alpha)}\cap S}
  \Parens{F(\bbundle) +(\bbundle'-\bbundle)\trans\nabla F(\bbundle)+\frac12(\bbundle'-\bbundle)\trans B(\bbundle'-\bbundle)}
\\
\label{eq:conv:1}
&
  =
  \min_{\bbundle'\in\Parens{\bbundle+\range(E_\alpha)}}
  \Parens{F(\bbundle) +(\bbundle'-\bbundle)\trans\nabla F(\bbundle)+\frac12(\bbundle'-\bbundle)\trans B(\bbundle'-\bbundle)}
\\
\label{eq:conv:1a}
&
  =
  \min_{\bdelta\in\R^{\card{\alpha}}} \Parens{F(\bbundle) +\bdelta\trans\nabla_\alpha F(\bbundle)+\frac12\bdelta\trans B_{\alpha,\alpha}\bdelta}
\\
&
\label{eq:conv:2}
  =
  F(\bbundle) -\frac12\bigParens{\nabla_\alpha F(\bbundle)}\trans B^+_{\alpha,\alpha}\bigParens{\nabla_\alpha F(\bbundle)}
\enspace.
\end{align}
In \Eq{conv:0a} it suffices to consider minimization over $S$, because the objective does not increase during the optimization and $\bbundle\in S$.
In \Eq{conv:0b} we use the fact that $F$ is smooth on $S$ with respect to $B$.
\Eq{conv:1} follows, because the minimum in \Eq{conv:1} is actually attained in $S$. We show this by contradiction.
Let $F'(\bbundle')$ denote the function minimized in \Eq{conv:1} and assume that
the minimum of $F'$ over $\bbundle+\range(E_\alpha)$ is attained at $\bbundle'\not\in S$. Since
$\bbundle\in S$, the line connecting $\bbundle$ and $\bbundle'$ intersects the boundary of $S$ at some point $\bbundle''$, where $F(\bbundle'')=\lambda$
by continuity of $F$. From the foregoing, we then have
\[
  F'(\bbundle')\le F'(\bbundle)=F(\bbundle)<\lambda=F(\bbundle'')\le F'(\bbundle'')
\enspace.
\]
This however contradicts the convexity of $F'$ along the line connecting $\bbundle$ and $\bbundle'$, because $\bbundle''$ lies between $\bbundle$ and $\bbundle'$,
and thus we should have $F'(\bbundle'')\le\max\set{F'(\bbundle),F'(\bbundle')}$.
This means that the minimum in \Eq{conv:1} is indeed attained somewhere in $S$.
In \Eq{conv:1a}, we make the substitution $\bbundle'-\bbundle=E_\alpha\bdelta$, and \Eq{conv:2} then follows by \Prop{pseudo},
because $\nabla_\alpha F(\bbundle)\in\cG_\alpha(F)\subseteq\range(B_{\alpha,\alpha})$.

Taking expectation over the uniformly random choice of the block $\alpha$, we have
\begin{align}
\notag
\obj(\bbundle) - \Ex{\alpha}{ \obj(\Psi_\alpha(\bbundle)) }
&\geq \frac{1}{2\card{\cA}} \sum_{\alpha \in \cA}\bigParens{\nabla_\alpha F(\bbundle)}\trans B^+_{\alpha,\alpha}\bigParens{\nabla_\alpha F(\bbundle)}
\\
\label{eq:conv:3}
&
=\frac{1}{2\card{\cA}} \bigParens{\nabla F(\bbundle)}\trans \cD(B)^+\bigParens{\nabla F(\bbundle)}
\enspace.
\end{align}

We can also apply the lower bound on $F$ and obtain
\begin{align}
\notag
  \obj(\Psi_\alpha(\bbundle))
%&
%  =
%  \min_{\bbundle'\in\Parens{\bbundle+\range(E_\alpha)}\cap S} \; F(\bbundle')
%\\
%\notag
&
  \ge
  \min_{\bbundle'\in\Parens{\bbundle+\range(E_\alpha)}\cap S}
  \Parens{F(\bbundle) +(\bbundle'-\bbundle)\trans\nabla F(\bbundle)+\frac12(\bbundle'-\bbundle)\trans A(\bbundle'-\bbundle)}
\\
\label{eq:conv:4}
&
  \ge
  \min_{\bbundle'\in\Parens{\bbundle+\range(E_\alpha)}}
  \Parens{F(\bbundle) +(\bbundle'-\bbundle)\trans\nabla F(\bbundle)+\frac12(\bbundle'-\bbundle)\trans A(\bbundle'-\bbundle)}
\\
\notag
&
  =
  \min_{\bdelta\in\R^{\card{\alpha}}} \Parens{F(\bbundle) +\bdelta\trans\nabla_\alpha F(\bbundle)+\frac12\bdelta\trans A_{\alpha,\alpha}\bdelta}
\\
&
\label{eq:conv:5}
  =
  F(\bbundle) -\frac12\bigParens{\nabla_\alpha F(\bbundle)}\trans A^+_{\alpha,\alpha}\bigParens{\nabla_\alpha F(\bbundle)}
\enspace.
\end{align}
The steps are analogous as in analyzing the upper bound except for \Eq{conv:4}, which is now more straightforward, since
the minimum over a larger set cannot lie above a minimum over a smaller set.

Taking expectation over $\alpha$, we thus obtain
\begin{align}
\label{eq:conv:6}
\obj(\bbundle) - \Ex{\alpha}{ \obj(\Psi_\alpha(\bbundle)) }
&
\le
\frac{1}{2\card{\cA}} \bigParens{\nabla F(\bbundle)}\trans \cD(A)^+\bigParens{\nabla F(\bbundle)}
\enspace.
\end{align}

The same reasoning that gave us bounds on $F(\Psi_\alpha(\bbundle))$ can be also used to bound the optimal value $F^\star$
by noting that $F^\star=F(\Psi_{\alpha^\star}(\bbundle))$ where $\alpha^\star$ is the block containing all the coordinates,
i.e., $\alpha^\star=[N]\times[K]$. (Note that $\alpha^\star$ is not necessarily in $\cA$.)
Eqs.~\eqref{eq:conv:2} and~\eqref{eq:conv:5} thus become
\begin{align}
\label{eq:conv:Fstar:lower}
  F^\star
&\le
  F(\bbundle) -\frac12\bigParens{\nabla F(\bbundle)}\trans B^+\bigParens{\nabla F(\bbundle)}
\enspace,
\\
\label{eq:conv:Fstar:upper}
  F^\star
&\ge
  F(\bbundle) -\frac12\bigParens{\nabla F(\bbundle)}\trans A^+\bigParens{\nabla F(\bbundle)}
\enspace.
\end{align}

To finish the proof, we take the conditional expectations in \Eq{conv:3}, use the definition of $\sigmalow$ and
take the conditional expectations in \Eq{conv:Fstar:upper} to obtain
\begin{align}
\notag
\ex{\obj(\bbundle^t)\bigGiven\bbundle^{t_0}}
- \ex{\obj(\bbundle^{t+1})\bigGiven\bbundle^{t_0}}
&
\ge
\frac{1}{2\card{\cA}} \ex{(\bg^t)\trans \cD(B)^+\bg^t\BigGiven\bbundle^{t_0}}
\\[2pt]
\notag
&
\ge
\frac{1}{2\card{\cA}}\cdot\sigmalow\ex{(\bg^t)\trans A^+\bg^t\BigGiven\bbundle^{t_0}}
\\[4pt]
\label{eq:conv:final:low}
&
\ge
\frac{\sigmalow}{\card{\cA}}\BigParens{\ex{\obj(\bbundle^t)\bigGiven\bbundle^{t_0}} - F^\star}
\enspace.
\end{align}
Similarly,
\begin{equation}
\notag
\ex{\obj(\bbundle^t)\bigGiven\bbundle^{t_0}}
- \ex{\obj(\bbundle^{t+1})\bigGiven\bbundle^{t_0}}
\le
\frac{\sigmahigh}{\card{\cA}}\BigParens{\ex{\obj(\bbundle^t)\bigGiven\bbundle^{t_0}} - F^\star}
\enspace.
\end{equation}
Since the objective never increases, we can actually write
\begin{equation}
\label{eq:conv:final:high}
\ex{\obj(\bbundle^t)\bigGiven\bbundle^{t_0}}
- \ex{\obj(\bbundle^{t+1})\bigGiven\bbundle^{t_0}}
\le
\min\Braces{\frac{\sigmahigh}{\card{\cA}}\;,\;1}\cdot\BigParens{\ex{\obj(\bbundle^t)\bigGiven\bbundle^{t_0}} - F^\star}
\enspace.
\end{equation}
The theorem now follows by rearranging terms in
Eqs.~\eqref{eq:conv:final:low} and~\eqref{eq:conv:final:high}.
\end{proof}

It remains to prove the following proposition, which was used in the proof:
\begin{proposition}
\label{prop:D:range}
For any positive-semidefinite matrix $M$, we have $\range(M)\subseteq\range(\cD(M))$.
\end{proposition}
\begin{proof}
For $\alpha\in\cA$, let $P_\alpha\coloneqq E_\alpha E_\alpha^T$ be the projection matrix into
$\range E_\alpha$, and note that
\[
  \cD(\cM)=\diag_{\alpha} M_{\alpha,\alpha}=\sum_{\alpha} P_\alpha M P_\alpha
\enspace.
\]
Let $\bu\in\R^{NK}$ and let $x_\alpha=\norm{M^{1/2}P_\alpha\bu}$. Then
\begin{align}
\notag
  \bu\trans M\bu
&\;=\;
  \norm{M^{1/2}\bu}^2
 \;=\;
  \BigNorm{M^{1/2}\bigParens{\sum_{\alpha\in\cA} P_\alpha \bu}}^2
\\
\notag
&\;\le\;
  \bigParens{\sum_{\alpha\in\cA}\norm{M^{1/2}P_\alpha\bu}}^2
 \;=\;
   \bigParens{\sum_{\alpha\in\cA} x_\alpha}^2
\\
\label{eq:D:range:AQ}
&\;\le\;
   \card{\cA}\cdot\sum_{\alpha\in\cA} x_\alpha^2
\\
\notag
&\;=\;
   \card{\cA}\cdot\sum_{\alpha\in\cA} \bu\trans P_\alpha M P_\alpha\bu
 \;=\;
   \card{\cA}\cdot\bigParens{\bu\trans\cD(M)\bu}
\enspace,
\end{align}
where in \Eq{D:range:AQ} we used the inequality between the arithmetic and
quadratic mean. Thus any $\bu\in\range(M)$ is also in $\range(\cD(M))$.
\end{proof}

\subsection{Proof of \Thm{conv}}
\label{app:proof:conv}

This proof builds on several propositions proved in \App{support:conv} below.
Let $\eps_{\max}=\sqrt{\lambda-F^\star}$. For $n=1,2,\dotsc$ define the level sets $S(n)\coloneqq S(F,F^\star+\eps_{\max}^2/n^2)$. Thus,
by \Prop{S:emax}, for some $c>0$, we have
\[
    \nabla^2 F(\bbundle)\within(1\pm \underbrace{c\eps_{\max}}_{c_1}/n)H^\star
\quad
    \text{for all $\bbundle\in S(n)$}
\enspace,
\]
where we have introduced the notation $c_1\coloneqq c\eps_{\max}$. By \Prop{AB}, if $n>c_1$, then the function $F$ is strongly convex
and smooth on $S(n)$ with respect to $A(n)\coloneqq(1-c_1/n)H^\star$ and $B(n)\coloneqq(1+c_1/n)H^\star$.
Thus, $S(n)$, $A(n)$, and $B(n)$ satisfy condition (1) of \Thm{asympt:0} for $n>c_1$.

In the remainder, we only consider
the level sets $S(n)$ for $n>c_1$.
For each of them define
\[
    \sigmalow(n)\coloneqq\frac{n-c_1}{n+c_1}\sigmalow
\enspace,
\quad
    \sigmahigh(n)\coloneqq\frac{n+c_1}{n-c_1}\sigmahigh
\enspace.
\]
We next argue that they satisfy condition (2) of \Thm{asympt:0}.
It suffices to verify that condition (2) of \Thm{asympt:0} holds for $t=t_0+\ell$; let's call this limited variant condition $(2')$.
If $(2')$ is satisfied then also (2) is satisfied, because if $\bbundle^{t_0}\in S(n)$,
then also $\bbundle^{t_0+k}\in S(n)$ for any $k\ge 1$, and so we can apply condition~$(2')$ at $t=t_0+k+\ell$ and by taking
the conditional expectation we obtain the original condition~(2).

To prove that condition (2) holds for $t=t_0+\ell$, note that $S(n)\subseteq S(1)=S(F,\lambda)$, so
\Eq{thm:conv} holds whenever $\bbundle^{t_0}=\bbundle^{t-\ell}\in S(n)$. So, assuming that $n>c_1$ and $\bbundle^{t_0}\in S(n)$, we obtain
\begin{align*}
    \sigmalow(n)
&=
    \frac{n-c_1}{n+c_1}\sigmalow
\\
&\le
      \frac{n/(n+c_1)}{n/(n-c_1)}
      \cdot
      \frac{\ex{\;(\bg^t)\trans \cD(H^\star)^+\bg^t\;\bigGiven\; \bbundle^{t_0} \;}}
           {\ex{\;(\bg^t)\trans (H^\star)^+\bg^t\;\bigGiven\; \bbundle^{t_0} \;}}
%\\
%&
=
      \frac{\ex{\;(\bg^t)\trans \cD(B(n))^+\bg^t\;\bigGiven\; \bbundle^{t_0} \;}}
           {\ex{\;(\bg^t)\trans A(n)^+\bg^t\;\bigGiven\; \bbundle^{t_0} \;}}
\\
    \sigmahigh(n)
&=
    \frac{n+c_1}{n-c_1}\sigmahigh
\\
&\le
      \frac{n/(n-c_1)}{n/(n+c_1)}
      \cdot
      \frac{\ex{\;(\bg^t)\trans \cD(H^\star)^+\bg^t\;\bigGiven\; \bbundle^{t_0} \;}}
           {\ex{\;(\bg^t)\trans (H^\star)^+\bg^t\;\bigGiven\; \bbundle^{t_0} \;}}
%\\
%&
=
      \frac{\ex{\;(\bg^t)\trans \cD(A(n))^+\bg^t\;\bigGiven\; \bbundle^{t_0} \;}}
           {\ex{\;(\bg^t)\trans B(n)^+\bg^t\;\bigGiven\; \bbundle^{t_0} \;}}
\enspace.
\end{align*}

Now invoking \Prop{t:t0} with $\Sin=S(n+1)$ and $\Sout=S(n)$, we obtain that if
$\bbundle^{t_0}\in S(n+1)$ then there exists $c>0$ such that for all $t\ge t_0$
\[
  \ex{F(\bbundle^t)\given\bbundle^{t_0}}-F^\star
  \le c\kappahigh(n)^{t-t_0}
\enspace,
\]
where $\kappahigh(n)\coloneqq 1-\sigmalow(n)/\card{\cA}$. \Prop{all:reached}
now implies that $S(n+1)$ is reached with probability~1,
so $\kappahigh(n)$ is a valid upper bound on the local convergence rate.

Since the optimization starts at a deterministic point $\bbundle^0=\zero$, and the randomization
is among a finite set of choices, there is only a finite set of allocation
vectors $\bbundle^t$ that can be reached at any given iteration $t$. If the optimization problem
is non-degenerate, then none of these allocations (for any $t$) are the
actual minimizers of $F$. In that case, \Prop{t:t0} yields that if
$\bbundle^{t_0}\in S(n+1)$, and it has a non-zero probability of occurring,
then for all $t_1\ge t_0$, there exists $c>0$ such that for all $t\ge t_1$
\[
  \ex{F(\bbundle^t)\given\bbundle^{t_1}}-F^\star
  \ge c\kappalow(n)^{t-t_1}
\enspace,
\]
where $\kappalow(n)\coloneqq 1-\sigmahigh(n)/\card{\cA}$.
Since $S(n+1)$ is reached with probability 1, this means
that any valid upper bound must be greater than $\kappalow(n)$.

The theorem now follows, because $\kappahigh(n)\to 1-\sigmalow/\card{\cA}$,
and $\kappalow\to 1-\sigmahigh/\card{\cA}$, as $n\to\infty$.

\subsection{Supporting Propositions for \Thm{conv}}
\label{app:support:conv}

%The proof builds on \Thm{asympt:0} as well as on the properties of \convplus functions.

Throughout the propositions below,
let $F$ be a \convplus function attaining a minimum. Let $P$ be a projection on $\cG(F)$.
Since $F$ attains a minimum, its gradient set includes zero, and therefore in \Prop{convplus}
we have $\ba=\zero$. This means that
\[
  F(\bbundle)=F(P\bbundle)
\enspace,
\]
so we can assume, without loss of generality, that $\bbundle^\star\in\cG(F)$.
Let $H^\star\coloneqq\nabla^2 F(\bbundle^\star)$. Finally,
recall that $S(F,\lambda)$ denotes a level set (see \Def{level:set}).

\begin{proposition}
\label{prop:S0:compact}
For any $\lambda>F^\star$, $S(F,\lambda)=S_0+\cG(F)^\perp$ where $S_0\subseteq\cG(F)$ is compact.
\end{proposition}
\begin{proof}
Since $F(\bbundle)=F(P\bbundle)$, any level set $S$ can be written as $S=S_0+\cG(F)^\perp$, where $S_0=S\cap\cG(F)$.
Now if $\lambda>F^\star$ then $S$ is non-empty and closed and hence so is $S_0$. It remains to argue that it is bounded.
By \convexityplus, $F$ is strictly convex on $\cG(F)$, so the minimum of $F$ on
the sphere $\set{\bbundle\in \cG(F):\:\norm{\bbundle-\bbundle^\star}=1}$ must be some $\lambda_1>F^\star$.
By convexity, $F(\bbundle)-F^\star\ge\lambda_1\norm{\bbundle-\bbundle^\star}$ for all $\bbundle\in\cG(F)$.
Since $S_0\subseteq\cG(F)$ and $F(\bbundle)\le\lambda$ for $\bbundle\in S_0$, the set $S_0$ must be bounded.
\end{proof}

\begin{proposition}
\label{prop:S:emax}
Let $\eps_{\max}>0$. Then there exists a constant $c>0$ such that for all $0<\eps\le\eps_{\max}$, and all $\bbundle\in S(F,F^\star+\eps^2)$, we have
\[
  \norm{P\bbundle-\bbundle^\star}\le c\eps
\enspace,
  \nabla^2 F(\bbundle)\within(1\pm c\eps) H^\star
\enspace.
\]
\end{proposition}
\begin{proof}
Let $S\coloneqq S(F,F^\star+\eps_{\max}^2)$ and $S_0\coloneqq S\cap\cG(F)$, which is compact by \Prop{S0:compact}. Let
\[
  \sigma_1\coloneqq\min_{\bbundle\in S_0} \lambda_{\min}(\nabla^2 F(\bbundle), P)
\enspace.
\]
Note that $\lambda_{\min}(\nabla^2 F(\bbundle),P)>0$
on the compact set $S_0$, and $\nabla^2 F(\cdot)$ and $\lambda_{\min}(\cdot,P)$ are
continuous, so $\sigma_1>0$. Therefore, $F$ is strictly convex on $S$ with the strict convexity
constant $\sigma_1$. Using the fact that $\nabla F(\bbundle^\star)=\zero$, we then have
for any $\bbundle\in S$,
\begin{align*}
 F(\bbundle)
 =F(P\bbundle)
 \ge
 F^\star+\frac{1}{2}\sigma_1\norm{P\bbundle-\bbundle^\star}^2
\enspace.
\end{align*}
Therefore, if $\bbundle\in S(F,F^\star+\eps^2)\subseteq S$ then
\[
 \norm{P\bbundle-\bbundle^\star}\le \eps\sqrt{2/\sigma_1}
\enspace.
\]
For the bound on the Hessian, note that
since the third derivative of $F$ is continuous, it is upper bounded on $S_0$. Therefore,
the Hessian is Lipschitz with some constant $L$ on $S_0$, and
so $\norm{\nabla^2 F(\bbundle)-H^\star}\le L\norm{P\bbundle-\bbundle^\star}$.
Thus, since $\range(\nabla^2 F(\bbundle)-H^\star)\subseteq\cG(F)$, we have
\[
  \nabla^2 F(\bbundle)\within H^\star\pm L\norm{P\bbundle-\bbundle^\star}P
\enspace.
\]
Since $\sigma_2 P\preceq H^\star$ for some $\sigma_2>0$, we obtain
\[
  \nabla^2 F(\bbundle)\within (1\pm L\sigma_2^{-1}\norm{P\bbundle-\bbundle^\star})H^\star
\enspace.
\]
Thus, for $\bbundle\in S(F,F^\star+\eps^2)\subseteq S$ we have
\[
  \nabla^2 F(\bbundle)\within (1\pm \eps L\sigma_2^{-1}\sqrt{2/\sigma_1})H^\star
\enspace.
\]
Setting $c\coloneqq\max\set{\sqrt{2/\sigma_1},\,L\sigma_2^{-1}\sqrt{2/\sigma_1}}$ then
proves the proposition.
\end{proof}

\begin{proposition}
\label{prop:AB}
Let $S$ be a convex set and $A$ and $B$ positive-semidefinite matrices such that
$A\preceq \nabla^2 F(\bbundle)\preceq B$ for
all $\bbundle\in S$. Then $F$ is strongly convex and smooth on $S$ with respect to $A$ and $B$.
\end{proposition}
\begin{proof}
Let $\bbundle\in S$, $\bbundle+\vdelta\in S$. Then from the 2nd-order Taylor
expansion, we have
\[
  F(\bbundle+\vdelta)-F(\bbundle)-\vdelta\nabla F(\bbundle)
  =
  \frac12\vdelta\trans[\nabla^2 F(\bbundle')]\vdelta
\enspace,
\]
where $\bbundle'\in S$. The proposition now follows, because
$A\preceq[\nabla^2 F(\bbundle')]\preceq B$.
\end{proof}

\begin{proposition}
\label{prop:all:reached}
For any $\lambda>F^\star$, with probability 1, the randomized block-coordinate descent algorithm
with the objective $F$ reaches an iteration $t$ in which $\bbundle^t\in S(F,\lambda)$.
\end{proposition}
\begin{proof}
Since the proposition holds for $\lambda\ge F(\bbundle^0)$, consider the
case $F^\star<\lambda<F(\bbundle^0)$, and in particular assume $F^\star<F(\bbundle^0)$.
We prove the proposition by applying \Thm{asympt:0}.

Let $S\coloneqq S(F,F(\bbundle^0)+1)$ and $S_0\coloneqq S\cap\cG(F)$, which is compact by
\Prop{S0:compact}. Let
\[
  c_{\min}\coloneqq\min_{\bbundle\in S_0} \lambda_{\min}(\nabla^2 F(\bbundle), P)
\enspace,
\quad
  c_{\max}\coloneqq\max_{\bbundle\in S_0} \lambda_{\max}(\nabla^2 F(\bbundle), P)
\enspace.
\]
Note that $\lambda_{\min}(\nabla^2 F(\bbundle),P)>0$
on the compact set $S_0$, and $\nabla^2 F(\cdot)$ and $\lambda_{\min}(\cdot,P)$ are
continuous, so $c_{\min}>0$. Similarly, since $\lambda_{\max}(\nabla^2 F(\bbundle), P)<\infty$
on $S_0$, the continuity yields $c_{\max}<\infty$.
Since $\nabla^2 F(\bbundle)=\nabla^2 F(P\bbundle)$, we have that for all $\bbundle\in S$
\[
   c_{\min} P
\preceq
   \nabla^2 F(\bbundle)
\preceq
   c_{\max} P
\enspace.
\]
Therefore, by \Prop{AB}, $F$ is strongly convex and smooth on $S$ with respect to $A\coloneqq c_{\min}P$ and
$B\coloneqq c_{\max}P$. Let
\[
  \sigmalow\coloneqq
  \lambda_{\min}\bigParens{\cD(B)^+,\,A^+}
  =\lambda_{\min}\bigParens{c_{\max}^{-1}\cD(P)^+,\,c_{\min}^{-1}P}
  =\frac{c_{\min}}{c_{\max}}\cdot\lambda_{\min}(\cD(P)^+,\,P)
\enspace,
\]
which is positive, because $\range(P)\subseteq\range(\cD(P))$ by \Prop{D:range}.

Now, by \Prop{t:t0}, with $\Sin=S(F,F(\bbundle^0))$, $\Sout=S(F,F(\bbundle^0)+1)=S$,
and the above matrices $A$ and~$B$, the scalar $\sigmalow$, and $\ell=0$, we obtain that
for some constant $c$ and $\kappa\coloneqq(1-\sigmalow/\card{\cA})<1$,
\[
  \ex{F(\bbundle^t)}-F^\star
  =\ex{F(\bbundle^t)\given \bbundle^0}-F^\star
  \le c\kappa^t
\enspace.
\]
To finish the proof, we will appeal to Borel-Cantelli lemma and show that the algorithm
must reach $S(F,\lambda)$ with probability 1. Specifically, note that by the Markov inequality
\begin{align*}
  \sum_{t=1}^\infty
  \prob{F(\bbundle^t)\ge\lambda}
&=
  \sum_{t=1}^\infty
  \prob{F(\bbundle^t)-F^\star\ge\lambda-F^\star}
\\
&\le
  \sum_{t=1}^\infty
  \frac{c\kappa^t}{\lambda-F^\star}
=
  \frac{c\kappa}{(1-\kappa)(\lambda-F^\star)}
<
\infty
\enspace,
\end{align*}
so with probability 1, only a finite number of the events $\set{F(\bbundle^t)\ge\lambda}$ will occur;
in other words, the level set $S(F,\lambda)$ is reached with probability 1.
\end{proof}

\subsection{Local Convergence of the Market}
\label{app:localmarketconv}

Throughout this section,
%and next,
we assume that $C$ is \convplus, which implies that $F$ is \convplus as well.
Our key tool for the analysis of the convergence error of the market is \Thm{conv}.
Therefore, we need to analyze the gradient and Hessian of $F$.
%In our analysis,
%we will use the matrix bounds of the form $A\within B\pm C$ to mean that $B-C\preceq A\preceq B+C$,
%where $A \preceq B$ iff $\bu\trans A\bu\le\bu\trans B\bu$ for all $\bu$.
%
We begin the analysis by deriving explicit expressions for
$\nabla F$ and $\nabla^2 F$ using the gradients and Hessians of $T$ and $C$. It will be convenient to do so for
trader-level blocks $\nabla_i$ and $\nabla^2_{ij}$.

Given an allocation vector $\bbundle\in R^{NK}$, the associated market price (the gradient of the cost)
will be denoted $\bprice_0(\bbundle)$ and the gradients of trader potentials will be denoted $\bprice_i(\bbundle)$:
\begin{align*}
%\label{eq:mu:0}
  \bprice_0(\bbundle)
  &\coloneqq
  \nabla C_b\bigParens{\textstyle\sum_{i=1}^N \bbundle_i}
  =
  \nabla C\bigParens{\textstyle\sum_{i=1}^N \bbundle_i/b}
\\
%\label{eq:mu:i}
  \bprice_i(\bbundle)
  &\coloneqq
  \nabla F_i(-\bbundle_i)
  =
  \nabla T(\bttheta_{i}-a_i\bbundle_i)
\qquad
\text{for $i\in[N]$}
\end{align*}
where $T$ is the log partition function.

The gradient of $F$ is composed of blocks
\begin{align}
\notag
  \nabla_i F(\bbundle)
   &=-\nabla F_i(-\bbundle_i)+\nabla C_b\bigParens{\textstyle\sum_i\bbundle_i}
\\
\label{eq:grad:i}
   &=-\bprice_i(\bbundle)+\bprice_0(\bbundle)
\enspace.
\end{align}
%where we plugged in the definitions of $\bprice_i$ and $\bprice_0$
%from \Eqs{mu:0}{mu:i}.

For the Hessian, first consider $\nabla_{ii} F$:
\begin{align}
\notag
  \nabla^2_{ii} F(\bbundle)
  &=\nabla^2 F_i(-\bbundle_i)+\nabla^2 C_b\bigParens{\textstyle\sum_i\bbundle_i}
\\
\notag
  &=a_i \nabla^2 T(\bttheta_{i}-a_i\bbundle_i)
    +
    \frac{1}{b} \nabla^2 C\BigParens{{\textstyle (\sum_i\bbundle_i)}/b}
\\
\label{eq:Hess:ii}
  &=a_i H_T(\bprice_i(\bbundle))
    +
    \frac{1}{b} H_C(\bprice_0(\bbundle))
\enspace.
\end{align}
Here, recall that for a \convplus function $f$, its Hessian at any given point is
only a function of the gradient at that point, which is denoted $H_f$.
In
\Eq{Hess:ii}, we express $\nabla^2 T$ and $\nabla^2 C$ using the respective functions
$H_T$ and $H_C$ and the definitions of $\bprice_i$ and $\bprice_0$.

For $i\ne j$, the block $\nabla^2_{ij} F$ is
\begin{align}
\notag
  \nabla^2_{ij} F(\bbundle)
   &=\nabla^2 C_b\bigParens{{\textstyle\sum_i\bbundle_i}}
    =\frac{1}{b} \nabla^2 C\BigParens{{\textstyle (\sum_i\bbundle_i)}/b}
\\
\label{eq:Hess:ij}
   &=\frac{1}{b} H_C(\bprice_0(\bbundle))
\enspace.
\end{align}

At any optimum $\bbundle^\star$, we have $\bprice_i(\bbundle^\star)=\bprice_0(\bbundle^\star)=\eqbprice$.
Thus, using the Kronecker product notation,
the Hessian of $F$ at $\bbundle^\star$ can be expressed as
\begin{equation}
\label{eq:Hess:rstar}
  \nabla^2 F(\bbundle^\star)
  = D\otimes H_T(\eqbprice)
    +\frac{\one\one\trans}{b}\otimes H_C(\eqbprice)
\enspace,
\end{equation}
where $D\coloneqq\diag_{i\in[N]} a_i$
%is the diagonal matrix enumerating risk aversion levels of the traders
and $\one$ is the $N$-dimensional all-ones vector

Using local Lipschitz property of $H_T$ and $H_C$ (\Prop{H:Lipschitz}) and the fact that $\norm{\eqbprice-\aggbprice}=O(b)$
(\Thm{bias:global}), we immediately obtain the following asymptotic expression for $\nabla^2 F(\bbundle^\star)$ as $b\to 0$:
\begin{proposition}
\label{prop:H}
Let $H^\star\coloneqq\nabla^2 F(\bbundle^\star)$ and $D\coloneqq\diag_i a_i$. Then
\[
  H^\star\within(1\pm O(b))\Parens{\;
     D\otimes H_T(\aggbprice)
     +(\one\one\trans)\otimes \frac{1}{b}H_C(\aggbprice)
     \;}
\enspace.
\]
% where $D=\diag_{i}a_i$.
% and $\eps_b=O(b)$ as $b\to 0$.
\end{proposition}
%\begin{proof}
%\mdcomment{TODO}
%\end{proof}

%To make the exposition more concise, we will write $H_F$, $H_C$ and $H_T$ instead of the more verbose
%$H_F(b;C)$, $H_C(b;C)$ and $H_T(b;C)$.
We next derive asymptotic formulas for matrices $\cD(H^\star)^+$ and $(H^\star)^+$, from which we will immediately
obtain a lower bound on strong convexity via \Eq{sigmalow}.

Recall that $\cA$ is the decomposition
of the coordinates $[N]\times[K]$, but for our two dynamics (\ASD and \SSD), this decomposition has
a special structure. This structure is described by a decomposition $\cB$ of $[K]$, which
is then applied to each trader, that is
$\cA=\bigSet{\set{i}\times\beta:\:i\in[N],\beta\in\cB}$. For $M\in\R^{K\times K}$, we use the notation $\cD_\cB(M)$
to describe $\diag_{\beta\in\cB} M_{\beta\beta}$. For $M\in\R^{NK\times NK}$, we continue writing $\cD(M)$ instead of a more
explicit $\cD_\cA(M)$.

In stating our results,
%in this and next section,
we use the following shorthands, some of which have been already introduced:
\[
  H^\star\coloneqq\nabla^2 F(\bbundle^\star),
  \quad
  H_T\coloneqq H_T(\aggbprice),
  \quad
  H_C\coloneqq H_C(\aggbprice),
  \quad
  D\coloneqq\diag_i a_i,
  \quad
  P=I_N-\one\one\trans/N.
\]
The matrix $P$ is the projection matrix on the set of centered vectors, i.e.,
vectors $\bu$ in $\R^N$ such that $\one\trans\bu=0$.
With this notation, the pseudoinverses $\cD(H^\star)^+$ and $(H^\star)^+$
are characterized in the following theorem:
\begin{theorem}
\label{thm:conv:1}
Let $M_1\coloneqq I_N\otimes \cD_\cB(H_C)^+$ and $M_2\coloneqq (PDP)^+\otimes H_T^+$.
Then, as $b\to 0$,
\begin{align}
\label{eq:conv:1:D}
  \cD(H^\star)^+
  &\within
  (1\pm O(b))\cdot bM_1
\enspace,
\\
\label{eq:conv:1:H}
  (H^\star)^+
  &\within
%(1\pm O(b))\cdot
  M_2\pm O(b)M_1
\enspace.
\end{align}
Local strong convexity is bounded from below
by
\begin{align}
  \sigmalow
&
  = b\cdot\lambda_{\min}\Parens{\;(M_2^{1/2})^+ M_1 (M_2^{1/2})^+\;}
  -O(b^2)
\\
\label{eq:conv:1:sigmalow}
&
  = b\cdot\lambda_{\min}(PDP)\cdot\lambda_{\min}\bigParens{\;H_T^{1/2}\cD_\cB(H_C)^+ H_T^{1/2}\;}
  -O(b^2)
%\;\;{-}\;\;O(b^2)
\enspace,
\end{align}
where $\lambda_{\min}(\cdot)$ denotes the smallest positive eigenvalue of a matrix.
\begin{comment}
Let $H^\star\coloneqq\nabla^2 F(\bbundle^\star)$, $H_T\coloneqq H_T(\aggbprice)$ and $H_C\coloneqq H_C(\aggbprice)$.
Then
\begin{align}
\label{eq:conv:1:D}
  \cD(H^\star)^+
  &\within
  (1\pm O(b))\cdot\bigParens{\;bI_N\otimes \cD_\cB(H_C)^+\;}
\enspace,
\\
\label{eq:conv:1:H}
  (H^\star)^+
  &\within
%(1\pm O(b))\cdot
  \bigParens{\;(PDP)^+\otimes H_T^+\;}\pm O(b)\cdot\bigParens{\; I_N\otimes \cD_\cB(H_C)^+\;}
\enspace,
\end{align}
where $D=\diag_i a_i$ and $P=I_N-\one\one\trans/N$.
Local strong convexity is bounded from below
by
%
\begin{equation}
\label{eq:conv:1:sigmalow}
  \sigmalow= b\cdot\lambda_{\min}(PDP)\cdot\lambda_{\min}\bigParens{\;H_T^{1/2}\cD_\cB(H_C)^+ H_T^{1/2}\;}-O(b^2)
%\;\;{-}\;\;O(b^2)
\enspace,
\end{equation}
%
where $\lambda_{\min}(\cdot)$ denotes the smallest positive eigenvalue of a matrix.
\end{comment}
\end{theorem}
The matrices $M_1$ and $M_2$ in the statement of the theorem do not depend on the liquidity parameter $b$. The matrix $M_2$,
which is the dominant part of the Hessian pseudoinverse $(H^\star)^+$, is also independent of the trader dynamics and the
cost function.
%, and only depends on trader utilities, namely, the risk aversion coefficients and the market-clearing prices.
On the other hand,
the matrix $M_1$ reflects the cost function and dynamics. The pseudoinverse $\cD(H^\star)^+$ approximately equals $bM_1$ as $b\to 0$.
The main implication is that $\sigmalow=\Omega(b)$. This yields linear convergence rate bound $\kappahigh=1-\Omega(b)$, which
suggests worse convergence as $b\to 0$. However, in the absence of a matching lower bound,
we cannot conclude that the actual convergence gets
worse as $b\to 0$. In \App{conv:ASD}, we derive a matching bound $\sigmahigh=O(b)$ for \ASD. Thus, for \ASD,
it is not possible to achieve a linear convergence rate better than $1-\Theta(b)$. (We conjecture similar behavior for \SSD.)
This means there is a tradeoff between convergence, which slows down as $b\to 0$, and the market-maker bias,
which gets smaller.

\begin{proof}[Proof of \Thm{conv:1}]
From \Prop{H}, we know that $H^\star\within(1\pm O(b))M$ where
\begin{equation}
\label{eq:def:M}
  M\coloneqq D\otimes H_T
     +(\one\one\trans)\otimes \frac{1}{b}H_C
\enspace.
\end{equation}
\Prop{H} also implies $\cD(H^\star)\within(1\pm O(b))\cD(M)$. Thus,
by \Prop{inv:ineq}, we obtain
\begin{align}
\label{eq:conv:DM}
   \cD(H^\star)^+
   &\within(1\pm O(b))\cD(M)^+
\enspace,
\\
\label{eq:conv:M}
   (H^\star)^+
   &\within(1\pm O(b)) M^+
\enspace.
\end{align}
The analysis will therefore focus on $M$ and convert to $H^\star$ only in the last step.

We begin with the analysis of $\cD(M)$. Specifically, consider the block $M_{\alpha\alpha}$
%=E_\alpha\trans H_F E_\alpha,
where $\alpha=\set{i}\times\beta$ for some $\beta\in\cB$. From the definition of $M$
\[
  M_{\alpha\alpha}=\frac{1}{b} H_{C,\beta\beta} + a_i H_{T,\beta\beta}.
\]
Since $\range(H_T)=\cG(T)\subseteq\cG(C)=\range(H_C)$, there is some constant $c_1$ such that $H_T\preceq c_1 H_C$, so we can write
\begin{equation}
\label{eq:convF:1}
  M_{\alpha\alpha}\within\Parens{\frac{1}{b}\pm a_i c_1} H_{C,\beta\beta}.
\end{equation}
Setting $c_2=(\max_i a_i)c_1$, and combining \Eq{convF:1} across all blocks $\alpha=\set{i}\times\beta$, we thus obtain
\[
  \cD(M)\within\Parens{\frac{1}{b}\pm c_2}\bigParens{\;I_N\otimes\cD_\cB(H_C)\;}
\enspace.
\]
Therefore, by \Prop{inv:ineq},
\begin{align*}
  \cD(M)^+
  &\within\frac{b}{1\pm b c_2}\bigParens{\;I_N\otimes\cD_\cB(H_C)^+\;}
\\
  &=b(1\pm O(b))\bigParens{\;I_N\otimes\cD_\cB(H_C)^+\;}
\enspace.
\end{align*}
The bound on $\cD(H^\star)^+$ now follows by \Eq{conv:M}.

We next bound $(H^\star)^+$ by analyzing $M^+$. First, decompose the matrix $D$ into blocks corresponding to
the ranges of the projection matrices $P$ and $I_N-P$.
%=\one\one\trans/N$.
Let $A=PDP$, $B=PD(I_N-P)$ and $X=(I_N-P)D(I_N-P)$. Thus,
\[
  D=A+B+B\trans+X
\enspace.
\]
Using the decomposition of $D$, we can decompose $M$ in order to carry out blockwise inversion:
\begin{align}
\notag
 M
 &=D\otimes H_T + \frac{N}{b}(I_N-P)\otimes H_C
\\
\label{eq:convF:2}
 &=A\otimes H_T + (B+B\trans)\otimes H_T + \Parens{X\otimes H_T +\frac{N}{b}(I_N-P)\otimes H_C}
\enspace.
\end{align}
We first analyze the Schur complement matrix, which appears in the blockwise inverse:
\begin{align*}
\notag
   S
&\coloneqq \Parens{X\otimes H_T +\frac{N}{b}(I_N-P)\otimes H_C}
   - (B\trans\otimes H_T)(A^+\otimes H_T^+)(B\trans\otimes H_T)
\\
&= \frac{N}{b}(I_N-P)\otimes H_C
    +
    \bigParens{X-B\trans A^+B}\otimes H_T
\enspace.
\end{align*}
As we argued before, $\range(H_T)\subseteq\range(H_C)$. Also $\range(X)\subseteq\range(I_N-P)$ and
$\range(B\trans A^+B)\subseteq\range(I_N-P)$, so for some $c_3>0$, we have
\[
  \bigParens{X-B\trans A^+B}\otimes H_T\within \pm c_3 (I_N-P)\otimes H_C
\enspace,
\]
and therefore
\begin{equation}
\label{eq:convF:2x}
  S\within\Parens{\frac{N}{b}\pm c_3}\bigParens{\;(I_N-P)\otimes H_C\;}
\enspace.
\end{equation}
We now apply blockwise inversion (\Prop{woodbury}) to \Eq{convF:2}, with the bounds of \Eq{convF:2x} on the Schur complement
to obtain
\[
 M^+
 \within
 A^+\otimes H_T^+
 +\frac{b}{N\pm b c_3} Y
\]
where
\[
 Y\coloneqq\BigParens{I_{NK}-(A^+\otimes H_T^+)(B\otimes H_T)}
   \bigParens{(I_N-P)\otimes H_C^+}
   \BigParens{I_{NK}-(A^+\otimes H_T^+)(B\otimes H_T)}\trans
\]
is a positive-semidefinite matrix. Finally, invoking \Eq{conv:M}, we obtain
\begin{align*}
 (H^\star)^+
 &\within
 (1\pm O(b))\bigParens{\;A^+\otimes H_T^+\;}
 +b(1\pm O(b)) Y
\\
 &\within
 \bigParens{\;A^+\otimes H_T^+\;}
 \pm O(b)\bigParens{\;A^+\otimes H_T^+ + Y\;}
\\
 &\within
 \bigParens{\;A^+\otimes H_T^+\;}
 \pm O(b)\bigParens{\;I_N\otimes \cD(H_C)^+\;}
\enspace,
\end{align*}
where in the last line we used that
$\range(A^+\otimes H_T^+ + Y)\subseteq\range(I_N\otimes H_C^+)\subseteq\range(I_N\otimes \cD_\cB(H_C)^+)$ because
$\range(H_T)\subseteq\range(H_C)\subseteq\range(\cD_\cB(H_C))$.

Finally, to prove \Eq{conv:1:sigmalow}, we use \Eq{sigmalow}. First, note that
$(H^\star)^+\within M_2\pm O(b)M_1$, so
\[
   \range(M_2)\subseteq\range(H^\star)=\cG(F)\subseteq\range(I_N\otimes H_C)\subseteq\range(M_1)
\enspace.
\]
Therefore, if $\bu\in\cG(F)\wo\range(M_2)$ we have $\bu\trans M_1\bu>0$, but $\bu\trans M_2\bu=0$, so
\[
\min_{\bu\in\cG(F)\wo\set{\zero}}\;
    \frac{\bu\trans M_1\bu}{\bu\trans M_2\bu}
=
\min_{\bu\in\range(M_2)\wo\set{\zero}}\;
    \frac{\bu\trans M_1\bu}{\bu\trans M_2\bu}
=
\lambda_{\min}(M_1,M_2)
>0
\]
and so
%
\begin{comment}
&= \min_{\bu\in\range(M_2)\wo\set{\zero}}
    \frac
    {\bu\trans \Parens{\;M_2^{1/2}(M_2^{1/2})^+\;}
               M_1
               \Parens{\;(M_2^{1/2})^+M_2^{1/2}\;}
     \bu}
    {\bu\trans M_2^{1/2} M_2^{1/2}\bu}
\\
&= \min_{\bv\in\range(M_2)\wo\set{\zero}}
    \frac{\bv\trans (M_2^{1/2})^+ M_1 (M_2^{1/2})^+ \bv}
    {\bv\trans\bv}
\end{comment}
%
\[
   \frac{\bu\trans M_2\bu}{\bu\trans M_1\bu}\le\lambda_{\min}^{-1}(M_1,M_2)
\]
for any $\bu\in\cG(F)\wo\set{\zero}$.

From the bounds in \Eqs{conv:1:D}{conv:1:H}, there exists a constant $c$ such that for $b$ sufficiently small, and for all $\bu\in\cG(F)\wo\set{\zero}$,
\begin{align*}
  \frac{\bu\trans \cD(H)^+\bu}{\bu\trans H^+\bu}
&\ge
  \frac{(1-cb)\cdot \bu\trans(bM_1)\bu}
      {\bu\trans M_2\bu + cb (\bu\trans M_1\bu)}
\\
&=
  b\cdot
  \frac{1-cb}
       {\frac{\bu\trans M_2\bu}{\bu\trans M_1\bu} + cb}
\\
&\ge
  b\cdot
  \frac{1-cb}
       {\lambda_{\min}^{-1}(M_1,M_2) + cb}
\\
&\ge b\cdot(\lambda_{\min}(M_1,M_2)-O(b))
\enspace.
\end{align*}
The bound on $\sigmalow$ now follows by \Eq{sigmalow}, after noting that
\begin{equation}
\tag*{\qedhere}
  \lambda_{\min}(M_1,M_2)
  =
  \lambda_{\min}\Parens{\;(M_2^{1/2})^+ M_1 (M_2^{1/2})^+\;}
  =
  \lambda_{\min}(PDP)\cdot\lambda_{\min}\bigParens{\;H_T^{1/2}\cD_\cB(H_C)^+ H_T^{1/2}\;}
\enspace.
\end{equation}
\end{proof}

\subsubsection{Tighter Analysis of the All-securities Dynamics}
\label{app:conv:ASD}

In \Thm{conv:1} we derived a worst-case bound on the curvature, valid
across all possible directions that a gradient can take. In our tighter analysis of \ASD, we derive a tighter bound on the expected curvature,
exploiting the fact that the updates are chosen uniformly at random. While our analysis only applies to \ASD, we conjecture
that a similar style of analysis can also work for \SSD.

We will index blocks by $i$ rather than $\alpha$, since each block consists of
all the coordinates controlled by the trader $i$.
%Thus, for instance, $\nabla_i$ is the partial gradient with
%respect to $\bbundle_i$.

We begin by a detailed analysis of how the block-coordinate updates affect the
value of the gradient.
Let $\bbundle$ by the current iterate.
Consider the update $\Psi_i$, which optimizes over the coordinates controlled by trader $i$ (see Eq.~\ref{eq:update}), and let $\bbundle'=\Psi_i(\bbundle)$ be the new iterate. By the optimality of the update,
we have
\[
 \bprice_i(\bbundle')
 =
 \nabla F_i(-\bbundle_i')
 =
 \nabla C_b\bigParens{\textstyle\sum_j\bbundle_j'}
 =
 \bprice_0(\bbundle')
\]
and therefore, by \Eq{grad:i}, for all $j\in[N]$,
\[
  \nabla_j F(\bbundle')
  =-\bprice_j(\bbundle')+\bprice_0(\bbundle')
  =-\bprice_j(\bbundle')+\bprice_i(\bbundle')
\enspace.
\]
Thus, after the first update, the gradient $\nabla F(\bbundle)$ can
be expressed using pairwise differences of $\bprice_j(\bbundle)$. When the update $\Psi_i$ is performed, the value of $\bprice_i$
changes, whereas $\bprice_j$ for $j\ne i$ is unchanged. The amount of change in $\bprice_i$
will be denoted as $\bdelta_i$:
\[
  \bdelta_i(\bbundle)\coloneqq \bprice_i(\bbundle')-\bprice_i(\bbundle)
%\nabla F_i(-\bbundle'_i) - \nabla F_i(-\bbundle_i)
\text{ where }
  \bbundle'=\Psi_i(\bbundle)
\enspace.
\]
This is locally bounded as follows:
\begin{lemma}
\label{lem:delta}
There exists constants $b_0,c>0$ such that for every $b\le b_0$ there exists a proper level set $S$
%\coloneqq S_{\lambda(b;C)}$, where $\lambda(b;C)>F^\star(b;C)$,
such that if $\bbundle\in S$ then
\[
  \Norm{\bdelta_i(\bbundle)}\le c b\bigNorm{\bprice_i(\bbundle)-\bprice_0(\bbundle)}
%\nabla F_i(-\bbundle_i)-\nabla C_b\bigParens{\textstyle\sum_i\bbundle_i}
\]
for all $i\in[N]$.
\end{lemma}
\begin{proof}
%\mdcomment{possibly replace by: $T$ and $C$ are \convplus in the neighborhood of $\aggbprice$}
Let $\eps\in(0,1)$. Since $T$ and $C$ are \convplus, by \Prop{H:Lipschitz} it is possible to pick $\delta$ such that
\[
   H_T(\bprice)\within(1\pm\eps) H_T(\aggbprice)
\enspace,
\qquad
   H_C(\bprice)\within(1\pm\eps) H_C(\aggbprice)
\]
whenever $\norm{\bprice-\aggbprice}\le\delta$. Pick $b_0$ small enough such that $\norm{\eqbprice(b;C)-\aggbprice}\le\delta/2$
for all $b\le b_0$. Fix some $b\le b_0$ and pick a level set $S$ such that for all $\bbundle\in S$, $\norm{\bprice_i(\bbundle)-\eqbprice}\le\delta/2$
for all $i$ and $\Norm{\bprice_0(\bbundle)-\eqbprice}\le\delta/2$. Thus, for any $\bbundle\in S$, we have
\[
   H_T(\bprice_i(\bbundle))\within(1\pm\eps) H_T(\aggbprice)
   \text{ for all $i$},
\qquad
   H_C(\bprice_0(\bbundle))\within(1\pm\eps) H_C(\aggbprice)
\enspace.
\]

Now let $\bbundle\in S$ and let $\bbundle'=\Psi_i(\bbundle)$. By the optimality of $\bbundle'$, we know that $\nabla_i F(\bbundle')=\zero$. From the
mean value theorem applied to $\nabla_i F$, there exists some $\bq$ on the line segment connecting $\bbundle$
and $\bbundle'$ such that
\[
   \zero=\nabla_i F(\bbundle')=\nabla_i F(\bbundle)+\nabla^2_{ii} F(\bq)(\bbundle'-\bbundle)
\enspace.
\]
Since $\bbundle'$ and $\bbundle$ differ only in block $i$, we obtain
\begin{equation}
\label{eq:convBlock:1}
  P(\bbundle'_i-\bbundle_i)=-\Parens{\nabla^2_{ii} F(\bq)}^+\nabla_i F(\bbundle)
\end{equation}
where $P$ is the projection on $\range\bigParens{\nabla^2_{ii} F(\bq)}=\cG(C)$. Now applying the mean value theorem to $\nabla F_i$, we obtain
that for some $\bq'$ on the line segment connecting $\bbundle$ and $\bbundle'$, we have
\begin{align}
\notag
  \bdelta_i(\bbundle)
  &=\nabla F_i(-\bbundle'_i) - \nabla F_i(-\bbundle_i)
\\
\notag
  &=\nabla^2 F_i(-\bq'_i)(-\bbundle'_i+\bbundle_i)
\\
\label{eq:convBlock:2}
  &=\nabla^2 F_i(-\bq'_i)\Parens{\nabla^2_{ii} F(\bq)}^+\nabla_i F(\bbundle)
\enspace,
\end{align}
where in \Eq{convBlock:2} we used \Eq{convBlock:1}.
Now both $\bbundle$ and $\bbundle'$ are in the level set $S$ and so is the line segment connecting them, which includes the
points $\bq$ and $\bq'$. Therefore,
\begin{equation}
\label{eq:convBlock:3}
  \nabla^2 F_i(-\bq'_i)=  a_i H_T(\bprice_i(\bq')) \preceq a_i(1+\eps) H_T(\aggbprice)
\end{equation}
and
\begin{align*}
  \nabla^2_{ii} F(\bq)
  &=  a_i H_T(\bprice_i(\bq)) + \frac{1}{b} H_C(\bprice_0(\bq))
\\
  &\succeq (1-\eps)\Parens{a_i H_T(\aggbprice) +\frac{1}{b} H_C(\aggbprice)}
\\
  &\succeq \frac{1-\eps}{b} H_C(\aggbprice)
\enspace.
\end{align*}
Thus, also
\begin{equation}
\label{eq:convBlock:4}
  \Parens{\nabla^2_{ii} F(\bq)}^+
  \preceq
  \frac{b}{1-\eps} H_C^+(\aggbprice)
\enspace.
\end{equation}
Plugging \Eq{convBlock:3} and \Eq{convBlock:4} into \Eq{convBlock:2}, we obtain
\[
  \Norm{\bdelta_i(\bbundle)}
  \le
  a_i(1+\eps)
  \Norm{H_T(\aggbprice)}
  \cdot
  \frac{b}{1-\eps}
  \Norm{H_C^+(\aggbprice)}
  \Norm{\nabla_i F(\bbundle)}
\enspace,
\]
finishing the proof, since $\nabla_i F(\bbundle)=-\bprice_i(\bbundle)+\bprice_0(\bbundle)$.
\end{proof}

Using the lemma, we can now prove bounds for \ASD:

\begin{theorem}
\label{thm:conv:2}
Consider the all-securities dynamics.
Let $M_1'\coloneqq P\otimes H_C^+$ and $M_2\coloneqq (PDP)^+\otimes H_T^+$.
Then for every sufficiently small $b$, there exists a proper level set $S$ such that
\begin{align}
\label{eq:conv:2:D}
\ex{(\bg^{t+1})\trans \cD(H^\star)^+\bg^{t+1}\bigGiven\bbundle^t}
&\within
(\bg^t)\trans \bigParens{\; \Parens{1\pm O(b)}\cdot 2bM_1'\;}\bg^t
\\
\label{eq:conv:2:H}
\ex{(\bg^{t+1})\trans (H^\star)^+\bg^{t+1}\bigGiven\bbundle^t}
&\within
(\bg^t)\trans\bigParens{\;M_2\pm O(b)M_1'\;}\bg^t
\end{align}
whenever $\bbundle^{t-1}\in S$.
%Here, $\bg^t=\nabla F(\bbundle^t)$, $D=\diag_i a_i$ and $P=I_N-\one\one\trans/N$.
Local strong convexity is bounded from below and above
by
\begin{align*}
  \sigmalow & = 2b\cdot\lambda_{\min}(PDP)\cdot\lambda_{\min}\bigParens{\;H_T^{1/2} H_C^+ H_T^{1/2}\;}-O(b^2)
\\
  \sigmahigh & = 2b\cdot\lambda_{\max}(PDP)\cdot\lambda_{\max}\bigParens{\;H_T^{1/2} H_C^+ H_T^{1/2}\;}+O(b^2)
%\;\;{-}\;\;O(b^2)
\enspace,
\end{align*}
where $\lambda_{\min}(\cdot)$ and $\lambda_{\max}(\cdot)$ denote
%, respectively,
the smallest and the largest positive eigenvalue of a matrix.
%
\begin{comment}
Let $H^\star\coloneqq\nabla^2 F(\bbundle^\star)$, $H_T\coloneqq H_T(\aggbprice)$ and $H_C\coloneqq H_C(\aggbprice)$.
Then for every sufficiently small $b$, there exists a proper level set $S$ such that
%
\begin{align}
\label{eq:conv:2:D}
\ex{(\bg^{t+1})\trans \cD(H^\star)^+\bg^{t+1}\bigGiven\bbundle^t}
&\within
\Parens{1\pm O(b)}\cdot  (\bg^t)\trans \Parens{\; 2bP \otimes H_C^+\;}\bg^t
\\
\label{eq:conv:2:H}
\ex{(\bg^{t+1})\trans (H^\star)^+\bg^{t+1}\bigGiven\bbundle^t}
&\within
(\bg^t)\trans\Parens{\;(PDP)^+\otimes H_T^+\;\pm\;O(b)\cdot (P\otimes H_C^+)\;}\bg^t
\end{align}
%
whenever $\bbundle^{t-1}\in S$.
%Here, $\bg^t=\nabla F(\bbundle^t)$, $D=\diag_i a_i$ and $P=I_N-\one\one\trans/N$.
The strong convexity is asymptotically bounded from below and above
by
%
\begin{align*}
  \sigmalow & = 2b\cdot\lambda_{\min}(PDP)\cdot\lambda_{\min}\bigParens{\;H_T^{1/2} H_C^+ H_T^{1/2}\;}-O(b^2)
\\
  \sigmahigh & = 2b\cdot\lambda_{\max}(PDP)\cdot\lambda_{\max}\bigParens{\;H_T^{1/2} H_C^+ H_T^{1/2}\;}+O(b^2)
%\;\;{-}\;\;O(b^2)
\enspace,
\end{align*}
%
where $\lambda_{\min}(\cdot)$ and $\lambda_{\max}(\cdot)$ denote, respectively, the smallest and the
largest positive eigenvalue of a matrix.
\end{comment}
\end{theorem}

As mentioned above, the key consequence of \Thm{conv:2} is the fact that \ASD converges at the rate
$1-\Theta(b)$. The key difference from \Thm{conv:1} is in the expression for $\cD(H^\star)^+$. While
$\cD(H^\star)^+\approx b(I_N\otimes H_C^+)$, as stated in \Thm{conv:1}, when we take an expectation over an update
in a single iteration, the action of $\cD(H^\star)^+$ is equivalent to that of the matrix $2b(P\otimes H_C^+)$. Thus, the averaging effect of an
expectation is to remove one rank from matrix $I_N$ and replace it by the centering matrix $P=I_N-\one\one\trans/N$ (while incurring
an extra factor of two). This has two consequences. First, the lower bound $\sigmalow$ is a factor of two larger. Second,
we can now obtain a non-trivial upper bound $\sigmahigh$, which would not be possible via an analog of \Eq{sigmalow},
because the range of $\cD(H^\star)^+$ is too large.
%%In the body of the paper,
%%we also show an example when the two bounds match (except for terms on the order of $O(b^2)$):
%%when all traders have identical risk aversions
%%and the cost function is \LMSR.
%, and for $\IND$ they are within a factor of two of each other.
%For \IND, under identical risk aversions, the two bounds are within a factor of two of each other.

%%%%%%
%%%%%%
%%%%%%
%%%%%%
%%%%%%
%\subsection{Proof of \Thm{conv:2}}
%%%%%%
%%%%%%

\begin{proof}[Proof of \Thm{conv:2}]
Consider $b_0$ and $c$ from \Lem{delta}, and let $b\le b_0$ and $S$ be the level set from the lemma. Assume that $\bbundle^{t-1}\in S$.
After the update in the iteration $t-1$, it is guaranteed that $\bprice_0(\bbundle^t)=\bprice_j(\bbundle^t)$ for some $j$. We analyze the update
in the following iteration, i.e., the iteration $t$.
We write $\bbundle$ for $\bbundle^t$ and
$\bbundle'$ for $\bbundle^{t+1}$. Let $\bg\coloneqq\nabla F(\bbundle)$, $\bprice_j\coloneqq\bprice_j(\bbundle)$ for $j\in [N]$, $\bprice_0\coloneqq\bprice_0(\bbundle)$,
and similarly define $\bg'$, $\bprice'_j$, $\bprice'_0$ for the iterate $\bbundle'$.

Recall from \Eq{grad:i} that the blocks of the gradient are
\[
  \bg_j=\bprice_j-\bprice_0
\enspace.
\]
Also recall that $P=I_N-\one\one\trans/N$. A key role in the analysis will be played by the \emph{centered} gradient vector $\bu\coloneqq(P\otimes I_K)\bg$ whose blocks are
\[
  \bu_j=\bprice_j-\bpriceavg
\]
where $\bpriceavg\coloneqq(\sum_j\bprice_j)/N$ is the average among $\bprice_j$. As the final part of the setup, let $\rho=\max_j\norm{\bprice_j-\bpriceavg}$,
and since $\bprice_0=\bprice_j$ for some $j\in[N]$, we also have $\rho\ge\norm{\bprice_0-\bpriceavg}$.

By \Thm{conv:1}, we have
\[
    \cD(H^\star)^+=\bigParens{1\pm O(b)}bM_1
\enspace,
\qquad
    (H^\star)^+=M_2\pm O(b)M_1
\enspace,
\]
where
\begin{align*}
  M_1 &= I_N\otimes H_C^+
\\
  M_2 &= (PDP)^+\otimes H_T^+
\enspace.
\end{align*}
For the first part of the theorem (Eqs.~\ref{eq:conv:2:D} and~\ref{eq:conv:2:H}),
it therefore suffices to show that
\begin{align}
\label{eq:M1:0}
 \ex{(\bg')\trans M_1\bg'\given\bbundle}
 &\within(1\pm O(b))\cdot\bg\trans\bigParens{\;2P\otimes H_C^+\;}\bg
\\
\label{eq:M2:0}
 \ex{(\bg')\trans M_2\bg'\given\bbundle}
 &\within \bg\trans M_2\bg \pm O(b)\cdot\bg\trans\bigParens{\;P\otimes H_C^+\;}\bg
\enspace.
\end{align}
Note that $P\otimes H_C^+=(P\otimes I_K) M_1 (P\otimes I_K)$ and $M_2=(P\otimes I_K) M_2 (P\otimes I_K)$,
so \Eqs{M1:0}{M2:0} are equivalent to
\begin{align}
\label{eq:M1}
\tag{\textup{$M_1$-bound}}
 \ex{(\bg')\trans M_1\bg'\given\bbundle}
 &\within(1\pm O(b))\cdot\bu\trans(2M_1)\bu
\\
\label{eq:M2}
\tag{\textup{$M_2$-bound}}
 \ex{(\bg')\trans M_2\bg'\given\bbundle}
 &\within \bu\trans M_2\bu\pm O(b)\cdot\bu\trans M_1\bu
\enspace,
\end{align}
which is what we will show next.

Assume that the $i$th block is chosen for an update in the iteration $t$ and write $\bdelta_i$ for $\bdelta_i(\bbundle)$. Note that
\[
  \bprice'_0 =\bprice'_i=\bprice_i+\bdelta_i
\enspace,
\qquad
  \bprice'_j =
\begin{cases}
   \bprice_j&\text{if $j\ne i$,}
\\
   \bprice_i+\bdelta_i&\text{if $j=i$.}
\end{cases}
\]
Thus, from \Eq{grad:i}, the blocks $\bg'_j$ can be written as
\[
  \bg'_j =
\begin{cases}
   \bprice_i-\bprice_j+\bdelta_i
   &\text{if $j\ne i$,}
\\
   \zero
   &\text{if $j=i$.}
\end{cases}
\]
Calculate:
\begin{align}
\notag
  (\bg')\trans M_1 \bg'
&=\sum_j (\bg'_j)\trans H_C^+\bg'_j
\\
\notag
&=\sum_{j\ne i} (\bprice_i-\bprice_j+\bdelta_i)\trans H_C^+(\bprice_i-\bprice_j+\bdelta_i)
\\
\notag
&=\sum_{j} (\bprice_i-\bprice_j+\bdelta_i)\trans H_C^+(\bprice_i-\bprice_j+\bdelta_i)
  -\bdelta_i\trans H_C^+\bdelta_i
\\
\notag
&=\sum_{j}\BigBracks{(\bprice_i-\bprice_j)\trans H_C^+(\bprice_i-\bprice_j)+2\bdelta_i\trans H_C^+(\bprice_i-\bprice_j)+\bdelta_i\trans H_C^+\bdelta_i}
  -\bdelta_i\trans H_C^+\bdelta_i
\\
\label{eq:convBlock:4:0}
&\within\sum_{j}\BigBracks{(\bprice_i-\bprice_j)\trans H_C^+(\bprice_i-\bprice_j)}
 \pm \BigBracks{
     4N\norm{\bdelta_i}\norm{H_C^+}\rho
     +
      N\norm{H_C^+}\norm{\bdelta_i}^2}
\\
\label{eq:convBlock:4:1}
&\within\sum_{j}\BigBracks{(\bprice_i-\bprice_j)\trans H_C^+(\bprice_i-\bprice_j)}
 \pm \BigBracks{
     8Nbc\norm{H_C^+}\rho^2
     +
     4Nb^2c^2\norm{H_C^+}\rho^2
     }
\\
\label{eq:convBlock:4:2}
&\within\sum_{j}\BigBracks{(\bprice_i-\bprice_j)\trans H_C^+(\bprice_i-\bprice_j)}
 \pm b c_1\rho^2
\enspace.
\end{align}
In \Eq{convBlock:4:0}, we used the fact that
\[
  \norm{\bprice_i-\bprice_j}\le\norm{\bprice_i-\bpriceavg}+\norm{\bprice_j-\bpriceavg}\le 2\rho
\enspace.
\]
In \Eq{convBlock:4:1}, we used \Lem{delta}, which implies that
\begin{equation}
\label{eq:convBlock:delta}
  \norm{\bdelta_i}\le b c\norm{\bprice_i-\bprice_0}
                  \le b c\BigParens{\norm{\bprice_i-\bpriceavg}+\norm{\bprice_0-\bpriceavg}}
                  \le 2 b c \rho
\enspace.
\end{equation}
And in \Eq{convBlock:4:1}, we set $c_1=8Nc\norm{H_C^+}+4N b_0 c^2\norm{H_C^+}$.

Recall that $\bu_j=\bprice_j-\bpriceavg$, so $\bprice_i-\bprice_j=\bu_i-\bu_j$. Thus, taking expectation over $i$, \Eq{convBlock:4:2} yields
\begin{align}
\notag
\Ex{i}{(\bg')\trans M_1\bg'}
&\within
  \frac{1}{N}\sum_{i=1}^N\sum_{j=1}^N\Bracks{(\bprice_i-\bprice_j)\trans H_C^+(\bprice_i-\bprice_j)}
\pm b c_1\rho^2
\\
\notag
&=
  \frac{1}{N}\sum_{i=1}^N\sum_{j=1}^N\Bracks{(\bu_i-\bu_j)\trans H_C^+(\bu_i-\bu_j)}
\pm b c_1\rho^2
\\
\notag
&=
  \frac{1}{N}\sum_{i=1}^N\sum_{j=1}^N\Bracks{\bu_i\trans H_C^+\bu_i-2\bu_i\trans H_C^+\bu_j+\bu_j H_C^+\bu_j}
\pm b c_1\rho^2
\\
\label{eq:convBlock:5}
&=
  \sum_{i=1}^N\Bracks{2\bu_i\trans H_C^+\bu_i}
\pm b c_1\rho^2
\\
\label{eq:convBlock:6}
&=
  2\bu\trans M_1\bu
\pm b c_1\rho^2
\enspace,
\end{align}
where \Eq{convBlock:5} follows because $\sum_i\bu_i=\zero$. To prove \eqref{eq:M1}, it remains to upper bound $\rho^2$.
Let $\sigma$ be the smallest eigenvalue of $H_C^+$ over $\cG(C)$; it must be greater than zero because $\range(H_C)=\cG(C)$. We can bound $\rho^2$ as
\begin{equation}
\label{eq:rho:bound}
  \rho^2\le\sum_{i=1}^N \norm{\bprice_i-\bpriceavg}^2
     \le\sigma^{-1}\sum_{i=1}^N (\bprice_i-\bpriceavg)\trans H_C^+(\bprice_i-\bpriceavg)
     =\sigma^{-1}\bu\trans M_1\bu
\enspace.
\end{equation}
Plugging this bound back into \Eq{convBlock:6} yields \eqref{eq:M1}.

We next prove \eqref{eq:M2}. Again, consider an update of the block $i$. Let $\bu'\coloneqq(P\otimes I_k)\bg'$ be the centered version of $\bg'$. Its blocks
are
\[
  \bu'_j=
  \begin{cases}
      \bu_j-\delta_i/N
  &\text{if $j\ne i$,}
\\
      \bu_j-\delta_i/N+\delta_i
  &\text{if $j=i$.}
  \end{cases}
\]
Thus,
\[
  \norm{\bu'-\bu}^2=(N-1)\cdot\frac{\norm{\delta_i}^2}{N^2}+\Parens{1-\frac{1}{N}}^2\norm{\delta_i}^2
  \le\norm{\delta_i}^2
  \le (2bc\rho)^2
\]
where the last inequality follows by \Eq{convBlock:delta}. Also, from the definition of $\rho$,
\[
  \norm{\bu}^2\le N\rho^2
\enspace.
\]
Now, we bound $(\bg')\trans M_2\bg'$. Since $M_2=(P\otimes I_K)M_2(P\otimes I_K)$, we can write
\begin{align}
\notag
(\bg')\trans M_2\bg'
&=
(\bu')\trans M_2\bu'
\\
\notag
&=
\bigParens{\bu+(\bu'-\bu)}\trans M_2\bigParens{\bu+(\bu'-\bu)}
\\
\notag
&\within
\bu\trans M_2\bu \pm\bigBracks{2\norm{\bu}\norm{M_2}\norm{\bu'-\bu}+\norm{M_2}\norm{\bu'-\bu}^2}
\\
\label{eq:convBlock:7}
&\within
\bu\trans M_2\bu \pm\bigBracks{4bc\sqrt{N}\norm{M_2}\rho^2 + 4b^2c^2\norm{M_2}\rho^2}
\\
\label{eq:convBlock:8}
&\within
\bu\trans M_2\bu \pm bc_2\rho^2
\enspace,
\end{align}
where in \Eq{convBlock:7} we used the previously derived bounds and in
\Eq{convBlock:8}, we set $c_2=4c\sqrt{N}\norm{M_2}+4b_0c^2\norm{M_2}$.
Finally, using the bound on $\rho^2$ from \Eq{rho:bound} in \Eq{convBlock:8}
yields \eqref{eq:M2}.

It remains to prove the bounds $\sigmalow$ and $\sigmahigh$. In particular, we will show that if $\bbundle^{t-1}\in S$ then
\begin{equation}
\label{eq:conv:2:bounds}
 \sigmalow
 \cdot
 \ex{(\bg')\trans (H^\star)^+\bg'\bigGiven\bbundle}
 \le
 \ex{(\bg')\trans \cD(H^\star)^+\bg'\bigGiven\bbundle}
 \le
 \sigmahigh
 \cdot
 \ex{(\bg')\trans (H^\star)^+\bg'\bigGiven\bbundle}
\end{equation}
then taking expectation over $\bbundle=\bbundle^t$, conditionally on $\bbundle^{t-1}$, we also obtain
\[
 \sigmalow
 \cdot
 \ex{(\bg')\trans (H^\star)^+\bg'\bigGiven\bbundle^{t-1}}
 \le
 \ex{(\bg')\trans \cD(H^\star)^+\bg'\bigGiven\bbundle^{t-1}}
 \le
 \sigmahigh
 \cdot
 \ex{(\bg')\trans (H^\star)^+\bg'\bigGiven\bbundle^{t-1}}
\]
which will yield the desired conclusion by \Thm{conv} (with the level set $S$ and $\ell=2$).

To prove \Eq{conv:2:bounds}, we first apply \Eqs{conv:2:D}{conv:2:H}:
\begin{align*}
\underbrace{
\frac
 {\ex{(\bg')\trans (H^\star)^+\bg'\bigGiven\bbundle}}
 {\ex{(\bg')\trans \cD(H^\star)^+\bg'\bigGiven\bbundle}}
}_{\eqqcolon z}
 &\within
 \Parens{1\pm O(b)}
 \cdot
 \frac
 {\bg\trans \bigParens{\; 2b(P\otimes I_K)M_1(P\otimes I_K)\;} \bg}
 {\bg\trans \bigParens{\;M_2\;\pm\;O(b)\cdot (P\otimes I_K)M_1(P\otimes I_K)\;} \bg}
\\
 &=
 \Parens{1\pm O(b)}
 \cdot
 \frac
 {\bu\trans \Parens{\; 2b M_1\;} \bu}
 {\bu\trans \bigParens{\;M_2\pm O(b) M_1\;} \bu}
\\
 &=
 \Parens{1\pm O(b)}
 \cdot
 \frac
 {2b}
 {\frac{\bu\trans M_2\bu}{\bu\trans M_1\bu}\;\pm\;O(b)}
\enspace.
\end{align*}
Now note that the blocks of $\bu$ take form $\bu_i=\bprice_i-(\sum_j\bprice_j)/N$, where $\bprice_i\in\cM\subseteq\range(H_T^+)$,
so $\bu\in\range(P\otimes H_T^+)=\range(M_2)\subseteq\range(M_1)$. This means that
\[
 0
 <
 \lambda_{\min}(M_1,M_2)
 \le
 \frac{\bu\trans M_1\bu}{\bu\trans M_2\bu}
 \le
 \lambda_{\max}(M_1,M_2)
\enspace,
\]
so we can write
\[
 z
 \within
 \Parens{1\pm O(b)}
 \cdot
 2b
 \cdot
 \frac{\bu\trans M_1\bu}{\bu\trans M_2\bu}
\enspace,
\]
and thus,
\begin{align*}
 2b\cdot\lambda_{\min}(M_1,M_2)-O(b^2)
 \;
 \le
 \;
 z
 \;
 \le
 \;
 2b\cdot\lambda_{\max}(M_1,M_2)+O(b^2)
\enspace.
\end{align*}
The theorem now follows, because
\begin{align*}
 \lambda_{\min}(M_1,M_2)
&=\lambda_{\min}\bigParens{\;(M_2^+)^{1/2}M_1(M_2^+)^{1/2}\;}
\\
&=\lambda_{\min}(PDP)\cdot\lambda_{\min}\bigParens{\;H_T^{1/2}H_C^+ H_T^{1/2}\;}
\enspace,
\end{align*}
and similarly for $\lambda_{\max}(M_1,M_2)$.
\end{proof}

\subsubsection{Tighter Relationship between Suboptimality and Convergence Error for \ASD}
\label{app:suboptimality:asd}

For \ASD, it is possible to establish a tighter relationship between the convergence error $\norm{\bprice^t-\eqbprice}$ and suboptimality $F(\bbundle^t)-F^\star$
than we proved in \App{subopt}. Specifically, we showed that $\norm{\bprice^t-\eqbprice}^2=O(F(\bbundle^t)-F^\star)$ for \convplus $C$. Under \ASD,
we obtain a matching lower bound for the one-step expectation, i.e., $\ex{\norm{\bprice^{t+1}-\eqbprice}^2\given\bbundle^t}=\Theta(F(\bbundle^t)-F^\star)$.

\begin{theorem}
\label{thm:errconnect:ASD}
Let $C$ be \convplus. Under \ASD, with probability one there exists $t_0$ such that for all $t\ge t_0$,
\[
  \ex{\norm{\bprice^{t+1}-\eqbprice}^2\given\bbundle^t}\within(1\pm O(b))\cdot \Theta(F(\bbundle^t)-F^\star)
\enspace.
\]
\end{theorem}
\begin{proof}
Recall that in the notation introduced in \App{subopt}, we have $\bprice^t=\bprice_0(\bbundle^t)$.
We use a similar notation as in the proof of \Thm{conv:2}. We write $\bbundle$ for $\bbundle^t$ and
$\bbundle'$ for $\bbundle^{t+1}$.
Let $\bprice_j\coloneqq\bprice_j(\bbundle)$ for $j\in [N]$, $\bprice_0\coloneqq\bprice_0(\bbundle)$,
and similarly define $\bprice'_j$, $\bprice'_0$ for the iterate $\bbundle'$. Finally, write $\bdelta_i$ for $\bdelta_i(\bbundle)$.

We assume that $t_0\ge 1$. Let $\rho=\max_j\norm{\bprice_j-\eqbprice}$,
and since $\bprice_0=\bprice_j$ for some $j\in[N]$, we also have $\rho\ge\norm{\bprice_0-\eqbprice}$. We will use the following loose bound
on $\rho^2$:
\[
  \rho^2\le\sum_{i=1}^N \norm{\bprice_i-\eqbprice}^2
\enspace.
\]
We start with the expression for the suboptimality in \Eq{subopt:convplus} and specialize it to \ASD:
\begin{align}
\notag
  F(\bbundle)-F^\star
&\within\Parens{1\pm\frac12}\sum_{i=1}^N
  \frac{1}{2a_i}\bigParens{\eqbprice-\bprice_i}\trans
     H_T^+(\eqbprice)\bigParens{\eqbprice-\bprice_i}
\\
\notag
&\qquad
  +
  b\Parens{1\pm\frac12}
  \frac{1}{2}\bigParens{\eqbprice-\bprice_0}\trans
     H_C^+(\eqbprice)\bigParens{\eqbprice-\bprice_0}
\\
\label{eq:suboptF:1}
&\within(c_1\pm c_2)\Parens{\sum_{i=1}^N\norm{\eqbprice-\bprice_i}^2+b\norm{\eqbprice-\bprice_0}^2}
\\
\label{eq:suboptF:2}
&\within(c_1\pm c_2)(1+b)\sum_{i=1}^N\norm{\eqbprice-\bprice_i}^2
\enspace.
\end{align}
In \Eq{suboptF:1}, we used the fact that $\lambda_{\min}(H_T^+(\eqbprice),P)>0$ and $\lambda_{\min}(H_C^+(\eqbprice),P)>0$, where $P$ is the projection
on the linear space parallel to $\cM$, which implies existence of constants $0\le c_2<c_1$ such that
\begin{gather*}
   (c_1-c_2)P\preceq \Parens{1\pm\frac12}\frac{1}{2a_i}H_T^+(\eqbprice)\preceq(c_1+c_2)P
\\
   (c_1-c_2)P\preceq \Parens{1\pm\frac12}\frac{1}{2}H_C^+(\eqbprice)\preceq(c_1+c_2)P
\enspace.
\end{gather*}
In \Eq{suboptF:2}, we used the upper bound $\norm{\eqbprice-\bprice_0}^2\le\rho^2\le\sum_{i=1}^N \norm{\bprice_i-\eqbprice}^2$. We now similarly
bound $\ex{\norm{\bprice^{t+1}-\eqbprice}^2\given\bbundle^t}$. Recall that in our notation $\bprice^{t+1}=\bprice_0(\bbundle^{t+1})=\bprice'_0$.
We assume that $b$ is sufficiently small, so \Lem{delta} applies, i.e., $b\le b_0$:
\begin{align}
\notag
  \ex{\norm{\bprice'_0-\eqbprice}^2\given\bbundle}
&=
  \frac1N \sum_{i=1}^N \norm{\bprice'_i-\eqbprice}^2
\\
\notag
&=
  \frac1N \sum_{i=1}^N \norm{\bprice_i+\bdelta_i-\eqbprice}^2
\\
\label{eq:suboptF:3}
&\within
  \frac1N \sum_{i=1}^N \BigBracks{ \norm{\bprice_i-\eqbprice}^2
                                   \pm 2\norm{\bdelta_i}\rho+\norm{\bdelta_i}^2 }
\\
\label{eq:suboptF:4}
&\within
  \frac1N \sum_{i=1}^N \BigBracks{ \norm{\bprice_i-\eqbprice}^2
                                   \pm 4cb\rho^2+4c^2b^2\rho^2 }
\\
\label{eq:suboptF:5}
&\within
  \Parens{\frac1N\pm c_3b}\sum_{i=1}^N \norm{\bprice_i-\eqbprice}^2
\enspace.
\end{align}
In \Eq{suboptF:3}, we applied the bound $\norm{\bprice_i-\eqbprice}\le\rho$. In \Eq{suboptF:4},
we used \Lem{delta} and applied the triangular inequality to obtain $\norm{\bprice_i-\bprice_0}\le\norm{\bprice_i-\eqbprice}+\norm{\bprice_0-\eqbprice}\le2\rho$.
Finally, in \Eq{suboptF:5}, we applied the bound $\rho^2\le\sum_{i=1}^N \norm{\bprice_i-\eqbprice}^2$ and set
$c_3=4c+4c^2b_0$. The theorem now follows by combining \Eq{suboptF:2} and \Eq{suboptF:5}. Note that, as before, we suppress the dependence on $N$ and $a_i$ within
the $O(\cdot)$ and $\Theta(\cdot)$ notation.
\end{proof}

\subsubsection{Summary of Local Convergence Results}

Here we summarize our local convergence results for \ASD and \SSD. Recall that $D\coloneqq\diag_{i\in[N]} a_i$,
$P\coloneqq I_N-\one\one\trans/N$, and
$\lambda_{\min}(\cdot)$ and $\lambda_{\max}(\cdot)$ to denote the smallest and the largest positive eigenvalues of a matrix.

\begin{theorem}
\label{thm:conv:summary}
Assume that $C$ is \convplus.
Let $H_T\coloneqq H_T(\aggbprice)$, $H_C\coloneqq H_C(\aggbprice)$, and $D_C$ be the diagonal matrix with the diagonal of $H_C$.
For all-securities dynamics, local strong convexity is bounded from below and above
by
\begin{align*}
  \sigmalow^{\ASD} & = 2b\cdot\lambda_{\min}(PDP)\cdot\lambda_{\min}\bigParens{\;H_T^{1/2} H_C^+ H_T^{1/2}\;}-O(b^2)\enspace,
\\
  \sigmahigh^{\ASD} & = 2b\cdot\lambda_{\max}(PDP)\cdot\lambda_{\max}\bigParens{\;H_T^{1/2} H_C^+ H_T^{1/2}\;}+O(b^2)
%\;\;{-}\;\;O(b^2)
\enspace.
\intertext{%
For single-security dynamics, local strong convexity is bounded from below by}
  \sigmalow^{\SSD}
  &= b\cdot\lambda_{\min}(PDP)\cdot\lambda_{\min}\bigParens{\;H_T^{1/2} D_C^+ H_T^{1/2}\;}
  -O(b^2)
%\;\;{-}\;\;O(b^2)
\enspace.
\end{align*}
\end{theorem}
\begin{proof}
The theorem follows immediately from Theorems~\ref{thm:conv:1} and~\ref{thm:conv:2}.
\end{proof}

Recall that by \Thm{conv}, the bounds on local strong convexity translate into bounds on local convergence rate as
$\kappahigh=1-\sigmalow/N$ and $\kappalow=1-\sigmahigh/N$ for \ASD, and $\kappahigh=1-\sigmalow/NK$ for \SSD.

So, for \ASD, \Thm{conv:summary} proves linear convergence with the rate $\gamma=1-\Theta(b)$. This means that the convergence gets
worse as $b\to 0$, leading to a trade-off with the bias, which decreases as $b\to 0$.
Our numerical experiments
in \Sec{experiments} and \App{experimentsNIPS} show that these bounds on the convergence rate are empirically quite tight.
%As we will see in numerical experiments, these bounds on the convergence rate are quite tight.
Below, we show an example when the two bounds match except for the $O(b^2)$ terms:
when all traders have identical risk aversions
and the cost function is \LMSR.
%, and for $\IND$ they are within a factor of two of each other.
%For \IND, under identical risk aversions, the two bounds are within a factor of two of each other.

For \SSD, we only present a lower bound on the local strong convexity, which suffices to establish a linear convergence rate. This bound is worse by a factor of two than the bound for
\ASD. This we believe is only an artifact of a looser analysis and we expect that the reasoning that gave rise to a tighter analysis of \ASD
can be generalized to \SSD. Our experiments in \App{experimentsNIPS} also suggest that our \SSD analysis is looser than the \ASD analysis.

\begin{comment}
we conjecture that the bounds we proved for \ASD also hold for \SSD,
with $H_C$ replaced by $D_C$. Our numerical experiments in \App{experimentsNIPS} suggest this might indeed be the case.

%We expect that the reasoning that gave rise to \Thm{conv:2} can be generalized to \SSD, along the lines of the following conjecture:
%
\begin{conjecture}
\label{conj:sigma}
For \SSD, local strong convexity is bounded from below and above
by
%
\begin{align*}
  \sigmalow & = 2b\cdot\lambda_{\min}(PDP)\cdot\lambda_{\min}\bigParens{\;H_T^{1/2} D_C^+ H_T^{1/2}\;}-O(b^2)
\\
  \sigmahigh & = 2b\cdot\lambda_{\max}(PDP)\cdot\lambda_{\max}\bigParens{\;H_T^{1/2} D_C^+ H_T^{1/2}\;}+O(b^2)
%\;\;{-}\;\;O(b^2)
\enspace.
\end{align*}
%
%where $\lambda_{\min}(\cdot)$ and $\lambda_{\max}(\cdot)$ denote
%the smallest and the largest positive eigenvalue of a matrix.
\end{conjecture}
\end{comment}

\begin{example}[Convergence of \LMSR under \ASD]
\label{ex:LMSR:ASD}
We next demonstrate the tightness of our bounds for \ASD.
Consider the setting when $N\ge 2$ and the risk aversion of all traders equals $a$. Then $PDP=aPI_NP=aP$, and since $P$ is a non-zero projection matrix,
we obtain $\lambda_{\min}(PDP)=a\lambda_{\min}(P)=a$ and similarly $\lambda_{\max}(PDP)=a$. Furthermore,
if the cost is \LMSR then $H_C=H_T$, and thus $H_T^{1/2} H_C^+ H_T^{1/2}=I_K-\one\one\trans/K$. Therefore, \Thm{conv:short} yields the bounds
$\sigmalow^{\ASD}=2ab-O(b^2)$ and $\sigmahigh^{\ASD}=2ab+O(b^2)$, whose main asymptotic terms match exactly. Thus,
the objective decreases at the rate $\gamma^t$ and the convergence error
at the rate $\gamma^{t/2}$ with $\gamma=1-2ab/N+O(b^2)$, and the linear term in $b$ cannot be improved.
\end{example}

\subsection{Proof of \Thm{conv:short}}

The theorem follows immediately from \Thm{conv:summary}, because $\kappahigh=1-\sigmalow/N$ and $\kappalow=1-\sigmahigh/N$ for \ASD (by \Thm{conv}).

\subsection{Proof of \Thm{conv:two}}

\begin{proof}[Proof of \Thm{conv:two}]

Fix the liquidity $b$ for \LMSR and $b'=b/\eta$ for \IND, where $\eta\in[1,2]$ (following \Thm{bias:two}).
We begin by deriving the relationship between the upper and lower bounds on the rate of convergence using \Thm{conv:short}. We will
write $\kappahigh$ and $\kappalow$ for the convergence rate bounds for \LMSR and $\kappahigh'$ and $\kappalow'$ for \IND.
We will start with \LMSR.

Following the same steps as in Example~\ref{ex:LMSR:ASD}, we have $\lambda_{\min}(PDP)=\lambda_{\max}(PDP)=a$, and since $H_\LMSR=H_T$, we obtain
\[
 \kappahigh=1-2ab/N+O(b^2)
\enspace,
 \kappalow=1-2ab/N-O(b^2)
\enspace.
\]
For \IND, we have $H_\IND=D_T$ where $D_T$ is the diagonal of $H_T$. By \Lem{two:conv} (see below), we obtain that
\begin{align*}
 \lambda_{\min}\bigParens{\;H_T^{1/2} H_\IND^+ H_T^{1/2}\;}
 &=
 \lambda_{\min}\bigParens{\;H_T^{1/2} D_T^+ H_T^{1/2}\;}
 \ge 1
\\
 \lambda_{\max}\bigParens{\;H_T^{1/2} H_\IND^+ H_T^{1/2}\;}
 &=
 \lambda_{\max}\bigParens{\;H_T^{1/2} D_T^+ H_T^{1/2}\;}
 \le 2
\enspace.
\end{align*}
Plugging these expressions, alongside $b'=b/\eta$, into \Thm{conv:short}, we obtain
\begin{align*}
 \kappahigh'&\le 1-2ab\,/\,\eta N+O(b^2) \le 1-ab/N+O(b^2)
\\
 \kappalow'&\ge 1-4ab\,/\,\eta N+O(b^2) \le 1-4ab/N-O(b^2)
\end{align*}

Next, note the following chain of inequalities which we will use to simplify our analysis. Let $\gamma=1-\alpha b+O(b^2)$, and let $t\ge t_0$ and $c>0$. Then
\begin{align*}
  c\gamma^{t-t_0}
&=   \exp\Braces{(\log c)+(t-t_0)\log\bigParens{1-\alpha b+O(b^2)}}
\\
&\le
  \exp\Braces{(\log c)-b(t-t_0)\bigParens{\alpha-O(b)}}
\\
&\le
  \exp\Braces{-bt\bigParens{\alpha-O(b)-\eps_t}}
\enspace,
\\
\intertext{%
where $\eps_t\to 0$ as $t\to\infty$. Similarly, for $\gamma=1-\alpha b-O(b^2)$, we can derive the following lower bound}
  c\gamma^{t-t_0}
&=   \exp\Braces{(\log c)+(t-t_0)\log\bigParens{1-\alpha b-O(b^2)}}
\\
&\ge
  \exp\Braces{(\log c)-b(t-t_0)\bigParens{\alpha+O(b)}}
\\
&\ge
  \exp\Braces{-bt\bigParens{\alpha+O(b)+\eps_t}}
\enspace,
\end{align*}
where $\eps_t\to 0$ as $t\to\infty$. In the remainder of the proof, we will write $\eps$, $\eps'$, $\eps''$ etc., to mean quantities that are $O(b)+\eps_t$
with some $\eps_t\to 0$ as $t\to\infty$.

We can now apply our convergence rate bounds to bound the suboptimality of the potential under both costs and liquidities.
For each of the two costs $C$,
let $F_C$ and $F_C^\star$ denote the corresponding potential and its optimal value, and $\bbundle^t_C$ and $\bprice^t_C$ be the corresponding iterates and
market prices. By \Prop{t:t0} and \Prop{all:reached}, we obtain that with probability 1, we will reach an iteration $t_0$ such that for all $t\ge t_0$
\begin{align*}
   \exp\Braces{-bt\bigParens{2a/N+\eps}}
&\le
   \ex{F_\LMSR(\bbundle^t_\LMSR)\given\bbundle_\LMSR^{t_0}}-F_\LMSR^\star
   \le
   \exp\Braces{-bt\bigParens{2a/N-\eps}}
\\
   \exp\Braces{-bt\bigParens{4a/N+\eps}}
&\le
   \ex{F_\LMSR(\bbundle^t_\IND)\given\bbundle_\IND^{t_0}}-F_\IND^\star
   \le
   \exp\Braces{-bt\bigParens{a/N-\eps}}
\enspace,
\end{align*}
where we used our bounds for $c\gamma^{t-t_0}$. By \Thm{errconnect:ASD}, we then also have, for some $\eps'$, and all $t>t_0$,
\begin{align*}
   \exp\Braces{-bt\bigParens{2a/N+\eps'}}
&\le
   \ex{\norm{\bprice^t_\LMSR-\eqbprice_\LMSR}^2\given\bbundle_\LMSR^{t_0}}
   \le
   \exp\Braces{-bt\bigParens{2a/N-\eps'}}
\\
   \exp\Braces{-bt\bigParens{4a/N+\eps'}}
&\le
   \ex{\norm{\bprice^t_\IND-\eqbprice_\IND}^2\given\bbundle_\IND^{t_0}}
   \le
   \exp\Braces{-bt\bigParens{a/N-\eps'}}
\enspace.
\end{align*}
Now writing $\Ex{t_0}{\cdot}$ instead of $\ex{\cdot\given\bbundle_\IND^{t_0}}$, $\ex{\cdot\given\bbundle_\IND^{t_0}}$, we obtain that
for a suitable $\eps''$
\begin{align}
\notag
   \Ex{t_0}{\bigNorm{\bprice^{2t(1+\eps'')}_\LMSR-\eqbprice_\LMSR}^2}
&\le
   \exp\Braces{-2bt\bigParens{2a/N-\eps'}(1+\eps'')}
\\
\notag
&=
   \exp\Braces{-bt\bigParens{4a/N-2\eps'+ 4a\eps''/N-2\eps'\eps''}}
\\
\label{eq:eps'':1}
&\le
   \exp\Braces{-bt\bigParens{4a/N+\eps'}}
\\[6pt]
\notag
&\le
   \Ex{t_0}{\norm{\bprice^t_\IND-\eqbprice_\IND}^2}
\\[6pt]
\notag
&\le
   \exp\Braces{-bt\bigParens{a/N-\eps'}}
\\
\label{eq:eps'':2}
&\le
   \exp\Braces{-bt\bigParens{a/N+\eps'/2-a\eps''/N-\eps'\eps''/2}}
\\
\notag
&=
   \exp\Braces{-b(t/2)\bigParens{2a/N+\eps'}(1-\eps'')}
\\
\notag
&\le
   \Ex{t_0}{\bigNorm{\bprice^{(t/2)(1-\eps'')}_\LMSR-\eqbprice_\LMSR}^2}
\enspace.
\end{align}
It remains to verify that we can choose a suitable $\eps''$ to guarantee inequalities~\eqref{eq:eps'':1}
and~\eqref{eq:eps'':2} for a sufficiently small $b$ and large $t$. One possibility is
to set $\eps''=3N\eps'/(2a)$, because then (for a sufficiently small $b$ and large $t$), we have
\begin{align}
\notag
  &\frac{4a\eps''}{N}
  \ge 3\eps'+ 2\eps'\eps''
\quad\implies
\quad
-2\eps'+ \frac{4a\eps''}{N}-2\eps'\eps''
\ge
  \eps'
\quad\implies
\text{\eqref{eq:eps'':1},}
\\
\tag*{\qedhere}
  &\frac{a\eps''}{N}
  \ge \frac32\eps'-\frac12\eps'\eps''
\quad\implies
\quad
-\eps'
\ge
  \frac{\eps'}{2}-\frac{a\eps''}{N}-\frac{\eps'\eps''}{2}
\quad\implies
\text{\eqref{eq:eps'':2}.}
\end{align}
\end{proof}

\begin{lemma}
\label{lem:two:conv}
Let $\bprice\in\R^K$ be a probability vector with non-zero entries, let $H\coloneqq\parens{\diag_{k\in[K]}\price_k}-\bprice\bprice\trans$ be the covariance
matrix
of the associated multinomial distribution, and $D\coloneqq\diag_{k\in[K]}\price_k(1-\price_k)$ be the
diagonal matrix consisting of the diagonal of $H$.
Then
\[
1\le\lambda_{\min}(H^{1/2}D^{-1}H^{1/2})
\quad
\text{and}
\quad
\lambda_{\max}(H^{1/2}D^{-1}H^{1/2})\le 2
\enspace.
\]
\end{lemma}
\begin{proof}
Let $\cL\coloneqq\range(D^{-1/2} H D^{-1/2})$. To prove the lemma, it suffices to show that for all $\bu\in\cL$
\begin{equation}
\label{eq:two:conv:toshow}
  \bu\trans\bu
  \le \bu\trans D^{-1/2} H D^{-1/2} \bu
  \le 2\bu\trans\bu
\enspace.
\end{equation}
This will imply that $1\le\lambda_{\min}(D^{-1/2} H D^{-1/2})$ and $\lambda_{\max}(D^{-1/2} H D^{-1/2})\le 2$. Thus, by \Eq{lambda:commute},
we will also have $1\le\lambda_{\min}(H^{1/2}D^{-1}H^{1/2})$ and $\lambda_{\max}(H^{1/2}D^{-1}H^{1/2})\le 2$.

Let $\bu\in\cL$ and $\bv=D^{-1/2}\bu$. \Eq{two:conv:toshow} can be rewritten as
\[
  \bv\trans D\bv
  \le \bv\trans H \bv
  \le 2\bv\trans D\bv
\enspace.
\]
We will next show that both inequalities hold.

\paragraph{Part 1: $\bv\trans D\bv\le \bv\trans H\bv$.}
We first rewrite the constraint $\bu\in\cL$ in terms of $\bv$.
%We begin by making the description of $\cL$ more explicit.
To start, note that $\bu\in\cL=\range(D^{-1/2} H D^{-1/2})$ iff
$\bu\perp\Null(D^{-1/2} H D^{-1/2})$. Next, note that $\by\in\Null(D^{-1/2} H D^{-1/2})$ iff $D^{-1/2}\by\in\Null(H)$,
which is equivalent to $\by\in D^{1/2}\Null(H)$. Now, $\Null(H)=\set{c\one:\:c\in\R}$, where $\one$ is the all-ones vector.
So,
\[
  \Null(D^{-1/2} H D^{-1/2})=D^{1/2}\Null(H)=\BigSet{\by:\: y_j=c\sqrt{\mu_j(1-\mu_j)} \text{ for $j\in[K]$, for some $c\in\R$}}
\enspace.
\]
Therefore, $\bu\perp\Null(D^{-1/2} H D^{-1/2})$ iff
\begin{equation}
\label{eq:L:explicit}
  \sum_{j\in[K]} u_j\sqrt{\mu_j(1-\mu_j)}=0
\enspace.
\end{equation}
Since $\bv=D^{-1/2}\bu$, we have $\bu=D^{1/2}\bv$, i.e., $u_j=v_j\sqrt{\mu_j(1-\mu_j)}$ for $j\in[K]$. Substituting this expression
into \Eq{L:explicit} yields
\begin{equation}
\label{eq:v:constr}
  \sum_{j\in[K]} v_j\mu_j(1-\mu_j)=0
\enspace.
\end{equation}
When $K=1$, then $D=H=\zero$, so in this case indeed $\bv\trans D\bv\le \bv\trans H\bv$. Next consider $K>1$. We will use the following
identity for $v_1\mu_1$, implied by \Eq{v:constr}:
\begin{equation}
\label{eq:v1}
  v_1\mu_1=-\frac{1}{1-\mu_1}\sum_{j\ge 2} v_j\mu_j(1-\mu_j)
\enspace.
\end{equation}
We next argue that $\bv\trans H\bv-\bv\trans D\bv\ge 0$:
\begin{align}
\label{eq:twoc:1}
\bv\trans H\bv-\bv\trans D\bv
  &=-\sum_{ \substack{j,k\in[K] \\ j\ne k} }
    \mu_j\mu_k v_j v_k
\\[2pt]
\notag
  &=-\sum_{ \substack{j,k\ge 2 \\ j\ne k} }
    \mu_j\mu_k v_j v_k
    -2\sum_{k\ge 2}\mu_1 v_1 \mu_k v_k
\\
\label{eq:twoc:2}
  &=-\sum_{ \substack{j,k\ge 2 \\ j\ne k} }
    \mu_j\mu_k v_j v_k
    +\frac{2}{1-\mu_1}\sum_{j,k\ge 2}\mu_j(1-\mu_j) v_j \mu_k v_k
\\
\label{eq:twoc:3}
  &=-\sum_{ \substack{j,k\ge 2 \\ j\ne k} }
    z_j z_k
    +\frac{2}{1-\mu_1}\sum_{j,k\ge 2}(1-\mu_j) z_j z_k
\\
\label{eq:twoc:4}
  &=-\sum_{ \substack{j,k\ge 2 \\ j\ne k} }
    z_j z_k
    +\frac{1}{1-\mu_1}\sum_{j,k\ge 2}\bigParens{(1-\mu_j)+(1-\mu_k)} z_j z_k
\\[2pt]
\notag
  &=\frac{1}{1-\mu_1}
    \BigBracks{\;
    \sum_{ \substack{j,k\ge 2 \\ j\ne k} }
    z_j z_k
    \BigParens{
     -(1-\mu_1)+(1-\mu_j)+(1-\mu_k)
    }
    +
    \sum_{j\ge 2}
    2(1-\mu_j)z_j^2
    \;}
\\[2pt]
\notag
  &=\frac{1}{1-\mu_1}
    \BigBracks{\;
    \sum_{ \substack{j,k\ge 2 \\ j\ne k} }
    z_j z_k
    \BigParens{
     (1-\mu_1-\mu_j-\mu_k)+2\mu_1
    }
    +
    \sum_{j\ge 2}
    z_j^2
    \BigParens{
     2(1-\mu_1-\mu_j) + 2\mu_1
    }
    \;}
\\[2pt]
\notag
  &=\frac{1}{1-\mu_1}
    \BigBracks{\;
    2\mu_1\sum_{j,k\ge 2}
    z_j z_k
    +
    \sum_{ \substack{j,k\ge 2 \\ j\ne k} }
    (1-\mu_1-\mu_j-\mu_k)
    z_j z_k
    +
    \sum_{j\ge 2}
    2(1-\mu_1-\mu_j) z_j^2
    \;}
\\[2pt]
\label{eq:twoc:5}
  &=\frac{1}{1-\mu_1}
    \BigBracks{\;
    2\mu_1\sum_{j,k\ge 2}
    z_j z_k
    +
    \sum_{ \substack{j,k\ge 2 \\ j\ne k} }
    \sum_{ \substack{\ell\ge 2\\ \ell\ne j,k} }
    \mu_\ell
    z_j z_k
    +
    \sum_{j\ge 2}
    \sum_{ \substack{\ell\ge 2\\ \ell\ne j} }
    2\mu_\ell z_j^2
    \;}
\\[2pt]
\notag
  &=\frac{1}{1-\mu_1}
    \BigBracks{\;
    2\mu_1\sum_{j,k\ge 2}
    z_j z_k
    +
    \sum_{ j,k\ge 2 }
    \sum_{ \substack{\ell\ge 2\\ \ell\ne j,k} }
    \mu_\ell
    z_j z_k
    +
    \sum_{j\ge 2}
    \sum_{ \substack{\ell\ge 2\\ \ell\ne j} }
    \mu_\ell z_j^2
    \;}
\\[2pt]
\notag
  &=\frac{1}{1-\mu_1}
    \BigBracks{\;
    2\mu_1\sum_{j,k\ge 2}
    z_j z_k
    +
    \sum_{ \ell\ge 2 }
    \mu_\ell
    \!\!\sum_{ \substack{j,k\ge 2 \\ j\ne\ell,k\ne\ell} }\!\!
    z_j z_k
    +
    \sum_{j\ge 2}
    \sum_{ \substack{\ell\ge 2\\ \ell\ne j} }
    \mu_\ell z_j^2
    \;}
\\[2pt]
\label{eq:twoc:6}
  &=\frac{1}{1-\mu_1}
    \BigBracks{\;
    2\mu_1\BigParens{\sum_{j\ge 2}z_j}^2
    +
    \sum_{ \ell\ge 2 }
    \mu_\ell
    \BigParens{\sum_{\substack{j\ge 2\\j\ne\ell}} z_j}^2
    +
    \sum_{ \substack{j,\ell\ge 2\\ \ell\ne j} }
    \mu_\ell z_j^2
    \;}
  \ge 0
\enspace.
\end{align}
In \Eq{twoc:1}, we use the fact that $D$ is the diagonal of $H$, so the right-hand side only sums over off-diagonal entries of $H$.
In \Eq{twoc:2}, we replaced $\mu_1 v_1$ using \Eq{v1}. In \Eq{twoc:3}, introduce the substitution $z_j\coloneqq \mu_j v_j$.
In \Eq{twoc:4}, we use the fact that $2\sum_{j,k\ge 2}(1-\mu_j)z_j z_k=\sum_{j,k\ge 2}(1-\mu_j)z_j z_k+\sum_{j,k\ge 2}(1-\mu_k)z_j z_k$.
In \Eq{twoc:5}, we use the fact that $\sum_{\ell\in[K]}\mu_\ell=1$ and so $(1-\mu_1-\mu_j-\mu_k)=\sum_{\ell\ne 1,j,k}\mu_\ell$
and $(1-\mu_1-\mu_j)=\sum_{\ell\ne 1,j}\mu_\ell$.
Finally, the inequality in \Eq{twoc:6} follows because $\mu_1,\mu_\ell\ge 0$.

\paragraph{Part 2: $\bv\trans H\bv\le 2\bv\trans D\bv$.} We show by direct calculation that $2\bv\trans D\bv-\bv\trans H\bv\ge 0$:
\begin{align}
\label{eq:twoc:21}
2\bv\trans D\bv-\bv\trans H\bv\ge 0
  &=2\BigBracks{\;\sum_{j\in[K]}\mu_j v_j^2-\sum_{j\in[K]}\mu_j^2 v_j^2\;}
    -\BigBracks{\;\sum_{j\in[K]}\mu_j v_j^2-\sum_{j,k\in[K]}\mu_j\mu_k v_j v_k\;}
\\[4pt]
\notag
  &=\sum_{j\in[K]}\mu_j v_j^2-2\sum_{j\in[K]}\mu_j^2 v_j^2
    +\sum_{j,k\in[K]}\mu_j\mu_k v_j v_k
\\[4pt]
\label{eq:twoc:22}
  &=\sum_{j,k\in[K]}\mu_j\mu_k v_j^2-2\sum_{j\in[K]}\mu_j^2 v_j^2
    +\sum_{j,k\in[K]}\mu_j\mu_k v_j v_k
\\[4pt]
\notag
  &=\sum_{\substack{j,k\in[K]\\ j\ne k}} \mu_j\mu_k v_j^2
    +\sum_{\substack{j,k\in[K]\\ j\ne k}} \mu_j\mu_k v_j v_k
\\[2pt]
\label{eq:twoc:23}
  &=\sum_{\substack{j,k\in[K]\\ j\ne k}}\BigParens{\frac12 \mu_j\mu_k v_j^2+\frac12 \mu_j\mu_k v_k^2+\mu_j\mu_k v_j v_k}
\\
\label{eq:twoc:24}
  &=\sum_{\substack{j,k\in[K]\\ j\ne k}}\frac12\mu_j\mu_k\Parens{v_j+v_k}^2
  \ge 0
\enspace.
\end{align}
In \Eq{twoc:21}, we just use the definition of $H$ and $D$. In \Eq{twoc:22}, we use that $\sum_{k\in[K]}\mu_k=1$. In \Eq{twoc:23},
we use that by symmetry $\sum_{j\ne k} \mu_j\mu_k v_j^2=\sum_{j\ne k} \mu_j\mu_k v_k^2$ and so $\sum_{j\ne k} \mu_j\mu_k v_j^2=\sum_{j\ne k} \mu_j\mu_k (v_j^2+v_k^2)/2$.
The final inequality in \Eq{twoc:24} follows because $\mu_j,\mu_k\ge 0$.
\end{proof}

\ifnips
\section{Additional Numerical Experiments}
\label{app:experimentsNIPS}

In \Sec{experiments}, we demonstrated that our asymptotic theory closely matches simulations for all-securities dynamics and single-peaked beliefs. Here we include experiments for an additional set of beliefs (uniform beliefs, defined below) and single-securities dynamics (defined in \App{trader:dynamics}).
Once again, we consider a setting in which there is a complete market over $\nsec = 5$ securities with $\nbuyers = 10$ traders who have exponential utilities, exponential-family beliefs and risk aversion coefficients $a_i = 1$ for $i \in [\nbuyers]$. Similar to \Sec{experiments}, we fix the ground-truth
%likelihood of each outcome (or equivalently, ground truth expected security payoffs) $\bprice^\TRUE$ and the corresponding
natural parameter $\btheta^\TRUE$ and independently sample the belief $\bttheta_i$ of each trader from $\textup{Normal}(\btheta^\TRUE,\sigma^2I_\nsec)$.
We consider two settings of the ground truth and beliefs:
% (only the second variant was considered in \Sec{experiments}):
\begin{itemize}
\item{\emph{Uniform Beliefs}: } All outcomes are equally likely.  We set $\btheta^\TRUE = \pmb{0}$ and $\sigma = 1$.
\item{\emph{Single-Peaked Beliefs}: } One outcome is more likely than the others. Here we set $\theta^\TRUE_1 = \log(1- \nu\cdot (k-1))$ and $\theta^\TRUE_k = \log(\nu)$ for $k \neq 1$. We use $\nu = 0.02$ and $\sigma = 5$.
\end{itemize}
\Fig{beliefs} shows the trader beliefs and market-clearing equilibrium prices (calculated via \Thm{eqprice-char:NIPS}) for both settings.
%from \Cref{defn:market_clearing}.
%Recall that we consider two types of trader dynamics -- \emph{all-securities dynamics} $\BCD$ and \emph{single-security dynamics} $\SCD$.
Note that in \Sec{experiments}, we gave results for the case of single-peaked beliefs and all-security dynamics whereas here we present results for all four combinations
of belief settings and trader dynamics.

\begin{figure}[t!]
\centering
\begin{subfigure}{.48\textwidth}
  \centering
  \includegraphics[width=.95\linewidth]{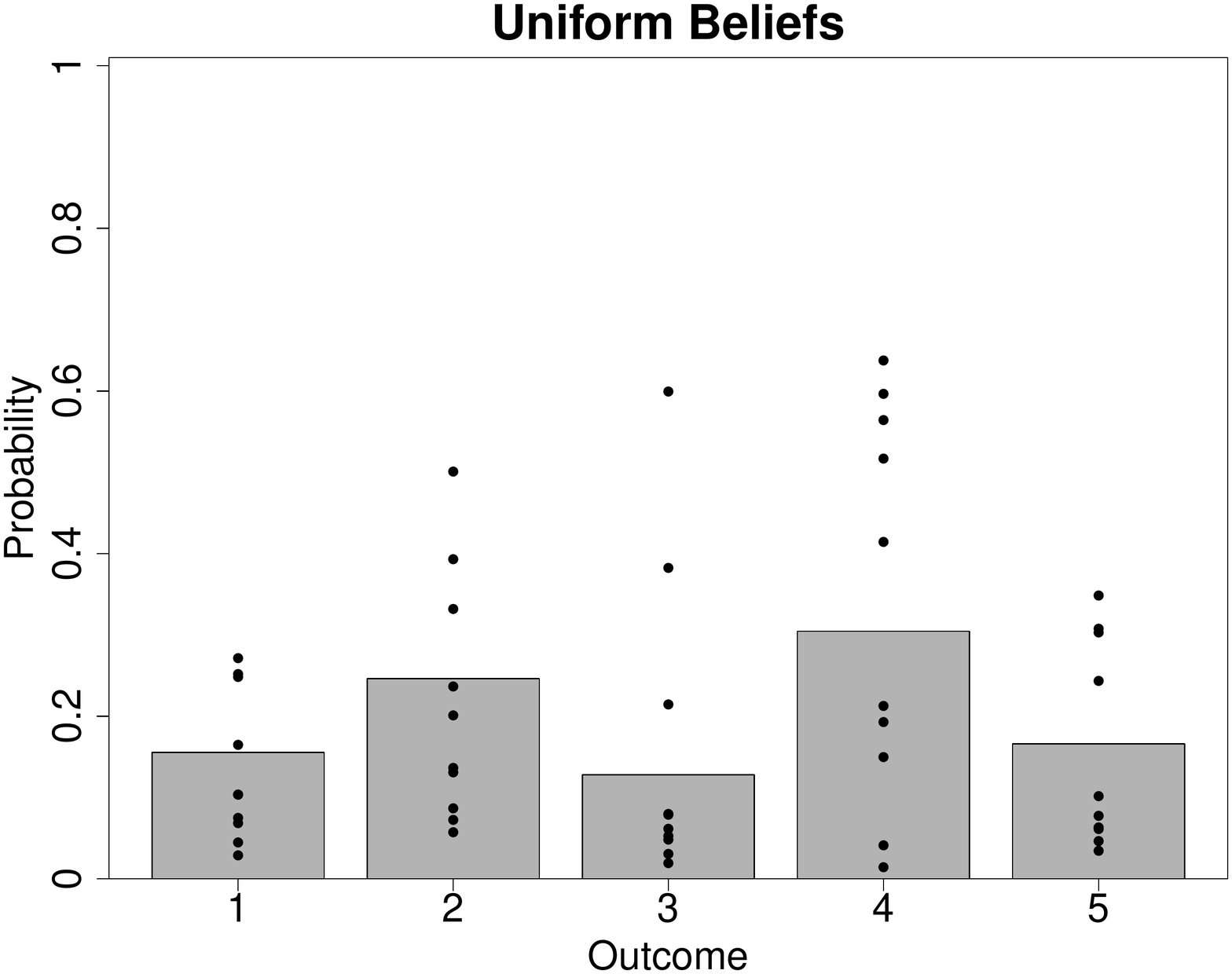}
%  \caption{Uniform beliefs}
 \end{subfigure}%
 \hfill
\begin{subfigure}{.48\textwidth}
  \centering
  \includegraphics[width=.95\linewidth]{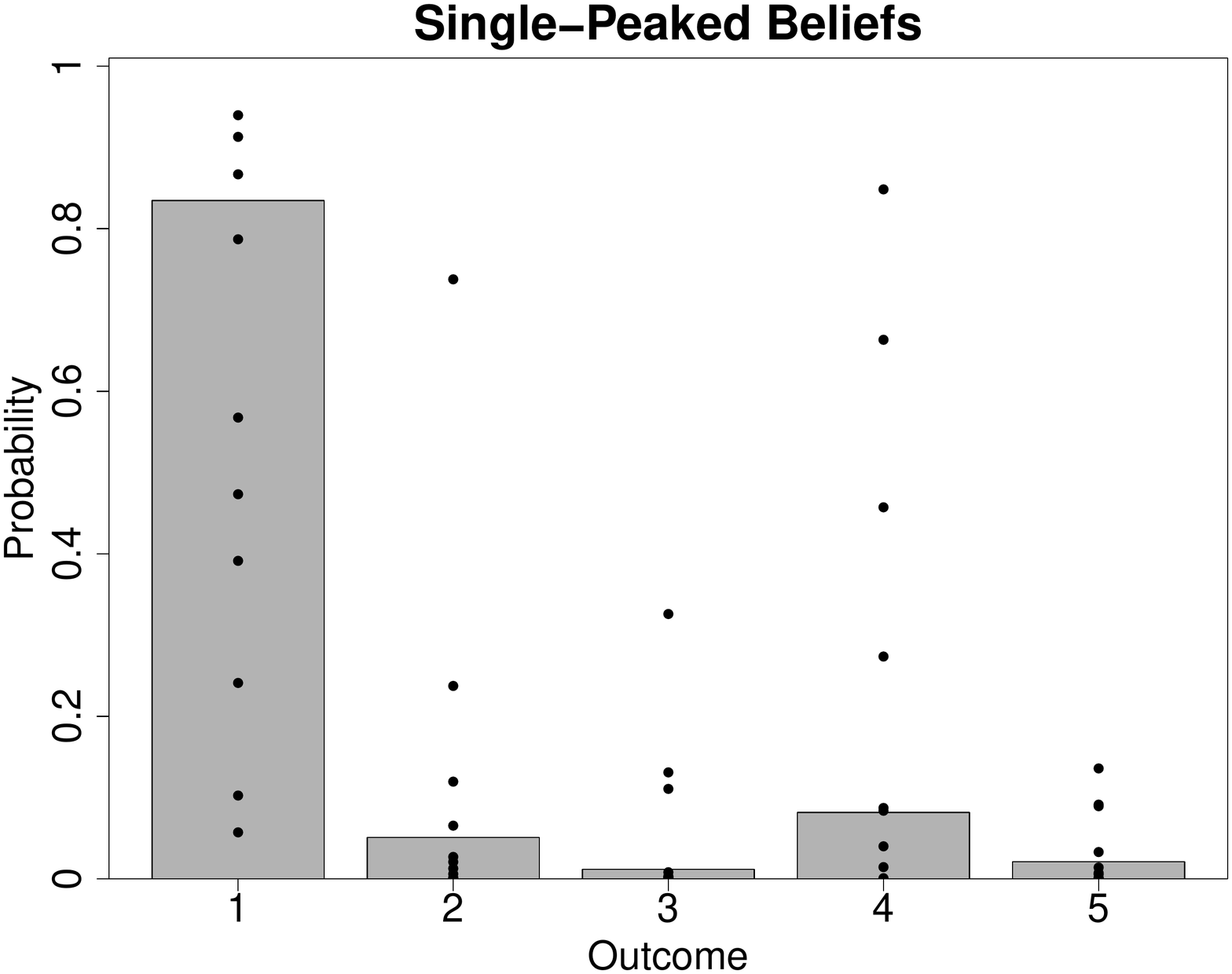}
%  \caption{Single-peaked beliefs}
 \end{subfigure}
\caption{%
Two sets of beliefs of the $\nbuyers = 10$ traders about
%the likelihood of each of
the $\nsec = 5$ outcomes. The beliefs were sampled once and then fixed in all experiments. The gray bars show the market-clearing equilibrium prices $\aggbprice$ as in \Cref{defn:market_clearing} and \Thm{eqprice-char:NIPS}.}
%Randomly sampled beliefs of the $\nbuyers = 10$ traders about the likelihood of each of the $\nsec = 5$ outcomes. The gray bars show the market-clearing equilibrium prices $\aggbprice$ as in Equation~\ref{eq:price-eq-form}.}
\label{fig:beliefs}
\end{figure}

\paragraph{Bias/convergence tradeoffs}
We first examine the tradeoff that arises between market-maker bias and convergence error as the liquidity parameter of the market is adjusted.  Since our main interest is in the effect of the cost function $C$ and liquidity parameter $\liq$ on error, we ignore the sources of error that do not depend on the choice of cost function, such as the sampling error.  \Fig{Uplots} shows the combined bias and convergence error,
$\norm{\aggbprice-\iterbprice{t}(\liq;C)}$, as a function of liquidity,
%when averaged over 20 random sequences of trade
for different beliefs and cost functions under $\BCD$ after different numbers of trades have occurred. (Other choices of norm lead to similar results.) Similarly, we give results for $\SCD$ in \Fig{UplotsSSD}. The minimum point on each curve tells us the optimal value of the liquidity parameter $\liq$ for the particular setting and number of trades. When the market has not been running long, larger values of $\liq$ lead to lower error.  On the other hand, smaller values of $\liq$ are preferable as the number of trades grows, with the combined error approaching 0 for small $\liq$.  The combined error of \LMSR is similar to that of the sum of independent LMSRs (\IND) under uniform beliefs, but \LMSR produces lower combined error for single-peaked beliefs.

\begin{figure}[t!]
\centering
\begin{subfigure}{.5\textwidth}
  \centering
  \includegraphics[width=.95\linewidth]{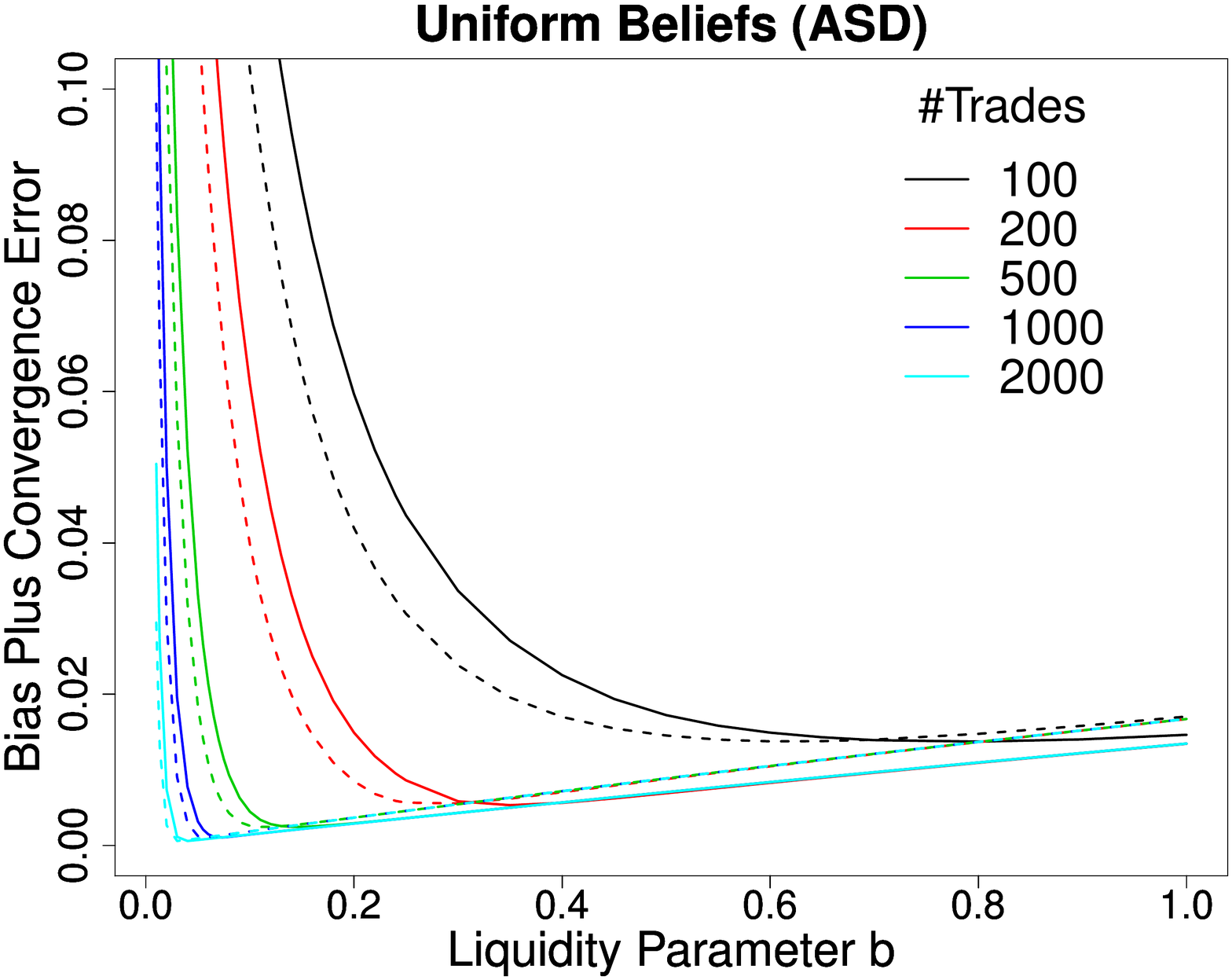}
%  \caption{Uniform Beliefs}
%  \label{fig:sub1}
\end{subfigure}%
\hfill
\begin{subfigure}{.5\textwidth}
  \centering
  \includegraphics[width=.95\linewidth]{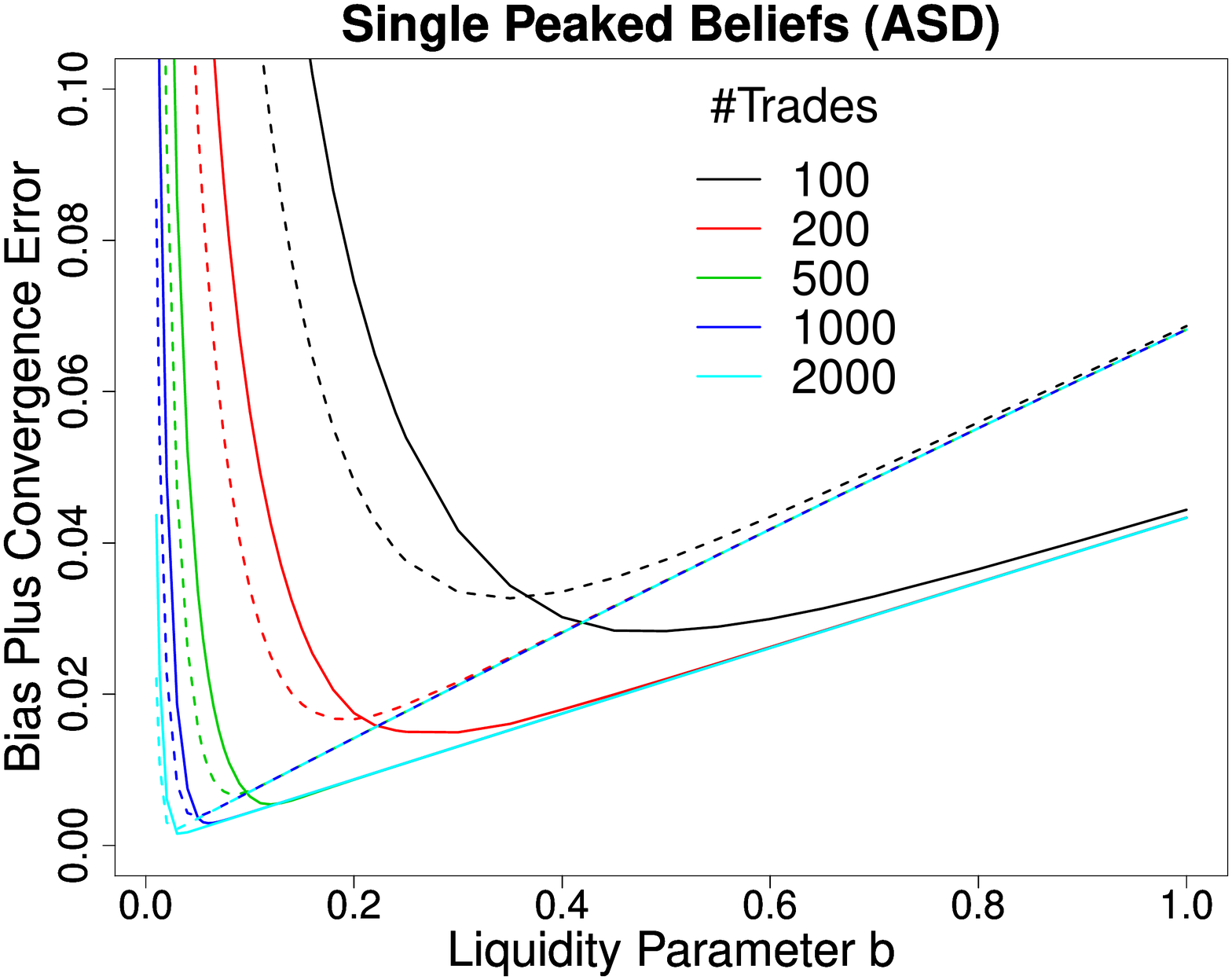}
%  \caption{Single Peaked Beliefs}
%  \label{fig:sub2}
\end{subfigure}
\caption{%
%An illustration of the tradeoff
The tradeoff between marker-maker bias and convergence error for $\BCD$ with different beliefs and cost functions. Solid lines show the total bias and convergence error of \LMSR after various numbers of trades, averaged over 20 random trade sequences. Dotted lines show the same for \IND. \label{fig:Uplots}}
\end{figure}

\begin{figure}[h]
\centering
\begin{subfigure}{.5\textwidth}
  \centering
  \includegraphics[width=.95\linewidth]{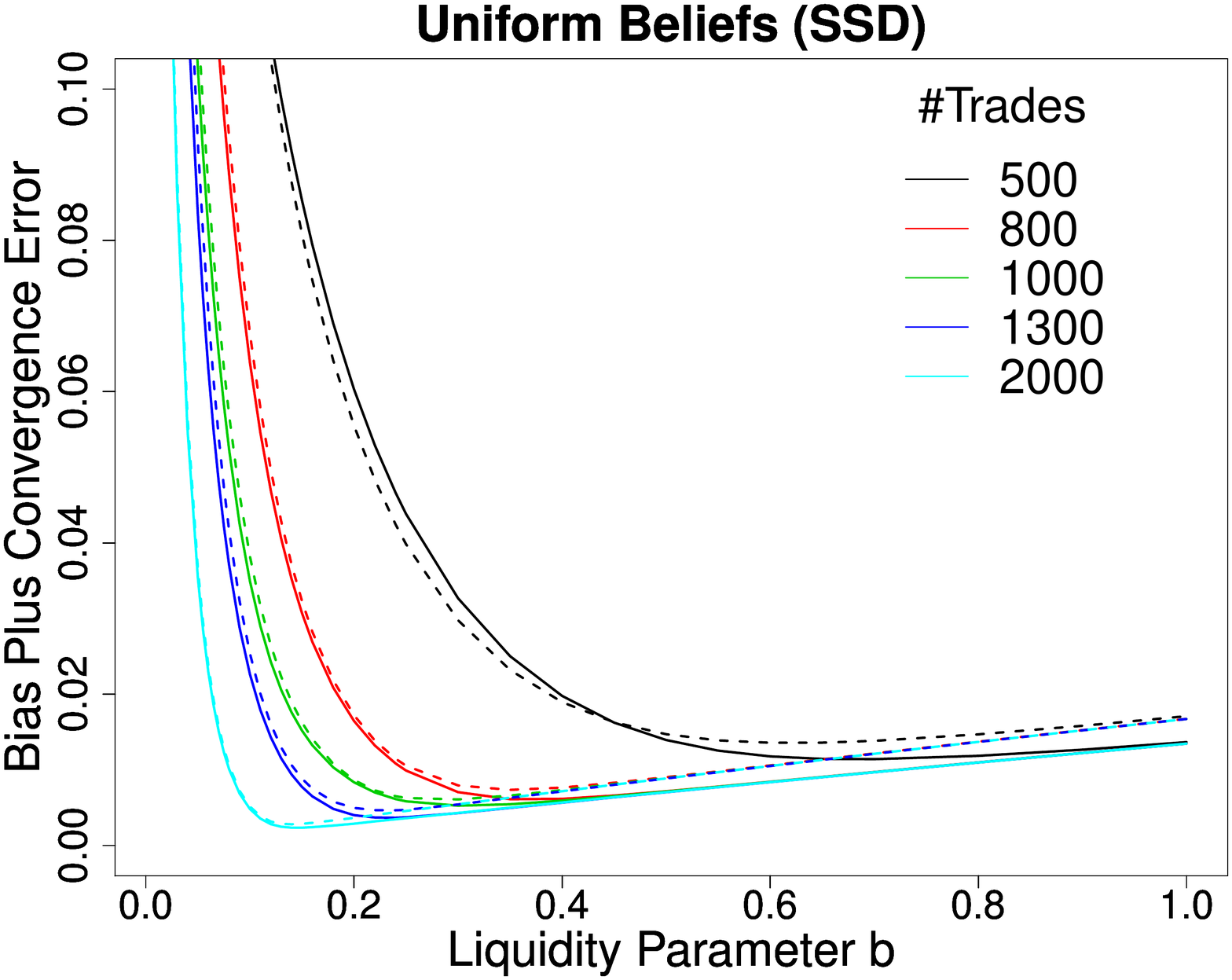}
%  \caption{Uniform Beliefs}
%  \label{fig:sub1}
\end{subfigure}%
\hfill
\begin{subfigure}{.5\textwidth}
  \centering
  \includegraphics[width=.95\linewidth]{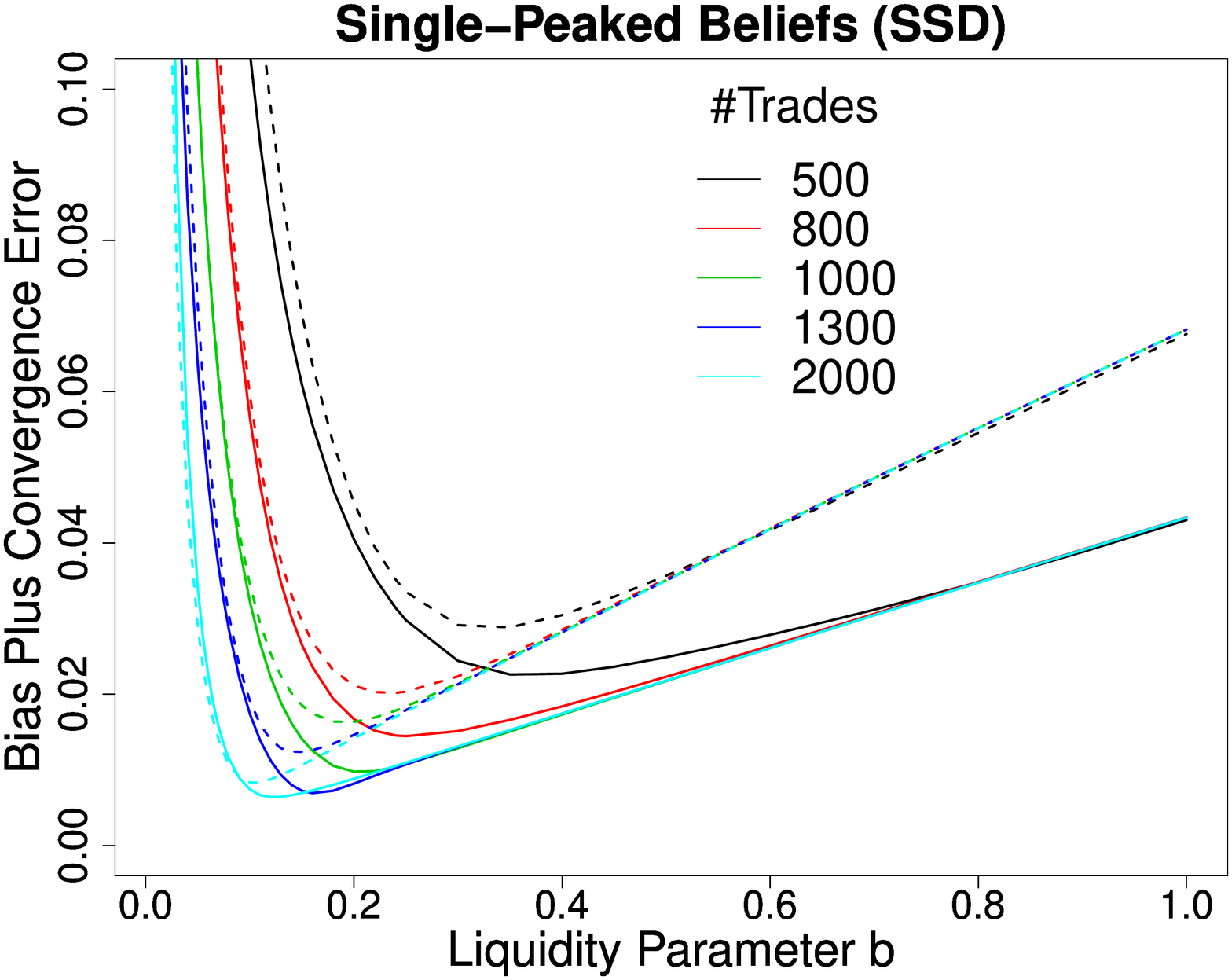}
%  \caption{With Single Peaked Beliefs}
%  \label{fig:sub2}
\end{subfigure}
\caption{%
%An illustration of the tradeoff
The tradeoff between marker-maker bias and convergence error for $\SCD$ with different beliefs and cost functions. Solid lines show the total bias and convergence error of \LMSR after various numbers of trades, averaged over 20 random trade sequences. Dotted lines show the same for \IND. \label{fig:UplotsSSD}}
\end{figure}

\paragraph{Market-maker bias}
We next focus in on the market-maker bias to empirically evaluate our bounds from Section~\ref{sec:bias}.  From Theorem~\ref{thm:bias:local}, we know that
$
\norm{\eqbprice(\liq;C)-\aggbprice} \approx b(\bar{a}/N)\norm{H_T(\aggbprice)\partial C^*(\aggbprice)}.
$
In \Fig{bias}, we plot the empirical bias $\norm{\eqbprice(\liq;C) - \aggbprice}$ as a function of $\liq$ for both \LMSR and \IND under uniform and single-peaked beliefs, and in each case compare this bias with the approximation implied by the theory.  We find that although \Cref{thm:bias:local} only gives an asymptotic guarantee as $\liq \to 0$, the approximation above is fairly accurate even for moderate values of $\liq$.
As \Thm{bias:two} shows,
%the slope of \IND is higher than the slope of \LMSR, i.e.,
the bias of \IND is higher than that of \LMSR at any fixed value of $b$, but by no more than a factor of two. The difference is greater for single-peaked beliefs than uniform beliefs.
Note that the bias is unaffected by the choice of trader dynamics.

\begin{figure}[t!]
\begin{subfigure}{.5\textwidth}
  \centering
  \includegraphics[width=.9\linewidth]{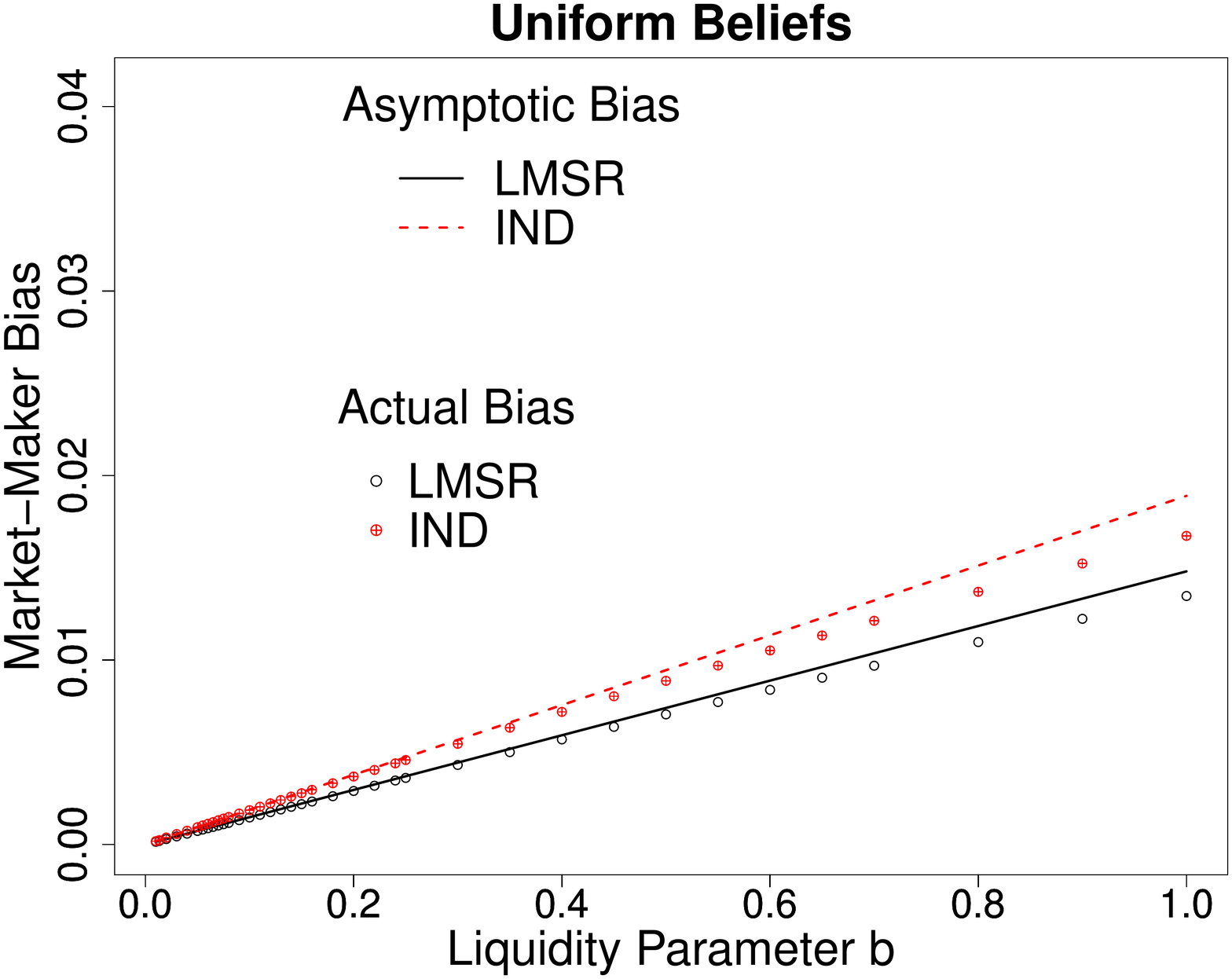}
%    \caption{Uniform Beliefs}
\end{subfigure}%
\hfill
\begin{subfigure}{.5\textwidth}
  \centering
    \includegraphics[width=.9\linewidth]{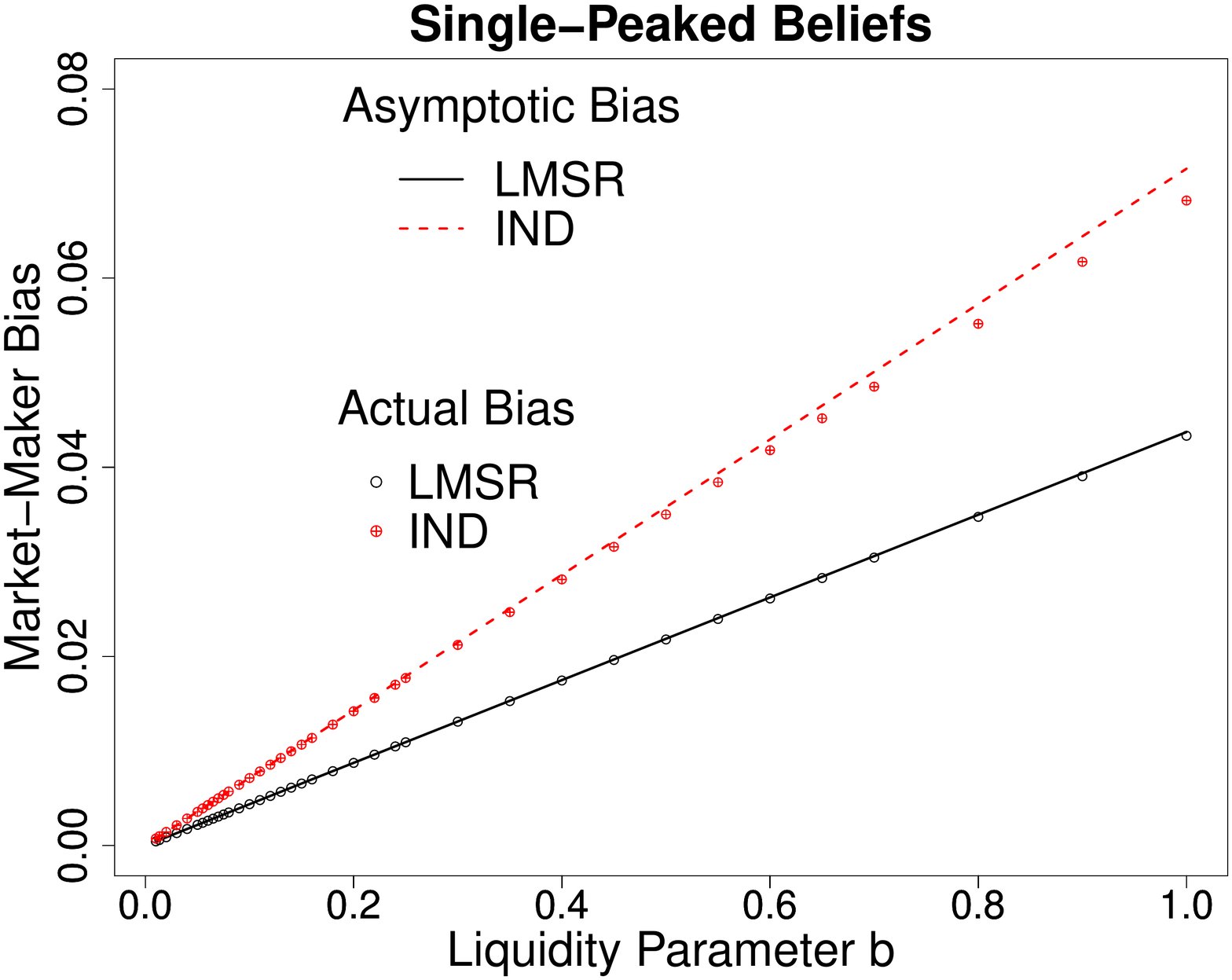}
%    \caption{Single Peaked Beliefs}
\end{subfigure}
\caption{Market-maker bias as a function of $\liq$ for different beliefs and cost functions. \label{fig:bias}}
 \end{figure}

\paragraph{Convergence error}
Finally, we turn to the convergence error.  We first show that the local linear convergence rate kicks in very quickly---essentially from the start of trade in our simulations.  We then examine the tightness of our bounds on the local convergence rate from Section~\ref{sec:convergence}.

From \Thm{errconnect:convplus}, we know that $F(\bbundle^t)-F^\star$ is an upper bound on $\norm{\bprice^t-\eqbprice}^2$, and under \ASD we also have $F(\bbundle^t)-F^\star=\Theta(\norm{\bprice^t-\eqbprice}^2)$ (by \Thm{errconnect:ASD}), where we suppress the implicit dependence on $C$ and $\liq$.
% \mdcomment{I removed the multiplier $\liq$, since we say that we suppress the dependence on $\liq$.}\jenn{Oh, good point. I meant in the notation we suppress the dependence. Is it misleading to leave out this term? It makes it harder to meaningfully compare different $\liq$ values in the plots.}
Rather than examining the convergence of prices directly, we examine convergence of the objective, which will be more convenient in the discussion below.  \Fig{semi_logB} shows the empirical value of $\hatex{F(\bbundle^t)}-F^\star$,
where the expectation is the empirical average over the 20 random sequences,
as a function of the number of trades, plotted on a log scale, for our two belief sets and cost functions under the all-securities trade dynamics. In all settings, the log of convergence error appears linear, matching the local asymptotic analysis in Section~\ref{sec:convergence}.  In other words, there exist some $\hat{c}$ and $\hat{\gamma}$ such that, empirically, we have for all $t$,
$
\hatex{F(\bbundle^t)}-F^\star \approx \hat{c} \hat{\gamma}^t.
$

To examine the tightness of the bounds from Section~\ref{sec:convergence}, we dig more deeply into the value of this empirical constant $\hat{\gamma}$, which depends on $C$ and $\liq$.  Since this approximation holds for any sufficiently large $t$, we can define $\hat{\gamma}$ by choosing some $t_1$ and $t_2$ and setting
\[
\hat{\gamma} =
\left(\frac{\hatex{\obj(\iterbbundle{t_2})} - \obj^\star }{\hatex{\obj(\iterbbundle{t_1})} - \obj^\star  }\right)^{1/(t_2 - t_1)} .
\]
If $\hat{\gamma}$ is the correct asymptotic convergence rate, then from \Cref{thm:conv}, we should have
$
1 - \sigmahigh/N \leq \hat{\gamma} \leq 1 - \sigmalow/ N
$
for values of $\sigmahigh$ and $\sigmalow$ that satisfy \Eq{thm:conv}, since $|\cA| = N$ for $\BCD$.
%(See \Cref{sect:ratio} for a more formal argument).
Rearranging terms, we would expect that, for sufficiently large $t_1$ and $t_2$,
\begin{equation}
\sigmalow \leq N \left( 1 - \left(\frac{\hatex{\obj(\iterbbundle{t_2})} - \obj^\star }{\hatex{\obj(\iterbbundle{t_1})} - \obj^\star  }\right)^{1/(t_2 - t_1)}\right) \leq \sigmahigh.
\label{eq:emp_sc}
\end{equation}
We refer to this quantity that is upper and lower bounded by $\sigmahigh$ and $\sigmalow$ as the \emph{empirical strong convexity} $\hat{\sigma}$. Note that $\sigmalow$ and $\sigmahigh$ implicitly depend on $\liq$ and $C$.

\ignore{
\begin{equation}
\alpha(t_1,t_2; \liq,C,\cA) \stackrel{\defn}{=}1-\left(\frac{\ex{\obj(\iterbbundle{t_2}(\liq,C))} - \obj^\star(\liq,C) }{\ex{\obj(\iterbbundle{t_1}(\liq,C))} - \obj^\star(\liq,C)  }\right)^{1/(t_2 - t_1)}.
\label{eq:alpha}
\end{equation}
}

\begin{figure}[t!]
\centering
\begin{subfigure}{.5\textwidth}
  \centering
  \includegraphics[width=.8\linewidth]{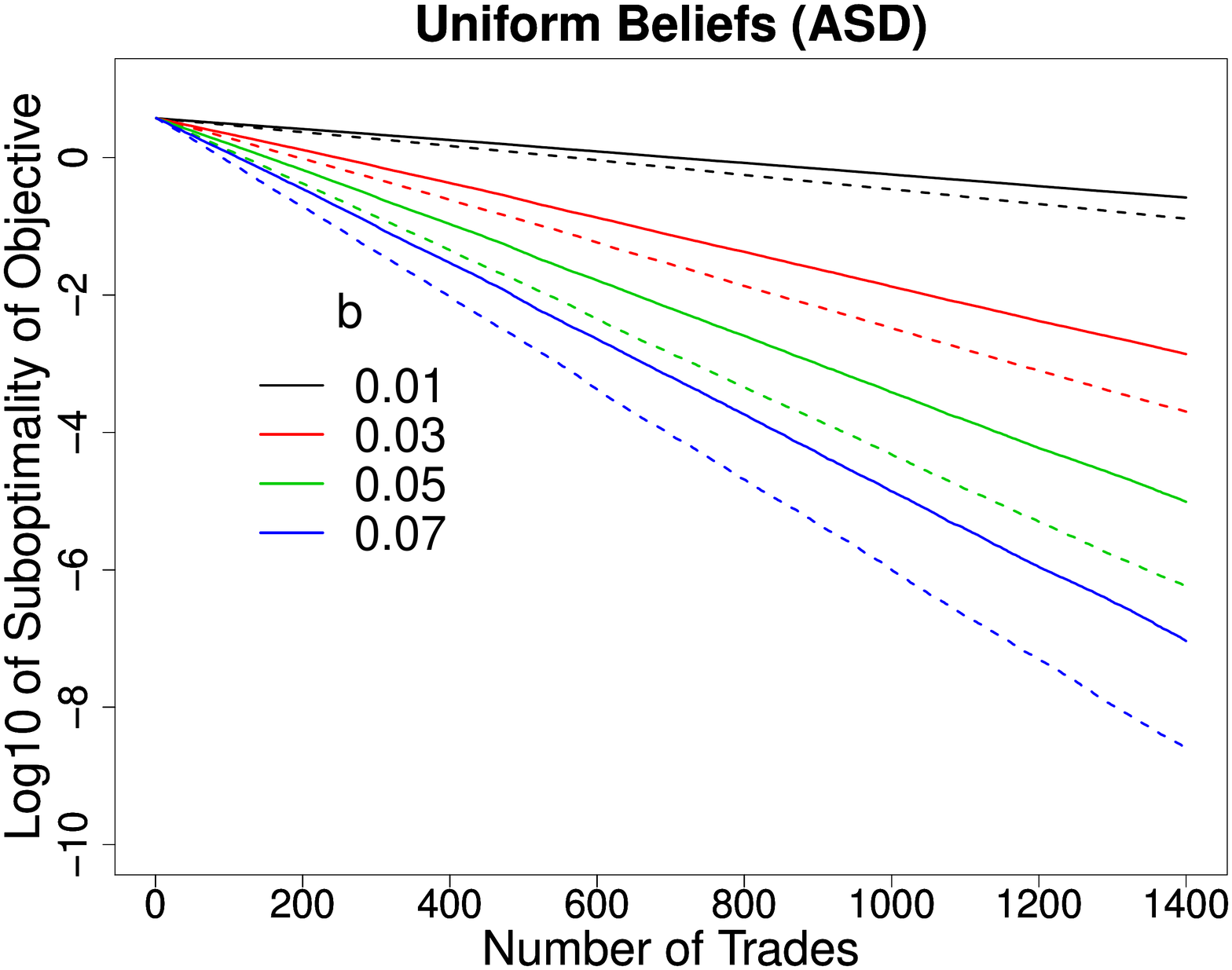}
  %\includegraphics[width=.8\linewidth]{Figures/BI_Semi_Plot0.pdf}
%  \caption{Uniform Beliefs}
%  \label{fig:sub1}
\end{subfigure}%
\hfill
\begin{subfigure}{.5\textwidth}
  \centering
  \includegraphics[width=.8\linewidth]{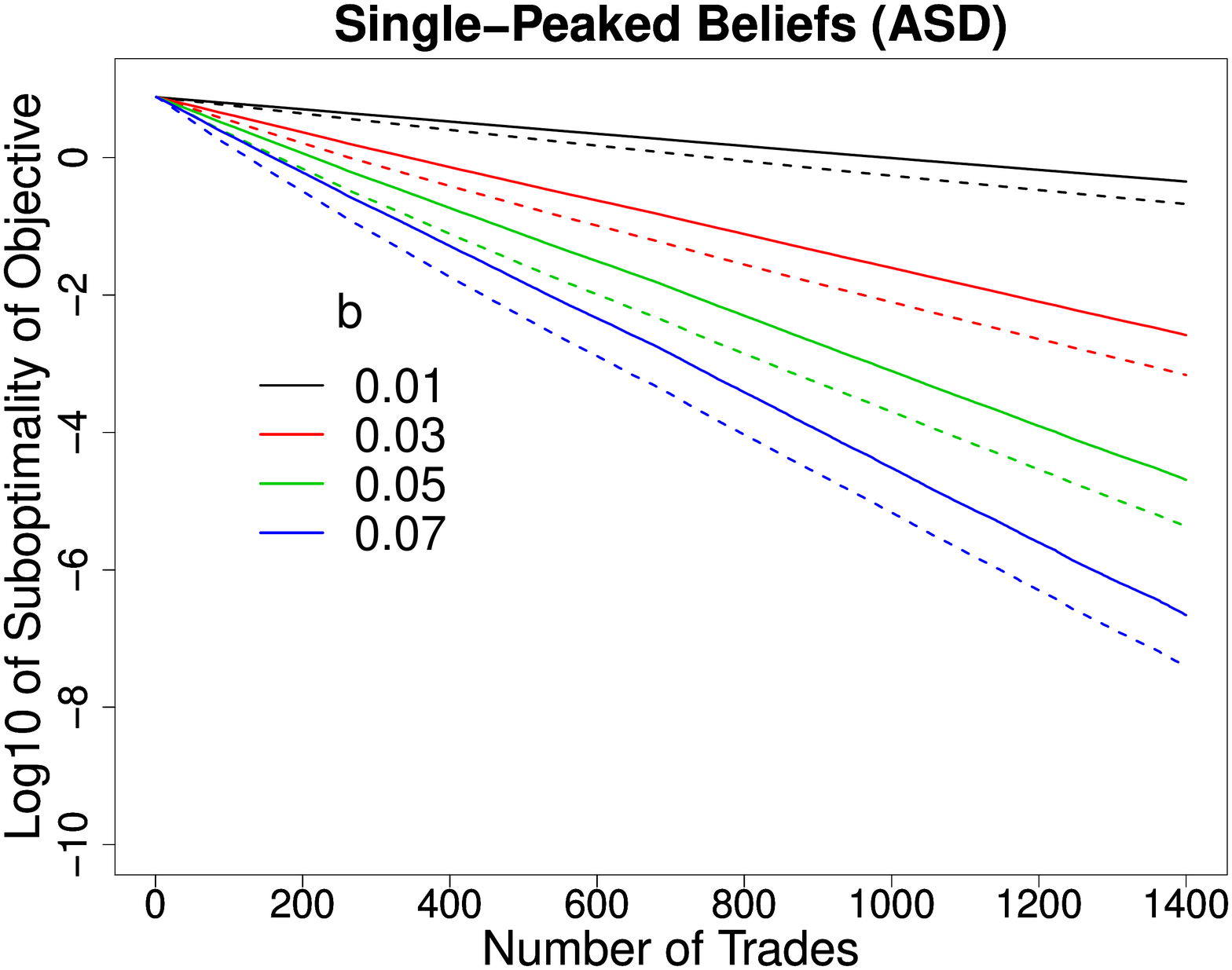}
 % \includegraphics[width=.8\linewidth]{Figures/BI_Semi_Plot1.pdf}
%  \caption{Single Peaked Beliefs}
%  \label{fig:sub2}
\end{subfigure}
\caption{Convergence in the objective value for various trader beliefs, cost functions, and liquidity parameters under $\BCD$. Solid lines show the log error in objective for \LMSR, dotted lines for \IND.\label{fig:semi_logB}}
\end{figure}

\begin{figure}[t!]
\centering
\begin{subfigure}{.5\textwidth}
  \centering
  \includegraphics[width=.8\linewidth]{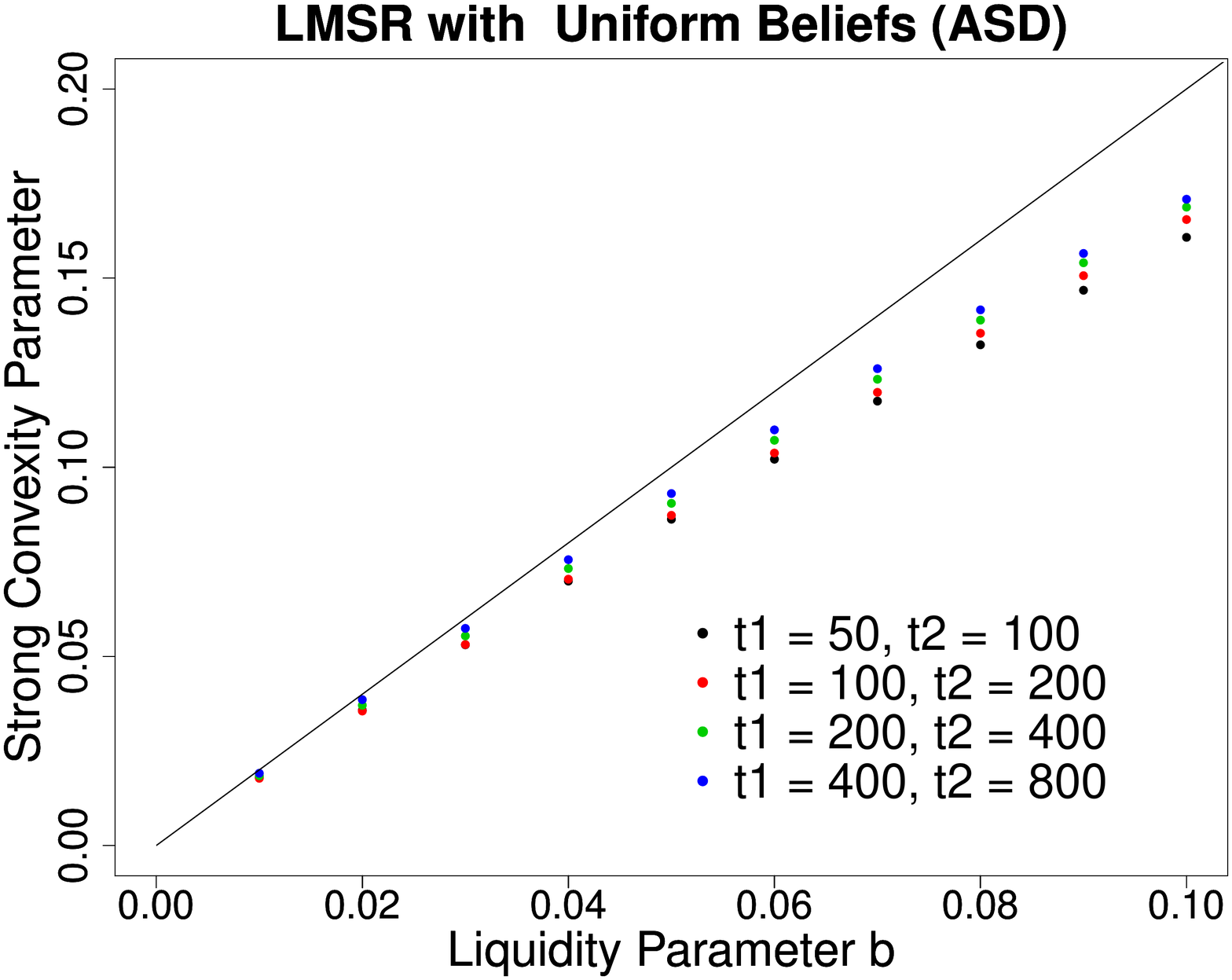}
  \includegraphics[width=.8\linewidth]{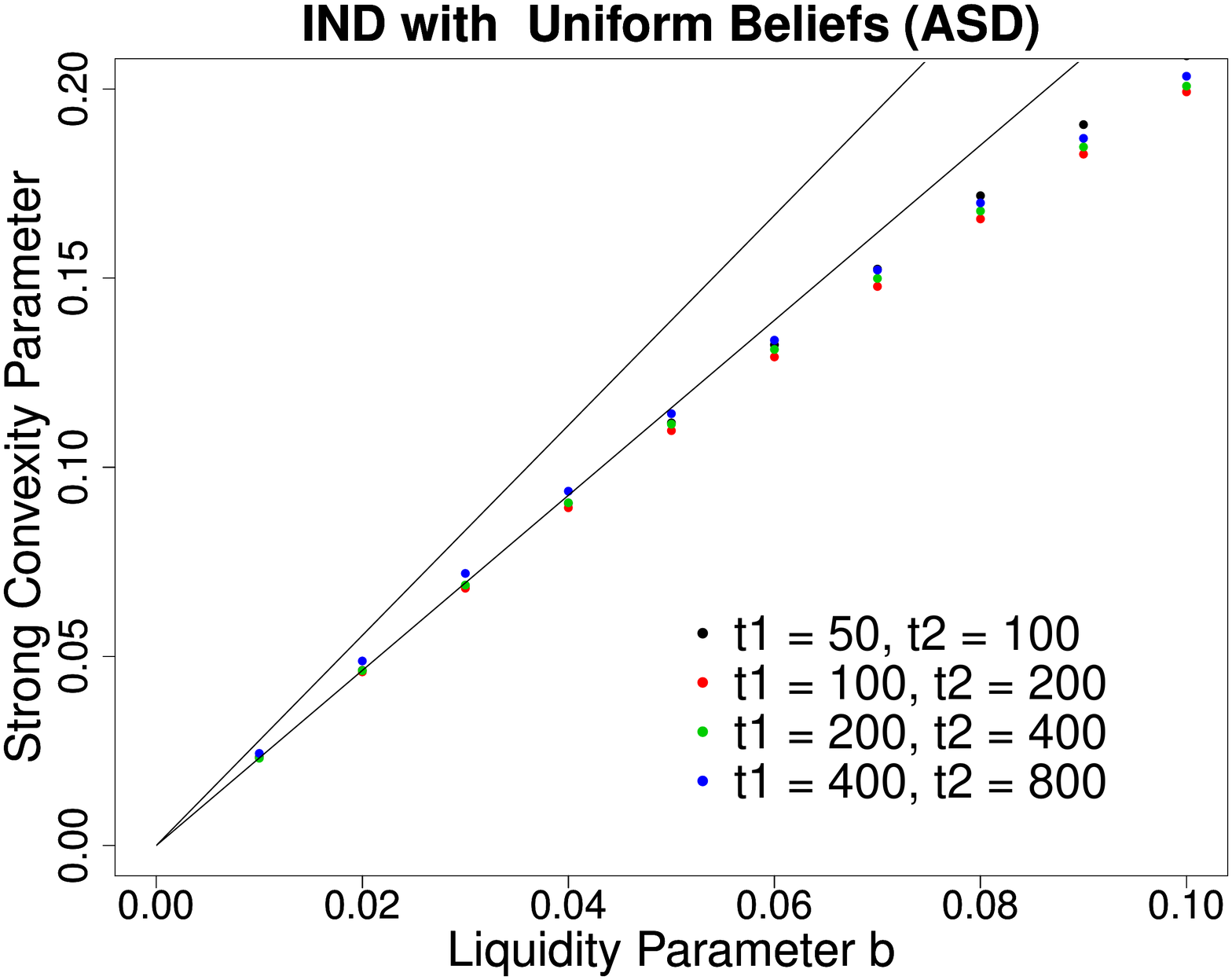}
%  \caption{Uniform Beliefs}
%  \label{fig:sub1}
\end{subfigure}%
\hfill
\begin{subfigure}{.5\textwidth}
  \centering
  \includegraphics[width=.8\linewidth]{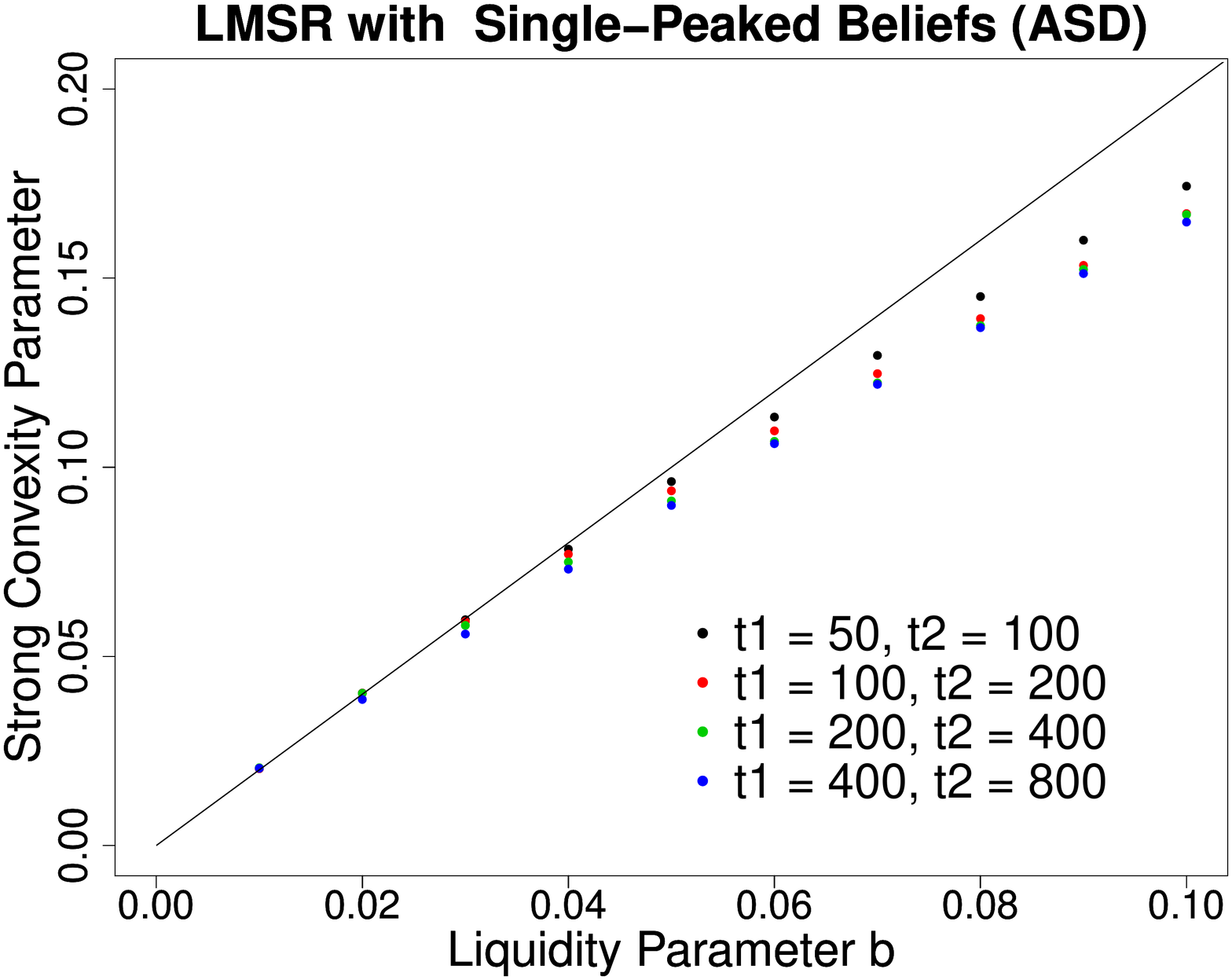}
  \includegraphics[width=.8\linewidth]{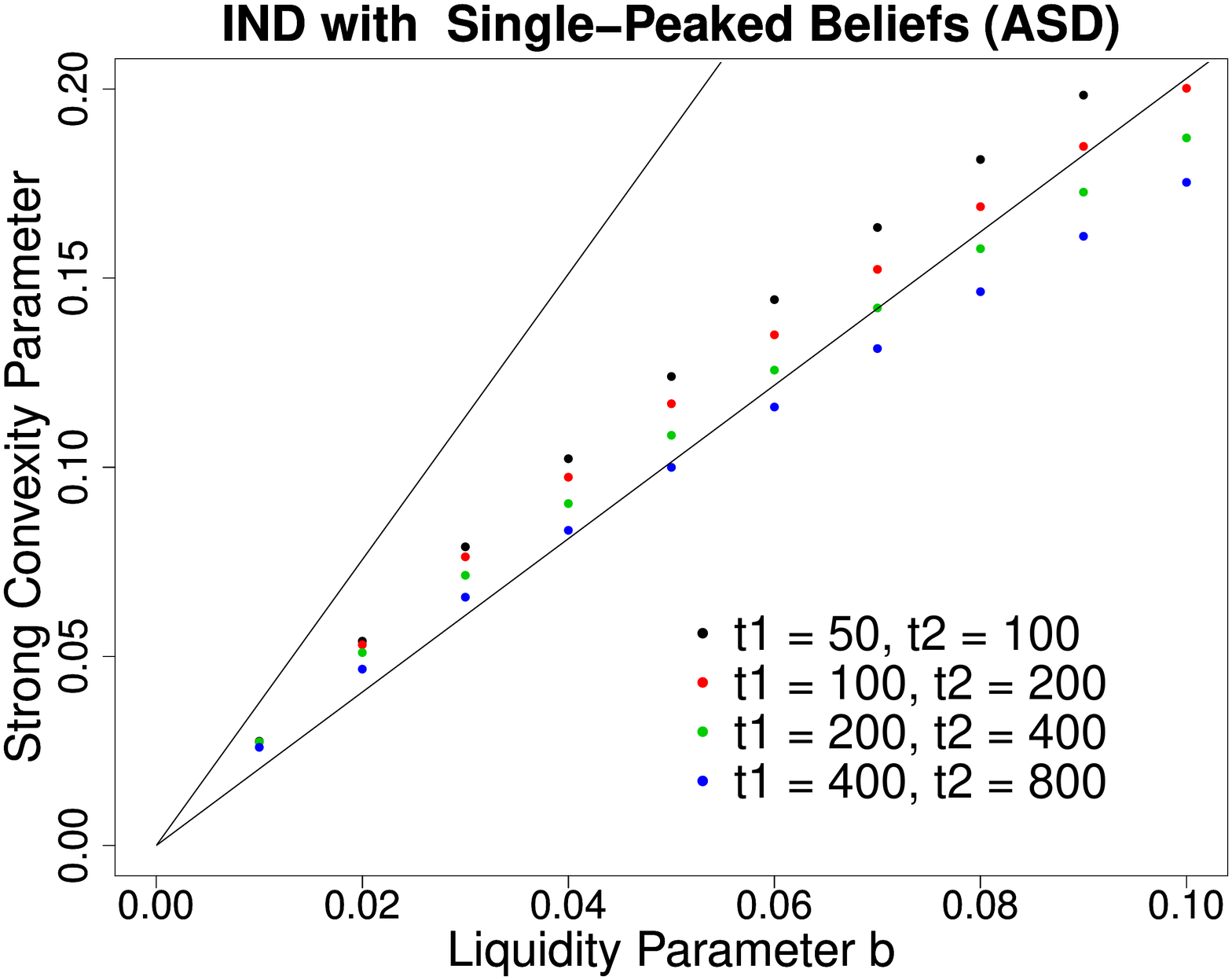}
  %\includegraphics[width=.8\linewidth]{Figures/BI_alpha_b_Plot1.pdf}
%  \caption{Single Peaked Beliefs}
%  \label{fig:sub2}
\end{subfigure}
\caption{The empirical strong convexity from \Eq{emp_sc} under $\BCD$ for various values of $t_1,t_2$ represented as dots, and asymptotic bounds for $\sigmalow$ and $\sigmahigh$ from \Thm{conv:summary} (ignoring $O(\liq^2)$ terms) represented as lines.  \label{fig:alpha_bB}}
\end{figure}

We can now check how well our theoretical lower and upper bounds on local strong convexity bound the empirical strong convexity $\hat{\sigma}$.  In \Fig{alpha_bB}, we plot $\hat{\sigma}$ as a function of $\liq$ using different values of $t_1$ and $t_2$ and compare it with the {asymptotic bounds} of $\sigmalow$ and $\sigmahigh$ computed as in \Thm{conv:summary}, dropping the terms that are $O(\liq^2)$. We would expect to see $\sigmalow \leq \hat{\sigma} \leq \sigmahigh$ as $\liq$ goes to 0, and indeed this is the case.  For \LMSR, the values of $\sigmahigh$ and $\sigmalow$ coincide, and the empirical values for $\hat{\sigma}$ agree for small $\liq$.

%%%%%%%%%
%%% SCD %%%
%%%%%%%%%
We now turn to our results for single-security dynamics (\SSD).  \Fig{semi_logS} shows the empirical value of $F(\bbundle^t)-F^\star$,
averaged over 20 random sequences of trade,
as a function of the number of trades, plotted on a log scale, for our two belief sets and cost functions under $\SCD$.
The plots show the convergence error for \LMSR and \IND right on top of each other, suggesting that the main asymptotic
term is driven by the diagonal of $H_C(\aggbprice)$, which is the same for both costs,
and which appears in the lower bound of \Thm{conv:summary} (with a multiplier that could be possibly improved).
%and the conjectured upper and lower bounds in Conjecture~\ref{conj:sigma}.
%
Similar to \ASD, we also evaluate the empirical strong convexity $\hat{\sigma}$.
In this case, we only have access to a lower bound (\Thm{conv:summary}), which our plots show to be a valid albeit a somewhat loose bound. All the bounds that we used
in the $\BCD$ and $\SCD$ strong convexity plots are summarized in \Cref{table:sigmas}.
%
%In this case, we only have access to a lower bound (\Thm{conv:summary}), which we believe is loose by a factor of two. In our plots in \Fig{alpha_bS}, we therefore also present our conjectured values for $\sigmalow$ and $\sigmahigh$ (from Conjecture~\ref{conj:sigma}) as dashed red lines.  \Cref{table:sigmas}
%summarizes the actual bounds we used in both $\BCD$ and $\SCD$ strong convexity plots.

\begin{figure}
\centering
\begin{subfigure}{.5\textwidth}
  \centering
  \includegraphics[width=.8\linewidth]{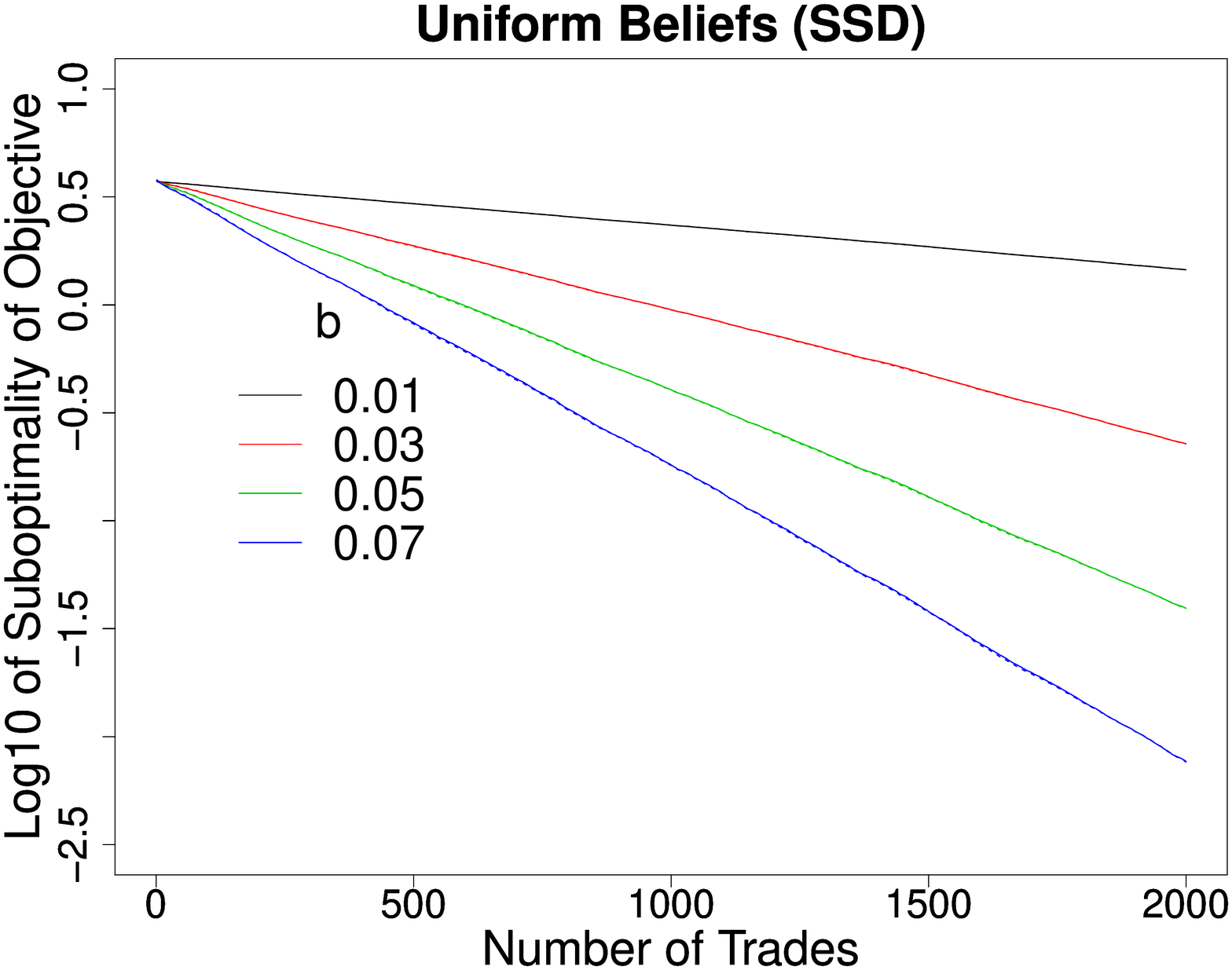}
  %\includegraphics[width=.8\linewidth]{Figures/SI_Semi_Plot0.pdf}
  %\caption{Uniform Beliefs}
  %\label{fig:sub1}
\end{subfigure}%
\hfill
\begin{subfigure}{.5\textwidth}
  \centering
  \includegraphics[width=.8\linewidth]{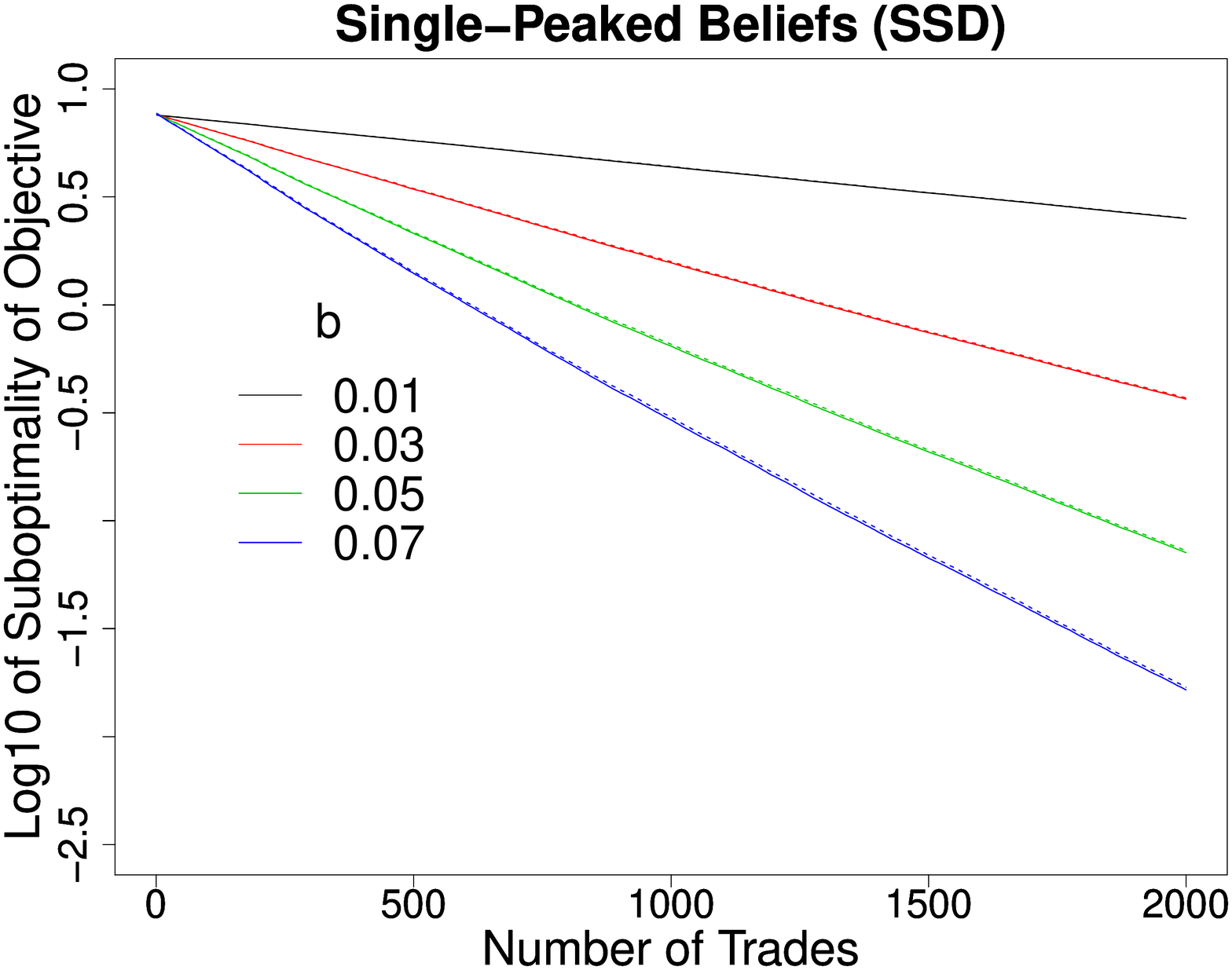}
  %\includegraphics[width=.8\linewidth]{Figures/SI_Semi_Plot1.pdf}
  %\caption{Single Peaked Beliefs}
  %\label{fig:sub2}
\end{subfigure}
\caption{Convergence in the objective value for various trader beliefs, cost functions, and liquidity parameters under $\SCD$. Solid lines show the log error in objective for \LMSR, dotted lines for \IND;
the \IND and \LMSR lines are right on top of each other.\label{fig:semi_logS}}
\end{figure}

\begin{figure}
\centering
\begin{subfigure}{.5\textwidth}
  \centering
  \includegraphics[width=.8\linewidth]{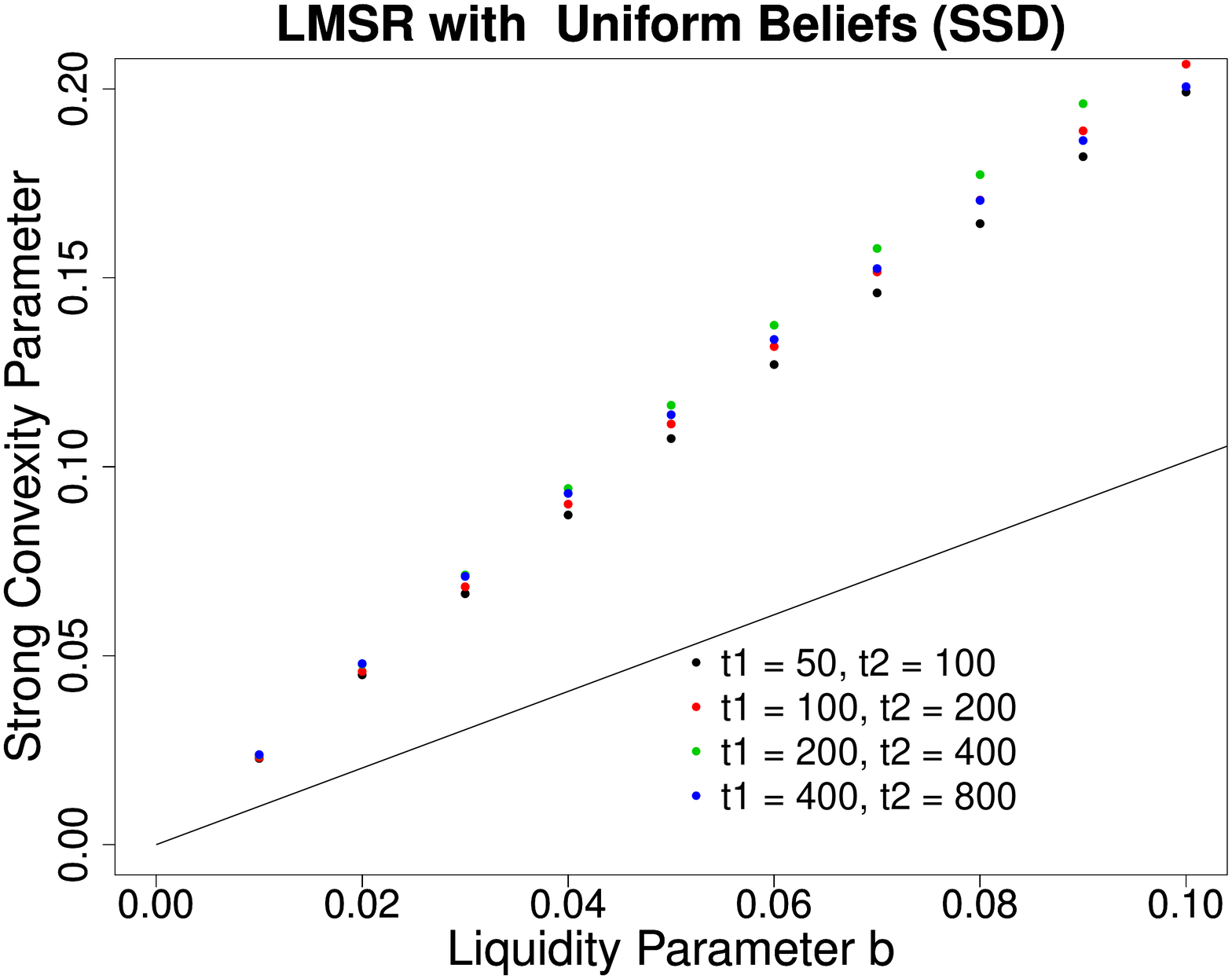}
   \includegraphics[width=.8\linewidth]{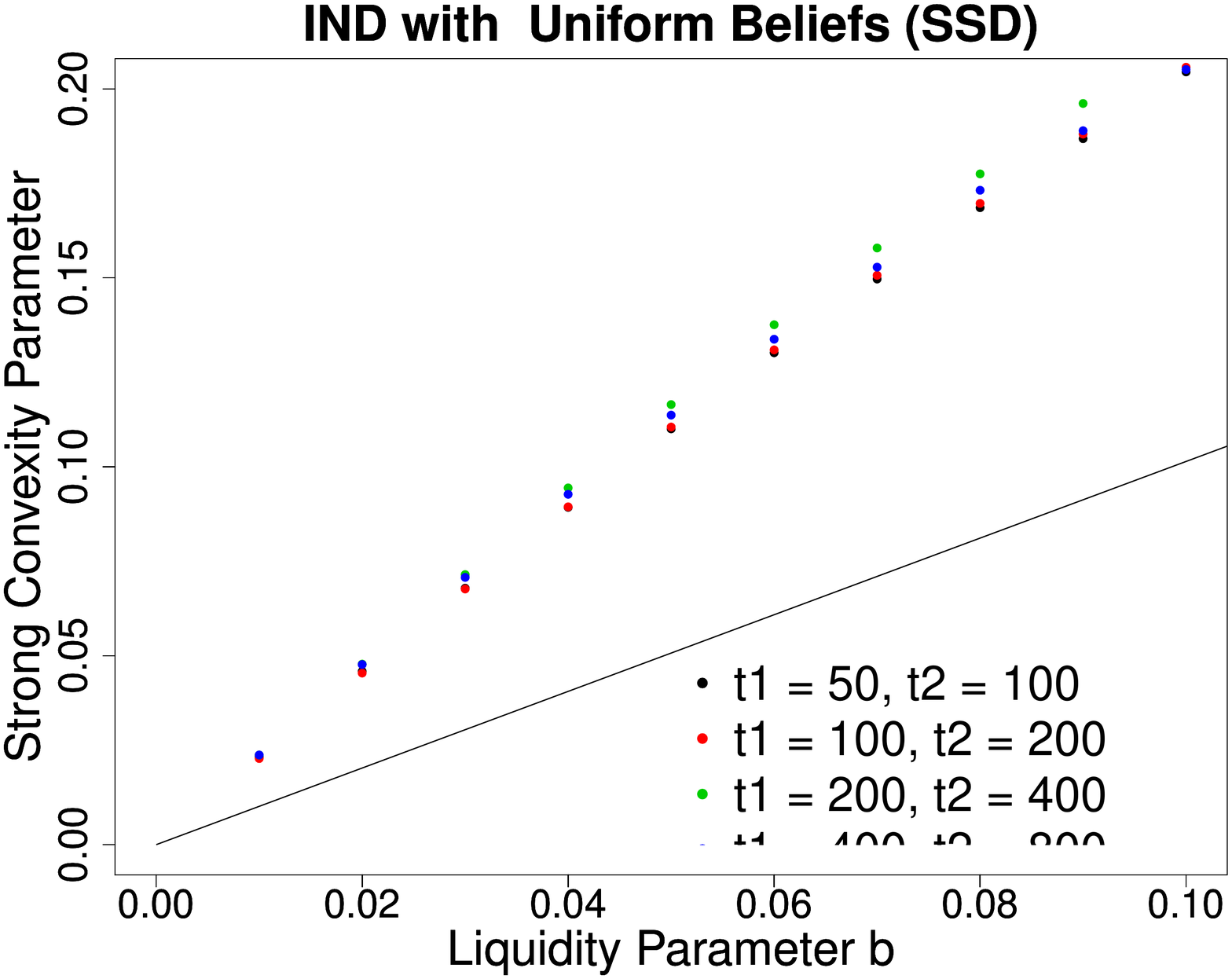}
  %\caption{}
  %\label{fig:sub1}
\end{subfigure}%
\hfill
\begin{subfigure}{.5\textwidth}
  \centering
  \includegraphics[width=.8\linewidth]{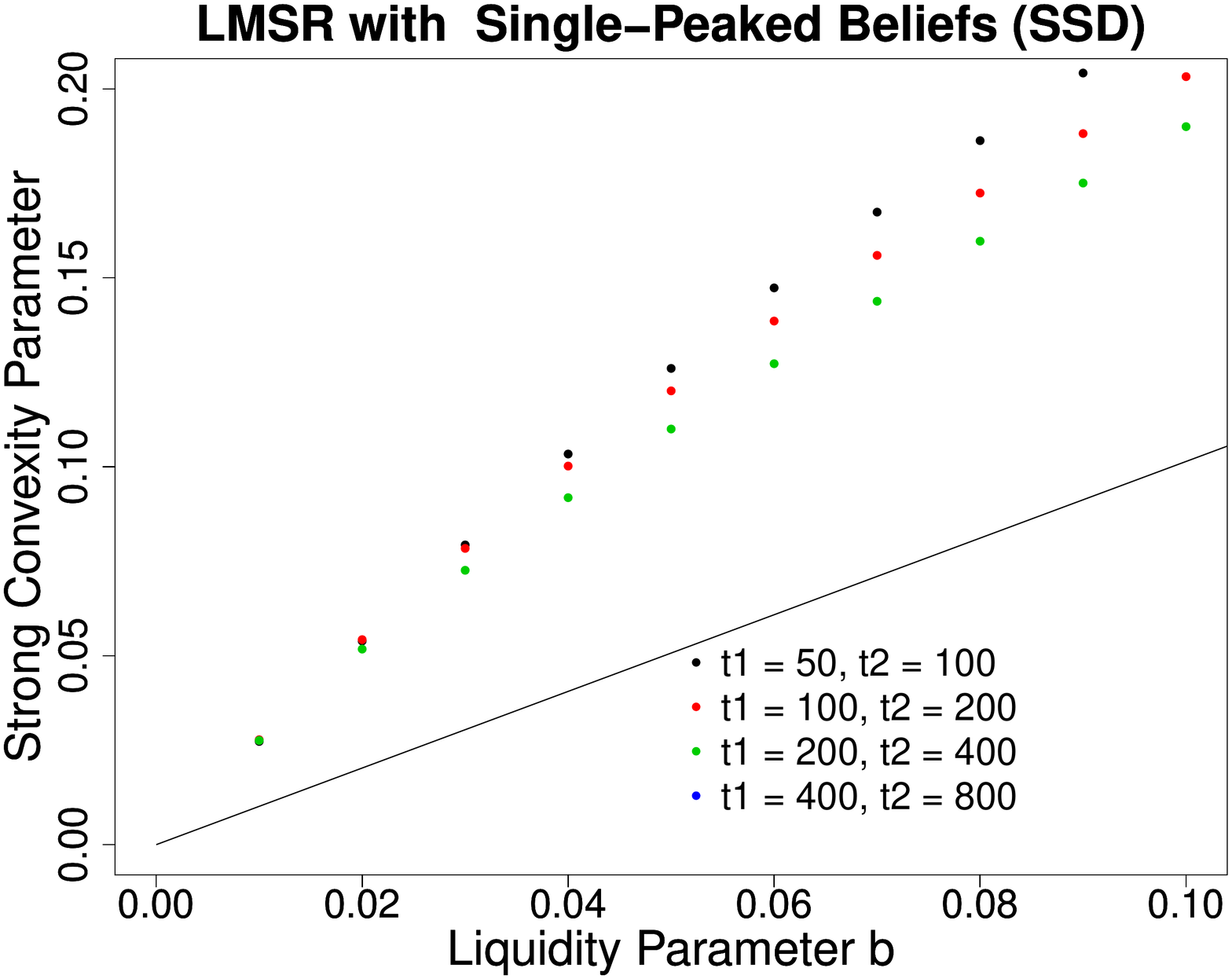}
  \includegraphics[width=.8\linewidth]{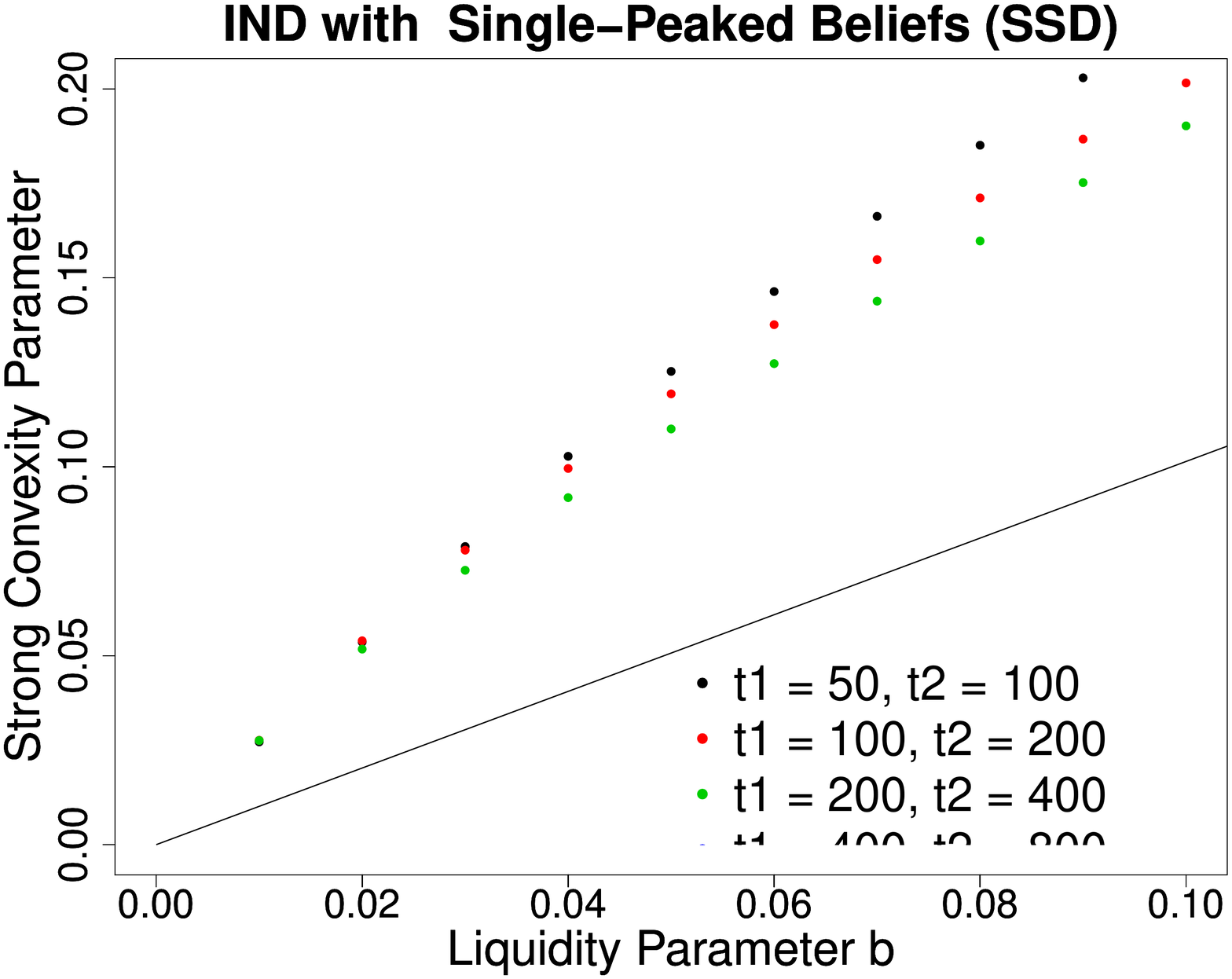}
 % \caption{}
  %\label{fig:sub2}
\end{subfigure}
\caption{The empirical strong convexity from \Eq{emp_sc} under $\SSD$ for various values of $t_1,t_2$ represented as dots, and asymptotic bound for $\sigmalow$ from \Thm{conv:summary} (ignoring the $O(\liq^2)$ term) represented as the black solid line.}
%The red dashed lines are conjectured values for $\sigmalow$ and $\sigmahigh$ from Conjecture~\ref{conj:sigma}.
\label{fig:alpha_bS}
\end{figure}

\begin{table}
    \caption{The bounds $\sigmalow$ and $\sigmahigh$ in the various cases we consider in our experiments, computed using \Thm{conv:summary}.\\
%the conjectured bounds for \SSD come from Conjecture~\ref{conj:sigma}.\\
    \label{table:sigmas}}
    \centering
    \begin{tabular}{ | c | c | c | c | c   |}
    \hline
    {\bf Beliefs} &{\bf Dynamics } & $C$ & $\sigmalow$ & $\sigmahigh$  \\
    \hline
    \multirow{4}{*}{Uniform Beliefs} & \multirow{2}{*}{$\BCD$} & $\LMSR$ & $2\liq$ & $2\liq$ \\
    							\cline{3-5}
    							& & $\IND$  & $2.31\liq$ & $2.78\liq$ \\
							\cline{2-5}
							& \multirow{2}{*}{$\SCD$} & $\LMSR$ & $1.01\liq$ & ---  \\
							\cline{3-5}
    							& & $\IND$  & $1.01\liq$ & ---  \\
    \hline
    \multirow{4}{*}{Single-Peaked Beliefs} & \multirow{2}{*}{$\BCD$} & $\LMSR$ & $2\liq$ & $2\liq$\\
    							\cline{3-5}
    							& & $\IND$  & $2.03\liq$ & $3.78\liq$ \\
							\cline{2-5}
							& \multirow{2}{*}{$\SCD$} & $\LMSR$ & $1.01 \liq$& --- \\
							\cline{3-5}
    							& & $\IND$  & $1.01\liq$ & ---  \\
    \hline
    \end{tabular}
    \end{table}

\ignore{
We have illustrated that our local characterizations of both bias and convergence error are tight, which justifies the use of the local analysis. Our theoretical framework therefore provides a meaningful way to compare the detailed quantitative error tradeoffs inherent in different choices of cost functions and liquidity levels.}

%%% Local Variables:
%%% mode: latex
%%% TeX-master: "main"
%%% End:

\else
\section{Additional Results from Numerical Experiments}

\subsection{Tightness of the Empirical Strong Convexity for Single-Peaked Beliefs}
\label{sec:asdAPP}
Recall that in Figure~\ref{fig:alpha_bB}, we plotted the empirical strong convexity $\hat{\sigma}$ as a function of $\liq$ for all-securities dynamics and uniform beliefs. We then also present the same plot for $\hat\sigma$ but for single-peaked beliefs in \Cref{fig:alpha_bB1}.  

\begin{figure}[t!]
\centering
\begin{subfigure}{.5\textwidth}
  \centering
  \includegraphics[width=.8\linewidth]{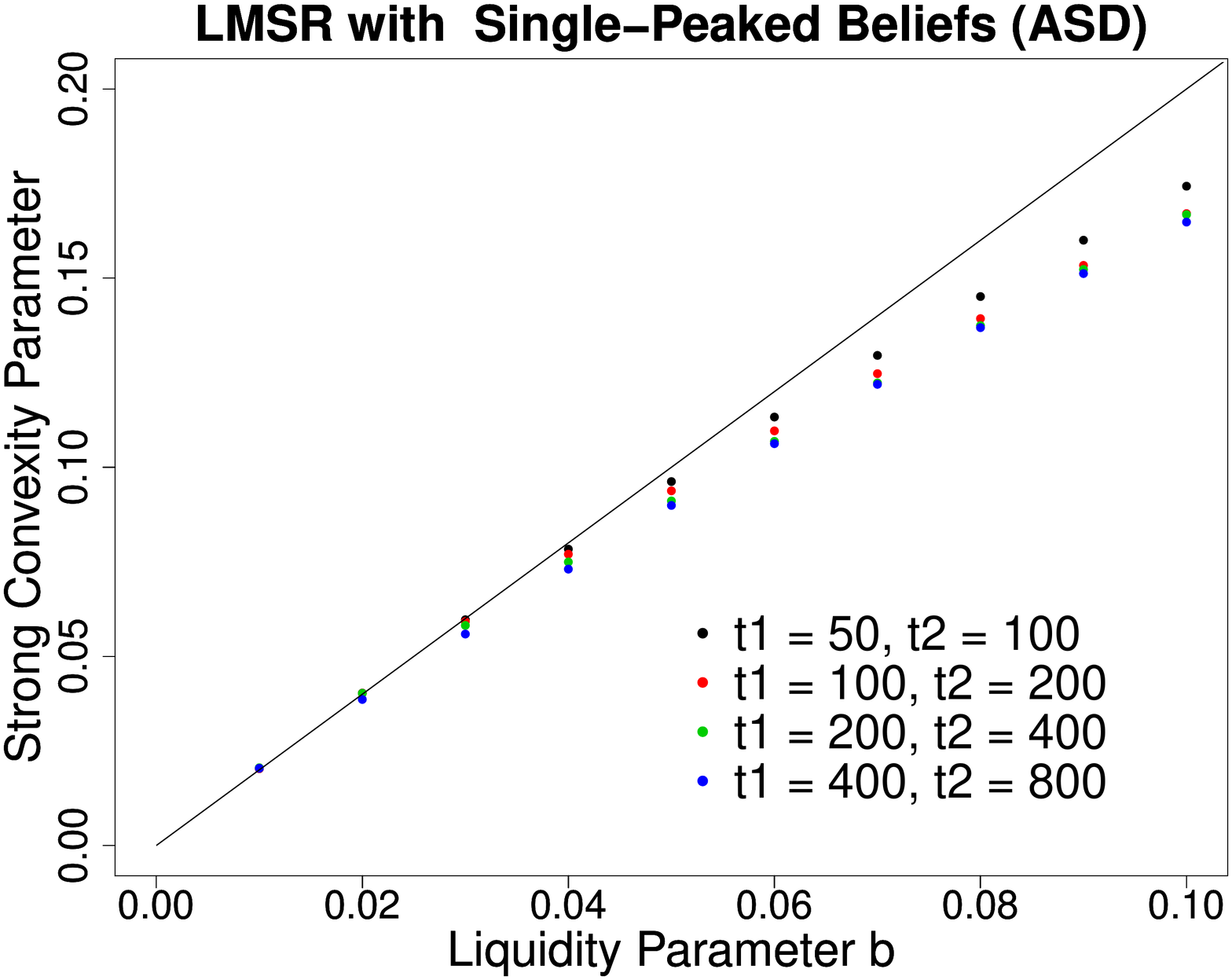}
  %\includegraphics[width=.8\linewidth]{Figures/BI_alpha_b_Plot0.pdf}
%  \caption{Uniform Beliefs}
%  \label{fig:sub1}
\end{subfigure}%
\hfill
\begin{subfigure}{.5\textwidth}
  \centering
  \includegraphics[width=.8\linewidth]{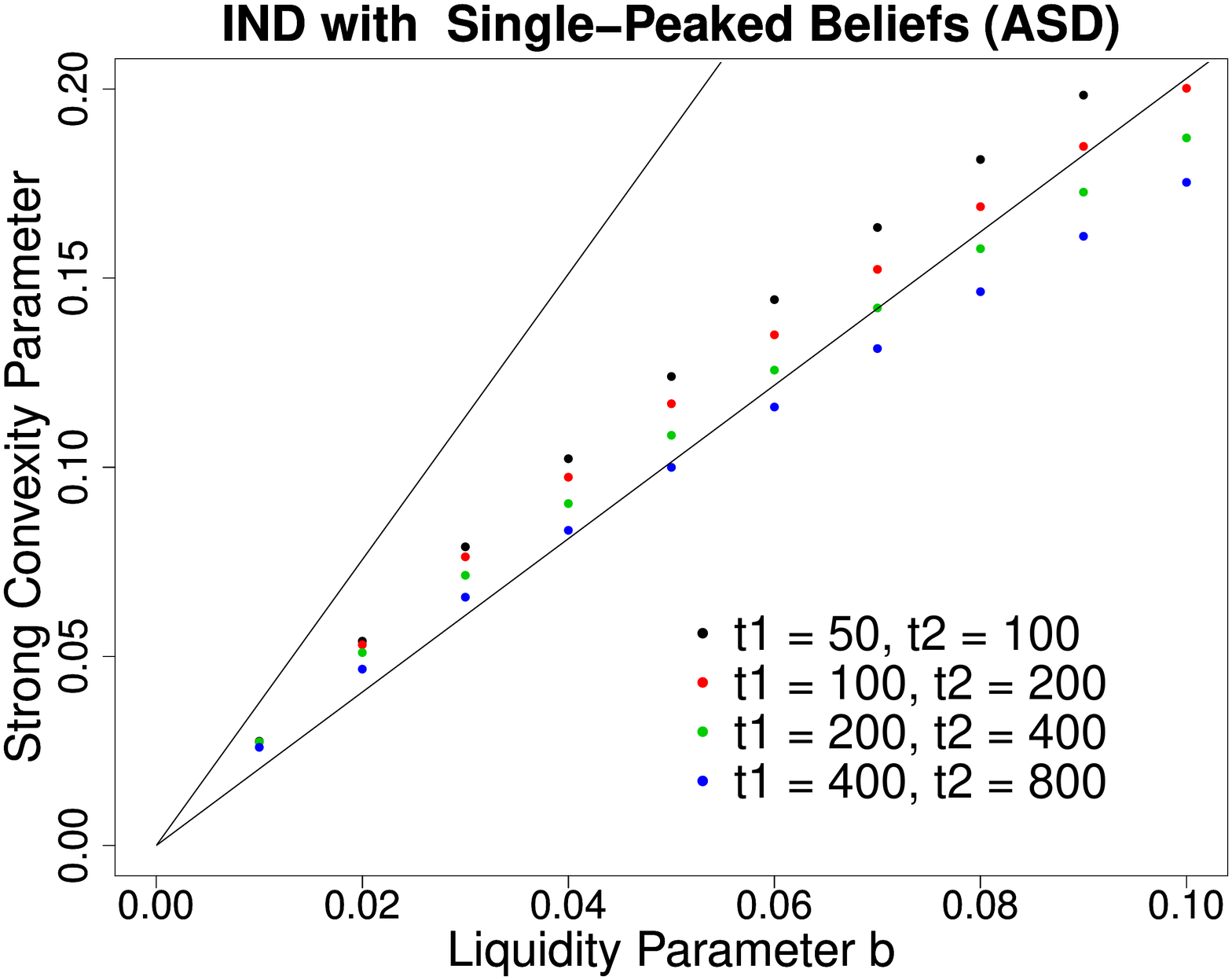}
  %\includegraphics[width=.8\linewidth]{Figures/BI_alpha_b_Plot1.pdf}
%  \caption{Single Peaked Beliefs}
%  \label{fig:sub2}
\end{subfigure}
\caption{ The empirical strong convexity from \eqref{eq:emp_sc} under \ASD with single-peaked beliefs for various values of $t_1,t_2$ represented as dots, asymptotic bounds for $\sigmalow$ and $\sigmahigh$ from \Cref{thm:conv:short} (ignoring $O(\liq^2)$ terms) represented as lines.  \label{fig:alpha_bB1}}
\end{figure}
\input{app-experiments}
\fi
%\fi

\end{document}